\documentclass[11pt]{article}
\usepackage{authblk}
\usepackage{natbib}
\usepackage{microtype}
\usepackage{graphicx}
\usepackage{subfigure}
\usepackage{booktabs}
\usepackage{comment}
\usepackage{algorithm2e}

\usepackage{hyperref}

\usepackage{amsmath}
\usepackage{amsthm}
\usepackage{amssymb}
\usepackage{amsfonts}
\usepackage{mathtools}
\usepackage{graphicx}
\usepackage{hyperref}
\usepackage{verbatim}

\usepackage{bbm}
\usepackage{color}
\usepackage{xspace}

\usepackage{multirow}

\usepackage[capitalize,noabbrev]{cleveref}

\theoremstyle{plain}
\newtheorem{theorem}{Theorem}[section]

\newtheorem{corollary}[theorem]{Corollary}
\theoremstyle{definition}
\newtheorem{definition}[theorem]{Definition}

\theoremstyle{remark}

\usepackage[textsize=tiny]{todonotes}

\def\R{\mathbb{R}}

\newcommand{\EE}[1]{\mathbb{E}\left[{#1}\right]} 

\def\R{\mathbb{R}}
\def\Z{\mathbf{Z}}
\def\W{\mathbf{W}}

\newcommand{\ident}{\mathbf{I}}

\newcommand\eqd{\stackrel{\mathclap{\normalfont\mbox{d}}}{=}}

\newcommand{\X}{\mathbf{X}}
\newcommand{\Y}{\mathbf{Y}}
\newcommand{\B}{\mathbf{B}}
\newcommand{\U}{\mathbf{U}}
\newcommand{\C}{\mathbf{C}}
\newcommand{\D}{\mathbf{D}}
\newcommand{\0}{\mathbf{0}}
\newcommand{\K}{\mathbf{K}}
\newcommand{\J}{\mathbf{J}}

\newcommand{\fdr}{\textnormal{FDR}}

\newcommand{\gfdr}{\fdr_{\textnormal{group}}}

\newcommand{\indep}{\perp \!\!\! \perp}
\newcommand{\iidsim}{\stackrel{\mathrm{iid}}{\sim}}

\title{Controlling FDR in selecting group-level simultaneous signals from multiple data sources with application to the National Covid Collaborative Cohort data}

\author[1]{Runqiu Wang}
\author[1]{Ran Dai\thanks{ran.dai@unmc.edu}}
\author[1]{Hongying Dai}
\author[2]{Evan French}
\author[1]{Cheng Zheng\thanks{cheng.zheng@unmc.edu}}
\author[ ]{on behalf of the N3C consortium}

\affil[1]{Department of Biostatistics, University of Nebraska Medical Center, Omaha, Nebraska, U.S.A.}

\affil[2]{Wright Center for Clinical and Translational Research, Virginia Commonwealth University, Richmond, Virginia, U.S.A.}

\begin{document}

\maketitle

\begin{abstract}

One challenge in exploratory association studies using observational data is that the associations between the predictors and the outcome are potentially weak and rare, and the candidate predictors have complex correlation structures. False discovery rate (FDR) controlling procedures can provide important statistical guarantees for replicability in predictor identification in exploratory research. In the recently established National COVID Collaborative Cohort (N3C), electronic health record (EHR) data on the same set of candidate predictors are independently collected in multiple different sites, offering opportunities to identify true associations by combining information from different sources. This paper presents a general knockoff-based variable selection algorithm to identify associations from unions of group-level conditional independence tests (simultaneous signals) with exact FDR control guarantees under finite sample settings. This algorithm can work with general regression settings, allowing heterogeneity of both the predictors and the outcomes across multiple data sources. We demonstrate the performance of this method with extensive numerical studies and an application to the N3C data.
\end{abstract}

\textbf{Keywords}: COVID-19, FDR, Multiple testing, Replicability, Variable selection, long COVID

\section{Introduction}
\label{sec:intro}

 With recent advances in biomedical research, data on the same set of candidate predictors are often collected independently from multiple sources, and there is a challenge in making reliable discoveries from such data jointly. {\color{black} For example, the electronic health record (EHR) data contains comprehensive information on patients' demographics, comorbidity, and medical history, providing great opportunities for observational studies. However, for multiple reasons (privacy protection, data storage capacity, data heterogeneity), EHR data from different sites are difficult to combine, bringing a challenge in effectively analyzing EHR data from multiple sites collectively. In addition, clinical concepts in the EHR data are often stored as groups of variables with complex dependence structures and data types.} In this paper, we introduce a knockoff-based framework to identify mutual signals from multiple independent studies and provide group-level variable selection accuracy guarantees under mild design and model assumptions.
 
 \subsection{The National COVID Collaborative Cohort example} 

Our methodology was motivated by a National COVID Collaborative Cohort (N3C) study. The N3C offers one of the largest collections of secure and de-identified clinical data in the United States for COVID-19 research \citep{Haendel-N3C}. Up to December 15, 2023, N3C has EHR information on over 17 million patients from 83 data-contributing sites (DCSs), with over 6 million confirmed COVID patients. With the accumulated COVID cohort data over time, long-term effects from SARS-CoV-2 infection have been identified and brought to attention. Some COVID-19 survivors present with persistent neurological, respiratory, or cardiovascular symptoms after the acute phase of the infection (Post-Acute Sequelae of COVID or ``long COVID”), regardless of the initial disease severity, vaccination status, and demographic and comorbidity status \citep{montani2022}. The identification of risk factors for long COVID has become an important question. Long COVID is not an illness that can be easily and universally defined. So far, two different long COVID indicator variables are recorded in the N3C database \citep{Pfaff2022}. The ``long COVID U9.09 diagnosis indicator" is a clinical diagnosis based on the International Classification of Diseases, Tenth Revision (ICD-10) codes; and the ``long COVID clinic visit indicator" indicates the patients' clinical visits for long COVID-related symptoms. These two long COVID definitions are highly related and both indicate patients' long COVID status. So far, the long COVID diagnosis information is only available from several DCSs, in the form of either ``long COVID U9.09 diagnosis indicator" or ``long COVID clinic visit indicator". 

To fully exploit the N3C long COVID data from multiple DCSs, we propose to develop a method to select mutual predictors (at the group level) from multiple data sources to identify reproducible risk factors for long COVID. For this purpose, we do not want to simply pool the data from different DCSs together. On one hand, the long COVID response variables have different definitions across different DCSs. {\color{black} On the other hand, data for clinical concepts are potentially recorded in variables with multiple data types; not all of them are available in every DCS, as the DCSs are heterogeneous and with different quality. Therefore, in the N3C EHR data analyses, there is a need for grouping variables with different data types, and allowing the group contents in different DCSs to be different.
For example, in N3C, there are multiple variables related to obesity. Some of them are continuous variables (BMI, body fat), and others are categorical (a four-level ordinal variable or a binary indicator for obesity). For diabetes, there are two related continuous variables (glucose and glycated hemoglobin) and two related categorical variables (complicated diabetes and uncomplicated diabetes). For high blood pressure, apart from a binary indicator variable, there are also longitudinal systolic and diastolic blood pressure measurements. The availability of these variables varies across the DCSs. Therefore, for the group indicating comorbidity (i.e. obesity), continuous variables (BMI, body fat) and categorical variables (obesity level) need to be considered as a group.} The N3C data with multiple independent DCSs provides us opportunities for reliably identifying signals, whereas the heterogeneous nature and the grouping structure of the data require strategic methodology planning. {\color{black}The increasing number of DCSs and sample sizes requires computational and communication efficiency, as well as data analysis capability to work in an online fashion, where analyses can be efficiently updated when data from more contributing sites become available.} It is also desirable to make reproducible discoveries. Novel false discovery rate (FDR) controlling methods are needed for these new data challenges. Our proposed method will be used to identify mutual risk factor signals for long COVID from patients' demographic and comorbidity information with an FDR control guarantee.

\subsection{Selecting group-level simultaneous signals}
 To formulate the group-level simultaneous signal identification problem mathematically, for a number $N \in \mathbb{N}$, denote $[N] = \{1, \cdots, N\}$. We are interested in $M$ domains of variables as candidate risk factors for an outcome. Our data is from $K$ independent datasets (the DCSs in the N3C example) $(\Y^1,\X^1)$, $\cdots$, $(\Y^K,\X^K)$, where {\color{black}$\Y^{k}\in \Omega_Y^{k}\subseteq \R^{n_k}$} (the long COVID response) and {\color{black} $\X^k\in \Omega_X^k\subseteq\R^{n_k\times p_k}$} (the candidate risk factors from the $M$ domains) for $k \in [K]$ {\color{black} with $\Omega_Y^{k}$ and $\Omega_X^k$ being the support of the distribution of $Y^{k}$ and $\X^k$ respectively}. Within the $k$-th dataset, there are $p_k$ variables (demographic, comorbidity, and medical record information), i.e. 
 \[(Y^k_{i},X^k_{i1},\cdots,X^k_{ip_k}) \iidsim \mathcal{D}_k, ~\text{for}~ i \in [n_k].\] {\color{black} for some arbitrary $p_k+1$ dimensional joint distribution $\mathcal{D}_k$s. For $k\in [K]$, we denote the $M$ domains as mutually exclusive groups of variables, with index set $G_{k1}, \cdots, G_{kM}$, where $G_{km}\subseteq [p_k]$ for all $k \in [K]$ and $m \in [M]$.}
 
 Across the $K$ experiments, both the outcome variables $\Y^k$s and the $\X_j^k$s for $k \in [K]$ and $j \in [p_k]$ can be of different data types, and $(\Y^k, \X_1^k,\cdots, \X_{p_k}^k)$ can have different distributions (heterogeneous). 
 For example, $\Y^k$s can be continuous or binary disease outcomes and $\X^k$s can be a mixture of continuous and categorical medical records from the EHR data. Furthermore, we do not assume $p_k$s to be identical across the $K$ datasets. For any $m \in [M]$, we also allow different group sizes ($|G_{km}|$s) across the K datasets. For example, in dataset $k$, $\X^k_{G_{km}}$ can be a group of dummy variables created for the categorical obesity level, and in dataset $l$, $\X^l_{G_{lm}}$ can be the continuously measured BMI.  
 
 Define the null hypothesis for the following test of group $m$ in dataset $k$ as $H_{0m}^k:= Y^k \indep \X^k_{G_{km}}|\X^k_{-G_{km}}$ where {\color{black}$\X^k_{-G_{km}}:= \X^k_{[p_k] \setminus G_{km}}$}, and the union null hypothesis $H_{0m} :=\cup_{k=1}^KH_{0m}^k$. Our goal is to control the FDR for the $M$ tests for the $H_{0m}$s. {\color{black}With} the group-level hypotheses, we define
\begin{equation}\label{eqn:H}
    \mathcal{S} = \{m\in [M]: H_{0m} ~\text{is false}\}~,  \text{and  $\mathcal{H} = \mathcal{S}^c = \{m\in [M]: H_{0m} ~\text{is true}\}$}.
\end{equation} 
We aim at developing a selection procedure returning a selection set of groups $\widehat{\mathcal{S}} \subseteq [M]$ {\color{black}to control the} group level FDR:  \begin{equation}\label{eqn:fdr}
\gfdr(\widehat{\mathcal{S}}) = \EE{\frac{|\widehat{\mathcal{S}} \cap \mathcal{H}|}{|\widehat{\mathcal{S}}|\vee 1}}.
\end{equation}

\subsection{Related prior work}
For risk factor identification using a single data source, knockoff-based methods have been developed for exact FDR control in selecting features with conditional associations with the response \citep{barber2015, candes2018}. {\color{black} The core concept of the knockoff-based method is to construct ``knockoff" copies of the covariates that retain their inner correlations. Unlike the original covariates, these knockoff copies are generated independently of the response variable. By incorporating these knockoff variables into the model, the method allows for the control of the False Discovery Proportion (FDP) during variable selection. Intuitively, if a variable represents a true signal, it is more likely to be selected over its knockoff copy; otherwise, the variable and its knockoff are equally likely to be selected. Thus, the FDP can be (conservatively) estimated by counting the number of knockoff variables included in the selected set.}

The original knockoff filter \citep{barber2015,barber2019} works on linear models assuming no knowledge of the design of covariates, the signal amplitude, or the noise level. It achieves exact FDR control under finite sample settings. For more general nonlinear models, \cite{candes2018} proposed the Model-X knockoff method, which allows the conditional distribution of the response to be arbitrary and completely unknown but requires some knowledge of the distribution of $\X$ \citep{huang2020}. Model-X knockoff method is also robust against errors in the estimation of the distribution of $\X$ \citep{barber2020}. There are also abundant publications on the construction of knockoffs with an approximated distribution of $\X$. \cite{romano2019} developed a Deep knockoff machine using deep generative models. \cite{liu2019} developed a Model-X generating method using deep latent variable models. \cite{bates2020} proposed an efficient general metropolized knockoff sampler. \cite{spector2020} proposed to construct knockoffs by minimizing the reconstructability of the features. Model-specific Knockoff methods have been proposed. \cite{dai2022kernel} proposed a kernel knockoff selection procedure for the nonparametric additive model. \cite{kormaksson2021} proposed the sequential knockoffs for continuous and categorical $\X$ variables. Knockoff-based methods have also been extended to test the null hypotheses at the group level. In this direction, group and multitask knockoff methods \cite{dai2016knockoff}, and prototype group knockoff methods \cite{chen2020} have been proposed. These group knockoff methods can also be used when there are categorical variables in $\X$ (see details in Section \ref{sec:groupknockoff}). Variants of knockoff methods have become useful tools in scientific research. For example, to identify the variations across the whole genome associated with a disease, \cite{sesia2019} developed a hidden Markov model knockoff method for FDR control in the genome-wide association study (GWAS). \cite{srinivasan2022} proposed a compositional knockoff filter for the analysis of microbiome compositional data.

{\color{black}
For simultaneous variable selection from multiple experiments, prior work focuses on individual-level variables without complex dependence structures. Methods based on the BH procedure \citep{heller2014, bogomolov2013, bogomolov2018}, local FDR-based methods \citep{chi2008, heller2014b} and a nonparametric method \citep{zhao2020} have been proposed. However, all these methods assume not only the independence of the experiments (in the N3C scenario, the DCSs), but also the independence or positive regression dependency (PRDS) property \citep{benjamini2001} of the p-values for the features within each experiment (in the N3C scenario, the demographic and comorbidity variables). This is not realistic for the patient information data in N3C. For example, age is correlated with many comorbidities such as heart disease and high blood pressure. More recently, \cite{dai2021multiple} proposed a simultaneous knockoff method for testing the union null hypotheses for feature selection at the individual level in $\X$, and \cite{li2021} proposed a multi-environment knockoff filter to find conditional associations that are consistent across environments. Both these knockoff methods can not be directly used when a group of $\X$ variables (for example, the group of diabetic variables HbA1c, glucose, complicated diabetes, and uncomplicated diabetes) needs to be selected together or categorical $\X$ variables with more than 2 categories (for example, dummy variables for the categorical COVID severity level) are present. 
}

\subsection{Our contribution}

In this paper, we propose a generalized simultaneous knockoff (\textit{GS knockoff}) framework, to establish exact FDR control in selecting mutual signals at the group level from multiple conditional independence tests, assuming very general conditional models. This extension is especially useful when using a general machine-learning model to select important groups of variables or categorical variables. The main contributions of this paper are summarized below:

\begin{itemize}
    \item We present a general knockoff-based algorithm for selecting simultaneous group-level features from multiple data sources; and show that it controls the exact group level FDR under mild conditions {\color{black}(see Sections \ref{sec:fixknockoff} and \ref{sec:method2} for details}) for $\X$ and $Y|\X$ under finite sample settings.
    \item We provide a collection of easy-to-implement group knockoff construction methods that are compatible with our framework, as well as powerful and simple-to-implement filter statistics.
    \item We demonstrate the FDR control property and the power of our method with extensive simulation settings. We also illustrate the application with the N3C long COVID research.
    \item Our method only requires the communication of the summary statistics from the individual datasets to identify the simultaneous signals, leading to the advantages of privacy-preserving, efficient distributed learning algorithms with the potential to work on online stream data {\color{black}(when a new DCS is added, instead of repeating the analysis using the expanded data, we can use the stored statistics from previous DCSs to efficiently update the results)}. 
\end{itemize}

\section{Group knockoff construction methods}\label{sec:groupknockoff}

In this section, we present a collection of methods for generating group knockoffs for an individual dataset. For notation simplicity, we omit the superscript $k$ in this section. We begin with some definitions.

\begin{definition}
(Swapping) For a set $S \subseteq [M]$, and for a vector $\mathbf{V} = (V_1, \cdots, V_{2M}) \in \R^{2M}$, $\mathbf{V}_{\textnormal{Swap}(S)}$ indicates the swapping of $V_j$ with $V_{j+M}$ for all $j \in S$.
\end{definition}
\begin{definition}
(Group Swapping) For a set $S \subseteq [M]$ and a group partition $G=\{G_1,\dots,G_M\}$ with $G_m\subseteq[p]$, and for a vector $\mathbf{V} = (V_1, \cdots, V_{2p}) \in \R^{2p}$, $\mathbf{V}_{\textnormal{GSwap}(S,G)}=\mathbf{V}_{\textnormal{Swap}(\cup_{m\in S}G_m)}$.
\end{definition}

\subsection{Group knockoff construction for Fixed-X knockoff approach}\label{sec:fixknockoff}

We first briefly review the group knockoff construction to work with the Fixed-X knockoff method. This group knockoff construction method has been proposed by {\color{black}Dai and Barber} \citep{dai2016knockoff}. 
{\color{black}The Fixed-X knockoff framework is predicated on a decentralized linear model structure \cite{barber2015}. This approach makes modest assumptions about the covariates and is tolerant of uncertainty in the magnitude of the regression coefficients, $\beta$. Furthermore, it is not contingent upon pre-established knowledge of the noise parameter, $\sigma^2$.}

\begin{definition} 
A Fixed-X group knockoff for a fixed design matrix $\X=(\X_1,\cdots,\X_p)$ with group partition $G=\{G_1,\cdots,G_M\}$ is a new design matrix $\widetilde{\X}=(\widetilde{\X}_1,\cdots,\widetilde{\X}_p)$ constructed with the following two properties:
\begin{enumerate}
\item $\widetilde{\X}^\top \widetilde{\X} = \Sigma\coloneqq \X^\top \X$
\item $\widetilde{\X}^\top \X = \Sigma - \B$, where $\B\succeq 0$ is group-block-diagonal meaning that $\B_{G_i,G_j} = 0$ for any two distinct groups $i\neq j$. 
\end{enumerate}
\end{definition}
Specifically, write $\B = \text{diag}\{\B_1,\dots,\B_M\}$ where $\B_m=\B_{G_m,G_m}$ and $\0\preceq \B \preceq 2\Sigma$. Here for $\J, \K \in \R^{m\times m}$, $\J \preceq \K$ if and only if $\K -\J$ is positive semidefinite. We can construct the fixed group knockoffs by setting \[\widetilde{\X} = \X(\ident_p -\Sigma^{-1}\B)+\widetilde{\U}\C\]
where $\widetilde{\U}$ is a $n\times p$ matrix orthogonal to the span of $\X$, while $\C^\top \C = 2\B - \B \Sigma^{-1}\B$ is a Cholesky decomposition. The condition $\0\preceq\B\preceq 2\Sigma$ guarantees the existence of such a Cholesky decomposition. We can select $\B$ using either the equivariant approach or the semidefinite programming (SDP) approach. For equivariant approach, we have $\B_m = b \cdot \Sigma_{G_m,G_m}$ where \[b = \min\left\{1, 2\lambda_{\min}\left(\D\Sigma \D\right)\right\}\] where $\lambda_{min}(\cdot)$ means the minimum eigen value and $\D = \text{diag}\{\Sigma_{G_1,G_1}^{-\frac{1}{2}},\dots,\Sigma_{G_M,G_M}^{-\frac{1}{2}}\}$. For SDP approach, we have $\B_m = b_m \cdot \Sigma_{G_m,G_m}$ and we can find $(b_1,\dots,b_M)$ that minimize $\sum_{m=1}^M (1-b_m)$ with the constraint $\B\preceq 2\Sigma$. In the non-group setting, it has been shown that the SDP approach can lead to a slight power increase.

We can also use an individual-level Fixed-X knockoff matrix which automatically satisfies the fixed group knockoff matrix requirement. However, the group-level condition is weaker and it allows more flexibility in constructing $\widetilde{\X}$. Such flexibility will enable more separation between a feature $\X_j$ and its knockoff $\widetilde{\X}_j$, which in turn can increase the power to detect true signals

The Fixed-X group knockoff method enjoys very relaxed assumptions on the covariates $\X$, the unknown regression coefficients, or the noise level, as $Y|\X$ follows a linear model (see an example in simulation E.4).  However, Fixed-X group knockoff can not work with the binary long COVID response in the N3C data example. We further extend the Model-X knockoff \cite{candes2018} to group knockoff construction.

\subsection{Group knockoff construction for Model-X knockoff approach} \label{sec:method2}

{\color{black}
The Model-X knockoff approach can work with arbitrary unknown dependence structure of $Y|\X$, {\color{black} assuming the knowledge of the distribution of $\X$ (or if the distribution of $\X$ can be well approximated)}\citep{candes2018} for variable selection at the individual level. However, in the N3C data, some clinical concepts are characterized as a group of variables with various data types (see the obesity, diabetes and high blood pressure examples in Section \ref{sec:intro}). When categorical variables exist in the data set, it is more natural to select all dummy variables for one categorical variable as a group. For the more complicated cases in EHR data, a group level factor can be defined as a mixture of continuous and categorical variables. In this section, we extend the group knockoff construction to work with the Model-X knockoff method settings.}

\begin{definition} \label{def:group_modelX}
A group Model-X knockoffs for the family of random variables $\X=(\X_1,\cdots,\X_p)$ with group partition $G=\{G_1,\cdots,G_M\}$ are a new family of random variables $\widetilde{\X}=(\widetilde{\X}_1,\cdots,\widetilde{\X}_p)$ constructed with the following two properties:
\begin{enumerate}
\item for any subset $S\subseteq[M]$, $(\X,\widetilde{\X})_{\text{GSwap}(S,G)}\eqd (\X,\widetilde{\X})$
\item $\widetilde{\X}\indep Y| \X$ if there is a response $Y$
\end{enumerate}
\end{definition}

\cite{candes2018} proposed a general algorithm to sample the model-X knockoff when each column is a single variable. We can extend it to allow each variable to be a multivariate random vector and thus the general algorithm to sample group knockoff can be given as below:\\
\begin{algorithm}
$m=1$; \\ 
\textbf{while} $m\leq M$, \textbf{do}
\begin{itemize}
    \item Sample $\widetilde{\X}_{G_m}$ from distribution $\mathcal{L}(\X_{G_m}|\X_{-G_m},\widetilde{\X}_{\cup_{j=1}^{m-1} G_j})$;
    \item $m=m+1$;
\end{itemize}
\textbf{end}
\caption{Model-X Group Knockoff construction}
\label{alg:groupModelX}
\end{algorithm}

The proof that this algorithm leads to knockoffs that satisfy the group Model-X knockoff properties (Definition \ref{def:group_modelX}) is given in Lemma 1 in Web Appendix A. Next, we present the sequential group knockoff construction algorithm to construct the Model-X group knockoffs.

\subsubsection{Sequential group knockoff construction}
We consider a general case where the group-wise $\X$ variables are composed of both continuous and categorical variables. 
For individual knockoff procedures with categorical $\X$ variables, a (individual) sequential knockoff construction has been proposed \citep{kormaksson2021}. We propose a sequential group knockoff construction algorithm that allows for both continuous and categorical variables to co-exist in one group in $\X$. Without loss of generality, we can assume that for each $\X_{G_m}$, it contains two components, the continuous component $\X_{G_m}^{con}$ and the categorical component $\X_{G_m}^{cat}$. {\color{black} We construct the knockoffs for the groups one by one, and in each group, we first generate knockoffs for the continuous components and then the categorical components.} We summarize the algorithm below:

\begin{algorithm}
                
$m=1$, 
\textbf{While} $m\leq M$, \textbf{do}
\begin{itemize}
    \item $\widetilde{\X}^{con}_{G_m}$ construction: {\color{black} If $\X_{G_m}^{con}$=$\emptyset$, ignore this step; Otherwise, }sample $\widetilde{\X}^{con}_{G_m} \sim \mathcal{N}(\widehat{\mu}_m,\widehat{\Sigma}_m)$ where $\widehat{\mu}_m$,$\widehat{\Sigma}_m$ are obtained by fitting a penalized multi-task linear regression of $\X^{con}_{G_m}$ on $[\X_{-G_m},\widetilde{\X}_{\cup_{j=1}^{m-1} G_j}]$.
    \item $\widetilde{\X}^{cat}_{G_m}$ construction: {\color{black} If $\X_{G_m}^{cat}$=$\emptyset$, ignore this step; Otherwise, }sample $\widetilde{\X}^{cat}_{G_m} \sim \text{Multinom}(\widehat{\pi})$, where $\widehat{\pi}$ are obtained by fitting a penalized multinomial logistic regression of $\X^{cat}_{G_m}$ on $[\X^{con}_{G_m},\X_{-G_m},\widetilde{\X}_{\cup_{j=1}^{m-1} G_j}]$ with predictions made on $[\widetilde{\X}^{con}_{G_m},\X_{-G_m},\widetilde{\X}_{\cup_{j=1}^{m-1} G_j}]$.
    \item $m=m+1$
\end{itemize}
\textbf{end}

\caption{Sequential Group Knockoff construction
        }
    \label{alg:seqknockoff}   
\end{algorithm}

In Lemma 2 in Web Appendix A, we show that when the model is correct, this satisfies the general Group Model-X Knockoff generation procedure (Algorithm \ref{alg:groupModelX}). For constructing $\widetilde{\X}^{con}_{G_m}$, the penalized multitask linear regression can be fitted using the method in Section 2.2 of \citep{dai2016knockoff}. For $\widetilde{\X}^{cat}_{G_m}$, the penalized multinomial regression is performed using the R package glmnet. More details can be found in Web Appendix D.

For the misspecified model cases, previous literature has shown the original Model-X knockoffs \citep{candes2018} and simultaneous knockoffs \citep{dai2021multiple} are robust to moderate model misspecifications. Theoretically, the misspecification problem has been further studied by Barber \emph{et al.}\citep{barber2020} and Huang and Janson \citep{huang2020}. Also, our simulation study in Section \ref{sec:simulation} reflects the scenario of creating Model-X knockoffs under approximated distribution and the numerical result shows robustness for FDR control.

\section{Generalized Simultaneous Knockoff Method}

{\color{black}
Intuitively, one might propose some naive methods to solve the mutual signal identification problem. For example, the \textit{intersection} strategy first selects signals specific to the individual datasets and then constructs the simultaneous signal set by taking the intersection of the signals selected from the multiple datasets. However, this method is not guaranteed to control the FDR \citep{katsevich2023filtering} (see an example in Figure \ref{fig:figure1}). Another strategy, the \textit{pooling} method aggregates data from the multiple datasets to construct a single dataset. This strategy has been used when data are homogeneous across the datasets \citep{kormaksson2021, sechidis2021biomarker}. However, datasets from multiple sites may face heterogeneous problems so the pooling might not be always meaningful; when the data types and dimensions of the (group level) variables from the multiple datasets are different, it is not possible to pool the datasets. Furthermore, the \textit{pooling} method also fails in controlling the FDR as defined in \eqref{eqn:fdr}.}

Our proposed \textit{GS knockoff} framework can work with general regression models as long as the settings for the individual datasets satisfy the Fixed-X or Model-X knockoff assumptions \citep{barber2015, candes2018}. Therefore it can work with a large spectrum of models, from linear regression models with very weak assumptions on $\X$, to machine learning models with some knowledge of the $\X$ distribution. For the group settings, we only assume that for all the $K$ datasets there are $M$ groups, but we do not require the group sizes to be the same across the datasets. Also, we do not require $\cup_{m=1}^M G_{km}=[M]$ so that we can adjust for confounding variables in the models. For example, in some study to on concurrent medications using EHR data, demographic information is always adjusted, but we are not interested in testing their associations with the outcome. 

\subsection{Preliminaries}
\begin{definition} \label{def:compatible}
A test statistics $[\Z, \widetilde{\Z}]$ is called group knockoff compatible with the group partition $G_1,\dots,G_M\subseteq[p]$ if it can be written as $[\Z, \widetilde{\Z}]=t([\X,\widetilde{\X}],Y])$ for some function $t(\cdot)$ such that for any $S\subseteq[M]$, $[\Z, \widetilde{\Z}]_{\text{Swap}(S)}=t([\X,\widetilde{\X}]_{\text{GSwap}(S,G)},Y])$.
\end{definition}
\begin{definition} \label{def:sufficiency}
A test statistics $[\Z, \widetilde{\Z}]$ satisfies the sufficiency requirement if it can be written as a function of $[\X,\widetilde{\X}]^\top [\X,\widetilde{\X}]$ and $[\X,\widetilde{\X}]^\top Y$.
\end{definition}
\begin{definition}\label{def:one-swap-flip}
(One swap flip sign function (OSFF)) A function $f: \R^{2MK}\rightarrow \R^{M}$ is called a one swap flip sign function (OSFF) if it satisfies that for all $k\in [K]$ and all $S \subseteq[M]$, \begin{equation*}
    f([\Z^1,\widetilde{\Z}^1],\cdots,[\Z^k,\widetilde{\Z}^k]_{\textnormal{Swap}(S)},\cdots,[\Z^K,\widetilde{\Z}^K])\\=f([\Z^1,\widetilde{\Z}^1],\cdots,[\Z^k,\widetilde{\Z}^k],\cdots,[\Z^K,\widetilde{\Z}^K])\odot \epsilon(S),
\end{equation*}
where $\Z^k,\widetilde{\Z}^k, {\color{black}\epsilon(S)} \in {\color{black}\R^M}$ for $k \in [K]$, {\color{black}$\epsilon(S)_j=-1$ for all $j \in S$, otherwise $\epsilon(S)_j=1$} and $\odot$ represents the Hadamard product (elementwise product).
\end{definition}

\subsection{Algorithm}\label{sec:algorithm}

The \textit{GS knockoff} procedure is described below:
\begin{itemize}
    
    \item \textit{Step 1: Group knockoff construction for the individual experiments.} Denote the knockoff matrices for $\X^1,\cdots,\X^K$ as $\widetilde{\X}^1, \cdots,\widetilde{\X}^K$. The $\widetilde{\X}^k$ matrices can be generated using the group knockoff construction methods as described in Section \ref{sec:groupknockoff}. When only individual features exist, methods for generating individual knockoffs \citep{barber2015, candes2018, romano2019, bates2020, spector2020} can also be used since satisfying individual knockoff requirements implies satisfying group knockoff requirements. However, using individual knockoff might cause the knockoff to be very similar to the original feature and thus has less power when the within-group variables are highly correlated. 
    
    \item \textit{Step 2: Test statistics calculation for the individual experiments.} For each experiment $k \in [K]$, choose and calculate statistics $[\Z^k, \widetilde{\Z}^k] \in \R^{2M}$ that are group knockoff compatible (Definition \ref{def:compatible}) with the group partition $G$ (and satisfy the sufficiency (Definition \ref{def:sufficiency}) requirement when fixed group knockoff construction is used). 
    For our analysis, we assume the true model is
    \begin{equation}
    g_k(\EE{Y_{i}^k})=\beta_0^k+\X^k_{i}\beta^{k},
    \end{equation}
    where $g_k(\cdot)$ is the link function for the $k$th experiment, $\X_i^k$ is the $i$th row of $\X^k$ and $\tilde{\X}_i^k$ is the $i$th row of $\tilde{\X}^k$. We fit the working model
    \begin{equation} g_k(\EE{Y_{i}^k})=\beta_0^k+\X_i^k \beta^k + \widetilde{\X}_i^k\widetilde{\beta}^{k},
    \end{equation}   
    by defining 
    \begin{eqnarray*}
    \left(\begin{array}{c}
         \widehat{\beta}^k_0\\ \widehat{\beta}^k(\lambda)  \\
        \widehat{\tilde{\beta}}^k(\lambda)
        \end{array}\right) &=& \arg\min_{(\beta_0^k,\beta^{k\top},\tilde{\beta}^{k\top})^\top} \sum_{i=1}^{n_k}\frac{(Y_i^k- {\color{black} g_k^{-1}(\beta_0^k+\X_i^k \beta^k+\widetilde{\X}_i^k\widetilde{\beta}^{k})})^2}{V^k_i}\\&&+\lambda \color{black}\sum_{m=1}^{M} \left(\sqrt{\sum_{j\in G_{km}}(\beta_{j}^k)^2}+\sqrt{\sum_{j\in G_{km}}(\tilde{\beta}_{j}^k)^2}\right),
    \end{eqnarray*}
    where $V^k_i=V^k(g_k^{-1}(\beta_0^k+\X_i^k \beta^k + \widetilde{\X}_i^k\widetilde{\beta}^{k}))$ and $V^k(\cdot)$ is the variance function specified for the generalized linear model (GLM) for $Y^k$. Then we define 
    \begin{equation}
    Z^k_m=\sup\{\lambda: \sum_{j\in G_{km}}\widehat{\beta}^k_j(\lambda)^2>0\}
    \end{equation}
   \begin{equation}
   \widetilde{Z}^k_m=\sup\{\lambda: \sum_{j\in G_{km}}\widehat{\tilde{\beta}}^k_{j}(\lambda)^2>0\}.
   \end{equation}
Denote $\Z^k = (Z_1^k,\cdots,Z_M^k)$ and $\widetilde{\Z}^k = (\widetilde{Z}_1^k,\cdots,\widetilde{Z}_M^k)$.
    
    \item \textit{Step 3: Calculation of the filter statistics {\color{black}$\W \in \R^M$.}} Choose an arbitrary OSFF $f$ as defined in Definition \ref{def:one-swap-flip} and calculate $\W=f([\Z^1,\widetilde{\Z}^1],\cdots,[\Z^K,\widetilde{\Z}^K])$. In this work, we use the difference function \citep{dai2021multiple}
    \begin{equation}
    \W=\odot_{k=1}^K [\Z^k-\widetilde{\Z}^k].
    \end{equation}
    Other choices of $\W$ construction can be found in Appendix A5 of \cite{dai2021multiple}. 
  
    \item \textit{Step 4: Threshold calculation and feature selection.} Using the filter statistics $\W$ from Step 3, we apply the knockoff+ filter \eqref{eqn:knockoff+} to obtain the selection set $\widehat{S}_+$ under the \textit{Generalized Simultaneous knockoff}+ procedure; or apply the knockoff filter \eqref{eqn:knockoff} to obtain $\widehat{S}$ under the \textit{Generalized Simultaneous knockoff} procedure.
    \begin{equation}\label{eqn:knockoff}
   \widehat{S} = \{j: W_j \geq \tau\},~\text{where }~\\ \tau = \min\left\{t \in \mathcal{W}_+: \frac{\#\{j:W_j \leq -t\}}{\#\{j:W_j\geq t\}\vee 1}\leq q\right\} .
\end{equation}
\begin{equation}\label{eqn:knockoff+}
   \widehat{S}_+ = \{j: W_j \geq \tau_+\}, ~\text{where }~ \\ \tau_+ = \min\left\{t \in \mathcal{W}_+: \frac{1+\#\{j:W_j \leq -t\}}{\#\{j:W_j\geq t\}\vee 1}\leq q\right\} .
\end{equation}
\end{itemize}
{\color{black}Here q is the target FDR level and $ \mathcal{W}_+= \{|W_j|:|W_j|>0\}$.}
 
\section{Main results}
\begin{theorem}\label{thm:1}
With the test statistics $[\Z^k,\widetilde{\Z}^k]$ for $k\in [K]$ satisfy the property that $[\Z^k,\widetilde{\Z}^k]\eqd [\Z^k,\widetilde{\Z}^k]_{Swap(S)}$ for any $S\in \mathcal{H}$ and $W=f([\Z^1,\widetilde{\Z}^1],\cdots,[\Z^K,\widetilde{\Z}^K])$ for an OSFF function $f$, the \textit{GS knockoff} procedure \eqref{eqn:knockoff} 
controls the modified group FDR defined as
\begin{equation}\label{eqn:mfdr}
\textnormal{m}\gfdr=\EE{\frac{|\widehat{S}\cap \mathcal{H}|}{|\widehat{S}|+1/q}}\leq q, \end{equation}
and the \textit{GS knockoff}+ procedure \eqref{eqn:knockoff+} controls the group FDR as defined in \eqref{eqn:fdr}.
\end{theorem}
The proof of Theorem \ref{thm:1} is in Web Appendix B.

\begin{corollary}\label{cor1} Under the specific choice of $[\Z^k,\widetilde{\Z}^k]$ in equations (5) and (6), and the choice of $W$ as in equation (7), we have that the \textit{GS knockoff} procedure controls the modified group FDR and the \textit{GS knockoff}+ procedure controls the group FDR as defined in \eqref{eqn:mfdr}.
\end{corollary}
The proof of Corollary \ref{cor1} is in Web Appendix C.
\section{Simulation} \label{sec:simulation}
To evaluate the performance of our proposed method, We simulated multiple settings with group-wise sparse predictor variables $\X^k \in \R^{n_k \times p_k}$ for $k \in [K]$. 

\textbf{Setting 1:} For $K= 3,4,5$, we have the same sample sizes $n_k = 1000{\color{black},200}$ and the same number of groups of features {\color{black}$M=40$} for $k \in [K]$. Within each group, there are 3 continuous variables and 1 categorical variable with 3 levels. In total, we have $p_k=160$ variables. 

\textbf{Setting 2:} For $K=4$, we vary the sample size and the types within the group across different sites. We set $n_1=2000, n_2=1200, n_3=700, n_4=600$. The types within the group across different sites are different. Site 1 encompasses 4 continuous variables per group. Site 2 also has 3 continuous variables and 1 categorical variable with 4 levels. Site 3 offers only 2 categorical features with 3 levels per group. Site 4 has 2 continuous and 2 binary categorical variables per group. In total, we have $p_1 = 160, p_2=160, p_3=80, p_4=160$.

For categorical variables with $L$ levels, we create $L-1$ dummy variables. Let $\bar{\X}^k$ denote the expanded design matrix of $\X^k$ after replacing each categorical variable with dummy variables. 

We consider the following three different models for $Y^k$s:
\begin{enumerate}
\item \textbf{Continuous}: For $k \in [K]$, {\color{black} $Y^k$s are continuous and simulated using linear regression models in all datasets.}
\begin{eqnarray*}\label{eqn:linmod}
Y^k&=&{\color{black}\bar{\X}^k\beta^{k}}+\varepsilon^k,
\end{eqnarray*}
where $\varepsilon^k\sim \mathcal{N}(0,\sigma_k^2)$, and $\sigma_k$ is the signal noise ratio. 
\item \textbf{Binary}: For $k \in [K]$, {\color{black}$Y^k$s are binary and simulated using logistic regression models for all datasets.}
\begin{eqnarray*}
Y^k\sim \textnormal{Bernoulli} \left(\frac{\exp(\alpha_k+{\color{black}\bar{\X}^k\beta^{k}})}{1+\exp(\alpha_k+{\color{black}\bar{\X}^k\beta^{k}})}\right).
\end{eqnarray*}
\item \textbf{Mixed}: {\color{black}$Y^k$s are either continuous or binary, $Y^k$s are either simulated from linear regression models or probit regression models.} We generate the latent outcome $\overline{Y}^k$s for $k \in [K]$ from the linear models:
\begin{eqnarray*}
\overline{Y}^k&=&{\color{black}\bar{\X}^k\beta^{k}}+\varepsilon^k,
\end{eqnarray*}
where $\varepsilon^k\sim \mathcal{N}(0,\sigma_k^2)$, and $\sigma_k$ is the signal noise ratio. Then for continuous outcome $Y^{k}$, we set $Y^{k}=\overline{Y}^{k}$; for binary outcome $Y^{k}$,
we set a threshold for $\overline{Y}^k$: \[Y^k = \mathbbm{1}\{\overline{Y}^k \geq 0\}. \]
\end{enumerate}

{\color{black}We also consider two scenarios for the signal strengths: 

\textbf{Scenario 1}: both directions and strengths of the simultaneous signals are the same among the $K$ datasets.

\textbf{Scenario 2}: only the directions of the simultaneous signals are the same but the signal strengths are different among the $K$ datasets.

For the coefficients $\beta^1, \cdots,\beta^K$ among the $K$ experiments, we explore three \textbf{choices} including \textbf{choice 1}: only simultaneous signals exist; \textbf{choice 2}: simultaneous signals and non-simultaneous signals exist in one dataset; \textbf{choice 3}: simultaneous signals and non-simultaneous signals exist in multiple datasets. These choices are frequently observed in the N3C database. The design structure for coefficients is shown in Figure \ref{fig:design}. More details on the data generation are provided in Web Appendix E.}

{\color{black}We perform the \textit{GS knockoff} procedure as described in Section \ref{sec:algorithm}. 
{\color{black}For individual dataset, when {\color{black}all the features are groups of continuous variables, and} Y$|$X follows a linear model, we use the Fixed-X knockoff approach (site 1 in \textbf{Setting 2}).  Otherwise, we use the sequential group knockoff (Algorithm \ref{alg:seqknockoff}) for group knockoff construction. }

In Step 4, we use the knockoff+ filter to control the FDR at $0.2$. We compare the proposed method with two alternative strategies \citep{dai2021multiple} for combining information from multiple datasets and two approaches that use individual knockoff constructions rather than group knockoff constructions: 
\begin{itemize}
    \item \textit{Pooling:} The multiple datasets are first pooled together, and the tests of the conditional associations are performed using the group knockoff methods for a single dataset.
    \item \textit{Intersection:} First, the group knockoff methods for single datasets are used to select signals from individual datasets. Then the intersection set of the selected signals from the multiple datasets is constructed as the simultaneous signal set.
    \item \textit{Individual (Lasso):} First, we construct the knockoff using the individual knockoff method. Then we fit the model using Lasso. If one signal is selected within a group, then the whole group will be selected.
   \item \textit{Individual (Group Lasso):} First, we construct the knockoff using individual knockoff. Then we fit the model using group Lasso. 
\end{itemize}

We run {\color{black}$500$} simulations under each of the following data settings. We first vary the signal sparsity levels of the mutual signals among K datasets ($s_0$), the number of groups of signals specifically for the $k$-th dataset ($s_k, ~\text{for}~ k \in [K]$), the number of groups of mutual signals in two datasets ($s_{ij}$, $i$-th and $j$-th datasets), {\color{black} three datasets ($s_{ijo}$, $i$-th, $j$-th, and $o$-th datasets), four datasets ($s_{ijop}$, and $i$-th, $j$-th, $o$-th and $p$-th datasets).} We also vary the within-group feature correlations $\rho_k$, and the ratio between the between-group correlations and within-group correlations $\gamma_k$. 
To validate the distribution of generating knockoffs, rather than assessing each group of predictors individually, we apply the Chi-square test to examine the symmetry of the filter statistics W distribution. More details on the data generation, simulation settings, {\color{black}and validation of knockoffs} are provided in Web Appendix E.


 {\color{black} Figure \ref{fig:figure1} compares the performances of five methods (\textit{GS knockoffs}, \textit{Pooling}, \textit{Intersection}, \textit{Individual (Group Lasso)}, and \textit{Individual (Lasso))} on \textbf{Setting 1} for the \textbf{Mixed} models setting ($\Y^k$s are either continuous or binary)} {\color{black}when $n_k=1000$.} We first demonstrate the performance of the methods as the sparsity level changes (a) for the mutual signals when no non-mutual signals exist, (b) when unique signals for each data set exist, (c) when mutual signals for 2 datasets exist. In Figure \ref{fig:figure1} (d)-(f), we demonstrate the effect of the (group) correlation structure of $\X$. The \textit{GS knockoff} method controls FDR in all the settings and has good power. The \textit{Pooling} method fails to control FDR when non-mutual signals exist (Figure \ref{fig:figure1} b-f). The \textit{Intersection} method fails to control FDR when mutual signals two datasets exist  (Figure \ref{fig:figure1} c). The \textit{Individual (Lasso)} method fails to control the group FDR in most settings (Figure \ref{fig:figure1} a-e), which is as expected theoretically. The \textit{Individual (Group Lasso)} method controls the FDR in all settings, which is also as expected theoretically; however, as the within-group correlation increases, there is a substantial power loss for this method (Figure \ref{fig:figure1} d). Simulation experiments for the continuous and binary settings show similar results (See Figures S1-S3 in Web Appendix F for more simulation results). 


Figure \ref{fig:figure2} shows simulation results for the $K=4$ and $K=5$ cases {\color{black} on \textbf{Setting 1} for the \textbf{Mixed} models setting} {\color{black}when $n_k=1000$.} Overall the results are consistent with the $K=3$ cases. As $K$ increases we see a {\color{black}slight} power decrease with all the three methods. The \textit{GS knockoff} method effectively controls the FDR and demonstrates good power. Although the \textit{Pooling} method has the highest power, it has high FDP when non-mutual signals exist (Figure \ref{fig:figure2}). The \textit{Intersection} method has comparable power with the \textit{GS knockoff} method but has no FDR control guarantee{\color{black}, especially for those mutual signals that only appear in a few sites.} Regarding the \textit{Individual} knockoff methods, the results are consistent with K=3. The group filter (\textit{Individual (Group Lasso)}) can control the group FDR but the power is very low when the within-group correlation is very strong while the individual filter (\textit{Individual (Lasso)}) fails to control the group FDR.

{\color{black}
Figure \ref{fig:figure3} displays simulation results from different sites with varied sample sizes and types for {\color{black} $K=4$ ($\X^k$ is simulated from data setting 2, $\Y^k$s are either continuous or binary}), showing consistency with previous scenarios of uniform sample sizes and types. The GS knockoff method's effectiveness remains unaffected by these differences, highlighting its robustness and adaptability to varied data conditions. This feature is particularly beneficial in multi-site studies, ensuring consistent and reliable results across diverse research environments.
}

The simulation results are consistent with our theoretical expectations. In terms of FDR, the proposed \textit{GS knockoff} method controls FDR across all designed settings while the other methods fail. The \textit{Pooling} method can control FDR when only simultaneous signals exist. {\color{black}The \textit{Intersection} method fails to control FDR in some settings when non-mutual signals exist, especially for settings when mutual signals for most but not all datasets are dominant.} In terms of power, the \textit{GS knockoff} method has good power, which is comparable to the \textit{Pooling} method and is slightly higher than the \textit{Intersection} method when only simultaneous signals exist. For the \textit{Individual} methods, when using an individual filter (i.e., \textit{Individual (Lasso)}), the group FDR is not always controlled. For the group filter (i.e., \textit{Individual (Group Lasso)}), the group FDR can be controlled but the power is less than the proposed \textit{GS knockoff} method when within group correlation is strong. {\color{black}Therefore, when only simultaneous signals are present across all sites, the \textit{Pooling} method outperforms others, offering controlled FDR and the highest power. Conversely, when non-simultaneous signals are present in only a few sites (e.g., unique to each site), the \textit{Intersection} method is superior, demonstrating comparable power to \textit{GS knockoff}, but with a lower FDR. However, when mutual signals are present in most, but not all datasets (e.g., appearing in 2 out of 3 sites, or 3 out of 4 sites), the \textit{Intersection} method fails to control FDR, whereas the \textit{GS knockoff} method effectively controls the FDR and provide satisfactory power (See Figure \ref{fig:figure3} right panel and Figure S6 right panel). The two individual knockoff methods do not offer any advantages compared with the \textit{GS knockoff} method.}
The performance of the methods on Scenarios 1 (same signal strengths) and 2 (different signal strengths), and different data settings (continuous, binary, and mixed) are similar. {\color{black}Additionally, the disparities in sample sizes and types at various sites do not impinge upon the efficacy of the proposed \textit{GS knockoff method}. This robustness underscores the method's adaptability to diverse data conditions, maintaining its performance regardless of sample size and type variations. {\color{black}Moreover, despite the limited sample size ($n_k=200$), the \textit{GS knockoffs} method consistently maintains a high stable power, outperforming all other methods.} This attribute of the \textit{GS knockoff} method is particularly advantageous in multi-site studies where such variability is common, ensuring reliable and stable results across different research settings. More simulation results are shown in Web Appendix F.}

\section{The N3C data analysis} \label{sec:realdata}
In this section, we demonstrate the application of our proposed \textit{GS knockoff} method to the N3C data for the selection of risk factors of long COVID from a collection of patient baseline demographic, comorbidity, and medication information (pre-conditions before the infection of acute COVID). Our data is from the N3C Knowledge Store Shared Project. The N3C enclave consists of EHR data for over 8 million patients with confirmed COVID-19 infection. It also contains high-dimensional patient demographics, comorbidity, medication, and socioeconomic information. As of December 15, 2023, there are over 83 DCSs in the N3C data enclave. The population is heterogeneous across the DCSs, and the data qualities are different. The long COVID indicator is not well recorded in a majority of the DCSs. There are only six DCSs with more than 1,000 long COVID cases reported and two of them have substantial missingness in the demographic and comorbidity information. The cohort is constructed by a matched case-control sampling of patients with confirmed COVID infections from four DCSs {\color{black}($n_1 = 11,797$ in site A1, $n_2=5,922$ in site A2, $n_3=3,175$ in site B1 and $n_4 = 2,749$ in site B2)}. Information on whether the patient has developed long COVID after the acute COVID has been recorded in the data sites differently. In data sites A1 and A2, a binary long COVID U9.09 diagnosis is provided as the long COVID outcome, whereas in sites B1 and B2, a binary long COVID clinical visit index is recorded as the long COVID outcome. These two long COVID indicators are highly related but not the same. A list of patient baseline information has been extracted as (group level) candidate risk factors ({\color{black}$M = 37$}). For some of the candidate variables, the data from the two sites are recorded differently (for example, the ``obesity" variable and ``diabetes" variable, see details in Web Appendix G). Our goal for this analysis is to identify mutual risk factors for these two outcomes. Details on the cohort construction and candidate risk factors can be found in Web Appendix G. 

We use the \textit{GS knockoff} method with the sequential group knockoff method for the knockoff construction, and the group knockoff+ filter with the $\gfdr$ controlled at 0.2. We also compare the result with the selection using the group knockoff filter and the \textit{intersection} method with knockoff+ filter. 

The GS knockoff+ method identifies 6 risk factors: age at COVID, obesity, systemic corticosteroids, depression, chronic lung disease, and usage of corticosteroids during COVID hospitalization. Using the group knockoff filter, {\color{black}five} additional risk factors are selected, namely, {\color{black}malignant cancer, antibody of COVID, the usage of Remdisivir during COVID hospitalization, emergence room indicator due to the COVID, and COVID severity type}. Because the long COVID indicators are not the same across DCSs, the \textit{pooling} method is not suitable for the analysis, the \textit{intersection} method with the knockoff+ filter selects {\color{black} age at COVID, obesity, systemic corticosteroids, depression, chronic lung disease, usage of corticosteroids during COVID hospitalization, malignant cancer, the antibody of COVID, dementia, metastatic solid tumor cancer.
}

{\color{black}
We also conduct a sensitivity analysis by adding the below 10 variables with permutations within each site into the original data: race, rheumatologic disease, kidney disease, heart failure, hemiplegia or paraplegia, psychosis, peptic ulcer, hypertension, tobacco smoker, solid organ or blood stem cell transplant.
The GS knockoff+ method identifies 5 risk factors: age at COVID, obesity, systemic corticosteroids, depression, and chronic lung disease. The intersection method identifies 9 risk factors: age at COVID, obesity, systemic corticosteroids, depression, chronic lung disease, metastatic solid tumor cancers, antibody of COVID, the usage of Remdisivir during COVID hospitalization, and severity type. The GS knockoff identifies 2 additional risk factors: sex and usage of corticosteroids during COVID hospitalization. No methods select permutation variables. The results show the stability of our proposed method.
}

Many of the risk factors identified using the GS knockoff method are also reported to be associated with long COVID in other independent studies. For example, older age has been found to be associated with a higher risk of long COVID, possibly due to the higher likelihood of severe initial COVID-19 illness and a slower, more complex recovery process in older people \citep{sudre2021attributes}. Obesity can lead to chronic inflammation and impair immune response, which may make individuals more susceptible to long-term effects of COVID-19 \citep{vimercati2021association}. Patients with pre-existing lung conditions may experience more severe COVID-19 symptoms and longer recovery times, leading to a higher risk of long COVID \citep{beltramo2021chronic}. Corticosteroids are often used in severe cases of COVID-19 to manage the body's immune response. However, their usage can also suppress the immune system, potentially leading to a longer recovery period and a higher likelihood of long COVID \citep{goel2022systemic}. There is a bidirectional relationship between COVID-19 and psychiatric disorders, with COVID-19 increasing the risk of psychiatric sequelae and a diagnosis of a mental health disorder increasing the risk of COVID-19 \citep{taquet2021bidirectional}.

\section{Discussion}
In this paper, we present a novel \textit{GS knockoff} method, which allows us to control FDR in testing the union null hypotheses on conditional associations between group-level candidate features and outcomes. Like other knockoff-based methods, the \textit{GS knockoffs} can work with very general conditional model settings and covariate structures within the individual datasets, assuming the independence between the datasets. This method allows us to collectively use information from datasets with different dependencies of $Y|\X$, different outcomes $Y$, and heterogeneous $\X$ structures, allowing for different data types and different group sizes across multiple datasets. The FDR control guarantee is exact for finite sample settings under the Fix-X or Model-X settings. 

When approximation error exists in the estimation of the $\X$ distributions, inflation on the FDR is expected \citep{barber2020}, with the inflation rate proportional to the exponential of the Kullback-Leibler divergence between the true distribution and the approximated distribution. With all the data settings we experimented numerically, we have sufficient sample size $n$ to approximate the $\X$ distribution. Therefore, we see well-controlled FDR in all our simulated settings. For potential application to ultra-high-dimensional data, an extension of the robustness result to our proposed group-level \textit{GS knockoff} is desirable.  

This method has broad applications beyond the N3C long COVID real data example.{\color{black}In EHR data from multiple data centers, some covariates are recorded differently among the centers, some groups of variables are of different group sizes and data types, and some groups are composed of both continuous and categorical variables, and the population distributions are different across the data sources. Our extensive simulations and the N3C data example show the \textit{GS knockoff} has satisfactory power and FDR control performance under different scenarios.}  Although we illustrate our methods with observation studies, we want to highlight that it can be useful for clinical trial data. When trials are homogeneous, the \textit{pooling} strategy is powerful and shows success in controlling the FDR when selecting predictive biomarker \citep{sechidis2021biomarker} and treatment effect modifiers for clinical trials \citep{katsevich2023filtering}. However, when trials are heterogeneous and group-level risk factors or effect modifiers are of concern, our proposed method provides a powerful tool. In addition, this method requires very limited information (only the test statistics) to be shared among the data centers, which benefits data collaboration under privacy protections. The general framework is compatible with all the existing knockoff and group knockoff construction methods. We develop a list of group knockoff construction methods to work with both the Fixed-X and Model-X knockoff approaches. Our framework can be implemented to extend other knockoff approaches. For example, for non-Gaussian mixed data, we construct group model-X knockoffs with the sequential group knockoff construction. However, there are alternative ways to construct group model-X knockoffs. For example, the Latent Gaussian Copula Knockoffs \citep{vasquez2023controlling} can also be extended for group knockoff construction by using second-order group Model-X construction instead of the original second-order Model-X construction algorithm for the latent Gaussian variables. As long as the group knockoffs can be constructed for Step 1, they can be used in our general simultaneous group knockoff framework. 

There are limitations of the current \textit{GS knockoff} method. First, the power is expected to decrease as the number of datasets and non-mutual signals increase. In Sections \ref{sec:simulation} and \ref{sec:realdata}, we demonstrate satisfying performance when $K=3$, $4$ or $5$. As $K$ further increases, the power will decrease, because we are testing the union null hypotheses. When $K$ is much higher, instead of pursuing simultaneous signals across all the datasets, one may be more interested in signals that are non-nulls in a fraction of the datasets. The multi-environment knockoff method \citep{li2021} can be extended for such applications. Second, the current \textit{Simultaneous} knockoff methods can only work with datasets that are mutually independent, which is satisfied by the N3C data. Methods allowing for overlapping samples across the datasets will be very useful for identifying signals for multiple outcomes using the same dataset. 

\section*{Acknowledgements}
This research is partly supported by the National Institute of General Medical Sciences under grants U54 GM115458 and U54 GM104942. 

\paragraph{N3C Attribution} The analyses described in this publication were conducted with data or tools accessed through the NCATS N3C Data Enclave \url{https://covid.cd2h.org} and N3C Attribution \& Publication Policy v 1.2-2020-08-25b supported by NCATS U24 TR002306, Axle Informatics Subcontract: NCATS-P00438-B. This research was possible because of the patients whose information is included within the data and the organizations (\url{https://ncats.nih.gov/n3c/resources/data-contribution/data-transfer-agreement-signatories}) and scientists who have contributed to the on-going development of this community resource \url{https://doi.org/10.1093/jamia/ocaa196}.

\paragraph{Disclaimer} The N3C Publication committee confirmed that this manuscript (msid:952.78) is in accordance with N3C data use and attribution policies; however, this content is solely the responsibility of the authors and does not necessarily represent the official views of the National Institutes of Health or the N3C program.

\paragraph{IRB}
The N3C data transfer to NCATS is performed under a Johns Hopkins University Reliance Protocol \# IRB00249128 or individual site agreements with NIH. The N3C Data Enclave is managed under the authority of the NIH; information can be found at \url{https://ncats.nih.gov/n3c/resources}.

\paragraph{Individual Acknowledgements For Core Contributors}
We gratefully acknowledge the following core contributors to N3C:

Adam B. Wilcox, Adam M. Lee, Alexis Graves, Alfred (Jerrod) Anzalone, Amin Manna, Amit Saha, Amy Olex, Andrea Zhou, Andrew E. Williams, Andrew Southerland, Andrew T. Girvin, Anita Walden, Anjali A. Sharathkumar, Benjamin Amor, Benjamin Bates, Brian Hendricks, Brijesh Patel, Caleb Alexander, Carolyn Bramante, Cavin Ward-Caviness, Charisse Madlock-Brown, Christine Suver, Christopher Chute, Christopher Dillon, Chunlei Wu, Clare Schmitt, Cliff Takemoto, Dan Housman, Davera Gabriel, David A. Eichmann, Diego Mazzotti, Don Brown, Eilis Boudreau, Elaine Hill, Elizabeth Zampino, Emily Carlson Marti, Emily R. Pfaff, Evan French, Farrukh M Koraishy, Federico Mariona, Fred Prior, George Sokos, Greg Martin, Harold Lehmann, Heidi Spratt, Hemalkumar Mehta, Hongfang Liu, Hythem Sidky, J.W. Awori Hayanga, Jami Pincavitch, Jaylyn Clark, Jeremy Richard Harper, Jessica Islam, Jin Ge, Joel Gagnier, Joel H. Saltz, Joel Saltz, Johanna Loomba, John Buse, Jomol Mathew, Joni L. Rutter, Julie A. McMurry, Justin Guinney, Justin Starren, Karen Crowley, Katie Rebecca Bradwell, Kellie M. Walters, Ken Wilkins, Kenneth R. Gersing, Kenrick Dwain Cato, Kimberly Murray, Kristin Kostka, Lavance Northington, Lee Allan Pyles, Leonie Misquitta, Lesley Cottrell, Lili Portilla, Mariam Deacy, Mark M. Bissell, Marshall Clark, Mary Emmett, Mary Morrison Saltz, Matvey B. Palchuk, Melissa A. Haendel, Meredith Adams, Meredith Temple-O'Connor, Michael G. Kurilla, Michele Morris, Nabeel Qureshi, Nasia Safdar, Nicole Garbarini, Noha Sharafeldin, Ofer Sadan, Patricia A. Francis, Penny Wung Burgoon, Peter Robinson, Philip R.O. Payne, Rafael Fuentes, Randeep Jawa, Rebecca Erwin-Cohen, Rena Patel, Richard A. Moffitt, Richard L. Zhu, Rishi Kamaleswaran, Robert Hurley, Robert T. Miller, Saiju Pyarajan, Sam G. Michael, Samuel Bozzette, Sandeep Mallipattu, Satyanarayana Vedula, Scott Chapman, Shawn T. O'Neil, Soko Setoguchi, Stephanie S. Hong, Steve Johnson, Tellen D. Bennett, Tiffany Callahan, Umit Topaloglu, Usman Sheikh, Valery Gordon, Vignesh Subbian, Warren A. Kibbe, Wenndy Hernandez, Will Beasley, Will Cooper, William Hillegass, Xiaohan Tanner Zhang. Details of contributions are available at \url{covid.cd2h.org/core-contributors}.

\paragraph{Data Partners with Released Data}
The following institutions whose data is released or pending: Available: Advocate Health Care Network — UL1TR002389: The Institute for Translational Medicine (ITM) • Boston University Medical Campus — UL1TR001430: Boston University Clinical and Translational Science Institute • Brown University — U54GM115677: Advance Clinical Translational Research (Advance-CTR) • Carilion Clinic — UL1TR003015: iTHRIV Integrated Translational health Research Institute of Virginia • Charleston Area Medical Center — U54GM104942: West Virginia Clinical and Translational Science Institute (WVCTSI) • Children’s Hospital Colorado — UL1TR002535: Colorado Clinical and Translational Sciences Institute • Columbia University Irving Medical Center — UL1TR001873: Irving Institute for Clinical and Translational Research • Duke University — UL1TR002553: Duke Clinical and Translational Science Institute • George Washington Children’s Research Institute — UL1TR001876: Clinical and Translational Science Institute at Children’s National (CTSA-CN) • George Washington University — UL1TR001876: Clinical and Translational Science Institute at Children’s National (CTSA-CN) • Indiana University School of Medicine — UL1TR002529: Indiana Clinical and Translational Science Institute • Johns Hopkins University — UL1TR003098: Johns Hopkins Institute for Clinical and Translational Research • Loyola Medicine — Loyola University Medical Center • Loyola University Medical Center — UL1TR002389: The Institute for Translational Medicine (ITM) • Maine Medical Center — U54GM115516: Northern New England Clinical \& Translational Research (NNE-CTR) Network • Massachusetts General Brigham — UL1TR002541: Harvard Catalyst • Mayo Clinic Rochester — UL1TR002377: Mayo Clinic Center for Clinical and Translational Science (CCaTS) • Medical University of South Carolina — UL1TR001450: South Carolina Clinical \& Translational Research Institute (SCTR) • Montefiore Medical Center — UL1TR002556: Institute for Clinical and Translational Research at Einstein and Montefiore • Nemours — U54GM104941: Delaware CTR ACCEL Program • NorthShore University HealthSystem — UL1TR002389: The Institute for Translational Medicine (ITM) • Northwestern University at Chicago — UL1TR001422: Northwestern University Clinical and Translational Science Institute (NUCATS) • OCHIN — INV-018455: Bill and Melinda Gates Foundation grant to Sage Bionetworks • Oregon Health \& Science University — UL1TR002369: Oregon Clinical and Translational Research Institute • Penn State Health Milton S. Hershey Medical Center — UL1TR002014: Penn State Clinical and Translational Science Institute • Rush University Medical Center — UL1TR002389: The Institute for Translational Medicine (ITM) • Rutgers, The State University of New Jersey — UL1TR003017: New Jersey Alliance for Clinical and Translational Science • Stony Brook University — U24TR002306 • The Ohio State University — UL1TR002733: Center for Clinical and Translational Science • The State University of New York at Buffalo — UL1TR001412: Clinical and Translational Science Institute • The University of Chicago — UL1TR002389: The Institute for Translational Medicine (ITM) • The University of Iowa — UL1TR002537: Institute for Clinical and Translational Science • The University of Miami Leonard M. Miller School of Medicine — UL1TR002736: University of Miami Clinical and Translational Science Institute • The University of Michigan at Ann Arbor — UL1TR002240: Michigan Institute for Clinical and Health Research • The University of Texas Health Science Center at Houston — UL1TR003167: Center for Clinical and Translational Sciences (CCTS) • The University of Texas Medical Branch at Galveston — UL1TR001439: The Institute for Translational Sciences • The University of Utah — UL1TR002538: Uhealth Center for Clinical and Translational Science • Tufts Medical Center — UL1TR002544: Tufts Clinical and Translational Science Institute • Tulane University — UL1TR003096: Center for Clinical and Translational Science • University Medical Center New Orleans — U54GM104940: Louisiana Clinical and Translational Science (LA CaTS) Center • University of Alabama at Birmingham — UL1TR003096: Center for Clinical and Translational Science • University of Arkansas for Medical Sciences — UL1TR003107: UAMS Translational Research Institute • University of Cincinnati — UL1TR001425: Center for Clinical and Translational Science and Training • University of Colorado Denver, Anschutz Medical Campus — UL1TR002535: Colorado Clinical and Translational Sciences Institute • University of Illinois at Chicago — UL1TR002003: UIC Center for Clinical and Translational Science • University of Kansas Medical Center — UL1TR002366: Frontiers: University of Kansas Clinical and Translational Science Institute • University of Kentucky — UL1TR001998: UK Center for Clinical and Translational Science • University of Massachusetts Medical School Worcester — UL1TR001453: The UMass Center for Clinical and Translational Science (UMCCTS) • University of Minnesota — UL1TR002494: Clinical and Translational Science Institute • University of Mississippi Medical Center — U54GM115428: Mississippi Center for Clinical and Translational Research (CCTR) • University of Nebraska Medical Center — U54GM115458: Great Plains IDeA-Clinical \& Translational Research • University of North Carolina at Chapel Hill — UL1TR002489: North Carolina Translational and Clinical Science Institute • University of Oklahoma Health Sciences Center — U54GM104938: Oklahoma Clinical and Translational Science Institute (OCTSI) • University of Rochester — UL1TR002001: UR Clinical \& Translational Science Institute • University of Southern California — UL1TR001855: The Southern California Clinical and Translational Science Institute (SC CTSI) • University of Vermont — U54GM115516: Northern New England Clinical \& Translational Research (NNE-CTR) Network • University of Virginia — UL1TR003015: iTHRIV Integrated Translational health Research Institute of Virginia • University of Washington — UL1TR002319: Institute of Translational Health Sciences • University of Wisconsin-Madison — UL1TR002373: UW Institute for Clinical and Translational Research • Vanderbilt University Medical Center — UL1TR002243: Vanderbilt Institute for Clinical and Translational Research • Virginia Commonwealth University — UL1TR002649: C. Kenneth and Dianne Wright Center for Clinical and Translational Research • Wake Forest University Health Sciences — UL1TR001420: Wake Forest Clinical and Translational Science Institute • Washington University in St. Louis — UL1TR002345: Institute of Clinical and Translational Sciences • Weill Medical College of Cornell University — UL1TR002384: Weill Cornell Medicine Clinical and Translational Science Center • West Virginia University — U54GM104942: West Virginia Clinical and Translational Science Institute (WVCTSI) Submitted: Icahn School of Medicine at Mount Sinai — UL1TR001433: ConduITS Institute for Translational Sciences • The University of Texas Health Science Center at Tyler — UL1TR003167: Center for Clinical and Translational Sciences (CCTS) • University of California, Davis — UL1TR001860: UCDavis Health Clinical and Translational Science Center • University of California, Irvine — UL1TR001414: The UC Irvine Institute for Clinical and Translational Science (ICTS) • University of California, Los Angeles — UL1TR001881: UCLA Clinical Translational Science Institute • University of California, San Diego — UL1TR001442: Altman Clinical and Translational Research Institute • University of California, San Francisco — UL1TR001872: UCSF Clinical and Translational Science Institute Pending: Arkansas Children’s Hospital — UL1TR003107: UAMS Translational Research Institute • Baylor College of Medicine — None (Voluntary) • Children’s Hospital of Philadelphia — UL1TR001878: Institute for Translational Medicine and Therapeutics • Cincinnati Children’s Hospital Medical Center — UL1TR001425: Center for Clinical and Translational Science and Training • Emory University — UL1TR002378: Georgia Clinical and Translational Science Alliance • HonorHealth — None (Voluntary) • Loyola University Chicago — UL1TR002389: The Institute for Translational Medicine (ITM) • Medical College of Wisconsin — UL1TR001436: Clinical and Translational Science Institute of Southeast Wisconsin • MedStar Health Research Institute — UL1TR001409: The Georgetown-Howard Universities Center for Clinical and Translational Science (GHUCCTS) • MetroHealth — None (Voluntary) • Montana State University — U54GM115371: American Indian/Alaska Native CTR • NYU Langone Medical Center — UL1TR001445: Langone Health’s Clinical and Translational Science Institute • Ochsner Medical Center — U54GM104940: Louisiana Clinical and Translational Science (LA CaTS) Center • Regenstrief Institute — UL1TR002529: Indiana Clinical and Translational Science Institute • Sanford Research — None (Voluntary) • Stanford University — UL1TR003142: Spectrum: The Stanford Center for Clinical and Translational Research and Education • The Rockefeller University — UL1TR001866: Center for Clinical and Translational Science • The Scripps Research Institute — UL1TR002550: Scripps Research Translational Institute • University of Florida — UL1TR001427: UF Clinical and Translational Science Institute • University of New Mexico Health Sciences Center — UL1TR001449: University of New Mexico Clinical and Translational Science Center • University of Texas Health Science Center at San Antonio — UL1TR002645: Institute for Integration of Medicine and Science • Yale New Haven Hospital — UL1TR001863: Yale Center for Clinical Investigation

\section*{Data Availability Statement}
The data that support the findings of this study are not publicly available. Any data request needs to be submitted to the N3C (\url{https://ncats.nih.gov/n3c/about/applying-for-access}).

\section*{Supporting Information}
The Web Appendix referenced in sections 2-5 is available with this paper. Data supporting the findings of this paper can be requested as described in the data availability statement. The R codes for the simulation of this paper are available at \url{https://github.com/RunqiuWang22/Generalized_Simultaneous_knockoff}.


\bibliography{example_paper}

\begin{thebibliography}{35}
\providecommand{\natexlab}[1]{#1}
\providecommand{\url}[1]{\texttt{#1}}
\expandafter\ifx\csname urlstyle\endcsname\relax
  \providecommand{\doi}[1]{doi: #1}\else
  \providecommand{\doi}{doi: \begingroup \urlstyle{rm}\Url}\fi

\bibitem[Barber \& Candès(2015)Barber and Candès]{barber2015}
Barber, R.~F. and Candès, E.~J.
\newblock Controlling the false discovery rate via knockoffs.
\newblock \emph{Ann. Statist.}, 43\penalty0 (5):\penalty0 2055--2085, 2015.
\newblock \doi{10.1214/15-AOS1337}.

\bibitem[Barber \& Candès(2019)Barber and Candès]{barber2019}
Barber, R.~F. and Candès, E.~J.
\newblock A knockoff filter for high-dimensional selective inference.
\newblock \emph{Ann. Statist.}, 47\penalty0 (5):\penalty0 2504--2537, 2019.
\newblock \doi{10.1214/18-AOS1755}.

\bibitem[Barber et~al.(2020)Barber, Candès, and Samworth]{barber2020}
Barber, R.~F., Candès, E.~J., and Samworth, R.~J.
\newblock Robust inference with knockoffs.
\newblock \emph{Ann. Statist.}, 48\penalty0 (3):\penalty0 1409--1431, 2020.
\newblock \doi{10.1214/19-AOS1852}.

\bibitem[Bates et~al.(2021)Bates, Candès, Janson, and Wang]{bates2020}
Bates, S., Candès, E., Janson, L., and Wang, W.
\newblock Metropolized knockoff sampling.
\newblock \emph{Journal of the American Statistical Association}, 116\penalty0
  (535):\penalty0 1413--1427, 2021.
\newblock \doi{10.1080/01621459.2020.1729163}.

\bibitem[Beltramo et~al.(2021)Beltramo, Cottenet, Mariet, Georges, Piroth,
  Tubert-Bitter, Bonniaud, and Quantin]{beltramo2021chronic}
Beltramo, G., Cottenet, J., Mariet, A.-S., Georges, M., Piroth, L.,
  Tubert-Bitter, P., Bonniaud, P., and Quantin, C.
\newblock Chronic respiratory diseases are predictors of severe outcome in
  covid-19 hospitalised patients: a nationwide study.
\newblock \emph{European Respiratory Journal}, 58\penalty0 (6), 2021.
\newblock \doi{10.1183/13993003.04474-2020}.

\bibitem[Benjamini \& Yekutieli(2001)Benjamini and Yekutieli]{benjamini2001}
Benjamini, Y. and Yekutieli, D.
\newblock {The control of the false discovery rate in multiple testing under
  dependency}.
\newblock \emph{The Annals of Statistics}, 29\penalty0 (4):\penalty0 1165 --
  1188, 2001.
\newblock \doi{10.1214/aos/1013699998}.
\newblock URL \url{https://doi.org/10.1214/aos/1013699998}.

\bibitem[Bogomolov \& Heller(2013)Bogomolov and Heller]{bogomolov2013}
Bogomolov, M. and Heller, R.
\newblock Discovering findings that replicate from a primary study of high
  dimension to a follow-up study.
\newblock \emph{Journal of the American Statistical Association}, 108\penalty0
  (504):\penalty0 1480--1492, 2013.
\newblock \doi{10.1080/01621459.2013.829002}.

\bibitem[Bogomolov \& Heller(2018)Bogomolov and Heller]{bogomolov2018}
Bogomolov, M. and Heller, R.
\newblock {Assessing replicability of findings across two studies of multiple
  features}.
\newblock \emph{Biometrika}, 105\penalty0 (3):\penalty0 505--516, 2018.
\newblock ISSN 0006-3444.
\newblock \doi{10.1093/biomet/asy029}.

\bibitem[Candès et~al.(2018)Candès, Fan, Janson, and Lv]{candes2018}
Candès, E., Fan, Y., Janson, L., and Lv, J.
\newblock Panning for gold: ‘model-x’ knockoffs for high dimensional
  controlled variable selection.
\newblock \emph{Journal of the Royal Statistical Society: Series B (Statistical
  Methodology)}, 80\penalty0 (3):\penalty0 551--577, 2018.
\newblock \doi{https://doi.org/10.1111/rssb.12265}.

\bibitem[Chen et~al.(2019)Chen, Hou, and Hou]{chen2020}
Chen, J., Hou, A., and Hou, T.~Y.
\newblock {A prototype knockoff filter for group selection with FDR control}.
\newblock \emph{Information and Inference: A Journal of the IMA}, 9\penalty0
  (2):\penalty0 271--288, 2019.
\newblock ISSN 2049-8772.
\newblock \doi{10.1093/imaiai/iaz012}.

\bibitem[Chi(2008)]{chi2008}
Chi, Z.
\newblock False discovery rate control with multivariate p -values.
\newblock \emph{Electron. J. Statist.}, 2:\penalty0 368--411, 2008.
\newblock \doi{10.1214/07-EJS147}.

\bibitem[Dai \& Barber(2016)Dai and Barber]{dai2016knockoff}
Dai, R. and Barber, R.
\newblock The knockoff filter for fdr control in group-sparse and multitask
  regression.
\newblock In \emph{International conference on machine learning}, pp.\
  1851--1859. PMLR, 2016.

\bibitem[Dai \& Zheng(2023)Dai and Zheng]{dai2021multiple}
Dai, R. and Zheng, C.
\newblock False discovery rate-controlled multiple testing for union null
  hypotheses: a knockoff-based approach.
\newblock \emph{Biometrics}, 79:\penalty0 3497--3509, 2023.
\newblock \doi{https://doi.org/10.1111/biom.13848}.

\bibitem[Dai et~al.(2023)Dai, Lyu, and Li]{dai2022kernel}
Dai, X., Lyu, X., and Li, L.
\newblock Kernel knockoffs selection for nonparametric additive models.
\newblock \emph{Journal of the American Statistical Association}, 118\penalty0
  (543):\penalty0 2158--2170, 2023.
\newblock \doi{10.1080/01621459.2022.2039671}.
\newblock URL \url{https://doi.org/10.1080/01621459.2022.2039671}.

\bibitem[Goel et~al.(2022)Goel, Goyal, Nagaraja, and Kumar]{goel2022systemic}
Goel, N., Goyal, N., Nagaraja, R., and Kumar, R.
\newblock Systemic corticosteroids for management of ‘long-covid’: an
  evaluation after 3 months of treatment.
\newblock \emph{Monaldi Archives for Chest Disease}, 92\penalty0 (2), 2022.
\newblock \doi{10.4081/monaldi.2021.1981}.

\bibitem[Haendel et~al.(2020)Haendel, Chute, Bennett, Eichmann, Guinney, Kibbe,
  Payne, Pfaff, Robinson, Saltz, Spratt, Suver, Wilbanks, Wilcox, Williams, Wu,
  Blacketer, Bradford, Cimino, Clark, Colmenares, Francis, Gabriel, Graves,
  Hemadri, Hong, Hripscak, Jiao, Klann, Kostka, Lee, Lehmann, Lingrey, Miller,
  Morris, Murphy, Natarajan, Palchuk, Sheikh, Solbrig, Visweswaran, Walden,
  Walters, Weber, Zhang, Zhu, Amor, Girvin, Manna, Qureshi, Kurilla, Michael,
  Portilla, Rutter, Austin, Gersing, and the N3C~Consortium]{Haendel-N3C}
Haendel, M.~A., Chute, C.~G., Bennett, T.~D., Eichmann, D.~A., Guinney, J.,
  Kibbe, W.~A., Payne, P. R.~O., Pfaff, E.~R., Robinson, P.~N., Saltz, J.~H.,
  Spratt, H., Suver, C., Wilbanks, J., Wilcox, A.~B., Williams, A.~E., Wu, C.,
  Blacketer, C., Bradford, R.~L., Cimino, J.~J., Clark, M., Colmenares, E.~W.,
  Francis, P.~A., Gabriel, D., Graves, A., Hemadri, R., Hong, S.~S., Hripscak,
  G., Jiao, D., Klann, J.~G., Kostka, K., Lee, A.~M., Lehmann, H.~P., Lingrey,
  L., Miller, R.~T., Morris, M., Murphy, S.~N., Natarajan, K., Palchuk, M.~B.,
  Sheikh, U., Solbrig, H., Visweswaran, S., Walden, A., Walters, K.~M., Weber,
  G.~M., Zhang, X.~T., Zhu, R.~L., Amor, B., Girvin, A.~T., Manna, A., Qureshi,
  N., Kurilla, M.~G., Michael, S.~G., Portilla, L.~M., Rutter, J.~L., Austin,
  C.~P., Gersing, K.~R., and the N3C~Consortium.
\newblock {The National COVID Cohort Collaborative (N3C): Rationale, design,
  infrastructure, and deployment}.
\newblock \emph{Journal of the American Medical Informatics Association},
  28\penalty0 (3):\penalty0 427--443, 08 2020.
\newblock ISSN 1527-974X.
\newblock \doi{10.1093/jamia/ocaa196}.
\newblock URL \url{https://doi.org/10.1093/jamia/ocaa196}.

\bibitem[Heller \& Yekutieli(2014)Heller and Yekutieli]{heller2014b}
Heller, R. and Yekutieli, D.
\newblock Replicability analysis for genome-wide association studies.
\newblock \emph{Ann. Appl. Stat.}, 8\penalty0 (1):\penalty0 481--498, 2014.
\newblock \doi{10.1214/13-AOAS697}.

\bibitem[Heller et~al.(2014)Heller, Bogomolov, and Benjamini]{heller2014}
Heller, R., Bogomolov, M., and Benjamini, Y.
\newblock Deciding whether follow-up studies have replicated findings in a
  preliminary large-scale omics study.
\newblock \emph{Proceedings of the National Academy of Sciences}, 111\penalty0
  (46):\penalty0 16262--16267, 2014.
\newblock ISSN 0027-8424.
\newblock \doi{10.1073/pnas.1314814111}.

\bibitem[Huang \& Janson(2020)Huang and Janson]{huang2020}
Huang, D. and Janson, L.
\newblock Relaxing the assumptions of knockoffs by conditioning.
\newblock \emph{Ann. Statist.}, 48\penalty0 (5):\penalty0 3021--3042, 2020.
\newblock \doi{10.1214/19-AOS1920}.

\bibitem[Katsevich et~al.(2023)Katsevich, Sabatti, and
  Bogomolov]{katsevich2023filtering}
Katsevich, E., Sabatti, C., and Bogomolov, M.
\newblock Filtering the rejection set while preserving false discovery rate
  control.
\newblock \emph{Journal of the American Statistical Association}, 118\penalty0
  (541):\penalty0 165--176, 2023.

\bibitem[Kormaksson et~al.(2021)Kormaksson, Kelly, Zhu, Haemmerle, Pricop, and
  Ohlssen]{kormaksson2021}
Kormaksson, M., Kelly, L.~J., Zhu, X., Haemmerle, S., Pricop, L., and Ohlssen,
  D.
\newblock Sequential knockoffs for continuous and categorical predictors: With
  application to a large psoriatic arthritis clinical trial pool.
\newblock \emph{Statistics in Medicine}, 40\penalty0 (14):\penalty0 3313--3328,
  2021.
\newblock \doi{https://doi.org/10.1002/sim.8955}.

\bibitem[Li et~al.(2021)Li, Sesia, Romano, Candès, and Sabatti]{li2021}
Li, S., Sesia, M., Romano, Y., Candès, E., and Sabatti, C.
\newblock {Searching for robust associations with a multi-environment knockoff
  filter}.
\newblock \emph{Biometrika}, 109\penalty0 (3):\penalty0 611--629, 11 2021.
\newblock ISSN 1464-3510.
\newblock \doi{10.1093/biomet/asab055}.

\bibitem[Liu \& Zheng(2019)Liu and Zheng]{liu2019}
Liu, Y. and Zheng, C.
\newblock Deep latent variable models for generating knockoffs.
\newblock \emph{Stat}, 8\penalty0 (1):\penalty0 e260, 2019.
\newblock \doi{https://doi.org/10.1002/sta4.260}.

\bibitem[Montani et~al.(2022)Montani, Savale, Noel, Meyrignac, Colle, Gasnier,
  Corruble, Beurnier, Jutant, Pham, Lecoq, Papon, Figueiredo, Harrois, Humbert,
  and Monnet]{montani2022}
Montani, D., Savale, L., Noel, N., Meyrignac, O., Colle, R., Gasnier, M.,
  Corruble, E., Beurnier, A., Jutant, E.-M., Pham, T., Lecoq, A.-L., Papon,
  J.-F., Figueiredo, S., Harrois, A., Humbert, M., and Monnet, X.
\newblock Post-acute covid-19 syndrome.
\newblock \emph{European Respiratory Review}, 31\penalty0 (163), 2022.
\newblock ISSN 0905-9180.
\newblock \doi{10.1183/16000617.0185-2021}.
\newblock URL \url{https://err.ersjournals.com/content/31/163/210185}.

\bibitem[Pfaff et~al.(2023)Pfaff, Madlock-Brown, Baratta, Bhatia, Davis,
  Girvin, Hill, Kelly, Kostka, Loomba, McMurry, Wong, Bennett, Moffitt, Chute,
  Haendel, {The N3C Consortium}, and {The RECOVER Consortium}]{Pfaff2022}
Pfaff, E.~R., Madlock-Brown, C., Baratta, J.~M., Bhatia, A., Davis, H., Girvin,
  A., Hill, E., Kelly, L., Kostka, K., Loomba, J., McMurry, J.~A., Wong, R.,
  Bennett, T.~D., Moffitt, R., Chute, C.~G., Haendel, M., {The N3C Consortium},
  and {The RECOVER Consortium}.
\newblock Coding long covid: characterizing a new disease through an icd-10
  lens.
\newblock \emph{BMC Med}, 21:\penalty0 58, 2023.
\newblock \doi{10.1186/s12916-023-02737-6}.

\bibitem[Romano et~al.(2020)Romano, Sesia, and Cand\`{e}s]{romano2019}
Romano, Y., Sesia, M., and Cand\`{e}s, E.
\newblock Deep knockoffs.
\newblock \emph{Journal of the American Statistical Association}, 115\penalty0
  (532):\penalty0 1861--1872, 2020.
\newblock \doi{10.1080/01621459.2019.1660174}.

\bibitem[Sechidis et~al.(2021)Sechidis, Kormaksson, and
  Ohlssen]{sechidis2021biomarker}
Sechidis, K., Kormaksson, M., and Ohlssen, D.
\newblock Using knockoffs for controlled predictive biomarker identification.
\newblock \emph{Statistics in Medicine}, 40\penalty0 (25):\penalty0 5453--5473,
  2021.
\newblock \doi{https://doi.org/10.1002/sim.9134}.

\bibitem[Sesia et~al.(2018)Sesia, Sabatti, and Candès]{sesia2019}
Sesia, M., Sabatti, C., and Candès, E.~J.
\newblock {Gene hunting with hidden Markov model knockoffs}.
\newblock \emph{Biometrika}, 106\penalty0 (1):\penalty0 1--18, 2018.
\newblock ISSN 0006-3444.
\newblock \doi{10.1093/biomet/asy033}.

\bibitem[Spector \& Janson(2022)Spector and Janson]{spector2020}
Spector, A. and Janson, L.
\newblock {Powerful knockoffs via minimizing reconstructability}.
\newblock \emph{The Annals of Statistics}, 50\penalty0 (1):\penalty0 252 --
  276, 2022.
\newblock \doi{10.1214/21-AOS2104}.
\newblock URL \url{https://doi.org/10.1214/21-AOS2104}.

\bibitem[Srinivasan et~al.(2021)Srinivasan, Xue, and Zhan]{srinivasan2022}
Srinivasan, A., Xue, L., and Zhan, X.
\newblock Compositional knockoff filter for high-dimensional regression
  analysis of microbiome data.
\newblock \emph{Biometrics}, 77\penalty0 (3):\penalty0 984--995, 2021.
\newblock \doi{https://doi.org/10.1111/biom.13336}.
\newblock URL \url{https://onlinelibrary.wiley.com/doi/abs/10.1111/biom.13336}.

\bibitem[Sudre et~al.(2021)Sudre, Murray, Varsavsky, Graham, Penfold, Bowyer,
  Pujol, Klaser, Antonelli, Canas, et~al.]{sudre2021attributes}
Sudre, C.~H., Murray, B., Varsavsky, T., Graham, M.~S., Penfold, R.~S., Bowyer,
  R.~C., Pujol, J.~C., Klaser, K., Antonelli, M., Canas, L.~S., et~al.
\newblock Attributes and predictors of long covid.
\newblock \emph{Nature medicine}, 27\penalty0 (4):\penalty0 626--631, 2021.
\newblock \doi{10.1038/s41591-021-01292-y}.

\bibitem[Taquet et~al.(2021)Taquet, Luciano, Geddes, and
  Harrison]{taquet2021bidirectional}
Taquet, M., Luciano, S., Geddes, J.~R., and Harrison, P.~J.
\newblock Bidirectional associations between covid-19 and psychiatric disorder:
  retrospective cohort studies of 62354 covid-19 cases in the usa.
\newblock \emph{The Lancet Psychiatry}, 8\penalty0 (2):\penalty0 130--140,
  2021.
\newblock \doi{10.1016/S2215-0366(20)30462-4}.

\bibitem[V{\'a}squez et~al.(2023)V{\'a}squez, M{\'a}rquez~Urbina,
  Gonz{\'a}lez~Far{\'\i}as, and Escarela]{vasquez2023controlling}
V{\'a}squez, A.~R., M{\'a}rquez~Urbina, J.~U., Gonz{\'a}lez~Far{\'\i}as, G.,
  and Escarela, G.
\newblock Controlling the false discovery rate by a latent gaussian copula
  knockoff procedure.
\newblock \emph{Computational Statistics}, pp.\  1--24, 2023.

\bibitem[Vimercati et~al.(2021)Vimercati, De~Maria, Quarato, Caputi, Gesualdo,
  Migliore, Cavone, Sponselli, Pipoli, Inchingolo,
  et~al.]{vimercati2021association}
Vimercati, L., De~Maria, L., Quarato, M., Caputi, A., Gesualdo, L., Migliore,
  G., Cavone, D., Sponselli, S., Pipoli, A., Inchingolo, F., et~al.
\newblock Association between long covid and overweight/obesity.
\newblock \emph{Journal of Clinical Medicine}, 10\penalty0 (18):\penalty0 4143,
  2021.
\newblock \doi{10.3390/jcm10184143}.

\bibitem[Zhao \& Nguyen(2020)Zhao and Nguyen]{zhao2020}
Zhao, S.~D. and Nguyen, Y.~T.
\newblock Nonparametric false discovery rate control for identifying
  simultaneous signals.
\newblock \emph{Electron. J. Statist.}, 14\penalty0 (1):\penalty0 110--142,
  2020.
\newblock \doi{10.1214/19-EJS1663}.

\end{thebibliography}


\begin{thebibliography}{4}
\providecommand{\natexlab}[1]{#1}
\providecommand{\url}[1]{\texttt{#1}}
\expandafter\ifx\csname urlstyle\endcsname\relax
  \providecommand{\doi}[1]{doi: #1}\else
  \providecommand{\doi}{doi: \begingroup \urlstyle{rm}\Url}\fi

\bibitem[Barber \& Candès(2015)Barber and Candès]{barber2015}
Barber, R.~F. and Candès, E.~J.
\newblock Controlling the false discovery rate via knockoffs.
\newblock \emph{Ann. Statist.}, 43\penalty0 (5):\penalty0 2055--2085, 2015.
\newblock \doi{10.1214/15-AOS1337}.

\bibitem[Candès et~al.(2018)Candès, Fan, Janson, and Lv]{candes2018}
Candès, E., Fan, Y., Janson, L., and Lv, J.
\newblock Panning for gold: ‘model-x’ knockoffs for high dimensional
  controlled variable selection.
\newblock \emph{Journal of the Royal Statistical Society: Series B (Statistical
  Methodology)}, 80\penalty0 (3):\penalty0 551--577, 2018.
\newblock \doi{https://doi.org/10.1111/rssb.12265}.

\bibitem[Dai \& Barber(2016)Dai and Barber]{dai2016knockoff}
Dai, R. and Barber, R.
\newblock The knockoff filter for fdr control in group-sparse and multitask
  regression.
\newblock In \emph{International conference on machine learning}, pp.\
  1851--1859. PMLR, 2016.

\bibitem[Dai \& Zheng(2023)Dai and Zheng]{dai2021multiple}
Dai, R. and Zheng, C.
\newblock False discovery rate-controlled multiple testing for union null
  hypotheses: a knockoff-based approach.
\newblock \emph{Biometrics}, 79:\penalty0 3497--3509, 2023.
\newblock \doi{https://doi.org/10.1111/biom.13848}.

\end{thebibliography}
\bibliographystyle{icml2022}


\pagebreak
\begin{figure}[!ht]
    \centering
    \includegraphics[scale=0.35]{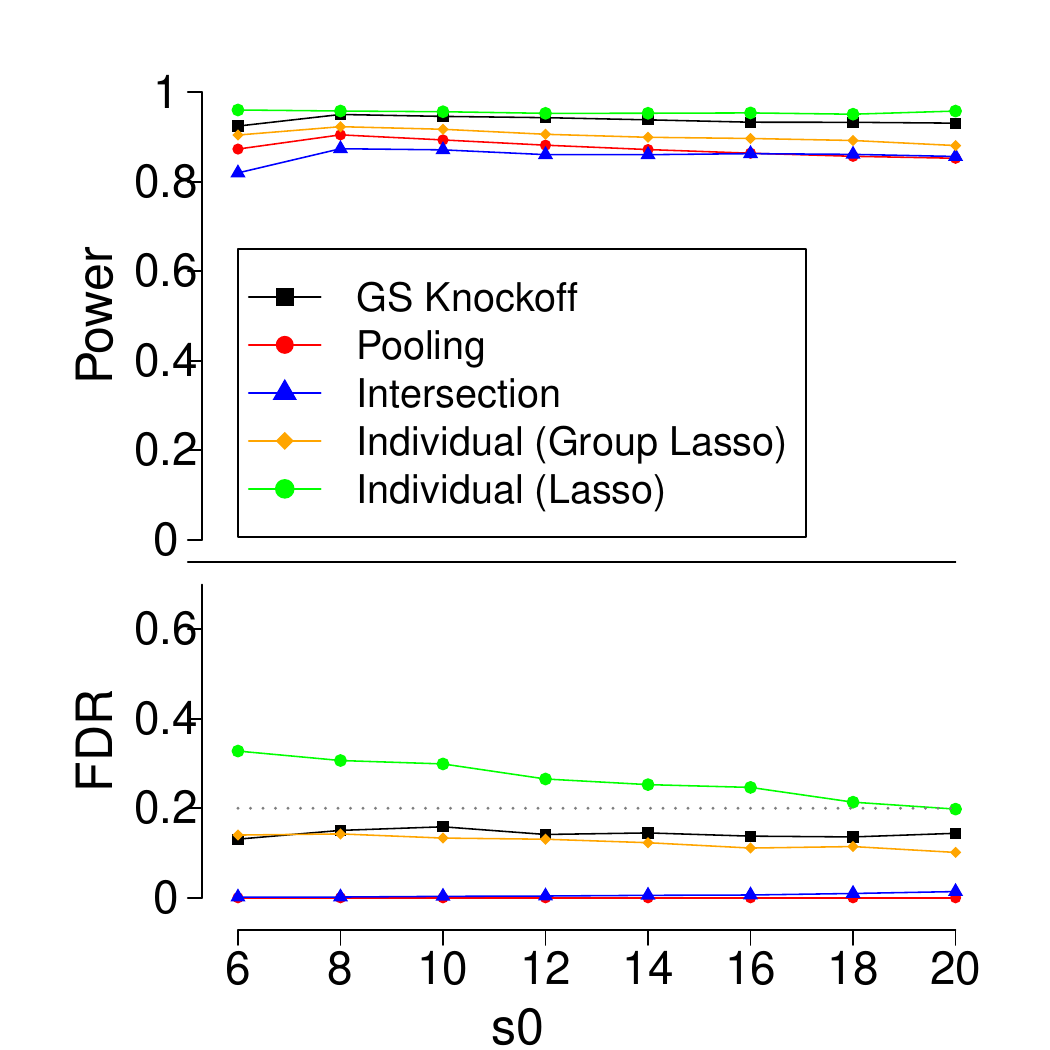}
    \includegraphics[scale=0.35]{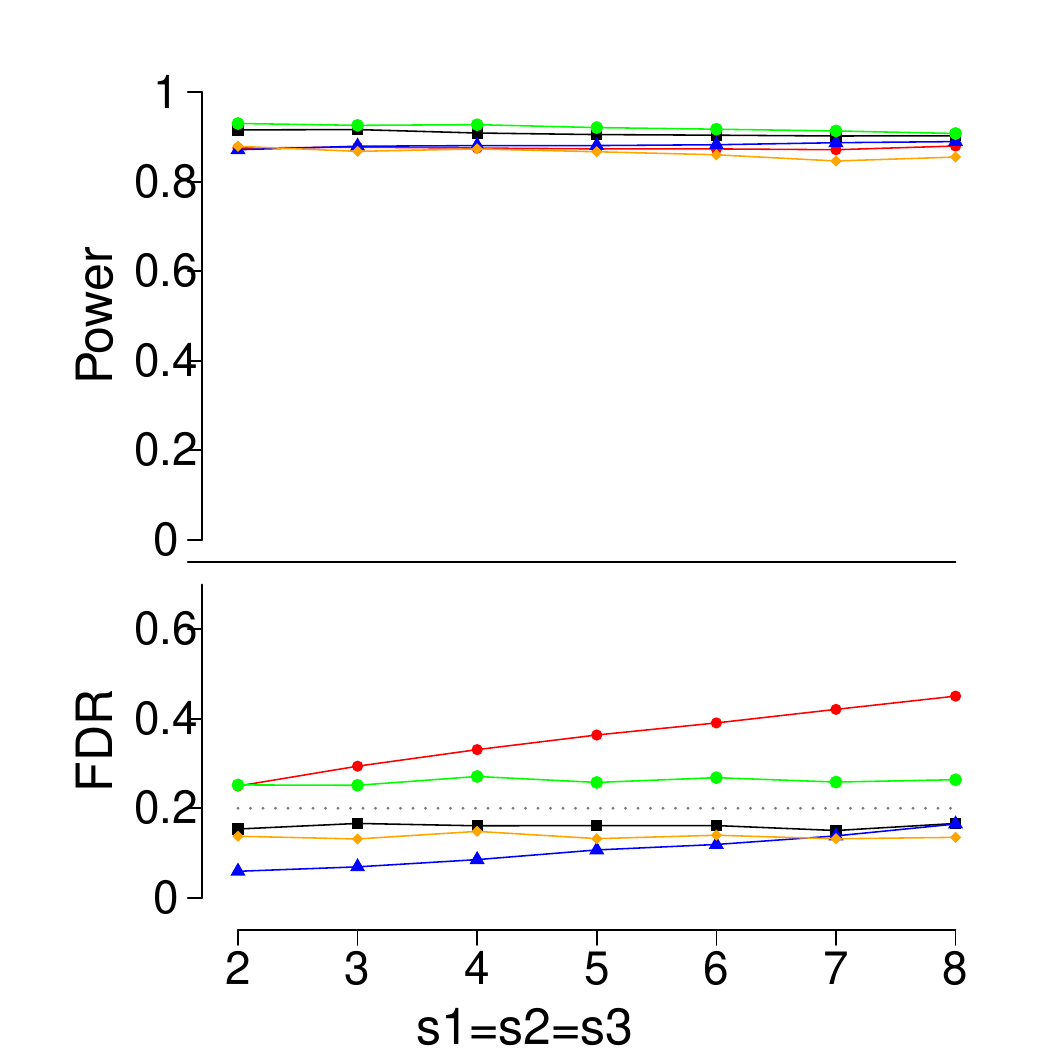}\\
    \includegraphics[scale=0.35]{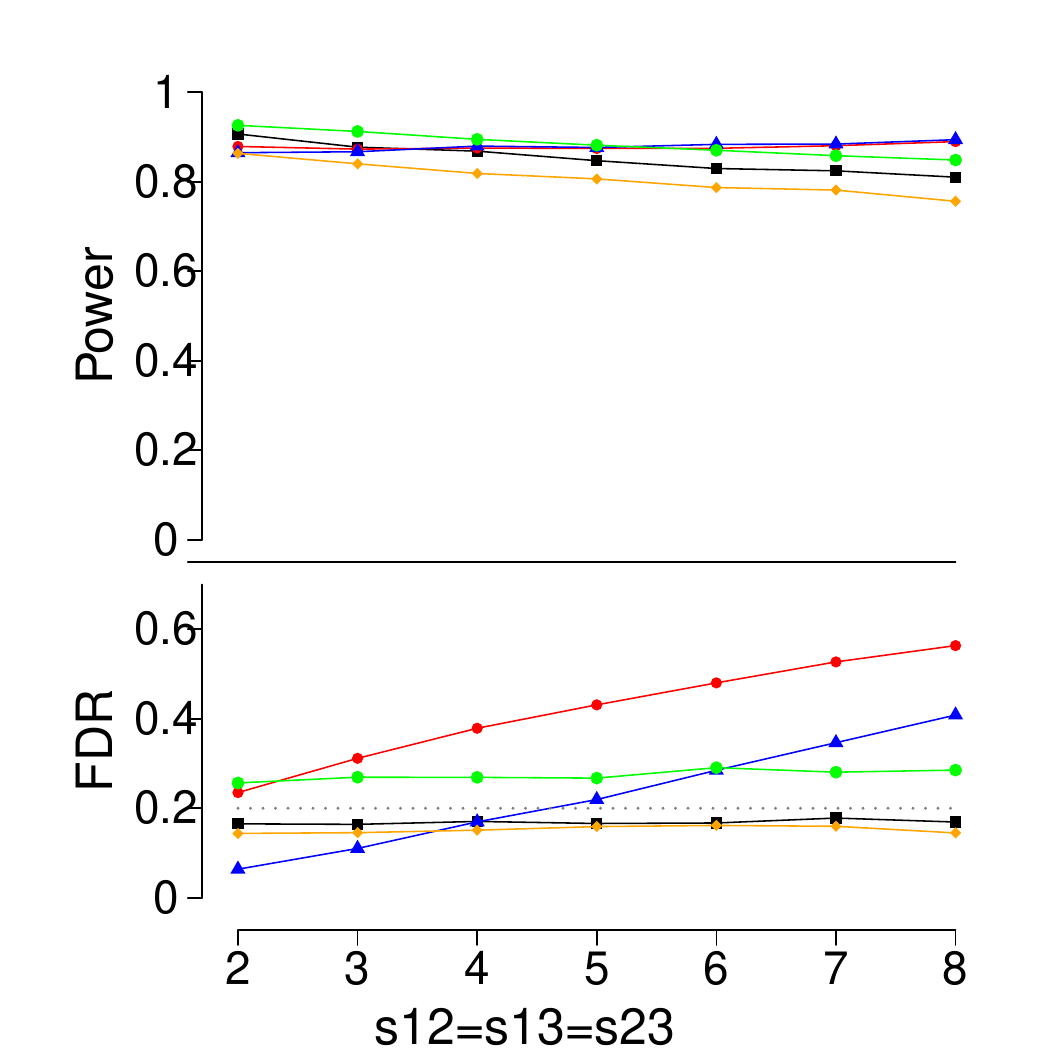}
    \includegraphics[scale=0.35]{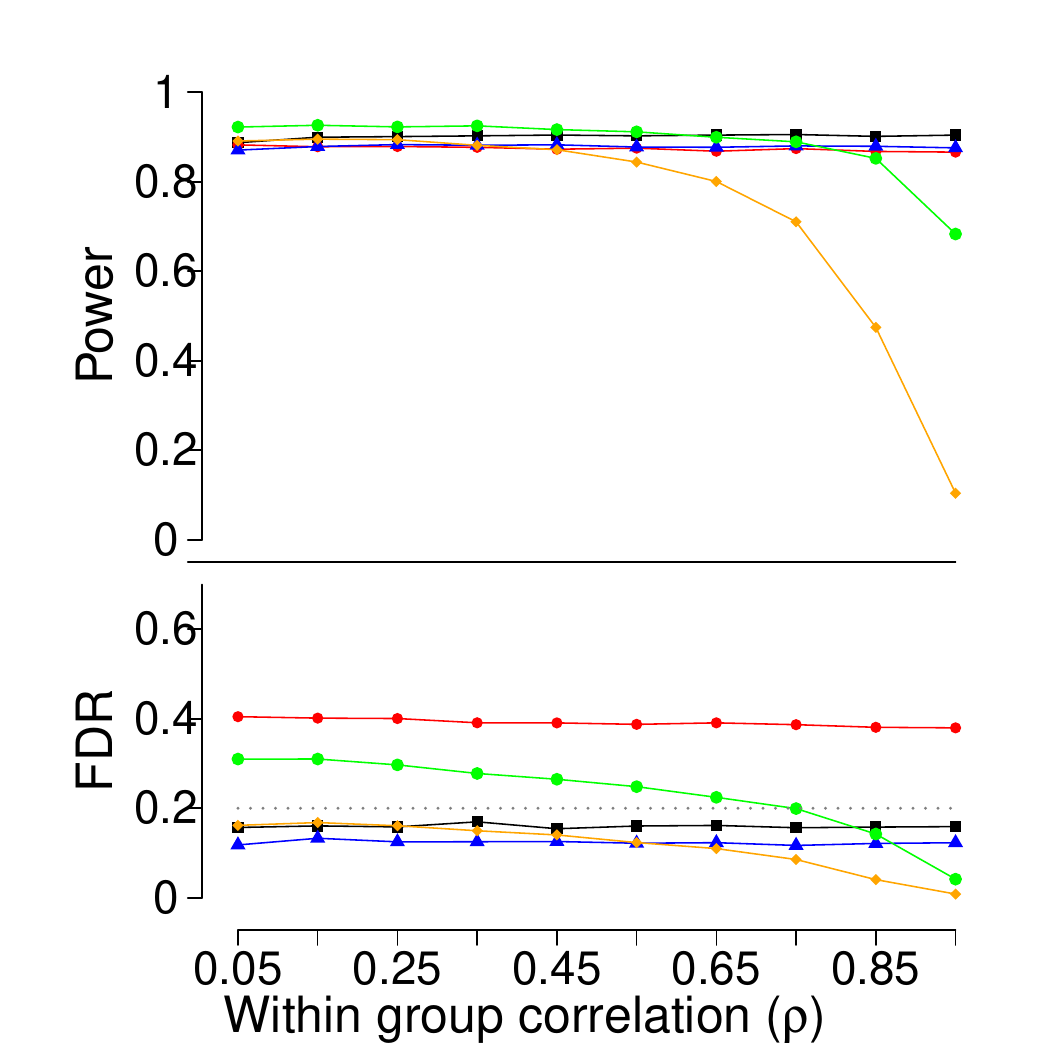}\\
    \includegraphics[scale=0.35]{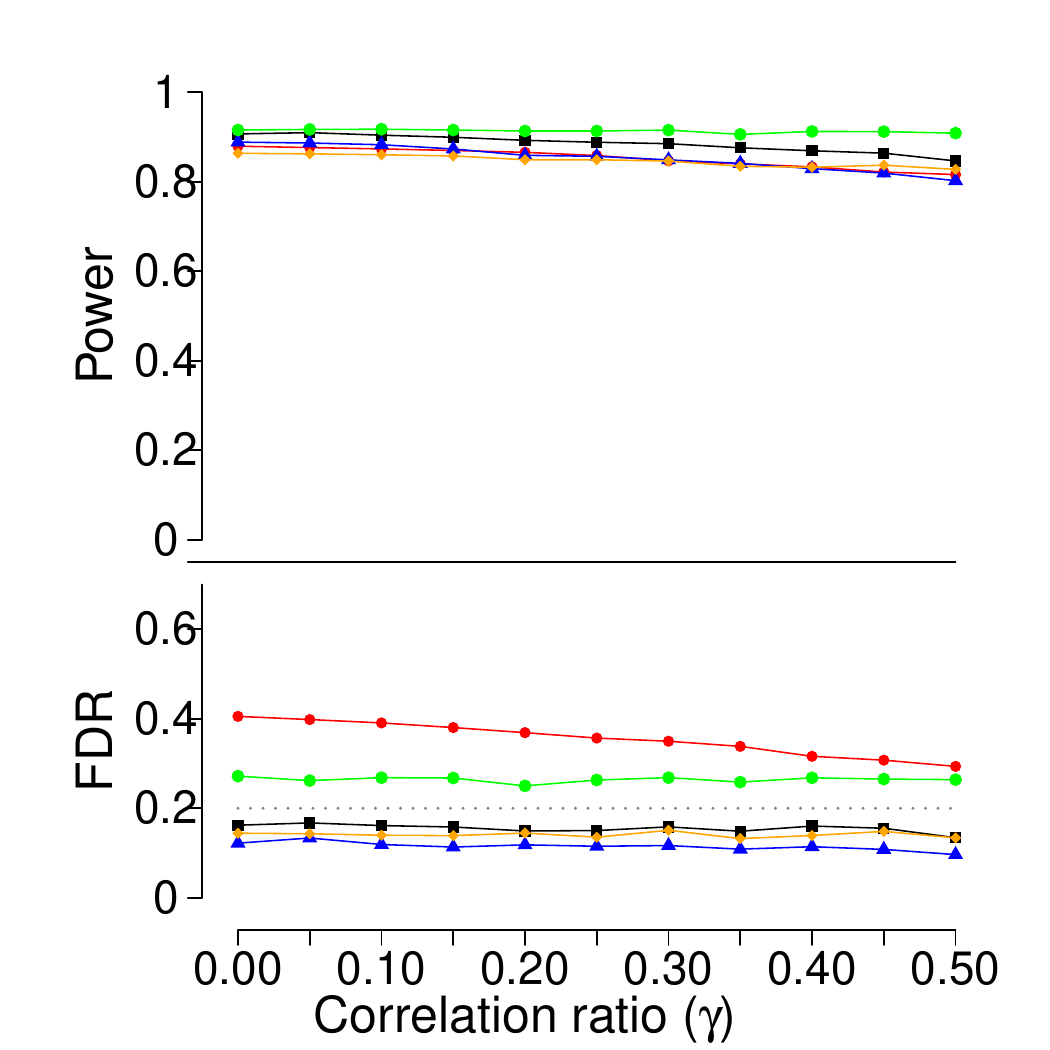}
    \includegraphics[scale=0.35]{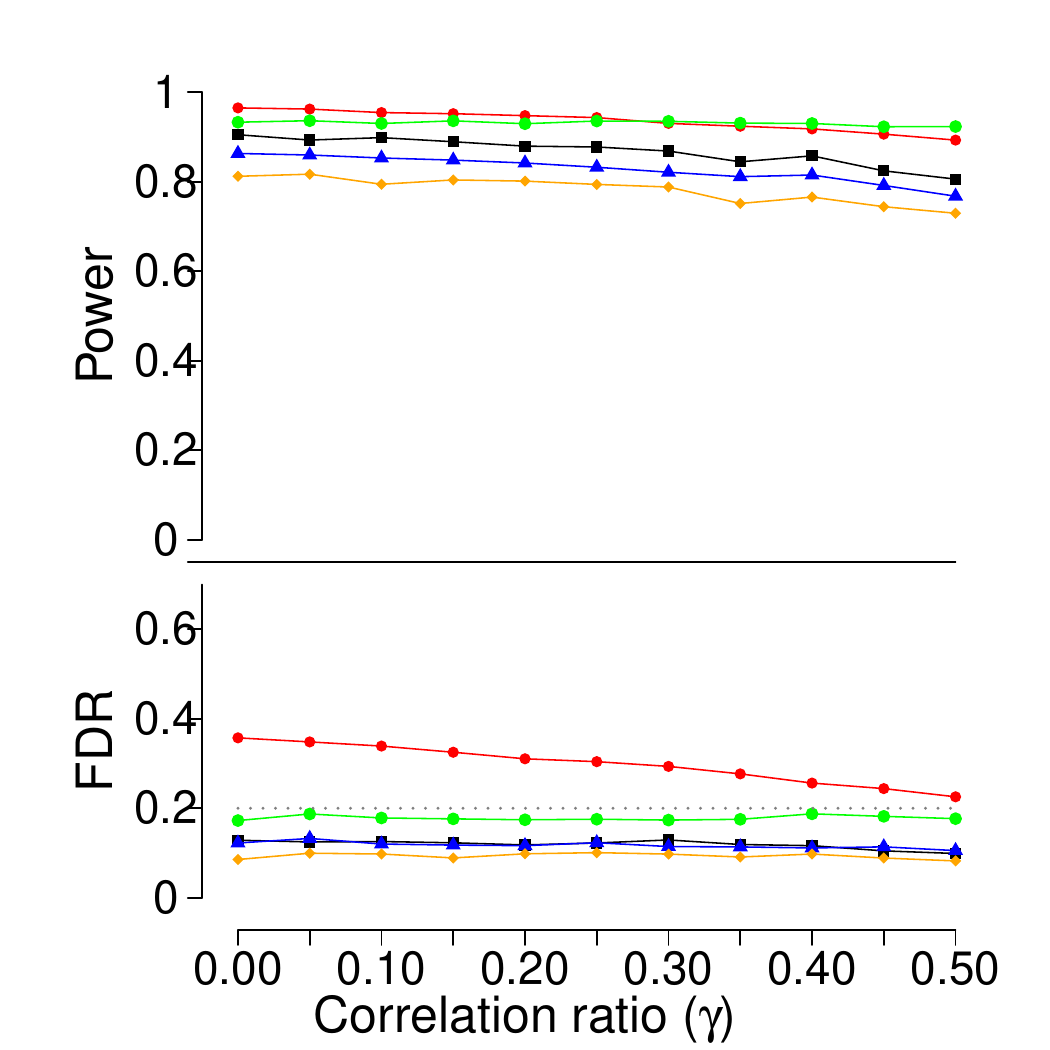}
    \caption{The power and the FDR for identifying group level simultaneous signals {\color{black} with data generated from \textbf{Setting 1} {\color{black}(the same sample size=1000)} for the \textbf{Mixed} models (K=3)} when varying (a) $s_0$ (\textbf{Scenario 1}); (b) $s_1=s_2=s_3$ (\textbf{Scenario 1}); (c) $s_{13}=s_{12}=s_{23}$ (\textbf{Scenario 1}); (d) within-group correlation $\rho$ (\textbf{Scenario 1}, \textbf{choice 2}); (e) Correlation ratio $\gamma$ (\textbf{Scenario 1}, \textbf{choice 2}); (f) Correlation ratio $\gamma$ (\textbf{Scenario 2}, \textbf{choice 2}). Details on the parameter settings for different \textbf{Scenarios} and \textbf{choices} are in Web Appendix E.}
    \label{fig:figure1}
\end{figure}

\pagebreak
\begin{figure}
    \centering
    \includegraphics[scale=0.35]{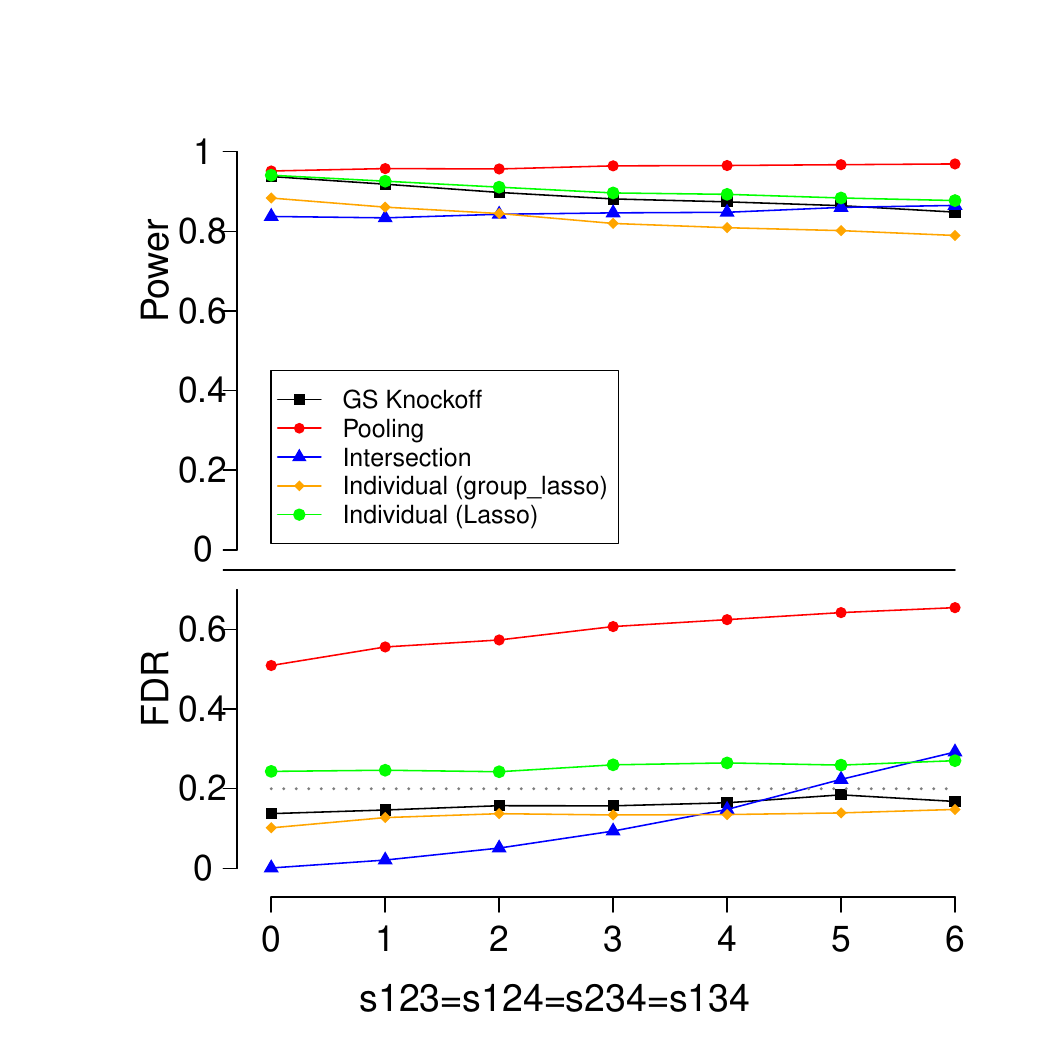}
    \includegraphics[scale=0.35]{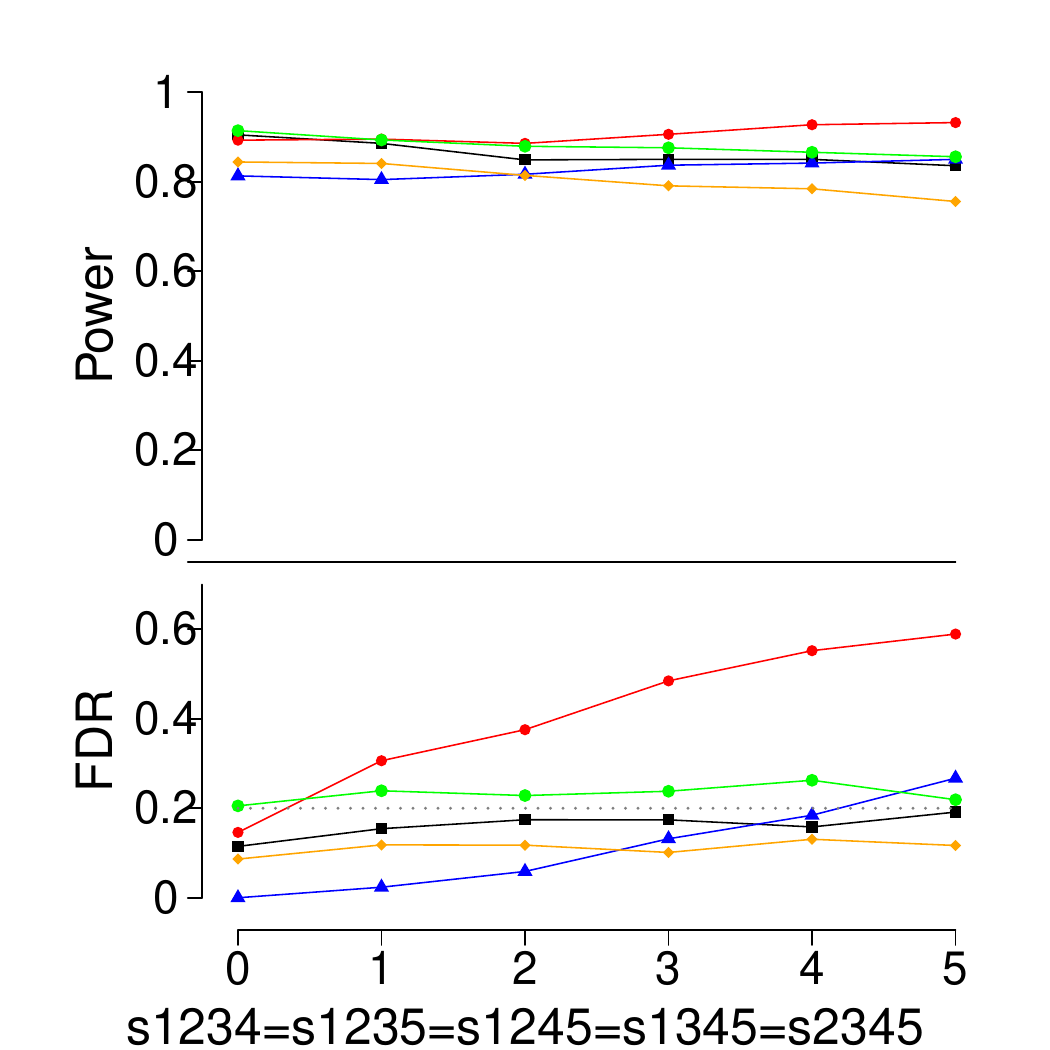}\\
    \includegraphics[scale=0.35]{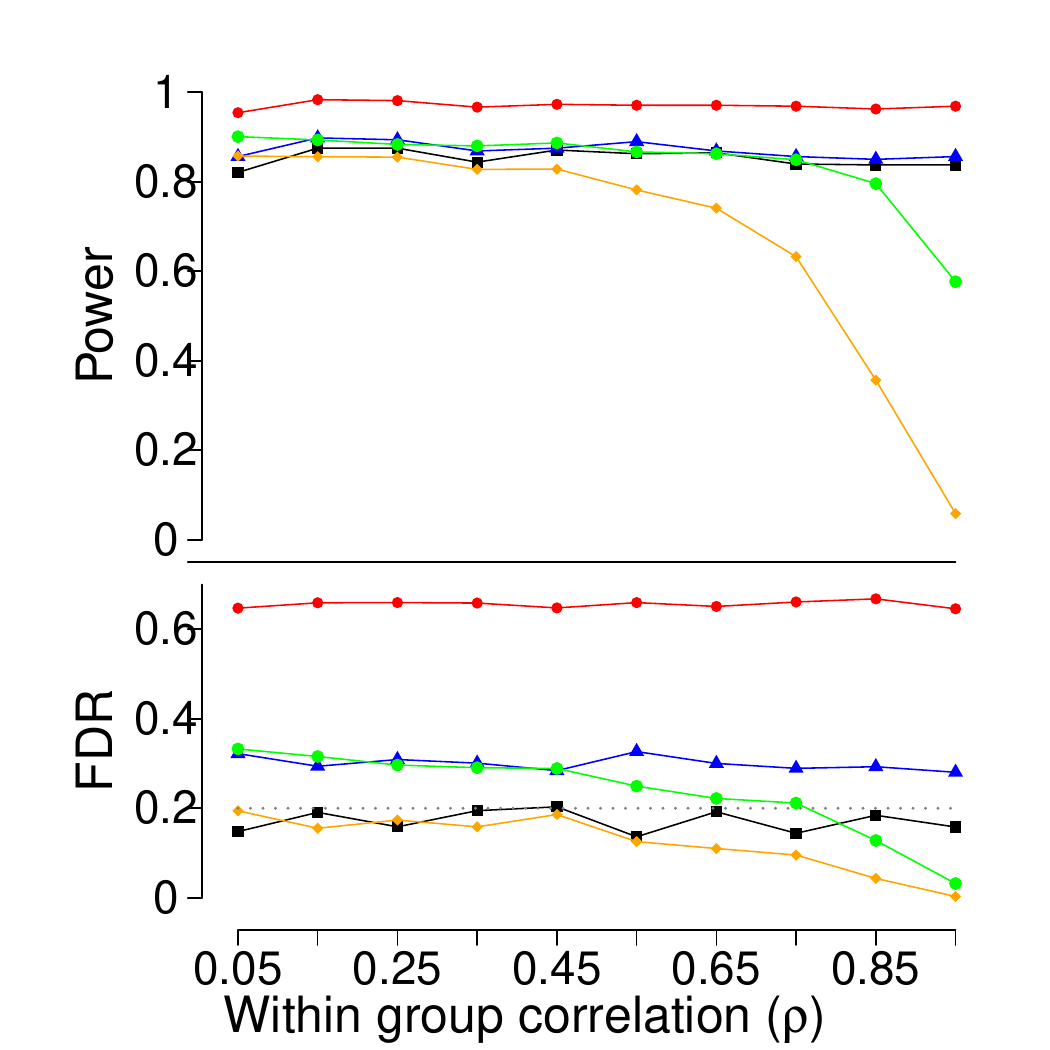}
    \includegraphics[scale=0.35]{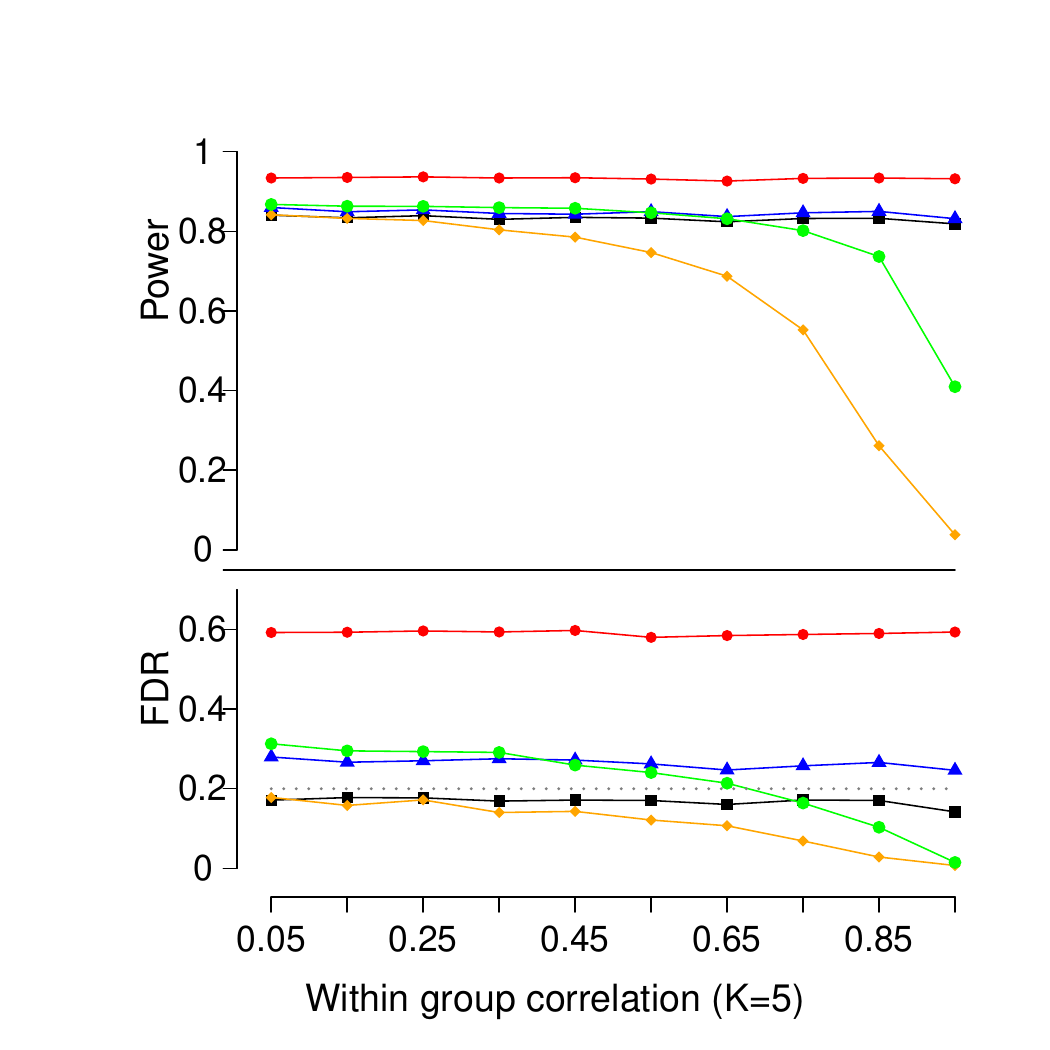}\\
    \includegraphics[scale=0.35]{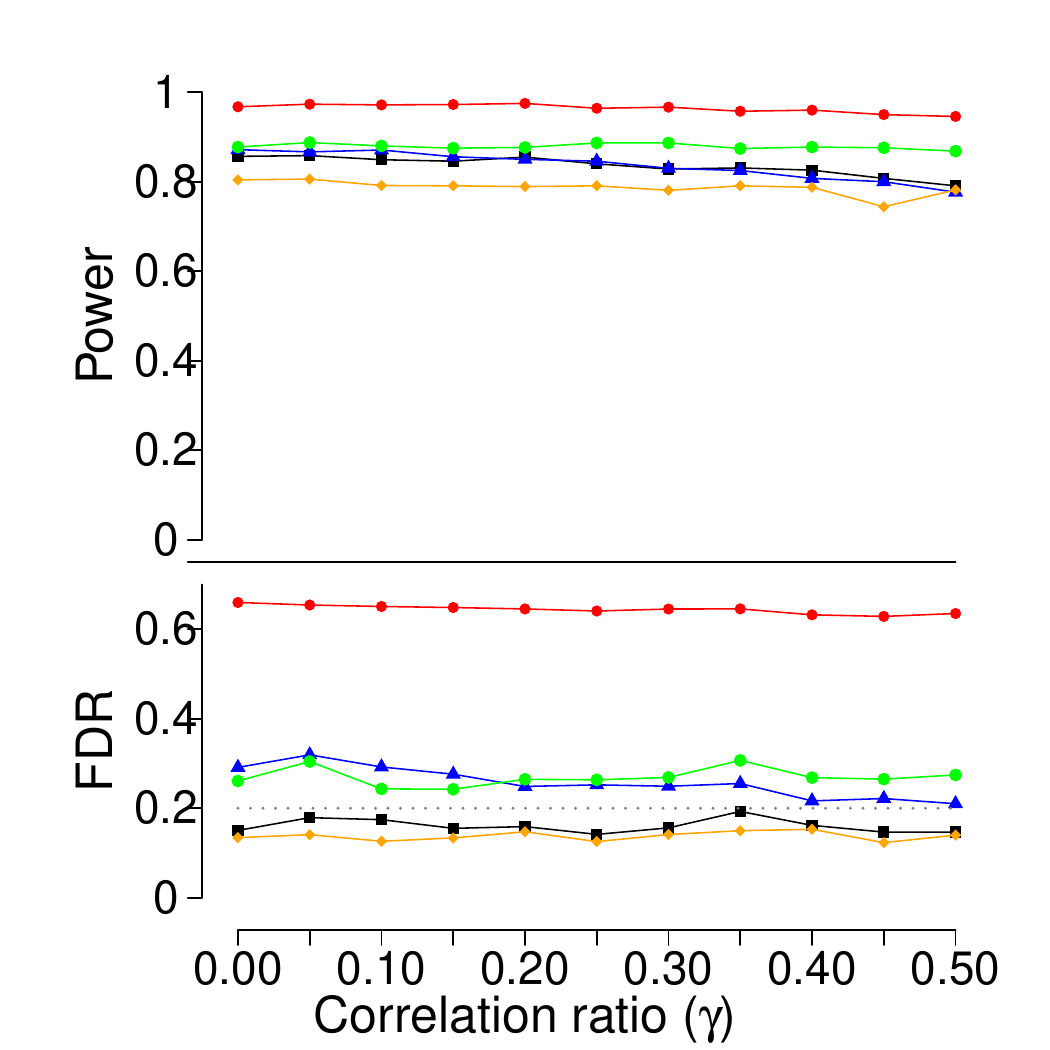}
    \includegraphics[scale=0.35]{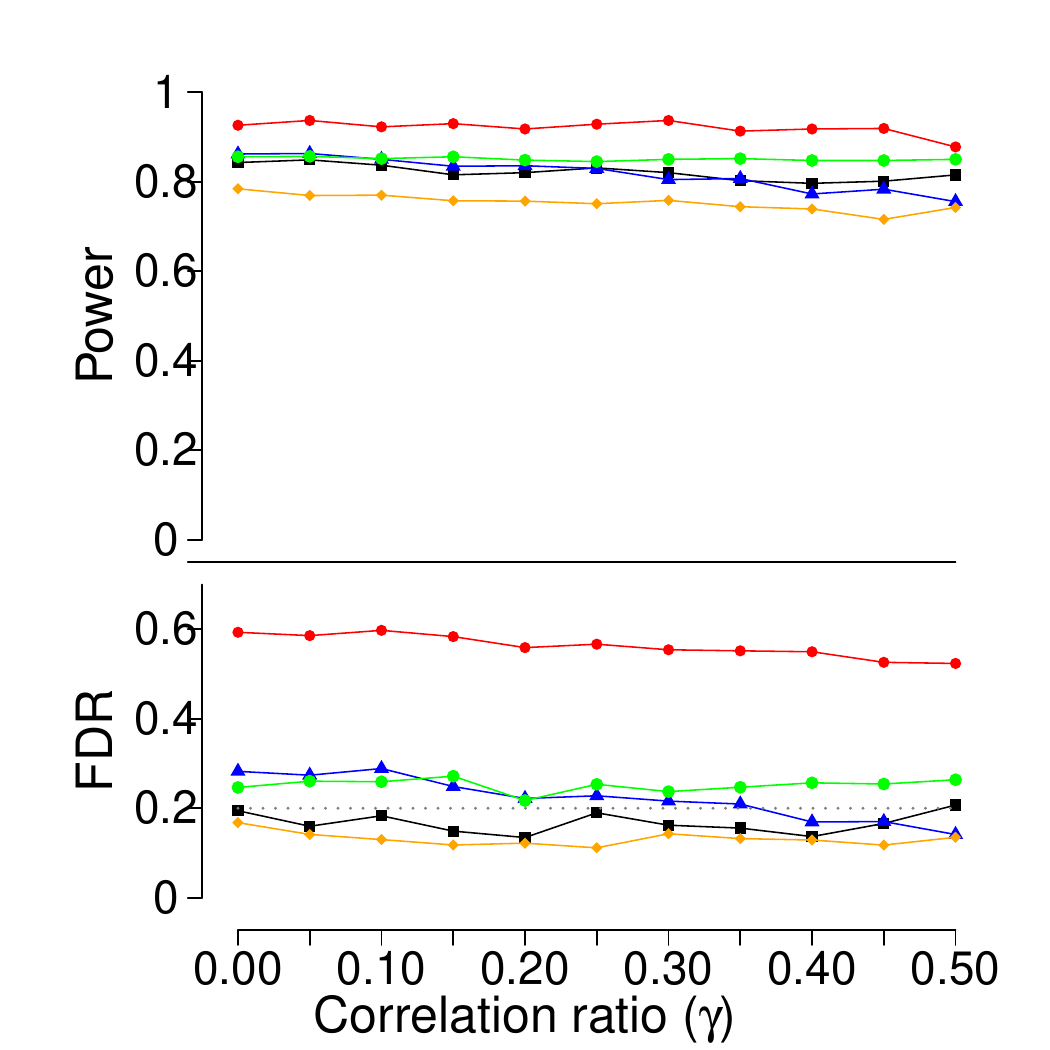}
    \caption{The power and the FDR for identifying group level simultaneous signals with {\color{black} data generated from \textbf{Setting 1} {\color{black}(the same sample size=1000)} for the \textbf{Mixed} models} on \textbf{Scenario 1} (same strengths) when K=4 (left column) and K=5 (right column). More details on the parameter settings are in Web Appendix E.}
    \label{fig:figure2}
\end{figure}

\pagebreak
\begin{figure}
    \centering
    \includegraphics[scale=0.35]{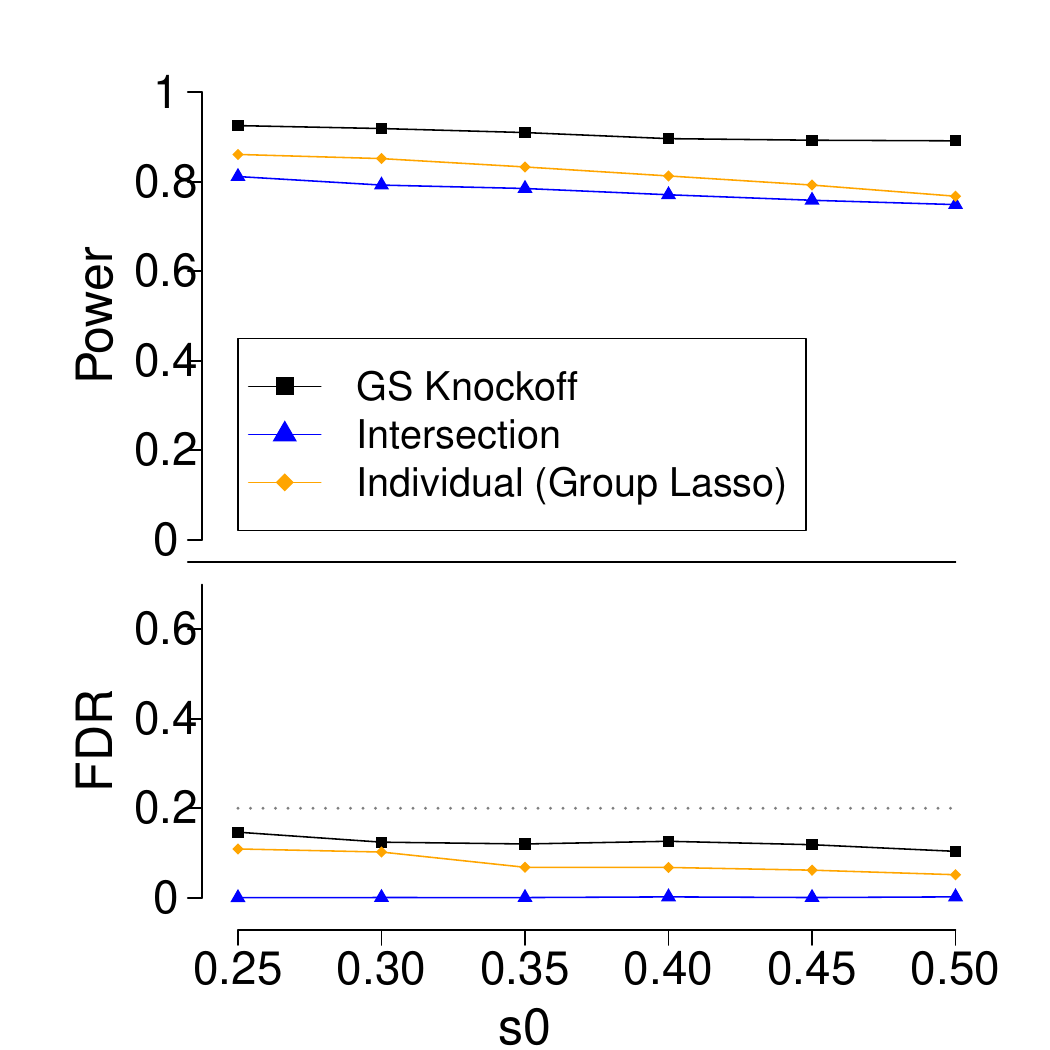}
    \includegraphics[scale=0.35]{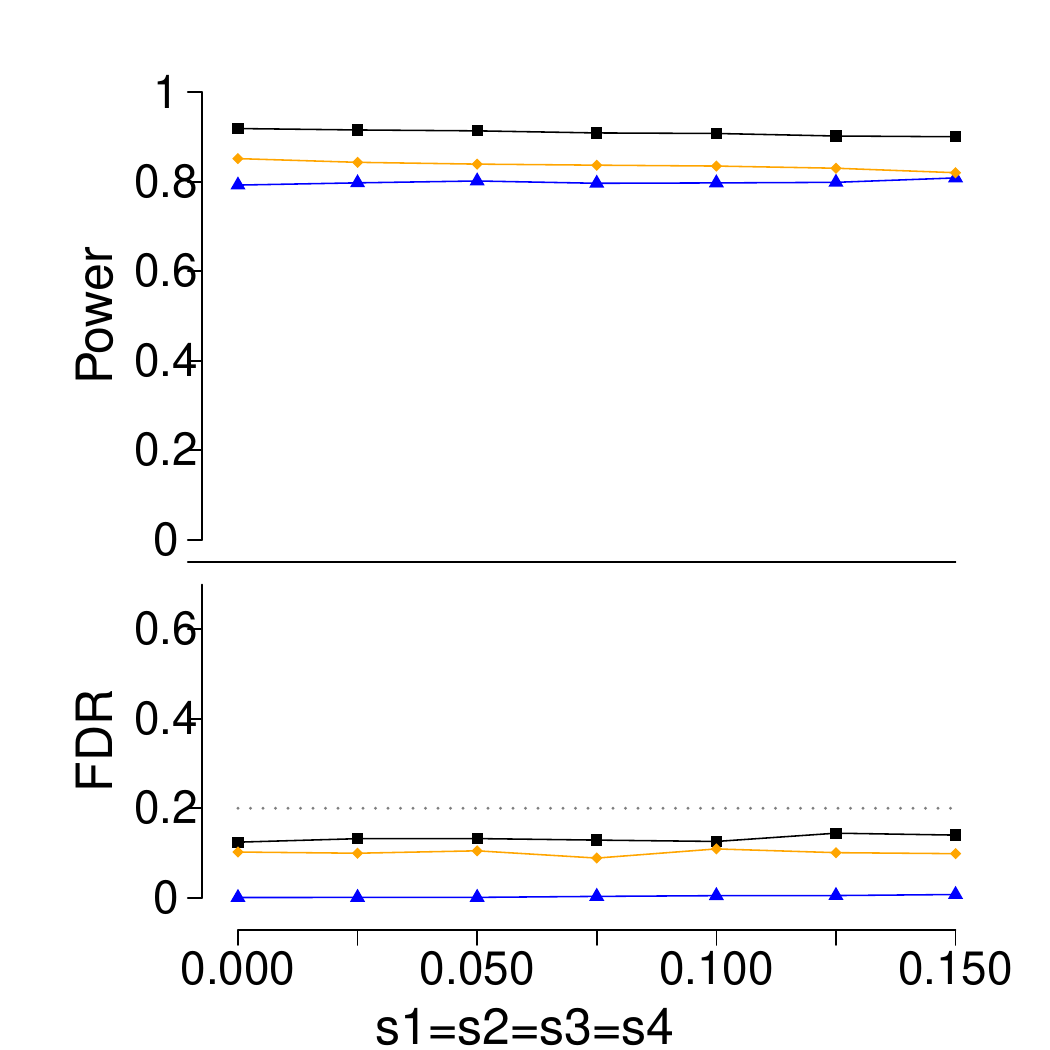}\\
    \includegraphics[scale=0.35]{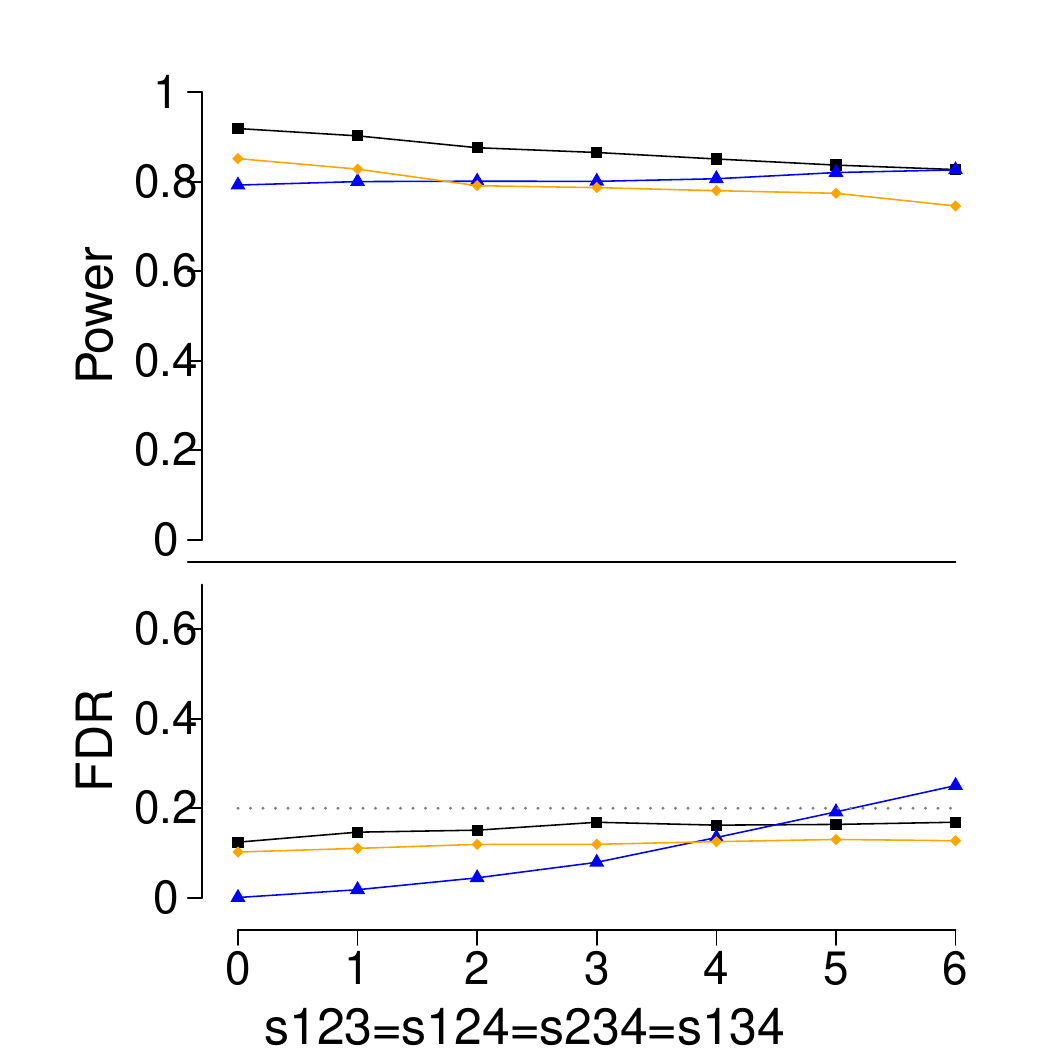}
    \includegraphics[scale=0.35]{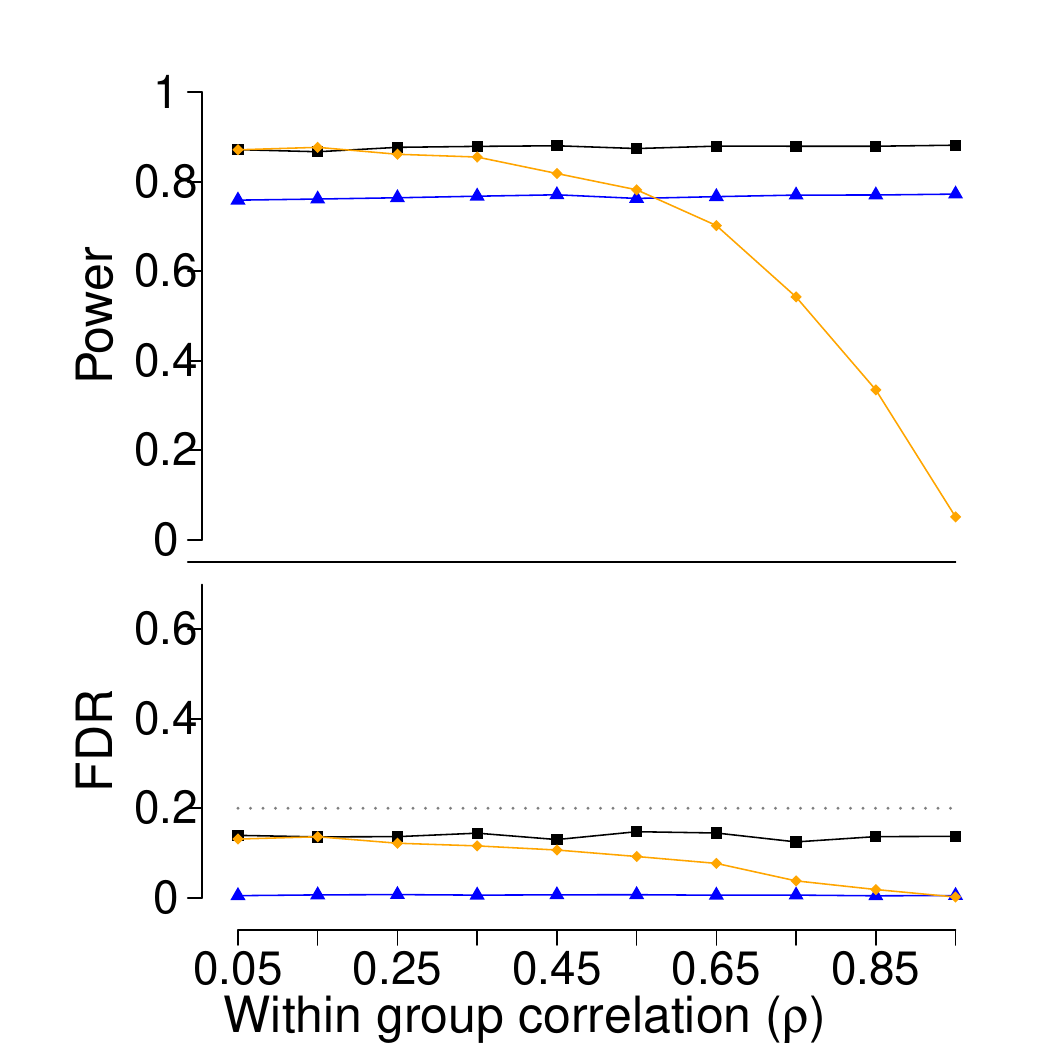}\\
    \includegraphics[scale=0.35]{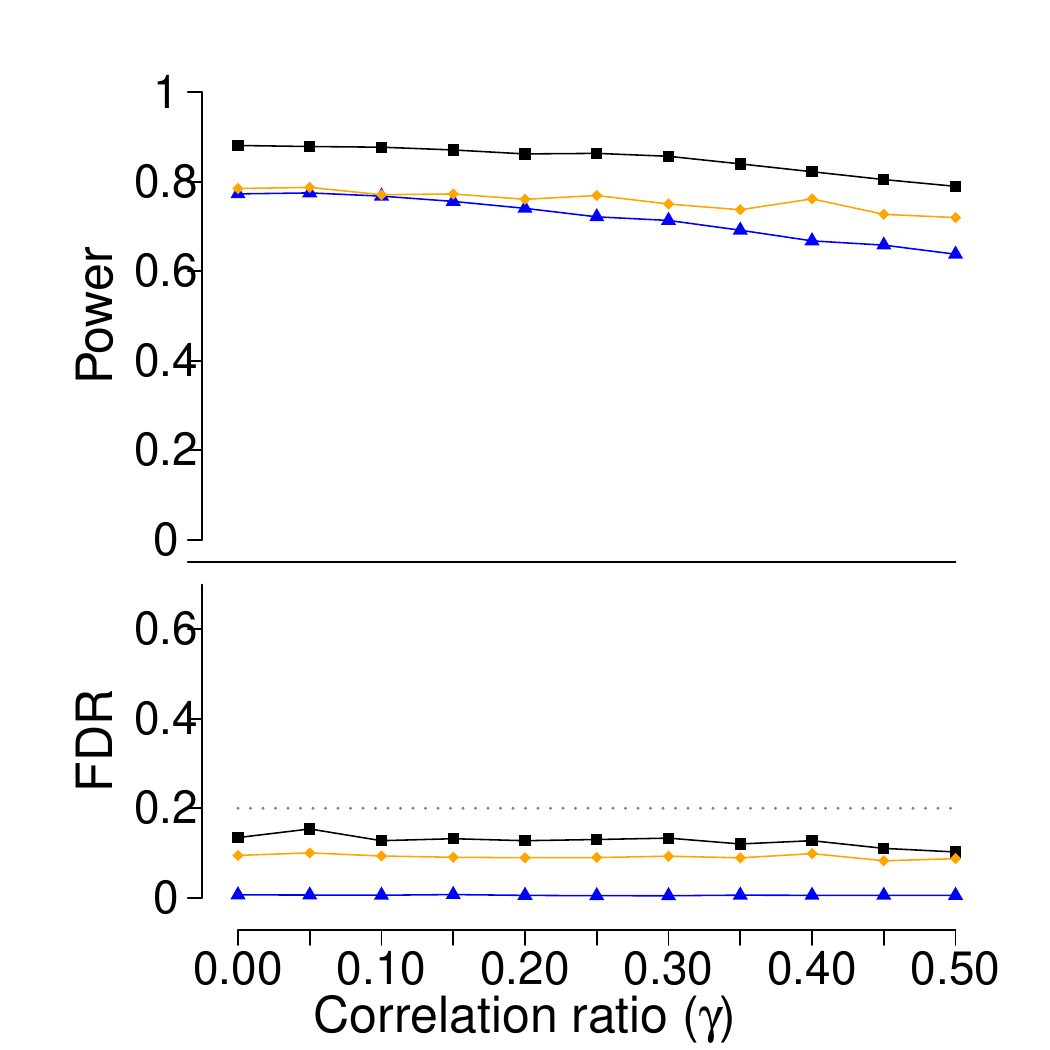}
    \includegraphics[scale=0.35]{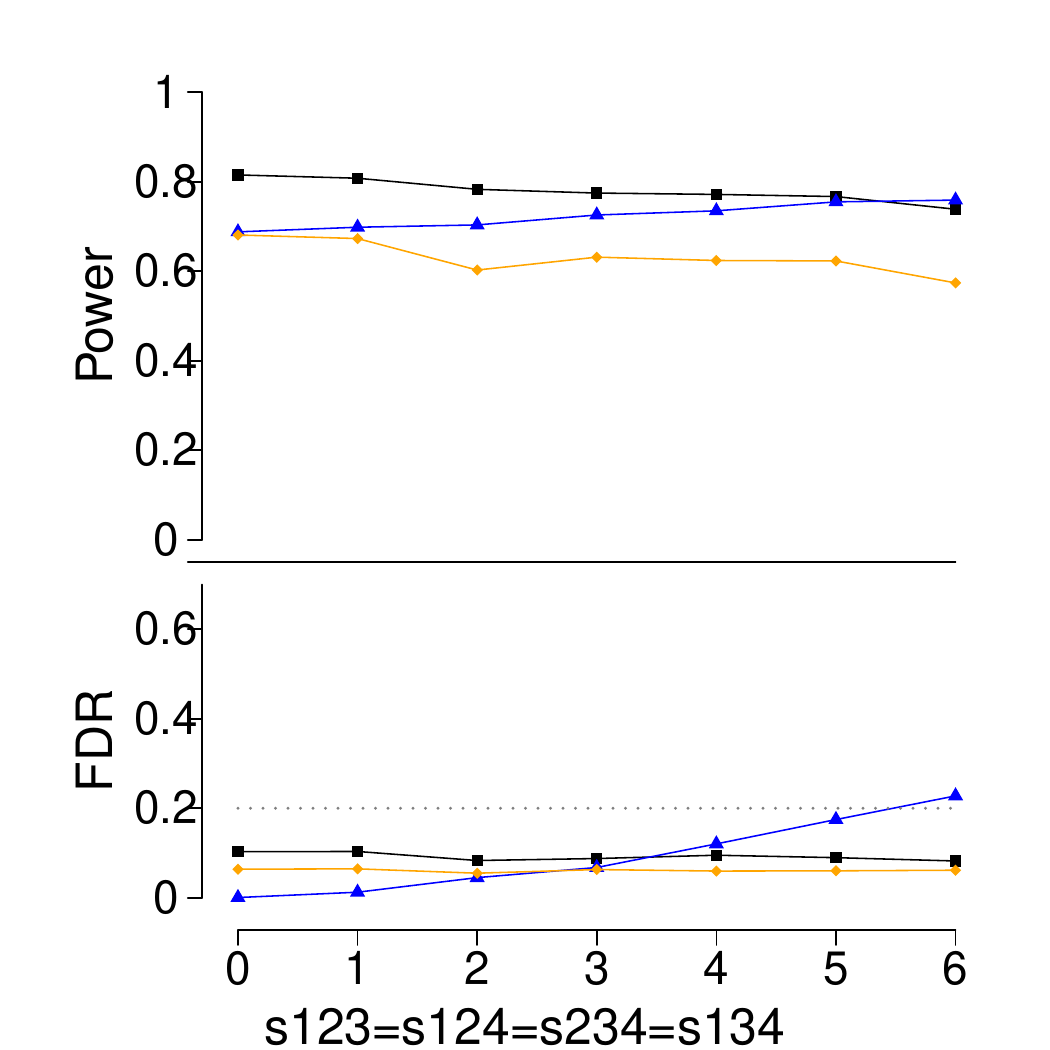}
       \caption{The power and the FDR for identifying group level simultaneous signals with {\color{black} with data generated from \textbf{Setting 2} (different numbers of variables and types of variables across the sites) for the \textbf{Mixed} models (K=4)} when varying (a) $s_0$ (Scenario 1); (b) $s_1=s_2=s_3=s_4$ (\textbf{Scenario 1}); (c) $s_{123}=s_{124}=s_{134}=s_{234}$ (\textbf{Scenario 1}); (d) within-group correlation $\rho$ (\textbf{Scenario 1}, \textbf{choice 2}); (e) Correlation ratio $\gamma$ (\textbf{Scenario 1}, \textbf{choice 2}); (f) $s_{123}=s_{124}=s_{134}=s_{234}$ (\textbf{Scenario 2}). Details on the parameter settings for different \textbf{Scenarios} and \textbf{choices} are in Web Appendix E.}
    \label{fig:figure3}
\end{figure}

\pagebreak
\begin{figure}
    \centering
    \includegraphics[scale=0.5]{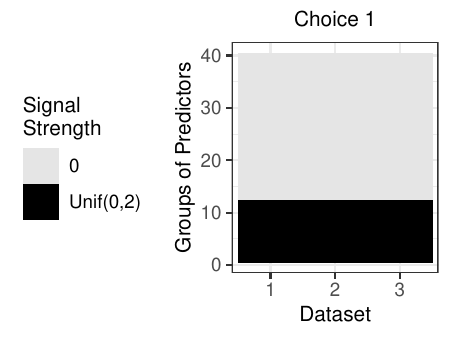}
    \includegraphics[scale=0.5]{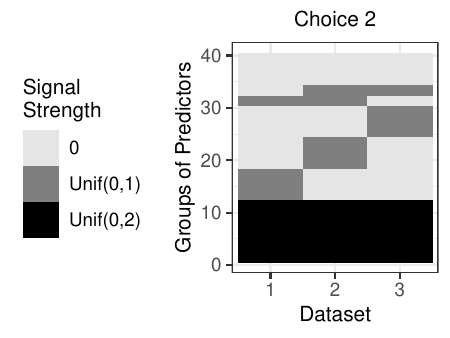}
    \includegraphics[scale=0.5]{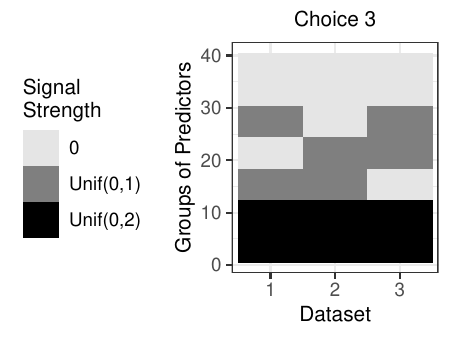}\\
    \includegraphics[scale=0.5]{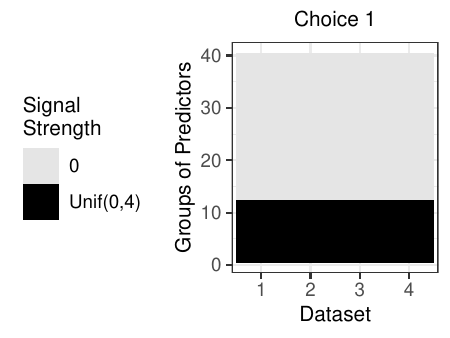}
    \includegraphics[scale=0.5]{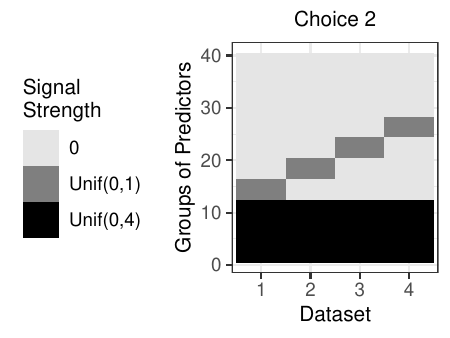}
    \includegraphics[scale=0.5]{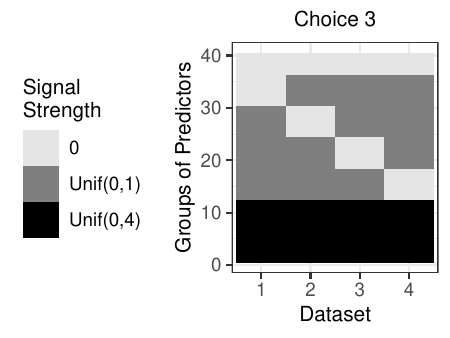}\\
    \includegraphics[scale=0.5]{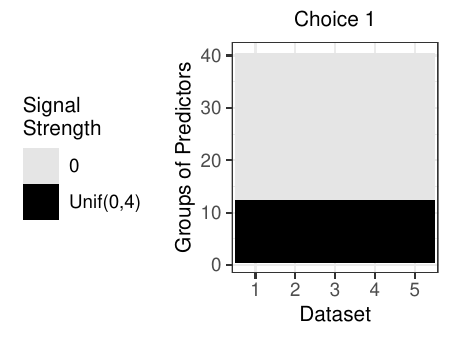}
    \includegraphics[scale=0.5]{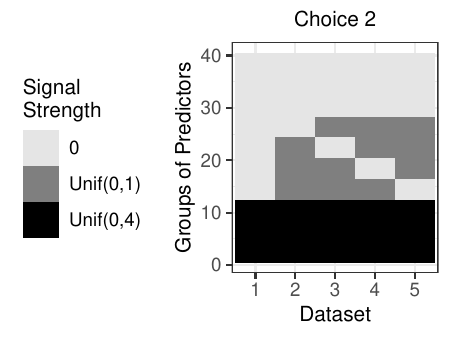}
    \includegraphics[scale=0.5]{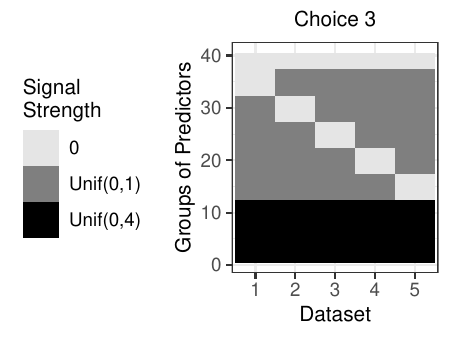}
    \caption{Design structure for coefficients when K=3 (row 1) for \textbf{choice 1} (only simultaneous signals exist), \textbf{choice 2} (simultaneous signals and non-simultaneous signals exist in one dataset and two datasets), and \textbf{choice 3} (simultaneous signals and non-simultaneous signals exist in two datasets); when K=4 (row 2) for \textbf{choice 1} (only simultaneous signals exist), \textbf{choice 2} (simultaneous signals and non-simultaneous signals exist in one dataset), and \textbf{choice 3} (simultaneous signals and non-simultaneous signals exist in three datasets); when K=5 (row 3) for \textbf{choice 1}(only simultaneous signals exist), \textbf{choice 2} (simultaneous signals and non-simultaneous signals exist in three datasets), and \textbf{choice 3} (simultaneous signals and non-simultaneous signals exist in four datasets). The black shades are simultaneous signals. The dark grey shades are signals that exist in partial of the datasets. The light grey areas are non-signals in all datasets.}
\label{fig:design}
\end{figure}
\end{document}


\maketitle

\section*{Web Appendix A: Technical Lemmas}\label{app:A}
\begin{lemma}\label{lem:gsknock}
For the $\widetilde{\X}$ generated from Algorithm 1, it satisfies the group Model-X knockoff requirements.
\end{lemma}
\begin{proof}
 {\color{black} The proof is a straightforward extension of Section E in \cite{candes2018} to the group-level signal selection case.} Since the generation is without looking at $Y$, so the second condition holds automatically. Here we just need to verify the first condition. Since $(\X_{G_1},\widetilde{\X}_{G_1})|\X_{-G_1}\eqd (\widetilde{\X}_{G_1},\X_{G_1})|\X_{-G_1}$, we have $(\X,\widetilde{\X}_{G_1})\eqd (\X,\widetilde{\X}_{G_1})_{\textnormal{GSwap}(\{1\},G)}$. Now, we use proof by induction similar as \citep{candes2018} to show that $(\X,\widetilde{\X}_{G_1},\dots,\widetilde{\X}_{G_m})\eqd (\X,\widetilde{\X}_{G_1},\dots,\widetilde{\X}_{G_m})_{\textnormal{GSwap}(S,G)}$ for any $S\in [m]$ after $m$ steps. Since the swap of multiple groups can be performed in multiple steps with a swap of one group at a time, we just need to prove the setting for $S=\{j\}$ for $j\in [m]$. Notice the measure for the joint distribution can be written as
\begin{eqnarray*}
dP(\X,\widetilde{\X}_{\cup_{l=1}^m G_l})&=&dP(\X,\widetilde{\X}_{\cup_{l=1}^{m-1} G_l})\frac{dP(\X,\widetilde{\X}_{{\color{black} G_m}},\widetilde{\X}_{\cup_{l=1}^{m-1} G_l})}{dP(\X,\widetilde{\X}_{\cup_{l=1}^{m-1} G_l})}\\
&=&dP(\X_{-{\color{black}G_m}},\X_{{\color{black}G_m}},\widetilde{\X}_{\cup_{l=1}^{m-1} G_l})\frac{dP(\X_{-{\color{black} G_m}},\widetilde{\X}_{{\color{black}G_m}},\widetilde{\X}_{\cup_{l=1}^{m-1} G_l})}{\int dP(\X_{-{\color{black} G_m}},du,\widetilde{\X}_{\cup_{l=1}^{m-1} G_l})},
\end{eqnarray*}
{\color{black} where the integration is over $du$.} When $j=m$, since changing $\X_{G_m}$ and $\widetilde{\X}_{G_m}$ will lead to the same numerator and the denominator does not depend on $\X_{G_m}$ and $\widetilde{\X}_{G_m}$, the group exchangeability property holds. When $j<m$, by induction, the function $dP(\cdot)$ is symmetric between $\X_{G_j}$ and $\widetilde{\X}_{G_j}$ and thus the group exchangeability also holds.

Now after the $M$ steps, we get the knockoff that satisfies condition 1 of the group Model-X knockoff construction.
\end{proof}

\begin{lemma}\label{lem:gseqknock}
For the $\widetilde{\X}$ generated from the sequential group knockoff, it satisfies the group Model-X knockoff requirements when the model is correctly specified and the true parameters are used.
\end{lemma}
\begin{proof}
When the model is correctly specified and true parameters are used, we have
\begin{eqnarray*}
\widetilde{\X}_{G_m}^{con}|\X_{-G_m}, \widetilde{\X}_{\cup_{j=1}^{m-1}G_j}\eqd \X_{G_m}^{con}|\X_{-G_m}, \widetilde{\X}_{\cup_{j=1}^{m-1}G_j}\\
\widetilde{\X}_{G_m}^{cat}|\widetilde{\X}_{G_m}^{{\color{black}con}}{\color{black}=x},\X_{-G_m}, \widetilde{\X}_{\cup_{j=1}^{m-1}G_j}\eqd \X_{G_m}^{{\color{black}cat}}|{\color{black}\X}_{G_m}^{\color{black}con}{\color{black}=x},\X_{-G_m}, \widetilde{\X}_{\cup_{j=1}^{m-1}G_j}
\end{eqnarray*}
which together implies
\begin{eqnarray*}
\widetilde{\X}_{G_m}|\X_{-G_m}, \widetilde{\X}_{\cup_{j=1}^{m-1}G_j}\eqd \X_{G_m}|\X_{-G_m}, \widetilde{\X}_{\cup_{j=1}^{m-1}G_j},
\end{eqnarray*}
and this follows the general group Model-X knockoff generation procedure. So applying Lemma 1 finishes the proof.
\end{proof}

The \textit{GS knockoff} procedure robustness against the misspecification of the distribution of $\X$. In real applications, when we have additional samples of $\X$ (for estimating the distribution of $\X$), we will be able to approximate the $\X$ distribution well. Theorem 2 of \citep{dai2021multiple} can be easily extended to show an FDR upper bound result for \textit{GS knockoffs}.


\begin{lemma}\label{lem:1}
Let $\W = f([\Z^1,\widetilde{\Z}^1],\cdots,[\Z^K,\widetilde{\Z}^K])$ where $f$ is an OSFF. Let $\epsilon \in \{\pm 1\}^M$ {\color{black} be an arbitrary sign sequence 
with $\epsilon_j = +1$ for all $j \in \mathcal{S}$ and $\epsilon_j \in \{\pm 1\}$} for all $j \in \mathcal{H}$. Then
$(W_1, \cdots, W_M) \eqd (W_1\cdot \epsilon_1, \cdots, W_M\cdot\epsilon_M)$. 
\end{lemma}
\begin{proof}
For any {\color{black} $V\subseteq \mathcal{H}$}, we can write it as the union of $K$ subsets {\color{black}$V=\cup_{k=1}^K V_k$, where $V_k\subseteq \mathcal{H}_k$} for $k = 1,\cdots,K$, and {\color{black}$V_{k1}\cap V_{k2} = \emptyset$} for all $k_1 \neq k_2$. In particular, we can let {\color{black}$V_k=S\cap \mathcal{H}_k\cap (\cup_{j=1}^{k-1}\mathcal{H}_j)^c$}. 
Since {\color{black}$V_k\subseteq \mathcal{H}_k$}, for $k\in [K]$, any statistics $[\Z^k, \widetilde{\Z}^k]=w([\X^k,\widetilde{\X}^k],\Y^k)$, because of Lemmas 1 and 2, the construction of knockoffs from Algorithms 1 and 2, satisfy 
$[\Z^k, \widetilde{\Z}^k] \eqd [\Z^k, \widetilde{\Z}^k]_{\textnormal{Swap}({\color{black}V}_k)} $. By the mutually independence between $[Z^1,\widetilde{Z}^1],\cdots, [Z^K,\widetilde{Z}^K]$, we have 

\[f\left([\Z^1,\widetilde{\Z}^1]_{\textnormal{Swap}({\color{black}V}_1)},\cdots,[\Z^K,\widetilde{\Z}^K]_{\textnormal{Swap}({\color{black}V}_K)}\right)\eqd f\left([\Z^1,\widetilde{\Z}^1],\cdots,[\Z^K,\widetilde{\Z}^K]\right).\] 

Using the definition of the OSFF, we have
\begin{eqnarray*}
&&f\left([\Z^1,\widetilde{\Z}^1]_{\textnormal{Swap}({\color{black}V}_1)},[\Z^2,\widetilde{\Z}^2]_{\textnormal{Swap}({\color{black}V}_2)},\cdots,[\Z^K,\widetilde{\Z}^K]_{\textnormal{Swap}({\color{black}V}_K)}\right)\\
&=&f\left([\Z^1,\widetilde{\Z}^1],[\Z^2,\widetilde{\Z}^2]_{\textnormal{Swap}({\color{black}V}_2)},\cdots,[\Z^K,\widetilde{\Z}^K]_{\textnormal{Swap}({\color{black}V}_K)}\right)\odot \epsilon({\color{black}V}_1)\\
&=&\cdots\\
&=&f\left([\Z^1,\widetilde{\Z}^1],[\Z^2,\widetilde{\Z}^2],\cdots,[\Z^K,\widetilde{\Z}^K]\right)\odot_{k=1}^K \epsilon({\color{black}V}_k)\\
&=&f\left([\Z^1,\widetilde{\Z}^1],[\Z^2,\widetilde{\Z}^2],\cdots,[\Z^K,\widetilde{\Z}^K]\right)\odot\epsilon({\color{black}V}).
\end{eqnarray*}

So we obtain \[\W=f\left([\Z^1,\widetilde{\Z}^1],\cdots,[\Z^K,\widetilde{\Z}^K]\right)\eqd f\left([\Z^1,\widetilde{\Z}^1],\cdots,[\Z^K,\widetilde{\Z}^K]\right)\odot\epsilon({\color{black}V})=\W\odot\epsilon({\color{black}V}).\] for any ${\color{black}V}\subseteq \mathcal{H}$. Therefore, by choosing ${\color{black}V}$ as the set $\{j:\epsilon_j=-1\}$, we have \[(W_1, \cdots, W_M) \eqd (W_1\cdot \epsilon_1, \cdots, W_M\cdot\epsilon_M).\] and thus we finish the proof of the lemma. {\color{black} This lemma implies that conditional on $(|W_1|,\cdots,|W_M|)$, the sign of $W_j$ for all $j\in \mathcal{H}$ i.i.d. $\sim{\pm 1}$ with equal probability of being 1 and being -1.}
\end{proof}

\begin{lemma}\label{lem:mar}
Assume $p_j\geq \textnormal{Uniform}[0,1]$ are i.i.d. for all nulls and are independent from non-nulls; {\color{black}that is, for all null $j$ and all $u \in$ [0,1], $P(p_j \leq u) \leq u$}. For ${\color{black}      o}=m,m-1,\cdots,1,0$, put $V^{+}({\color{black} o})=\#\{j:1\leq j\leq {\color{black} o}, p_j\leq 1/2, j\in \mathcal{H}\}$ and $V^{-}({\color{black} o})=\#\{ j:1\leq j\leq {\color{black} o}, p_j> 1/2, j\in \mathcal{H}\}$ with the convention that $V^{\pm}(0)=0$. Let $\mathcal{F}_{\color{black} o}$ be the filtration defined by knowing all the non-null $p$-values, as well as $V^{\pm}({\color{black} o}')$ for all ${\color{black} o}'\geq {\color{black} o}$. Then the process $M({\color{black} o})=\frac{V^{+}({\color{black} o})}{1+V^{-}({\color{black} o})}$ is a super-martingale running backward in time with respect to $\mathcal{F}_{\color{black} o}$. For any fixed $q$, $\widehat{{\color{black} o}}=\widehat{{\color{black} o}}_{+}$ or $\widehat{{\color{black} o}}=\widehat{{\color{black} o}}_{0}$ as defined in the proof of Theorem 1 are stopping times, and as consequences
\begin{eqnarray*}
\EE{\frac{\#\{j\leq \widehat{{\color{black} o}}:p_j\leq 1/2, j\in \mathcal{H}\}}{1+\#\{j\leq \widehat{{\color{black} o}}:p_j>1/2, j\in \mathcal{H}\}}}\leq 1
\end{eqnarray*}
\end{lemma}

\begin{proof}
The filtration $\mathcal{F}_{\color{black} o}$ contains the information of whether ${\color{black} o}$ is null and non-null process is known exactly. If ${\color{black} o}$ is non-null, then $M({\color{black} o}-1)=M({\color{black} o})$ and if ${\color{black} o}$ is null, we have
\begin{eqnarray*}
M({\color{black} o}-1)=\frac{V^{+}({\color{black} o})-\mathbbm{1}_{p_{\color{black} o}\leq 1/2}}{1+V^{-}({\color{black} o})-(1-\mathbbm{1}_{p_{\color{black} o}\leq 1/2})}=\frac{V^{+}({\color{black} o})-\mathbbm{1}_{p_{\color{black} o}\leq 1/2}}{\left(V^{-}({\color{black} o})+\mathbbm{1}_{p_{\color{black} o}\leq 1/2}\right)\vee 1}
\end{eqnarray*}
Given that nulls are \iid, we have
\begin{eqnarray*}
\PP{\mathbbm{1}_{p_{\color{black} o}\leq 1/2}|\mathcal{F}_{\color{black} o}}=\frac{V^{+}({\color{black} o})}{\left(V^{+}({\color{black} o})+V^{-}({\color{black} o})\right)}.
\end{eqnarray*}
So when ${\color{black} o}$ is null, we have
\begin{eqnarray*}
\EE{M({\color{black} o}-1)|\mathcal{F}_{\color{black} o}}&=&\frac{1}{V^{+}({\color{black} o})+V^{-}({\color{black} o})}\left[V^{+}({\color{black} o})\frac{V^{+}({\color{black} o})-1}{V^{-}({\color{black} o})+1}+V^{-}({\color{black} o})\frac{V^{+}({\color{black} o})}{V^{-}({\color{black} o})\vee 1}\right]\\
&=&\frac{V^{+}({\color{black} o})}{1+V^{-}({\color{black} o})}\mathbbm{1}_{V^{-}({\color{black} o})>0}+(V^{+}({\color{black} o})-1)\mathbbm{1}_{V^{-}({\color{black} o})=0}\\
&\leq &M({\color{black} o})
\end{eqnarray*}
This finishes the proof of super-martingale property. $\widehat{{\color{black} o}}$ is a stopping time with respect to $\{\mathcal{F}_{\color{black} o}\}$ since $\{\widehat{{\color{black} o}}\geq {\color{black} o}\}\in \mathcal{F}_{\color{black} o}$. So we have $\EE{M(\widehat{{\color{black} o}})}\leq \EE{M(m)}=\EE{\frac{\#\{j:p_j\leq 1/2, j\in \mathcal{H}\}}{1+\#\{j:p_j>1/2, j\in \mathcal{H}\}}}$.

Let $X=\#\{j:p_j\leq 1/2, j\in \mathcal{H}\}$, given that $p_j\geq \textnormal{Uniform}[0,1]$ independently for all nulls, we have $X\leq_d \textnormal{Binomial}(N,1/2)$. Let $Y\sim \textnormal{Binomial}(N,1/2)$ where $N$ is the total number of nulls. Given that $f(x)=\frac{x}{1+N-x}$ is non-decreasing, we have
\begin{eqnarray*}
\EE{\frac{X}{1+N-X}}&\leq &\EE{\frac{Y}{1+N-Y}}\\
&=&\sum_{i=1}^N (1/2)^{N}\frac{N!}{i!(N-i)!}\frac{i}{1+N-i}\\
&=&\sum_{i=1}^N \PP{Y=i-1}\\
&\leq &1.
\end{eqnarray*}
\end{proof}

\section*{Web Appendix B: Proof of the main theorem (Theorem 1
)} \label{sec:pf-thm1}
The proof of Theorem 1 follows the proof idea in \citep{barber2015}.
Let ${\color{black}l}=\#\{j:W_j\neq 0\}$ and assume without loss of generality that $|W_1|\geq |W_2|\geq\cdots\geq |W_{\color{black}l}|>0$. Define p-values $p_j=1$ if $W_j<0$ and $p_j=1/2$ if $W_j>0$, then Lemma 3 implies that null $p$-values are $i.i.d.$ with $p_j\geq \textnormal{Uniform}[0,1]$ and are independent from nonnulls. 

We first show the result for the knockoff+ threshold. Define $V=\#\{j\leq \widehat{k}_{+}: p_j\leq 1/2, j\in \mathcal{H}\}$ and $R=\#\{j\leq \widehat{k}_{+}:p_j\leq 1/2\}$ where $\widehat{k}_{+}$ satisfy that $|W_{\widehat{k}_{+}}|=\tau_{+}$ where $\tau_{+}$ is defined in {\color{black}equation (9) of the main paper}, we have
\begin{eqnarray*}
&&\EE{\frac{V}{R\vee 1}}=\EE{\frac{V}{R\vee 1}\mathbbm{1}_{\widehat{k}_{+}>0}}\\
&=&\EE{\frac{\#\{j\leq \widehat{k}_{+}: p_j\leq 1/2, j\in \mathcal{H}\}}{1+\#\{j\leq \widehat{k}_{+}: p_j> 1/2, j\in \mathcal{H}\}}\left(\frac{1+\#\{j\leq \widehat{k}_{+}: p_j> 1/2, j\in \mathcal{H}\}}{\#\{j\leq \widehat{k}_{+}:p_j\leq 1/2\}\vee 1}\right)\mathbbm{1}_{\widehat{k}_{+}>0}}\\
&\leq&\EE{\frac{\#\{j\leq \widehat{k}_{+}: p_j\leq 1/2, j\in \mathcal{H}\}}{1+\#\{ j\leq \widehat{k}_{+}: p_j> 1/2, j\in \mathcal{H}\}}}q\leq q,
\end{eqnarray*}
where the first inequality holds by the definition of $\widehat{k}_{+}$ and the second inequality holds by Lemma 4.

Similarly, for the knockoff threshold, we have $V=\#\{j\leq \widehat{k}_0: p_j\leq 1/2, j\in \mathcal{H}\}$ and $R=\#\{j\leq \widehat{k}_0:p_j\leq 1/2\}$ where $\widehat{k}_0$ satisfies that $|W_{\widehat{k}_0}|=\tau$ where $\tau$ is defined as in {\color{black}equation (8) of the main paper}, then
\begin{eqnarray*}
&& \EE{\frac{V}{R+q^{-1}}}\\
&=&\EE{\frac{\#\{ j\leq \widehat{k}_0: p_j\leq 1/2, j\in \mathcal{H}\}}{1+\#\{j\leq \widehat{k}_0: p_j> 1/2, j\in \mathcal{H}\}}\left(\frac{1+\#\{ j\leq \widehat{k}_0: p_j> 1/2, j\in \mathcal{H}\}}{\#\{j\leq \widehat{k}_0:p_j\leq 1/2\}+q^{-1}}\right)\mathbbm{1}_{\widehat{k}_0>0}}\\
&\leq&\EE{\frac{\#\{j\leq \widehat{k}_0: p_j\leq 1/2, j\in \mathcal{H}\}}{1+\#\{ j\leq \widehat{k}_0: p_j> 1/2, j\in \mathcal{H}\}}}q\leq q,
\end{eqnarray*}
where the first inequality holds by the definition of $\widehat{k}_0$ and the second inequality holds by Lemma 4.

\section*{Web Appendix C: Proof of Corollary 1
} \label{sec:pf-cor1}
\begin{proof}
First, we noticed that for $m\in S$,
\begin{eqnarray*}
W_m&=&f([\Z^1,\widetilde{\Z}^1],\cdots,[\Z^k,\widetilde{\Z}^k]_{\textnormal{Swap}(S)},\cdots,[\Z^K,\widetilde{\Z}^K])_m\\
&=&\prod_{j\in [K]\backslash k} (Z^j_m-\widetilde{Z}^j_m)(\widetilde{Z}_m^k-Z_m^k)\\
&=&-\prod_{j\in [K]} (Z^j_m-\widetilde{Z}^j_m)\\
&=&-f([\Z^1,\widetilde{\Z}^1],\cdots,[\Z^K,\widetilde{\Z}^K])_m,
\end{eqnarray*}
and for $m\notin S$
\begin{eqnarray*}
W_m&=&f([\Z^1,\widetilde{\Z}^1],\cdots,[\Z^k,\widetilde{\Z}^k]_{\textnormal{Swap}(S)},\cdots,[\Z^K,\widetilde{\Z}^K])_m\\
&=&\prod_{j\in [K]\backslash k} (Z^j_m-\widetilde{Z}^j_m)(Z_m^k-\widetilde{Z}_m^k)\\
&=&\prod_{j\in [K]} (Z^j_m-\widetilde{Z}^j_m)\\
&=&f([\Z^1,\widetilde{\Z}^1],\cdots,[\Z^K,\widetilde{\Z}^K])_m,
\end{eqnarray*}
So we have
\begin{eqnarray*}
W&=&f([\Z^1,\widetilde{\Z}^1],\cdots,[\Z^k,\widetilde{\Z}^k]_{\textnormal{Swap}(S)},\cdots,[\Z^K,\widetilde{\Z}^K])\\
&=&f([\Z^1,\widetilde{\Z}^1],\cdots,[\Z^k,\widetilde{\Z}^k],\cdots,[\Z^K,\widetilde{\Z}^K])\odot \epsilon(S)
\end{eqnarray*}
and thus the OSFF assumption is satisfied. 

When Model-X group knockoff construction is used, based on the construction of group Model-X knockoff, by Lemmas 1 and 2 we have $[\X^k \widetilde{\X}^k]\eqd [\X^k \widetilde{\X}^k]_{\text{GSwap}(S,G)}$, $\widetilde{\X}^k\indep Y^k | \X^k$. For any $S\subseteq\mathcal{H}$, we have $[\X^k, \widetilde{\X}^k]|Y^k\eqd [\X^k, \widetilde{\X}^k]_{\text{GSwap}(S,G)}|Y^k$. By the definition of knockoff-compatible statistics, we have 
\begin{eqnarray*}
&&[\Z^k,\widetilde{\Z}^k]_{\text{Swap}(S)}=t([\X^k, \widetilde{\X}^k]_{\text{GSwap}(S,G)}, Y^k)\\
&\eqd&t([\X^k, \widetilde{\X}^k], Y^k)\\
&=&[\Z^k,\widetilde{\Z}^k]
\end{eqnarray*}
 for any $S\subseteq\mathcal{H}$. When fixed group knockoff construction is used, by the definition of knockoff compatible statistics and sufficiency requirement, we have
 \begin{eqnarray*}
&&[\Z^k,\widetilde{\Z}^k]_{\text{Swap}(S)}=t([\X^k, \widetilde{\X}^k]^\top_{\text{GSwap}(S,G)}[\X^k, \widetilde{\X}^k]_{\text{GSwap}(S,G)}, [\X^k, \widetilde{\X}^k]_{\text{GSwap}(S,G)}^\top Y^k)\\
&\eqd&t([\X^k, \widetilde{\X}^k]^\top[\X^k, \widetilde{\X}^k], [\X^k, \widetilde{\X}^k]^\top Y^k)\\
&=&[\Z^k,\widetilde{\Z}^k].
\end{eqnarray*}

 Applying Theorem 1, we obtain the conclusion for this corollary.
\end{proof}

\section*{Web Appendix D: Details for the group knockoff construction algorithms} \label{app:algorithm}

For the sequential group knockoff construction in Algorithm 2, we have the following steps. 
For the continuous $\X$, the knockoff distribution can be generated by fitting the penalized multitask linear regression

\[ \widehat{\mathbf{B}}_m = \arg\min_{\mathbf{B}_m} \norm{ \X_{G_m}^{con} - [\X_{-G_m},\widetilde{\X}_{\cup_{j=1}^{m-1}G_j} ] \mathbf{B}_m }_{Fro}^2 + \lambda \norm{\mathbf{B}_m}_{l_1 / l_2} ,\]

where $\norm{\cdot}_{Fro}$ is the Frobenius norm, and the $\norm{\cdots}_{l_1 / l_2}$ is defined as $\norm{\mathbf{B}}_{l_1 / l_2} = \sum_i \sqrt{\sum_j B_{ij}}$. Then
\begin{align*}
    \widehat{\mathbf{\mu}}_m &= [\X_{-G_m},\widetilde{\X}_{\cup_{j=1}^{m-1}G_j} ] \widehat{\mathbf{B}}_m, \text{~~and} \\
    \widehat{\Sigma}_m &= \frac{1}{n} (\X_{G_m}^{con} -  \widehat{\mathbf{\mu}}_m)^{\top} (\X_{G_m}^{con} -  \widehat{\mathbf{\mu}}_m) .
\end{align*}

The high-dimensional multitask regression problem can be reformulated into a group lasso problem as described in Section 2.2 of \citep{dai2016knockoff}.

For the penalized multinomial regression to construct $\X_{G_m}^{cat}$, we use the penalty and fit the model with the R package glmnet. {\color{black} The penalized log-likelihood is 
\[
\widehat{\mathbf{B}}_m = \arg\min_{\mathbf{B}_m} \{- \frac{1}{n} \sum_{i=1}^{n} \log p_{i} + \lambda \norm{\mathbf{B}_m}_{l_1 / l_2} \},
\]
where \( p_i = \frac{\exp(\eta_{il})}{\sum_{l'=1}^{L} \exp(\eta_{il'})}\) for $l=\X_{iG_m}^{cat}$ and $\eta_{il}$ is the $(i,l)$th element of $[\X_{G_m}^{con},\X_{-G_m},\widetilde{\X}_{\cup_{j=1}^{m-1}G_j}]\mathbf{B}_m$.

We use a 10-fold cross-validation method to select penalty parameters with minimal mean square error for the continuous outcome or misclassification error rate for the categorical outcome.
} 

\section*{Web Appendix E: Additional simulation details} \label{app:sim}
\subsection*{E.1: Simulation for {\color{black} \textbf{Setting 1} when $K=3$}}
\subsection*{E.1.1: Data Generation}

We set each dataset to have the same sample size {\color{black}{(1)$n_k=1000$; (2)$n_k=200$}} and the same number of groups of features {\color{black}$M=40$} for each dataset. We have 4 features per group with 75\% continuous variables and 25\% categorical variables leading to a total of $p_k=160$ variables. {\color{black} We simulate independent $X^k$s for $k \in [K]$ such that
\[\X_i^k \sim \mathcal{N} (\mathbf{0}, \Sigma^{k}) ~\text{for}~ i \in [n_k],\]
where $\Sigma^k \in \R^{p_k \times p_k}$ with diagonal elements $\Sigma_{jj}^{k}=1$ for $j \in [p_k]$, within-group correlations $\Sigma_{ji}^{k}=\rho_{k}$ for $i \neq j$ in the same group (i.e., $\{i,j\}\subset G_{km}$ for some $m\in [M]$) and between-group correlations $\Sigma_{ji}^{k}=\gamma_{k}\rho_{k}$ for $j	
\neq i$ in different groups (i.e., there is no $m\in [M]$ such that $\{i,j\}\subset G_{km}$). Within each group, we randomly select one variable and transform it into a three-level categorical variable by breaking it down using the 25th and 75th percentiles. Then, we create 2 dummy variables for this categorical variable and consider the 3 continuous variables and 2 dummy variables as a group.} Let $\bar{\X}^k$ denote the expanded design matrix of $\X^k$ after replacing each categorical variable with dummy variables, then we have $\bar{p}_k=200$ for $k \in [K]$ columns in total and 5 columns per each group. The group index for the expanded design matrix will be $G_{km} = \{5m-4, \cdots, 5m\}$ for $k \in [K], m \in [M]$). 

Next, we generate the coefficients $\beta^1, \cdots,\beta^K$ for the $K$ experiments. 
We denote $s_0$ as the number of groups of simultaneous signals among the $K$ datasets, $s_k$ as the number of groups of signals specifically for the $k$-th datasets, $s_{ij}$ as the number of groups of mutual signals in $i$-th and $j$-th datasets. We consider two \textbf{scenarios}: (1) both directions and strengths of the mutual signals are the same among the $K$ datasets, and (2) only the directions of the mutual signals are the same but the signal strengths are different among the K datasets. For each dataset, the signal strengths within each group $m \in [M]$ are identical.

For \textbf{Scenario 1}, we sample $\omega_{j} \in \mathbb{R}^{s_j}, j \in \{0,1,2,3,12,13,23\}$, with their elements $\omega_{ji}\sim \textnormal{Uniform}[0,{\color{black}A_j}]$ independent for $ i=1,\cdots, s_j$. Then we sample $\epsilon \in \{-1,1\}^M$ where $\epsilon_m$ are independently sampled from Rademacher distribution for $l=1,\cdots,M$. With $K=3$, the coefficients $\beta^1, \beta^2,\beta^3$ are determined by:

\begin{eqnarray*}
\beta^1&=&((\omega_0^\top, \omega_1^\top, \mathbf{0}_{s_2}^\top, \mathbf{0}_{s_3}^\top, \omega_{12}^\top, \omega_{13}^\top, \mathbf{0}_{s_{23}}^\top, \mathbf{0}_{p_s}^\top)^\top  \odot \epsilon )\otimes \mathbf{1}_5,\\
\beta^2&=&((\omega_0^\top , \mathbf{0}_{s_1}^\top,\omega_2^\top, \mathbf{0}_{s_3}^\top, \omega_{12}^\top, \mathbf{0}_{s_{13}}^\top, \omega_{23}^\top, \mathbf{0}_{p_s}^\top)^\top  \odot \epsilon)\otimes \mathbf{1}_5 ,\\
\beta^3&=&((\omega_0^\top , \mathbf{0}_{s_1}^\top, \mathbf{0}_{s_2}^\top, \omega_3^\top, \mathbf{0}_{s_{12}}^\top, \omega_{13}^\top, \omega_{23}^\top, \mathbf{0}_{p_s}^\top)^\top  \odot \epsilon)\otimes \mathbf{1}_5,
\end{eqnarray*}
where $\odot$ is the Hadamard product, and $p_s=M-s_0-s_1-s_{2}-s_{3}-s_{12}-s_{13}-s_{23}$.\\

For \textbf{Scenario 2}, we generate $\omega_{jk} \in \mathbb{R}^{s_j}$ for $k\in [K],j \in \{0,1,2,3,12,13,23\}$ from $\textnormal{Uniform}[0,{\color{black}A_j}]$ independently; for example, we sample $\omega_{0ki} \sim \textnormal{Uniform}[0,{\color{black}A_0}]$ independently for $k \in [K]$ and $i=1,\dots, s_0$. We generate $\epsilon$ the same way as described in \textbf{Scenario 1}. The coefficients $\beta^1, \beta^2, \beta^3$ are determined by:

\begin{eqnarray*}
\beta^1&=&((\omega_{01}^\top , \omega_{11}^\top, \mathbf{0}_{s_2}^\top, \mathbf{0}_{s_3}^\top, \omega_{121}^\top, \omega_{131}^\top, \mathbf{0}_{s_{23}}^\top, \mathbf{0}_{p_s}^\top)^\top  \odot \epsilon )\otimes \mathbf{1}_5,\\
\beta^2&=&((\omega_{02}^\top , \mathbf{0}_{s_1}^\top,\omega_{22}^\top, \mathbf{0}_{s_3}^\top, \omega_{122}^\top, \mathbf{0}_{s_{13}}^\top, \omega_{232}^\top, \mathbf{0}_{p_s}^\top)^\top  \odot \epsilon)\otimes \mathbf{1}_5 ,\\
\beta^3&=&((\omega_{03}^\top , \mathbf{0}_{s_1}^\top, \mathbf{0}_{s_2}^\top, \omega_{33}^\top, \mathbf{0}_{s_{12}}^\top, \omega_{133}^\top, \omega_{233}^\top, \mathbf{0}_{p_s}^\top)^\top  \odot \epsilon)\otimes \mathbf{1}_5,
\end{eqnarray*}
where $\odot$ is the Hadamard product, and $p_s=M-s_0-s_1-s_{2}-s_{3}-s_{12}-s_{13}-s_{23}$.\\

For \textbf{continuous} setting, $Y^k$s are obtained from the following linear model: 
\begin{eqnarray*}\label{eqn:linmod}
Y^k&=&\bar{\X}^k\beta^{k}+\varepsilon^k,
\end{eqnarray*}
where $\varepsilon^k\sim \mathcal{N}(0,\sigma_k^2)$ for $k=1,2,3$, and $\sigma_k$ is the signal noise ratio. 

For \textbf{binary} setting, $Y^k$s are obtained from logistic models: 
\begin{eqnarray*}
Y^k\sim \textnormal{Bernoulli} \left(\frac{\exp(\alpha_k+\bar{\X}^k\beta^{k})}{1+\exp(\alpha_k+\bar{\X}^k\beta^{k})}\right),
\end{eqnarray*}
where $k=1,2,3.$ 

For \textbf{mixed} setting, we generate the latent outcome $\overline{Y}^k$s for $k \in [K]$ from the linear models:
\begin{eqnarray*}
\overline{Y}^k&=&\bar{\X}^k\beta^{k}+\varepsilon^k,
\end{eqnarray*}
where $\varepsilon^k\sim \mathcal{N}(0,\sigma_k^2)$ for $k=1,2,3$, and $\sigma_k$ is the signal noise ratio. Then for continuous outcome $Y^k$, we set $Y^k=\overline{Y}^k$; For binary outcome $Y^k$, we set a threshold for $\overline{Y}^k.$

Here we set $Y^1=\overline{Y}^1$, $Y^3=\overline{Y}^3$ and set a threshold for $\overline{Y}^2$ to construct the binary $Y^2$: \[Y^2 = \mathbbm{1}\{\overline{Y}^2 \geq 0\}.\]. 

\subsection*{E.1.2: Parameter settings} \label{par:E1.2}
We conduct simulations to check the effects of sparsity levels $s_0, s_1, s_2,s_3,s_{12},s_{13},s_{23}$ and different correlation structures. We set the targeted FDR $q=0.2$. We set the amplitude of signals {\color{black}$A_0=2,A_1=A_2=A_3=A_{12}=A_{13}=A_{23}=1$}, the within-group correlations $\rho_k=0.5, ~\text{for}~ k \in [K]$, and the between-group correlations are set to be $\gamma_k\rho_k, ~\text{for}~ k \in [K]$, with the default correlation ratio $\gamma_k=0.1$, and the signal noise parameter $\sigma_1=1$, $\sigma_2=2,\sigma_3=1$ for continuous and mixed settings, $\alpha_1=1$, $\alpha_2=2,\alpha_3=1$ for binary setting. To understand the effects of sparsity levels, within and between group correlations respectively, we vary one of three kinds of parameters (sparsity levels parameters, within-group correlation parameters, and correlation ratio parameters) in each simulation study and fix the other two kinds of parameters.  {\color{black}To maintain the broad applicability of our study, we explore three choices including only simultaneous signals, both simultaneous signals and non-simultaneous signals exist in one dataset and two datasets, and both simultaneous signals and non-simultaneous signals exist in two datasets. These choices are frequently observed in the N3C database.} 
\begin{itemize}
\item Sparsity level parameters: $s_0, s_1, s_2,s_3,s_{12},s_{13},s_{23}.$ We fix $\gamma_k=0.1$, and $\rho_k=0.5$.
 {\color{black}\begin{align*}
        &\text{1. Fixing $s_1=s_2=s_3=s_{12}=s_{13}=s_{23}=0$, we vary $s_0 = 6,8,10,12,14,16,18,20$;}\\
        &\text{2. Fixing $s_0 =12,s_{12}=s_{23}=1,s_{13}=0$, we vary $s_1=s_2=s_3= 2,3,4,5,6,7,8$;}\\
        &\text{3. Fixing $s_0 =12,s_1=s_2=s_3=0$, we vary $s_{12}=s_{13}=s_{23}= 2,3,4,5,6,7,8$.}\\
    \end{align*}}
\item Within-group correlation parameters: $\rho_1, \rho_2,\rho_3$. We fix $\gamma_k=0.1$ and vary within-group correlations $\rho_1=\rho_2=\rho_3 \in \{0.05, 0.15, \dots, 0.95\}$ for the following three choice of $s_0, s_1, s_2,s_3,s_{12},s_{13},s_{23}$.
{\color{black}\begin{align*}
       \textbf{choice 1}:& s_0=12,s_1=s_2=s_3=s_{12}=s_{13}=s_{23}=0;\\
       \textbf{choice 2}:& s_0=12,s_1=s_2=s_3=6,s_{12}=s_{23}=2,s_{13}=0;\\
       \textbf{choice 3}:& s_0=12,s_1=s_2=s_3=0,s_{12}=s_{13}=s_{23}=6.
    \end{align*}}

\item Correlation ratio parameters: $\gamma_1, \gamma_2,\gamma_3$. We fix $\rho_k=0.5$ and vary correlation ratio $\gamma_1=\gamma_2=\gamma_3 \in \{0,0.05,0.1,\dots,0.5\}$ for the following three choice of $s_0, s_1, s_2,s_3,s_{12},s_{13},s_{23}$. Then, the between-group correlations are calculated as $\rho_k\gamma_k$.
{\color{black}\begin{align*}
       \textbf{choice 1}: & s_0=12,s_1=s_2=s_3=s_{12}=s_{13}=s_{23}=0;\\
       \textbf{choice 2}: & s_0=12,s_1=s_2=s_3=6,s_{12}=s_{23}=2,s_{13}=0;\\
       \textbf{choice 3}: & s_0=12,s_1=s_2=s_3=0,s_{12}=s_{13}=s_{23}=6.
    \end{align*}}
\end{itemize}

\subsection*{E.2: Simulation for {\color{black} \textbf{Setting 1} when $K=4$}}
\subsection*{E.2.1: Data Generation}
Our design matrices $\X^k$s {\color{black}{with $n_k=1000$}}, the coefficients $\beta^1, \cdots, \beta^K \in \R^{p_k}$ and outcome variable $Y_i^k$ are generated the same as $K=3$ setting with $j$ extension to $j \in \{\color{black}{0,1,2,3,4,12,13,14,23,24,34,123,124,134,234}\}$. With $K=4$, the coefficients $\beta^1, \beta^2,\beta^3,\beta^4$ for \textbf{Scenario 1} are determined by:

{\color{black}\begin{eqnarray*}
\beta^1&=&((\omega_0^\top , \omega_1^\top, \mathbf{0}_{s_2}^\top, \mathbf{0}_{s_3}^\top, \mathbf{0}_{s_4}^\top,\omega_{12}^\top, \omega_{13}^\top,\omega_{14}^\top,\mathbf{0}_{s_{23}}^\top, \mathbf{0}_{s_{24}}^\top,\mathbf{0}_{s_{34}}^\top,\omega_{123}^\top,\omega_{124}^\top,\omega_{134}^\top,\mathbf{0}_{234}^\top,\mathbf{0}_{p_s}^\top)^\top  \odot \epsilon)\otimes \mathbf{1}_5 ,\\
\beta^2&=&((\omega_0^\top , \mathbf{0}_{s_1}^\top,\omega_2^\top, \mathbf{0}_{s_3}^\top, \mathbf{0}_{s_4}^\top,\omega_{12}^\top, \mathbf{0}_{s_{13}}^\top, \mathbf{0}_{s_{14}}^\top, \omega_{23}^\top, \omega_{24}^\top, \mathbf{0}_{s_{34}}^\top, \omega_{123}^\top,\omega_{124}^\top,\mathbf{0}_{134}^\top,\omega_{234}^\top,\mathbf{0}_{p_s}^\top)^\top  \odot \epsilon)\otimes \mathbf{1}_5 ,\\
\beta^3&=&((\omega_0^\top , \mathbf{0}_{s_1}^\top, \mathbf{0}_{s_2}^\top, \omega_3^\top, \mathbf{0}_{s_4}^\top, \mathbf{0}_{s_{12}}^\top, \omega_{13}^\top, \mathbf{0}_{s_{14}}^\top, \omega_{23}^\top, \mathbf{0}_{s_{24}}^\top,\omega_{34}^\top,\omega_{123}^\top,\mathbf{0}_{124}^\top,\omega_{134}^\top,\omega_{234}^\top, \mathbf{0}_{p_s}^\top)^\top  \odot \epsilon)\otimes \mathbf{1}_5,\\
\beta^4&=&((\omega_0^\top , \mathbf{0}_{s_1}^\top, \mathbf{0}_{s_2}^\top, \mathbf{0}_{s_3}^\top, \omega_{4}^\top, \mathbf{0}_{s_{12}}^\top, \mathbf{0}_{s_{13}}^\top, \omega_{14}^\top, \mathbf{0}_{s_{23}}^\top, \omega_{24}^\top,\omega_{34}^\top, \mathbf{0}_{123}^\top,\omega_{124}^\top,\omega_{134}^\top,\omega_{234}^\top,\mathbf{0}_{p_s}^\top)^\top  \odot \epsilon)\otimes \mathbf{1}_5,
\end{eqnarray*}}

where $\odot$ is the Hadamard product, and $p_s=M-s_0-s_1-s_2-s_3-s_4-s_{12}-s_{13}-s_{14}-s_{23}-s_{24}-s_{34}-{\color{black}s_{123}-s_{124}-s_{134}-s_{234}}$.\\

for \textbf{Scenario 2}, 
{\color{black}\begin{eqnarray*}
\beta^1&=&((\omega_{01}^\top , \omega_{11}^\top, \mathbf{0}_{s_2}^\top, \mathbf{0}_{s_3}^\top, \mathbf{0}_{s_4}^\top,\omega_{121}^\top, \omega_{131}^\top,\omega_{141}^\top,\mathbf{0}_{s_{23}}^\top, \mathbf{0}_{s_{24}}^\top,\mathbf{0}_{s_{34}}^\top,\omega_{1231}^\top,\omega_{1241}^\top,\omega_{1341}^\top,\mathbf{0}_{234}^\top,\mathbf{0}_{p_s}^\top)^\top  \odot \epsilon)\otimes \mathbf{1}_5 ,\\
\beta^2&=&((\omega_{02}^\top , \mathbf{0}_{s_1}^\top,\omega_{22}^\top, \mathbf{0}_{s_3}^\top, \mathbf{0}_{s_4}^\top,\omega_{122}^\top, \mathbf{0}_{s_{13}}^\top, \mathbf{0}_{s_{14}}^\top, \omega_{232}^\top, \omega_{242}^\top, \mathbf{0}_{s_{34}}^\top, \omega_{1232}^\top,\omega_{1242}^\top,\mathbf{0}_{134}^\top,\omega_{2342}^\top,\mathbf{0}_{p_s}^\top)^\top  \odot \epsilon)\otimes \mathbf{1}_5 ,\\
\beta^3&=&((\omega_{03}^\top , \mathbf{0}_{s_1}^\top, \mathbf{0}_{s_2}^\top, \omega_{33}^\top, \mathbf{0}_{s_4}^\top, \mathbf{0}_{s_{12}}^\top, \omega_{133}^\top, \mathbf{0}_{s_{14}}^\top, \omega_{233}^\top, \mathbf{0}_{s_{24}}^\top,\omega_{343}^\top, \omega_{1233}^\top,\mathbf{0}_{124}^\top,\mathbf{0}_{1343}^\top,\omega_{2343}^\top,\mathbf{0}_{p_s}^\top)^\top  \odot \epsilon)\otimes \mathbf{1}_5,\\
\beta^4&=&((\omega_{04}^\top , \mathbf{0}_{s_1}^\top, \mathbf{0}_{s_2}^\top, \mathbf{0}_{s_3}^\top, \omega_{44}^\top, \mathbf{0}_{s_{12}}^\top, \mathbf{0}_{s_{13}}^\top, \omega_{144}^\top, \mathbf{0}_{s_{23}}^\top, \omega_{244}^\top,\omega_{344}^\top, \mathbf{0}_{123}^\top,\omega_{1244}^\top,\omega_{1344}^\top,\omega_{2344}^\top,\mathbf{0}_{p_s}^\top)^\top  \odot \epsilon)\otimes \mathbf{1}_5,
\end{eqnarray*}}

where $\odot$ is the Hadamard product, and $p_s=M-s_0-s_1-s_2-s_3-s_4-s_{12}-s_{13}-s_{14}-s_{23}-s_{24}-s_{34}-{\color{black}s_{123}-s_{124}-s_{134}-s_{234}}$.\\

\textbf{Continuous} and \textbf{Binary} setting are the same as $K=3$. For \textbf{mixed} setting, we set $Y^1=\overline{Y}^1$, $Y^3=\overline{Y}^3$ and set a threshold for $\overline{Y}^2$ and  $\overline{Y}^4$ to construct the binary $Y^2$ and $Y^4$: \[Y^2 = \mathbbm{1}\{\overline{Y}^2 \geq 0\};\]\[Y^4 = \mathbbm{1}\{\overline{Y}^4 \geq 0\}.\]

\subsection*{E.2.2: Parameter settings} \label{par:E2.2}
We also conduct simulations to check the effects of sparsity levels and different correlation structures and set the targeted FDR $q=0.2$. As default, we set the amplitude of signals {\color{black}$A_0=4, A_1=A_2=A_3=A_4=A_{12}=A_{13}=A_{14}=A_{23}=A_{24}=A_{34}=A_{123}=A_{124}=A_{134}=A_{234}=1$}, the within-group correlations $\rho_1=0.5,\rho_2=0.4,\rho_3=0.5,\rho_4=0.6$, and the correlation ratios are set to be $\gamma_1=0.1,\gamma_2=0.15,\gamma_3=0.1,\gamma_4=0.05$, and the signal noise parameter $\sigma_1=1$, $\sigma_2=2,\sigma_3=1,\sigma_4=1$ for continuous and mixed settings, $\alpha_1=1, \alpha_2=2, \alpha_3=1, \alpha_4=1$ for binary setting. To understand the effects of sparsity levels, within and between group correlations respectively, we vary one of three kinds of parameters (sparsity levels parameters, within-group correlation parameters, and correlation ratio parameters) in each simulation study and fix the other two kinds of parameters. {\color{black}Similarly, to avoid loss of generality, we explore three choices including only simultaneous signals, both simultaneous signals and non-simultaneous signals exist in one dataset, and both simultaneous signals and non-simultaneous signals exist in three datasets.} 
\begin{itemize}
\item Sparsity level parameters: $s_0, s_1, s_2,s_3,s_4,s_{12},s_{13},s_{14},s_{23},s_{24},s_{34},s_{123},s_{124},s_{134},s_{234}.$ 
{\color{black}
\begin{align*}
    &\text{1. Vary $s_0 = 6,8,10,12,14,16,18,20$,}\\
    &\text{fixing $s_1=s_2=s_3=s_4=s_{12}=s_{13}=s_{14}=s_{23}=s_{24}=s_{34}=s_{123}=s_{124}=s_{134}=s_{234}=0$;}\\
    &\text{2. Vary $s_1=s_2=s_3=s_4=0,1,2,3,4,5,6$,}\\
    &\text{Fixing $s_0 =12,s_{12}=s_{13}=s_{14}=s_{23}=s_{24}=s_{34}=s_{123}=s_{124}=s_{134}=s_{234}=0$;}\\
    &\text{3. Vary $s_{123}=s_{124}=s_{134}=s_{234}= 0,1,2,3,4,5,6$;}\\
    &\text{Fixing $s_0 =12,s_1=s_2=s_3=s_4=s_{12}=s_{13}=s_{14}=s_{23}=s_{24}=s_{34}=0$.}
\end{align*}
}

\item Within-group correlation parameters: $\rho_1, \rho_2,\rho_3,\rho_4$. We fix {\color{black}$\gamma_1=0.1,\gamma_2=0.15,\gamma_3=0.1,\gamma_4=0.05$} and vary within-group correlations $\rho_1=\rho_2=\rho_3=\rho_4 \in \{0.05, 0.15, \dots, 0.95\}$ for the following choice of $s_0, s_1, s_2,s_3,s_4,s_{12},s_{13},s_{14},s_{23},s_{24},s_{34},s_{123},s_{124},s_{134},s_{234}$.
{\color{black}
\begin{align*}
       \textbf{choice 1}:& s_0=12,s_1=s_2=s_3=s_4=s_{12}=s_{13}=s_{14}=s_{23}=s_{24}=s_{34}=0,\\&s_{123}=s_{124}=s_{134}=s_{234}=0;\\
       \textbf{choice 2}:& s_0=12,s_1=s_2=s_3=s_4=4, s_{12}=s_{13}=s_{14}=s_{23}=s_{24}=s_{34}=0,\\&s_{123}=s_{124}=s_{134}=s_{234}=0;\\
       \textbf{choice 3}: & s_0=12,s_{123}=s_{124}=s_{134}=s_{234}=6,s_1=s_2=s_3=s_4=0,\\&s_{12}=s_{13}=s_{14}=s_{23}=s_{24}=s_{34}=0.
    \end{align*}
}

\item Correlation ratio parameters: $\gamma_1, \gamma_2,\gamma_3,\gamma_4$. We fix {\color{black}$\rho_1=0.5,\rho_2=0.4,\rho_3=0.5,\rho_4=0.6$} and vary correlation ratio $\gamma_1=\gamma_2=\gamma_3=\gamma_4 \in \{0,0.05,0.1,\dots,0.5\}$ for the following choice of $s_0, s_1, s_2,s_3,s_4,s_{12},s_{13},s_{14},s_{23},s_{24},s_{34},\\s_{123},s_{124},s_{134},s_{234}$. Then, the between-group correlations are calculated as $\rho_k\gamma_k$.
{\color{black}

\begin{align*}
       \textbf{choice 1}:& s_0=12,s_1=s_2=s_3=s_4=s_{12}=s_{13}=s_{14}=s_{23}=s_{24}=s_{34}=0,\\&s_{123}=s_{124}=s_{134}=s_{234}=0;\\
       \textbf{choice 2}:& s_0=12,s_1=s_2=s_3=s_4=4, s_{12}=s_{13}=s_{14}=s_{23}=s_{24}=s_{34}=0,\\&s_{123}=s_{124}=s_{134}=s_{234}=0;\\
       \textbf{choice 3}: & s_0=12,s_{123}=s_{124}=s_{134}=s_{234}=6,s_1=s_2=s_3=s_4=0,\\&s_{12}=s_{13}=s_{14}=s_{23}=s_{24}=s_{34}=0.
\end{align*}
}

\end{itemize}
\subsection*{E.3: Simulation for {\color{black} \textbf{Setting 1} when $K=5$}}
\subsection*{E.3.1: Data Generation}
Similar with K=4, our design matrices $\X^k$s {\color{black}with $n_k=1000$}, the coefficients $\beta^1, \cdots, \beta^K \in \R^{p_k}$ and outcome variable $Y_i^k$ are generated the same as $K=3$ setting with $j$ extension to {\color{black}$j \in \{0,1,2,3,4,5,12,13,14,15,23,24,25,34,35,45,123,124,125,134,\\135,145,234,235,245,345,1234,1235,1245,1345,2345\}$}. With $K=5$, the coefficients $\beta^1, \beta^2,\beta^3,\beta^4,\beta^5$ for \textbf{Scenario 1} are determined by:\\
{\color{black}
$\beta^1=(\boldsymbol{\omega_1} \odot \epsilon)\otimes \mathbf{1}_5 ,
\beta^2 =(\boldsymbol{\omega_2}\odot \epsilon) \otimes \mathbf{1}_5,
\beta^3 =(\boldsymbol{\omega_3}\odot \epsilon) \otimes \mathbf{1}_5,
\beta^4 =(\boldsymbol{\omega_4}\odot \epsilon) \otimes \mathbf{1}_5,
\beta^5 =(\boldsymbol{\omega_5}\odot \epsilon) \otimes \mathbf{1}_5.$ where
$\boldsymbol{\omega_1}=
\begin{pmatrix}\omega_0 \\ \mathbf{0}_{m_s}\\ \omega_{s_{123}}\\ \omega_{s_{124}}\\ \omega_{s_{125}}\\ \omega_{s_{134}}\\ \omega_{s_{135}}\\ \omega_{s_{145}}\\ \mathbf{0}_{s_{234}}\\ \mathbf{0}_{s_{235}}\\ \mathbf{0}_{s_{245}}\\ \mathbf{0}_{s_{345}}\\ \omega_{s_{1234}}\\ \omega_{s_{1235}}\\ \omega_{s_{1245}}\\ \omega_{s_{1345}}\\ \mathbf{0}_{s_{2345}}\\ \mathbf{0}_{p_s}
\end{pmatrix}$,
$\boldsymbol{\omega_2}=
\begin{pmatrix}\omega_0\\ \mathbf{0}_{m_s}\\ \omega_{s_{123}}\\ \omega_{s_{124}}\\ \omega_{s_{125}} \\ \mathbf{0}_{s_{134}}\\ \mathbf{0}_{s_{135}}\\ \mathbf{0}_{s_{145}}\\ \omega_{s_{234}}\\ \omega_{s_{235}}\\ \omega_{s_{245}}\\ \mathbf{0}_{s_{345}}\\ \omega_{s_{1234}}\\ \omega_{s_{1235}}\\\omega_{s_{1245}}\\ \mathbf{0}_{s_{1345}}\\\omega_{s_{2345}}\\\mathbf{0}_{p_s}
\end{pmatrix}$,
$\boldsymbol{\omega_3}=
\begin{pmatrix}\omega_0\\ \mathbf{0}_{m_s}\\\omega_{s_{123}}\\ \mathbf{0}_{s_{124}}\\\mathbf{0}_{s_{125}}\\ \omega_{s_{134}}\\ \omega_{s_{135}}\\\mathbf{0}_{s_{145}}\\ \omega_{s_{234}}\\ \omega_{s_{235}}\\ \mathbf{0}_{s_{245}}\\ \omega_{s_{345}}\\\omega_{s_{1234}}\\ \omega_{s_{1235}}\\ \mathbf{0}_{s_{1245}}\\ \omega_{s_{1345}}\\ \omega_{s_{2345}}\\\mathbf{0}_{p_s}
\end{pmatrix}$,
$\boldsymbol{\omega_4}=
\begin{pmatrix}\omega_0\\ \mathbf{0}_{m_s}\\ \mathbf{0}_{s_{123}}\\\omega_{s_{124}}\\\mathbf{0}_{s_{125}}\\ \omega_{s_{134}}\\\mathbf{0}_{s_{135}}\\ \omega_{s_{145}}\\ \omega_{s_{234}}\\ \mathbf{0}_{s_{235}}\\ \omega_{s_{245}}\\ \omega_{s_{345}}\\ \omega_{s_{1234}}\\ \mathbf{0}_{s_{1235}}\\ \omega_{s_{1245}}\\ \omega_{s_{1345}}\\\omega_{s_{2345}}\\ \mathbf{0}_{p_s}
\end{pmatrix}$,
$\boldsymbol{\omega_5}=
\begin{pmatrix}\omega_0\\ \mathbf{0}_{m_s}\\ \mathbf{0}_{s_{123}}\\ \mathbf{0}_{s_{124}}\\ 
\omega_{s_{125}}\\\mathbf{0}_{s_{134}}\\ \omega_{s_{135}}\\ \omega_{s_{145}}\\\mathbf{0}_{s_{234}}\\ \omega_{s_{235}}\\ \omega_{s_{245}}\\ \omega_{s_{345}}\\ \mathbf{0}_{s_{1234}}\\ \omega_{s_{1235}}\\ \omega_{s_{1245}}\\ \omega_{s_{1345}}\\ \omega_{s_{2345}}\\ \mathbf{0}_{p_s}
\end{pmatrix}.$\\
}

Besides, $\odot$ is the Hadamard product, {\color{black}$m_s = s_1+s_2+s_3+s_4+s_5+s_{12}+s_{13}+s_{14}+s_{15}+s_{23}+s_{24}+s_{25}+s_{34}+s_{35}+s_{45}$,} {\color{black}and $p_s=M-s_0-m_s-s_{123}-s_{124}-s_{125}-s_{134}-s_{135}-s_{145}-s_{234}-s_{235}-s_{245}-s_{345}-s_{1234}-s_{1235}-s_{1245}-s_{1345}-s_{2345}.$}

For \textbf{Scenario 2}, 
{\color{black}
$\beta^1 =( \boldsymbol{\widetilde{\omega}_1} \odot\epsilon)\otimes \mathbf{1}_5 ,
\beta^2 = (\boldsymbol{\widetilde{\omega}_2}\odot \epsilon) \otimes \mathbf{1}_5,
\beta^3 = (\boldsymbol{\widetilde{\omega}_3}\odot \epsilon) \otimes \mathbf{1}_5,
\beta^4 = (\boldsymbol{\widetilde{\omega}_4}\odot \epsilon) \otimes \mathbf{1}_5,
\beta^5 = (\boldsymbol{\widetilde{\omega}_5}\odot \epsilon) \otimes \mathbf{1}_5.$ where\\
$\boldsymbol{\widetilde{\omega}_1}=
\begin{pmatrix}\omega_{01}\\ \mathbf{0}_{m_s}\\ \omega_{s_{1231}}\\ \omega_{s_{1241}}\\\omega_{s_{1251}}\\\omega_{s_{1341}}\\\omega_{s_{1351}}\\\omega_{s_{1451}}\\\mathbf{0}_{s_{234}}\\\mathbf{0}_{s_{235}}\\\mathbf{0}_{s_{245}}\\\mathbf{0}_{s_{345}}\\\omega_{s_{12341}}\\\omega_{s_{12351}}\\\omega_{s_{12451}}\\\omega_{s_{13451}}\\\mathbf{0}_{s_{2345}}\\\mathbf{0}_{p_s}
\end{pmatrix}$,
$\boldsymbol{\widetilde{\omega}_2}=
\begin{pmatrix}\omega_{02}\\ \mathbf{0}_{m_s}\\ \omega_{s_{1232}}\\ \omega_{s_{1242}}\\\omega_{s_{1252}}\\ \mathbf{0}_{s_{134}}\\ \mathbf{0}_{s_{135}}\\ \mathbf{0}_{s_{145}}\\ \omega_{s_{2342}}\\ \omega_{s_{2352}}\\ \omega_{s_{2452}}\\ \mathbf{0}_{s_{345}}\\ \omega_{s_{12342}}\\ \omega_{s_{12352}}\\ \omega_{s_{12452}}\\ \mathbf{0}_{s_{1345}}\\ \omega_{s_{23452}}\\ \mathbf{0}_{p_s}
\end{pmatrix}$,
$\boldsymbol{\widetilde{\omega}_3}=
\begin{pmatrix}\omega_{03}\\\mathbf{0}_{m_s}\\\omega_{s_{1233}}\\ \mathbf{0}_{s_{124}}\\ \mathbf{0}_{s_{125}}\\ \omega_{s_{1343}}\\ \omega_{s_{1353}}\\ \mathbf{0}_{s_{145}}\\ \omega_{s_{2343}}\\ \omega_{s_{2353}}\\ \mathbf{0}_{s_{245}}\\ \omega_{s_{3453}}\\ \omega_{s_{12343}}\\ \omega_{s_{12353}}\\ \mathbf{0}_{s_{1245}}\\ \omega_{s_{13453}}\\ \omega_{s_{23453}}\\\mathbf{0}_{p_s}
\end{pmatrix}$,
$\boldsymbol{\widetilde{\omega}_4}=
\begin{pmatrix}\omega_{04} \\\mathbf{0}_{m_s}\\ \mathbf{0}_{s_{123}}\\ \omega_{s_{1244}}\\ \mathbf{0}_{s_{125}}\\ \omega_{s_{1344}}\\ \mathbf{0}_{s_{135}}\\ \omega_{s_{1454}}\\ \omega_{s_{2344}}\\ \mathbf{0}_{s_{235}}\\ \omega_{s_{2454}}\\ \omega_{s_{3454}}\\ \omega_{s_{12344}}\\ \mathbf{0}_{s_{1235}}\\ \omega_{s_{12454}}\\ \omega_{s_{13454}}\\ \omega_{s_{23454}}\\ \mathbf{0}_{p_s},\\
\end{pmatrix}$,
$\boldsymbol{\widetilde{\omega}_5}=
\begin{pmatrix}\omega_{05}\\\mathbf{0}_{m_s}\\ \mathbf{0}_{s_{123}}\\ \mathbf{0}_{s_{124}}\\ 
\omega_{s_{1255}}\\ \mathbf{0}_{s_{134}}\\ \omega_{s_{1355}}\\ \omega_{s_{1455}}\\\mathbf{0}_{s_{234}}\\ \omega_{s_{2355}}\\ \omega_{s_{2455}}\\ \omega_{s_{3455}}\\ \mathbf{0}_{s_{1234}}\\ \omega_{s_{12355}}\\ \omega_{s_{12455}}\\ \omega_{s_{13455}}\\ \omega_{s_{23455}}\\ \mathbf{0}_{p_s}
\end{pmatrix}$.
}
Besides, $\odot$ is the Hadamard product, {\color{black}$m_s = s_1+s_2+s_3+s_4+s_5+s_{12}+s_{13}+s_{14}+s_{15}+s_{23}+s_{24}+s_{25}+s_{34}+s_{35}+s_{45}$,} {\color{black}and $p_s=M-s_0-m_s-s_{123}-s_{124}-s_{125}-s_{134}-s_{135}-s_{145}-s_{234}-s_{235}-s_{245}-s_{345}-s_{1234}-s_{1235}-s_{1245}-s_{1345}-s_{2345}.$}

\textbf{Continuous} and \textbf{Binary} setting are the same as $K=3$. For \textbf{mixed} setting, we set $Y^1=\overline{Y}^1$, $Y^3=\overline{Y}^3$ and set a threshold for $\overline{Y}^2$,  $\overline{Y}^4$ and $\overline{Y}^5$ to construct the binary $Y^2$, $Y^4$ and $Y^5$: \[Y^2 = \mathbbm{1}\{\overline{Y}^2 \geq 0\}.\]\[Y^4 = \mathbbm{1}\{\overline{Y}^4 \geq 0\}.\]
\[Y^5 = \mathbbm{1}\{\overline{Y}^5 \geq 0\}.\]

\subsection*{E.3.2: Parameter settings} \label{par:E3.2}
We also conduct simulations to check the effects of sparsity levels and different correlation structures and set the targeted FDR $q=0.2$. As default, we set the amplitude of signals {\color{black}$A_0=4,A_1=A_2=A_3=A_4=A_5=A_{12}=A_{13}=A_{14}=A_{15}=A_{23}=A_{24}=A_{25}=A_{34}=A_{35}=A_{45}=A_{123}=A_{124}=A_{125}=A_{134}=A_{135}=A_{145}=A_{234}=A_{235}=A_{245}=A_{345}=A_{1234}=A_{1235}=A_{1245}=A_{1345}=A_{2345}=1$}, the within-group correlations $\rho_1=0.5,\rho_2=0.4,\rho_3=0.5,\rho_4=0.6,\rho_5=0.5$, and the correlation ratios are set to be $\gamma_1=0.1,\gamma_2=0.15,\gamma_3=0.1,\gamma_4=0.05,\gamma_5=0.1$, and the signal noise parameter $\sigma_1=1$, $\sigma_2=2,\sigma_3=1,\sigma_4=1,\sigma_5=1$ for continuous and mixed settings, $\alpha_1=1$, $\alpha_2=2,\alpha_3=1,\alpha_4=1,\alpha_5=1$ for binary setting. To understand the effects of sparsity levels, within and between group correlations respectively, we vary one of three kinds of parameters (sparsity levels parameters, within-group correlation parameters, and correlation ratio parameters) in each simulation study and fix the other two kinds of parameters. {\color{black}Similarly, to avoid loss of generality, we explore three choices, including only simultaneous signals exist, simultaneous signals and non-simultaneous signals exist in three datasets, and simultaneous signals and non-simultaneous signals exist in four datasets.} 
\begin{itemize}
\item Sparsity level parameters: $s_0, s_1, s_2,s_3,s_4,s_5,s_{12},s_{13},s_{14},s_{15},s_{23},s_{24},s_{25},s_{34},s_{35},s_{45},\\s_{123},s_{124},s_{125},s_{134},s_{135},s_{145},s_{234},s_{235},s_{245},s_{345},s_{1234}=s_{1235}=s_{1245}=s_{1345}=s_{2345}.$ 
{\color{black}
 \begin{align*}
       &\text{1. Vary $s_0 = 10,12,14,16,18,20,$}\\
       &\text{Fixing } s_1=s_2=s_3=s_4=s_5=s_{12}=s_{13}=s_{14}=s_{15}=s_{23}=s_{24}=s_{25}=s_{34}=s_{35}=s_{45}=0,\\
    &s_{123}=s_{124}=s_{125}=s_{134}=s_{135}=s_{145}=s_{234}=s_{235}=s_{245}=s_{345}=0,\\&s_{1234}=s_{1235}=s_{1245}=s_{1345}=s_{2345}=0\\
       &\text{2. Vary $s_{234}=s_{235}=s_{245}=s_{345}= 0,1,2,3,4,5,6$,}\\
    &\text{Fixing $s_0=12,s_1=s_2=s_3=s_4=s_5=s_{12}=s_{13}=s_{14}=s_{15}=s_{23}=s_{24}=s_{25}=0,$}\\ 
    &\text{$s_{34}=s_{35}=s_{45}=s_{123}=s_{124}=s_{125}=s_{134}=s_{135}=s_{145}=0$,}\\
       &\text{3. Vary $s_{1234}=s_{1235}=s_{1245}=s_{1345}=s_{2345}=0,1,2,3,4,5$,}\\
       &\text{Fixing $s_1=s_2=s_3=s_4=s_5=s_{12}=s_{13}=s_{14}=s_{15}=s_{23}=s_{24}=s_{25}=s_{34}=s_{35}=s_{45}=0$,}\\
       &\text{$s_{123}=s_{124}=s_{125}=s_{134}=s_{135}=s_{145}=s_{234}=s_{235}=s_{245}=s_{345}=0$;}\\
    \end{align*}
}
\item Within-group correlation parameters: $\rho_1, 
 \rho_2, \rho_3, \rho_4, \rho_5$.\\
We fix {\color{black}$\gamma_1=0.1,\gamma_2=0.15,\gamma_3=0.1,\gamma_4=0.05,\gamma_5=0.1$} and vary within-group correlations $\rho_1=\rho_2=\rho_3=\rho_4=\rho_5 \in \{0.05, 0.15, \dots, 0.95\}$ for the following choice of {\color{black}$s_0, s_1, s_2,s_3,s_4,s_5,s_{12},s_{13},s_{14},s_{15},s_{23},s_{24},s_{25},s_{34},s_{35},s_{45},\\s_{123},s_{124},s_{125},s_{134},s_{135},s_{145},s_{234},s_{235},s_{245},s_{345},s_{1234},s_{1235},s_{1245},s_{1345},s_{2345}.$ }
{\color{black}
\begin{align*}
  \textbf{choice 1}:&\ s_0=12,s_1=s_2=s_3=s_4=s_5=0,\\
  &\ s_{12}=s_{13}=s_{14}=s_{15}=s_{23}=s_{24}=s_{25}=s_{34}=s_{35}=s_{45}=0,\\
  &\ s_{123}=s_{124}=s_{125}=s_{134}=s_{135}=s_{145}=s_{234}=s_{235}=s_{245}=s_{345}=0,\\
  &\ s_{1234}=s_{1235}=s_{1245}=s_{1345}=s_{2345}=0.\\
  \textbf{choice 2}:&\ s_0=12,s_{234}=s_{235}=s_{245}=s_{345}=4,\\
  &\ s_1=s_2=s_3=s_4=s_5=0,\\
  &\ s_{12}=s_{13}=s_{14}=s_{15}=s_{23}=s_{24}=s_{25}=s_{34}=s_{35}=s_{45}=0,\\
  &\ s_{123}=s_{124}=s_{125}=s_{134}=s_{135}=s_{145}=s_{1234}=s_{1235}=s_{1245}=s_{1345}=s_{2345}=0,\\
   \textbf{choice 3}:&\ s_0=12,s_{1234}=s_{1235}=s_{1245}=s_{1345}=s_{2345}=5.\\
  &\ s_1=s_2=s_3=s_4=s_5=0,\\
  &\ s_{12}=s_{13}=s_{14}=s_{15}=s_{23}=s_{24}=s_{25}=s_{34}=s_{35}=s_{45}=0,\\
  &\ s_{123}=s_{124}=s_{125}=s_{134}=s_{135}=s_{145}=s_{234}=s_{235}=s_{245}=s_{345}=0.
\end{align*}
}

\item Correlation ratio parameters: $\gamma_1, \gamma_2,\gamma_3,\gamma_4,\gamma_5$. 
We fix {\color{black}$\rho_1=0.5,\rho_2=0.4,\rho_3=0.5,\rho_4=0.6,\rho_5=0.5$} and vary correlation ratio $\gamma_1=\gamma_2=\gamma_3=\gamma_4=\gamma_5 \in \{0,0.05,0.1,\dots,0.5\}$ for the following choice of $s_0, s_1, s_2,s_3,s_4,s_5,s_{12},s_{13},s_{14},s_{15},s_{23},s_{24},s_{25},s_{34},s_{35},s_{45},s_{123},s_{124},s_{125},s_{134},s_{135},s_{145},\\s_{234},s_{235},s_{245},s_{345},s_{1234},s_{1235},s_{1245},s_{1345},s_{2345}.$ Then, the between-group correlations are calculated as $\rho_k\gamma_k$.
{\color{black}
\begin{align*}
  \textbf{choice 1}:&\ s_0=12,s_1=s_2=s_3=s_4=s_5=0,\\
  &\ s_{12}=s_{13}=s_{14}=s_{15}=s_{23}=s_{24}=s_{25}=s_{34}=s_{35}=s_{45}=0,\\
  &\ s_{123}=s_{124}=s_{125}=s_{134}=s_{135}=s_{145}=s_{234}=s_{235}=s_{245}=s_{345}=0,\\
  &\ s_{1234}=s_{1235}=s_{1245}=s_{1345}=s_{2345}=0.\\
  \textbf{choice 2}:&\ s_0=12,s_{234}=s_{235}=s_{245}=s_{345}=4,\\
  &\ s_1=s_2=s_3=s_4=s_5=0,\\
  &\ s_{12}=s_{13}=s_{14}=s_{15}=s_{23}=s_{24}=s_{25}=s_{34}=s_{35}=s_{45}=0,\\
  &\ s_{123}=s_{124}=s_{125}=s_{134}=s_{135}=s_{145}=s_{1234}=s_{1235}=s_{1245}=s_{1345}=s_{2345}=0,\\
   \textbf{choice 3}:&\ s_0=12,s_{1234}=s_{1235}=s_{1245}=s_{1345}=s_{2345}=5.\\
  &\ s_1=s_2=s_3=s_4=s_5=0,\\
  &\ s_{12}=s_{13}=s_{14}=s_{15}=s_{23}=s_{24}=s_{25}=s_{34}=s_{35}=s_{45}=0,\\
  &\ s_{123}=s_{124}=s_{125}=s_{134}=s_{135}=s_{145}=s_{234}=s_{235}=s_{245}=s_{345}=0.
\end{align*}
}

\end{itemize}
{\color{black}
\section*{E.4: Simulation for {\color{black} \textbf{Setting 2} when $K=4$}}
\subsection*{E.4.1: Data Generation}
To ensure broad applicability, we perform a simulation study that reflects our real data application, varying the sample size and the types within the group across different sites. 
We set $n_1=2000, n_2=1200, n_3=700, n_4=600$. The types within the group across the sites are different. Site 1 encompasses four continuous variable features per group. Site 2 also has four features but with a mix of 75\% continuous and 25\% four-leveled categorical variables (3 continuous + 3 dummy = 6 variables in the expanded design). Site 3 differs, offering only two categorical features per group, each with three levels (2 dummy + 2 dummy = 4 variables in the expanded design). Lastly, Site 4 aligns with the first two in feature count but balances the variable types, with two continuous and two binary categorical features (2 continuous + 2 dummy = 4 variables in the expanded design). The coefficients $\beta^1, \beta^2,\beta^3,\beta^4$ for \textbf{Scenario 1} are updated as:
{\color{black}\begin{eqnarray*}
\beta^1&=&((\omega_0^\top , \omega_1^\top, \mathbf{0}_{s_2}^\top, \mathbf{0}_{s_3}^\top, \mathbf{0}_{s_4}^\top,\omega_{12}^\top, \omega_{13}^\top,\omega_{14}^\top,\mathbf{0}_{s_{23}}^\top, \mathbf{0}_{s_{24}}^\top,\mathbf{0}_{s_{34}}^\top,\omega_{123}^\top,\omega_{124}^\top,\omega_{134}^\top,\mathbf{0}_{234}^\top,\mathbf{0}_{p_s}^\top)^\top  \odot \epsilon)\otimes \mathbf{1}_4 ,\\
\beta^2&=&((\omega_0^\top , \mathbf{0}_{s_1}^\top,\omega_2^\top, \mathbf{0}_{s_3}^\top, \mathbf{0}_{s_4}^\top,\omega_{12}^\top, \mathbf{0}_{s_{13}}^\top, \mathbf{0}_{s_{14}}^\top, \omega_{23}^\top, \omega_{24}^\top, \mathbf{0}_{s_{34}}^\top, \omega_{123}^\top,\omega_{124}^\top,\mathbf{0}_{134}^\top,\omega_{234}^\top,\mathbf{0}_{p_s}^\top)^\top  \odot \epsilon)\otimes \mathbf{1}_6 ,\\
\beta^3&=&((\omega_0^\top , \mathbf{0}_{s_1}^\top, \mathbf{0}_{s_2}^\top, \omega_3^\top, \mathbf{0}_{s_4}^\top, \mathbf{0}_{s_{12}}^\top, \omega_{13}^\top, \mathbf{0}_{s_{14}}^\top, \omega_{23}^\top, \mathbf{0}_{s_{24}}^\top,\omega_{34}^\top,\omega_{123}^\top,\mathbf{0}_{124}^\top,\omega_{134}^\top,\omega_{234}^\top, \mathbf{0}_{p_s}^\top)^\top  \odot \epsilon)\otimes \mathbf{1}_4,\\
\beta^4&=&((\omega_0^\top , \mathbf{0}_{s_1}^\top, \mathbf{0}_{s_2}^\top, \mathbf{0}_{s_3}^\top, \omega_{4}^\top, \mathbf{0}_{s_{12}}^\top, \mathbf{0}_{s_{13}}^\top, \omega_{14}^\top, \mathbf{0}_{s_{23}}^\top, \omega_{24}^\top,\omega_{34}^\top, \mathbf{0}_{123}^\top,\omega_{124}^\top,\omega_{134}^\top,\omega_{234}^\top,\mathbf{0}_{p_s}^\top)^\top  \odot \epsilon)\otimes \mathbf{1}_4,
\end{eqnarray*}}

where $\odot$ is the Hadamard product, and $p_s=M-s_0-s_1-s_2-s_3-s_4-s_{12}-s_{13}-s_{14}-s_{23}-s_{24}-s_{34}-{\color{black}s_{123}-s_{124}-s_{134}-s_{234}}$.\\

for \textbf{Scenario 2}, 
{\color{black}\begin{eqnarray*}
\beta^1&=&((\omega_{01}^\top , \omega_{11}^\top, \mathbf{0}_{s_2}^\top, \mathbf{0}_{s_3}^\top, \mathbf{0}_{s_4}^\top,\omega_{121}^\top, \omega_{131}^\top,\omega_{141}^\top,\mathbf{0}_{s_{23}}^\top, \mathbf{0}_{s_{24}}^\top,\mathbf{0}_{s_{34}}^\top,\omega_{1231}^\top,\omega_{1241}^\top,\omega_{1341}^\top,\mathbf{0}_{234}^\top,\mathbf{0}_{p_s}^\top)^\top  \odot \epsilon)\otimes \mathbf{1}_4 ,\\
\beta^2&=&((\omega_{02}^\top , \mathbf{0}_{s_1}^\top,\omega_{22}^\top, \mathbf{0}_{s_3}^\top, \mathbf{0}_{s_4}^\top,\omega_{122}^\top, \mathbf{0}_{s_{13}}^\top, \mathbf{0}_{s_{14}}^\top, \omega_{232}^\top, \omega_{242}^\top, \mathbf{0}_{s_{34}}^\top, \omega_{1232}^\top,\omega_{1242}^\top,\mathbf{0}_{134}^\top,\omega_{2342}^\top,\mathbf{0}_{p_s}^\top)^\top  \odot \epsilon)\otimes \mathbf{1}_6 ,\\
\beta^3&=&((\omega_{03}^\top , \mathbf{0}_{s_1}^\top, \mathbf{0}_{s_2}^\top, \omega_{33}^\top, \mathbf{0}_{s_4}^\top, \mathbf{0}_{s_{12}}^\top, \omega_{133}^\top, \mathbf{0}_{s_{14}}^\top, \omega_{233}^\top, \mathbf{0}_{s_{24}}^\top,\omega_{343}^\top, \omega_{1233}^\top,\mathbf{0}_{124}^\top,\mathbf{0}_{1343}^\top,\omega_{2343}^\top,\mathbf{0}_{p_s}^\top)^\top  \odot \epsilon)\otimes \mathbf{1}_4,\\
\beta^4&=&((\omega_{04}^\top , \mathbf{0}_{s_1}^\top, \mathbf{0}_{s_2}^\top, \mathbf{0}_{s_3}^\top, \omega_{44}^\top, \mathbf{0}_{s_{12}}^\top, \mathbf{0}_{s_{13}}^\top, \omega_{144}^\top, \mathbf{0}_{s_{23}}^\top, \omega_{244}^\top,\omega_{344}^\top, \mathbf{0}_{123}^\top,\omega_{1244}^\top,\omega_{1344}^\top,\omega_{2344}^\top,\mathbf{0}_{p_s}^\top)^\top  \odot \epsilon)\otimes \mathbf{1}_4,
\end{eqnarray*}}

where $\odot$ is the Hadamard product, and $p_s=M-s_0-s_1-s_2-s_3-s_4-s_{12}-s_{13}-s_{14}-s_{23}-s_{24}-s_{34}-{\color{black}s_{123}-s_{124}-s_{134}-s_{234}}$.\\

We focus on the \textbf{mixed} model setting, and set $Y^1=\overline{Y}^1$, $Y^3=\overline{Y}^3$ and set a threshold for $\overline{Y}^2$ and  $\overline{Y}^4$ to construct the binary $Y^2$ and $Y^4$: \[Y^2 = \mathbbm{1}\{\overline{Y}^2 \geq 0\}.\]\[Y^4 = \mathbbm{1}\{\overline{Y}^4 \geq 0\}.\]

\subsection*{E.4.2: Parameter settings} \label{par:E4.2}
We conduct simulations to check the effects of sparsity levels and different correlation structures and set the targeted FDR $q=0.2$. As default, we set the amplitude of signals {\color{black}$A_0=4, A_1=A_2=A_3=A_4=A_{12}=A_{13}=A_{14}=A_{23}=A_{24}=A_{34}=A_{123}=A_{124}=A_{134}=A_{234}=1$}, the within-group correlations $\rho_1=0.5,\rho_2=0.4,\rho_3=0.5,\rho_4=0.6$, and the correlation ratios are set to be $\gamma_1=0.1,\gamma_2=0.15,\gamma_3=0.1,\gamma_4=0.05$, and the signal noise parameter $\sigma_1=1$, $\sigma_2=2,\sigma_3=1,\sigma_4=1$. To understand the effects of sparsity levels, within and between group correlations respectively, we vary one of three kinds of parameters (sparsity levels parameters, within-group correlation parameters, and correlation ratio parameters) in each simulation study and fix the other two kinds of parameters. {\color{black}Similarly, to avoid loss of generality, we explore three choices including only simultaneous signals, both simultaneous signals and non-simultaneous signals exist in one dataset, and both simultaneous signals and non-simultaneous signals exist in three datasets.} 
\begin{itemize}
\item Sparsity level parameters: $s_0, s_1, s_2,s_3,s_4,s_{12},s_{13},s_{14},s_{23},s_{24},s_{34},s_{123},s_{124},s_{134},s_{234}.$ 
{\color{black}
\begin{align*}
    &\text{1. Vary $s_0 = 6,8,10,12,14,16,18,20$,}\\
    &\text{fixing $s_1=s_2=s_3=s_4=s_{12}=s_{13}=s_{14}=s_{23}=s_{24}=s_{34}=s_{123}=s_{124}=s_{134}=s_{234}=0$;}\\
    &\text{2. Vary $s_1=s_2=s_3=s_4=0,1,2,3,4,5,6$,}\\
    &\text{Fixing $s_0 =12,s_{12}=s_{13}=s_{14}=s_{23}=s_{24}=s_{34}=s_{123}=s_{124}=s_{134}=s_{234}=0$;}\\
    &\text{3. Vary $s_{123}=s_{124}=s_{134}=s_{234}= 0,1,2,3,4,5,6$;}\\
    &\text{Fixing $s_0 =12,s_1=s_2=s_3=s_4=s_{12}=s_{13}=s_{14}=s_{23}=s_{24}=s_{34}=0$.}
\end{align*}
}

\item Within-group correlation parameters: $\rho_1, \rho_2,\rho_3,\rho_4$. We fix {\color{black}$\gamma_1=0.1,\gamma_2=0.15,\gamma_3=0.1,\gamma_4=0.05$} and vary within-group correlations $\rho_1=\rho_2=\rho_3=\rho_4 \in \{0.05, 0.15, \dots, 0.95\}$ for the following choice of $s_0, s_1, s_2,s_3,s_4,s_{12},s_{13},s_{14},s_{23},s_{24},s_{34},s_{123},s_{124},s_{134},s_{234}$.
{\color{black}
\begin{align*}
       \textbf{choice 1}:& s_0=12,s_1=s_2=s_3=s_4=s_{12}=s_{13}=s_{14}=s_{23}=s_{24}=s_{34}=0,\\&s_{123}=s_{124}=s_{134}=s_{234}=0;\\
       \textbf{choice 2}:& s_0=12,s_1=s_2=s_3=s_4=4, s_{12}=s_{13}=s_{14}=s_{23}=s_{24}=s_{34}=0,\\&s_{123}=s_{124}=s_{134}=s_{234}=0;\\
       \textbf{choice 3}: & s_0=12,s_{123}=s_{124}=s_{134}=s_{234}=5,s_1=s_2=s_3=s_4=0,\\&s_{12}=s_{13}=s_{14}=s_{23}=s_{24}=s_{34}=0.
    \end{align*}
}

\item Correlation ratio parameters: $\gamma_1, \gamma_2,\gamma_3,\gamma_4$. We fix {\color{black}$\rho_1=0.5,\rho_2=0.4,\rho_3=0.5,\rho_4=0.6$} and vary correlation ratio $\gamma_1=\gamma_2=\gamma_3=\gamma_4 \in \{0,0.05,0.1,\dots,0.5\}$ for the following choice of $s_0, s_1, s_2,s_3,s_4,s_{12},s_{13},s_{14},s_{23},s_{24},s_{34},\\s_{123},s_{124},s_{134},s_{234}$. Then, the between-group correlations are calculated as $\rho_k\gamma_k$.
{\color{black}

\begin{align*}
       \textbf{choice 1}:& s_0=12,s_1=s_2=s_3=s_4=s_{12}=s_{13}=s_{14}=s_{23}=s_{24}=s_{34}=0,\\&s_{123}=s_{124}=s_{134}=s_{234}=0;\\
       \textbf{choice 2}:& s_0=12,s_1=s_2=s_3=s_4=4, s_{12}=s_{13}=s_{14}=s_{23}=s_{24}=s_{34}=0,\\&s_{123}=s_{124}=s_{134}=s_{234}=0;\\
       \textbf{choice 3}: & s_0=12,s_{123}=s_{124}=s_{134}=s_{234}=5,s_1=s_2=s_3=s_4=0,\\&s_{12}=s_{13}=s_{14}=s_{23}=s_{24}=s_{34}=0.
\end{align*}
}

\end{itemize}

Given the disparity in types within the group across various sites, it is clear that \textit{pooling} \textit{Individual (Lasso)} methods are not feasible to be applied. We compare the performance of the proposed method with \textit{Intersection} and \textit{Individual (Group Lasso)}.
}

{\color{black}
\subsection*{E.5: Additional simulation for power comparisons}
Due to the high dimensionality of our data, we apply the Chi-square test to examine the goodness of fit of the filter statistics W distribution from the family of symmetric distribution, instead of assessing each group of predictors individually. We use the parameter settings as the description in section E.1.2 for K=3, E.2.2 for K=4, and E.3.2 for K=5. In order to avoid the redundant presentation of results, we only show results with default parameter settings for the amplitude of signals, the within-group correlations, the correlation ratios, and the signal noise parameters for the mixed model settings. The results are summarized in Table S1 for K=3, Table S2 for K =4, and Table S3 for K =5 respectively.

The symmetrical nulls numerically demonstrate that the sign of $W_j$ is unrelated to its size when signals are absent from all three datasets, maintaining a probability $P\{W_j>0\}=\frac{1}{2}$. This characteristic is essential for the effective application of the false discovery rate (FDR) control theorem. Conversely, all of our non-nulls have corresponded asymmetric distributions of $W_j$, indicating that the presence of signals in all datasets correlates with an increased probability of being positive, thereby enhancing the test's power. Although we have a few nulls that are asymmetric, the results are still under the nominal FDR at 0.2.

\section*{Web Appendix F: Additional simulation results} \label{app:sim-res}
\subsection*{F.1: Additional simulation results for {\color{black} Setting 1 when }K=3,4,5}
In Figure \ref{fig:figure3-1}, we show results for continuous setting with \textbf{Scenario 1} (same strengths) for K=3 {\color{black}when $n_k=1000$}. When only simultaneous signals among K datasets exist ($s_1=s_2=s_3=s_{12}=s_{13}=s_{23}=0$), only \textit{Individual (Lasso)} (i.e., individual sequential knockoff methods with individual filter (Lasso)) fails to control FDR. When both simultaneous signals among K datasets, and non-mutual signals exist, \textit{Individual (Lasso)} and \textit{Pooling}, and \textit{Intersection} methods fail to control FDR (Figure \ref{fig:figure3-1} right column). In general, only our proposed \textit{GS knockoffs} and \textit{Individual (Group Lasso)} (i.e., individual sequential knockoff methods with group lasso fitting)) control FDR across all the settings. However, \textit{Individual (Group Lasso)} method loses the power when the within-group correlation is strong.

Figures \ref{fig:figure3-2} and \ref{fig:figure3-3} show the results for binary and mixed data settings with \textbf{Scenario 1} (same strengths) for K=3 {\color{black}when $n_k=1000$}. The results are consistent with the continuous case, with power slightly lower for all methods than the continuous case. \textit{Pooling} method fails to control FDR when specific signals in $k-$th dataset or mutual signals in two datasets exist. \textit{Intersection} method fails to control FDR when mutual signals among K datasets and non-mutual signals exist. For the individual knockoff methods, the FDR is not guaranteed to be controlled for \textit{Individual (Lasso)}. For \textit{Individual (Group Lasso)}, the FDR can be controlled but the power is lower than the proposed (\textit{GS knockoffs} methods when within-group correlation is strong. 

{\color{black}Figure \ref{fig:figure3-} presents the results for the mixed data settings under \textbf{Scenario 1} (same strengths) with \(K=3\) and a small sample size of \(n_k=200\). The results also indicate that the proposed \textit{GS knockoffs} method achieves the best performance. Despite the limited sample size, the \textit{GS knockoffs} method maintains a stable power, while the \textit{Pooling} and \textit{Intersection} methods exhibit lower power than the \textit{GS knockoffs}. Furthermore, the \textit{Individual (Lasso)} and \textit{Individual (Group Lasso)} methods display lower power (around 0.1), particularly when the within-group correlation is strong.
}

{\color{black}
Figures \ref{fig:figure3-4} and \ref{fig:figure3-5} show the results for mixed data settings with \textbf{Scenario 1} (same strengths) for K=4 and K=5, respectively. The findings align with those observed in the K=3 scenarios. An increase in K results in a marginal reduction in power across all three methodologies. The \textit{GS knockoff} successfully maintains FDR control while exhibiting robust power. Despite the \textit{Pooling} method achieving the highest power, it exhibits a substantial false discovery proportion (FDP) in the presence of non-mutual signals. The \textit{Intersection} method shows similar power with the \textit{GS knockoff} but lacks a guaranteed FDR control. Regarding the individual knockoff approaches, their outcomes agree with the K=3 case. The group filter, namely \textit{Individual (Group Lasso)}, is capable of regulating the group FDR, but its power significantly diminishes under high within-group correlation. Conversely, the individual filter (\textit{Individual (Lasso)}) does not effectively manage group FDR.
\subsection*{F.2: Additional simulation results {\color{black} for Setting 2 when K=4}}
Figure \ref{fig:figure3-6} presents the simulation results conducted with varying sample sizes and types across different sites. The results are consistent with what we observed in the scenarios with the same sample sizes and the same types within the group across the sites. Notably, the disparities in sample sizes and types at various sites do not impinge upon the efficacy of the proposed \textit{GS knockoff method}. This robustness underscores the method's adaptability to diverse data conditions, maintaining its performance regardless of sample size and type variations. This attribute of the \textit{GS knockoff} method is particularly advantageous in multi-site studies where such variability is common, ensuring reliable and stable results across different research settings.

}


\section*{Web Appendix G: Additional information for real data analysis} \label{app:real}
The detailed cohort generating inclusion and exclusion steps are illustrated in Figure \ref{fig:figure6}.
There are {\color{black}40 candidate} risk factors (37 group-level risk factors) included in this analysis: demographic information include sex (``Male", ``Felmale"), age at COVID (continuous), race (``Hispanic or Latino Any Race", ``Others"), binary indicators include mild liver disease, moderate severe liver disease, rheumatologic disease, dementia, congestive heart failure, substanceuse disorder, kidney disease, malignant cancer, cerebrovascular disease, peripheralvascular disease, heart failure, hemiplegia or paraplegia, psychosis, coronaryartery disease, systemic corticosteroids, depression, metastatic solid chronic lung disease, peptic ulcer, myocardial infarction, cardiomyopathies, hypertension, other immunocompromised, negative antibody, pulmonary embolism, tobacco smoker, solid organ or blood stem cell transplant, and some COVID related information include number of COVID vaccine dose, the usage of corticosteroids during the hospitalization, remdesivir usage during COVID,  COVID associated emergency department visit, and severity type (``Asymptomatic", ``Mild", ``Moderate", ``Severe"). Those indicators indicate patients have been diagnosed with those diseases or symptoms before or on the day they were confirmed to get an acute COVID infection (i.e. all the indicators are in two categories: ``Yes", or ``No"). {\color{black} In our analysis there are two group-level variables, obesity and diabetes. The obesity group includes two variables: obesity indicator, and BMI (continuous). Sites A1, B1, and B2 only collect the obesity indicator for obesity information, whereas site A2 also has the BMI information. The diabetes group is composed of three variables: diabetes complicated indicator, diabetes uncomplicated indicator, and pre-COVID glucose level (continuous).  Sites A1, B1, and B2 have all three diabetic variables in the diabetes group, while site A2 lacks the glucose level variable.} The full data dictionary can be found in \url{https://unite.nih.gov/workspace/report/ri.report.main.report.855e1f58-bf44-4343-9721-8b4c878154fe}.


\bibliography{example_paper}
\bibliographystyle{icml2022}

\newpage
\begin{table}[h]
\centering
\caption{The Chi-square test of symmetry for the distribution of the filter statistics W for \textbf{K=3} when \textbf{(a) Only simultaneous signals exist in three datasets} ($s_0 = 0, s_1=s_2=s_3=s_{12}=s_{13}=s_{23}=0$); \textbf{(b) Simultaneous signals and non-simultaneous signals exist in one and two datasets} ($s_0=12, s_1=s_2=s_3 =6, s_{12}=s_{23}=2,s_{13}=0$); \textbf{(c) Simultaneous signals and non-simultaneous signals exist in two datasets} ($s_0=12, s_1=s_2=s_3=0, s_{12}=s_{13}=s_{23} =6$).}
\label{table:1}
\begin{tabular}{c|c|c}
\hline
\multicolumn{3}{c}{(a) Only simultaneous signals exist in three datasets } \\ \hline
 & Symmetric & Not-symmetric \\ \hline
Non-nulls  & 0 & 12 \\ \hline
Nulls & 28 & 0 \\ \hline
\multicolumn{3}{c}{(b) Simultaneous signals and non-simultaneous signals exist in one and two datasets} \\ \hline
 & Symmetric & Not-symmetric \\ \hline
Non-nulls  & 0 & 12 \\ \hline
Nulls & 28 & 0 \\ \hline
\multicolumn{3}{c}{(c)  Simultaneous signals and non-simultaneous signals exist in two datasets} \\ \hline
 & Symmetric & Not-symmetric \\ \hline
Non-nulls  & 0 & 12 \\ \hline
Nulls & 28 & 0 \\ \hline
\end{tabular}
\end{table}

\newpage
\begin{table}[h]
\centering
\caption{The Chi-square test of symmetry for the distribution of the filter statistics W for \textbf{K=4} when \textbf{(a) Only simultaneous signals exist in four datasets} ($s_0=12, s_1=s_2=s_3=s_4=s_{12}=s_{13}=s_{14}=s_{23}=s_{24}=s_{34}=s_{123}=s_{124}=s_{134}=s_{234}=0$); \textbf{(b) Simultaneous signals and non-simultaneous signals exist in one dataset} ($s_0=12, s_1=s_2=s_3=s_4=4, s_{12}=s_{13}=s_{14}=s_{23}=s_{24}=s_{34}=s_{123}=s_{124}=s_{134}=s_{234}=0$); \textbf{(c) Simultaneous signals and non-simultaneous signals exist in three datasets} ($s_0=12, s_1=s_2=s_3=s_4=s_{12}=s_{13}=s_{14}=s_{23}=s_{24}=s_{34}=0, s_{123}=s_{124}=s_{134}=s_{234}=6$).}

\label{table:2}
\begin{tabular}{c|c|c}
\hline
\multicolumn{3}{c}{(a) Only simultaneous signals exist in four datasets} \\ \hline
 & Symmetric & Not-symmetric \\ \hline
Non-nulls  & 0 & 12 \\ \hline
Nulls & 27 & 1 \\ \hline
\multicolumn{3}{c}{(b) Simultaneous signals and non-simultaneous signals exist in one dataset} \\ \hline
 & Symmetric & Not-symmetric \\ \hline
Non-nulls  & 0 & 12 \\ \hline
Nulls & 28 & 0 \\ \hline
\multicolumn{3}{c}{(c) Simultaneous signals and non-simultaneous signals exist in three datasets } \\ \hline
 & Symmetric & Not-symmetric \\ \hline
Non-nulls  & 0 & 12 \\ \hline
Nulls & 28 & 0 \\ \hline
\end{tabular}
\end{table}

\newpage
\begin{table}[h]
\centering
\caption{The Chi-square test of symmetry for the distribution of the filter statistics W for \textbf{K=5} when \textbf{(a) Only simultaneous signals exist in five datasets} ($s_0=12, s_1=s_2=s_3=s_4=s_5=s_{12}=s_{13}=s_{14}=s_{15}=s_{23}=s_{24}=s_{25}=s_{34}=s_{35}=s_{45}=s_{123}=s_{124}=s_{125}=s_{134}=s_{135}=s_{145}=s_{234}=s_{235}=s_{245}=s_{345}=s_{1234}=s_{1235}=s_{1245}=s_{1345}=s_{2345}=0$); 
\textbf{ (b)Simultaneous signals and non-simultaneous signals exist in three datasets} ($s_0=12, s_{234}=s_{235}=s_{245}=s_{345}=4, s_1=s_2=s_3=s_4=s_5=s_{12}=s_{13}=s_{14}=s_{15}=s_{23}=s_{24}=s_{25}=s_{34}=s_{35}=s_{45}=s_{123}=s_{124}=s_{125}=s_{134}=s_{135}=s_{145}=s_{1234}=s_{1235}=s_{1245}=s_{1345}=s_{2345}=0$; \textbf{(c)Simultaneous signals and non-simultaneous signals exist in four datasets} ($s_0 = 12, s_{1234}=s_{1235}=s_{1245}=s_{1345}=s_{2345}=4, s_1=s_2=s_3=s_4=s_5=s_{12}=s_{13}=s_{14}=s_{15}=s_{23}=s_{24}=s_{25}=s_{34}=s_{35}=s_{45}=s_{123}=s_{124}=s_{125}=s_{134}=s_{135}=s_{145}=s_{234}=s_{235}=s_{245}=s_{345}=0.$)}
\label{table:3}
\begin{tabular}{c|c|c}
\hline
\multicolumn{3}{c}{(a) Only simultaneous signals exist in five datasets } \\ \hline
 & Symmetric & Not-symmetric \\ \hline
Non-nulls  & 0 & 12 \\ \hline
Nulls & 27 & 1 \\ \hline
\multicolumn{3}{c}{(b) Simultaneous signals and non-simultaneous signals exist in three datasets } \\ \hline
 & Symmetric & Not-symmetric \\ \hline
Non-nulls  & 0 & 12 \\ \hline
Nulls & 27 & 1 \\ \hline
\multicolumn{3}{c}{(c) Simultaneous signals and non-simultaneous signals exist in four datasets } \\ \hline
 & Symmetric & Not-symmetric \\ \hline
Non-nulls  & 0 & 12 \\ \hline
Nulls & 28 & 0 \\ \hline
\end{tabular}
\end{table}
}

\begin{figure}
    \centering
    \includegraphics[scale=0.23]{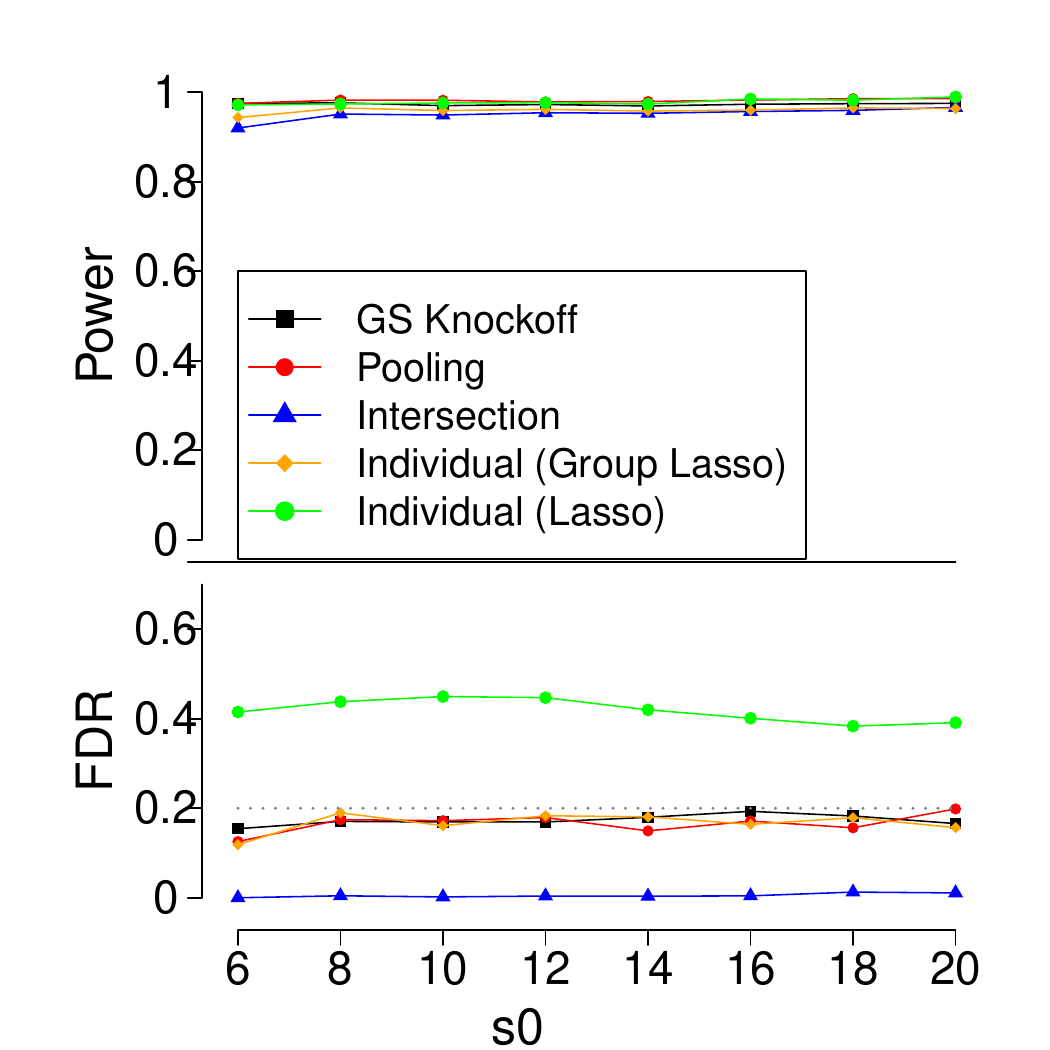}
    \includegraphics[scale=0.23]{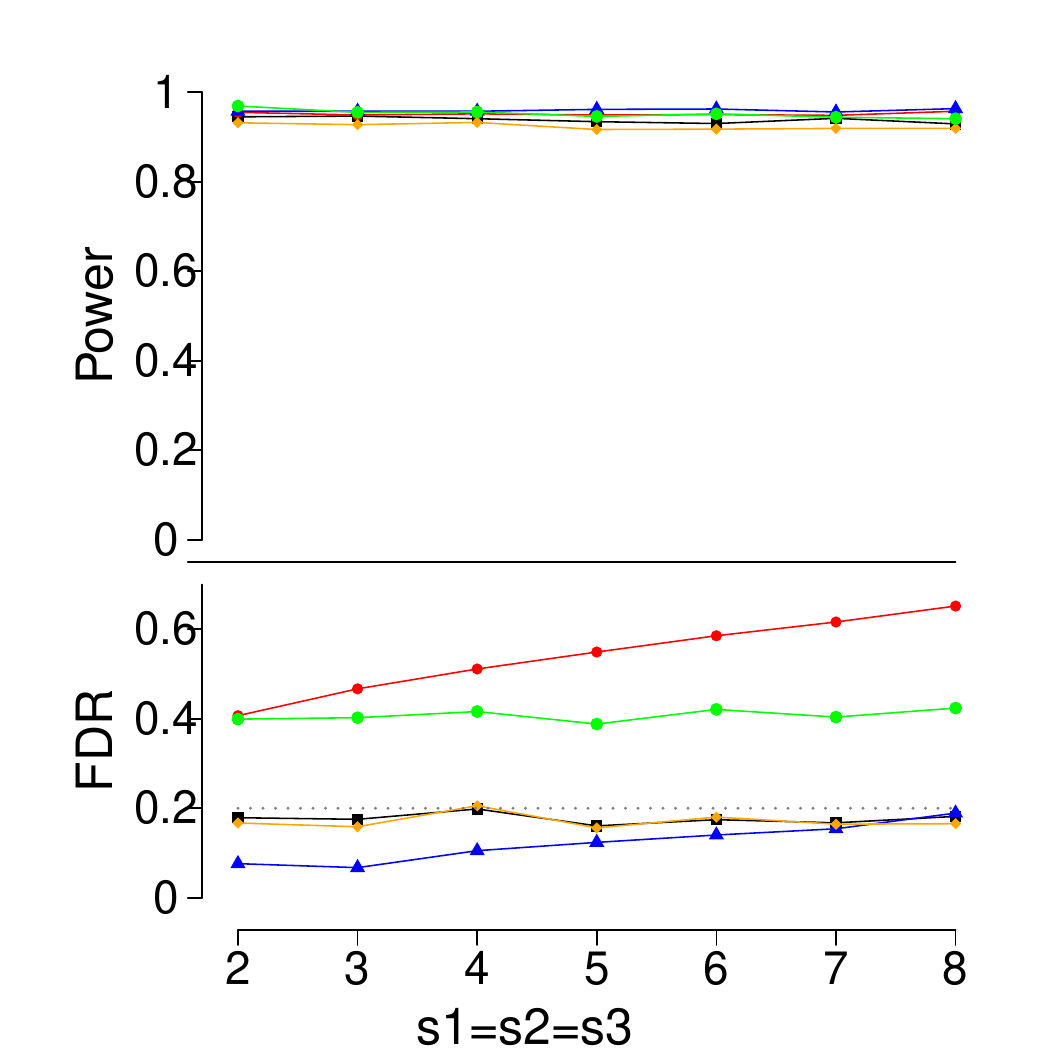}
    \includegraphics[scale=0.23]{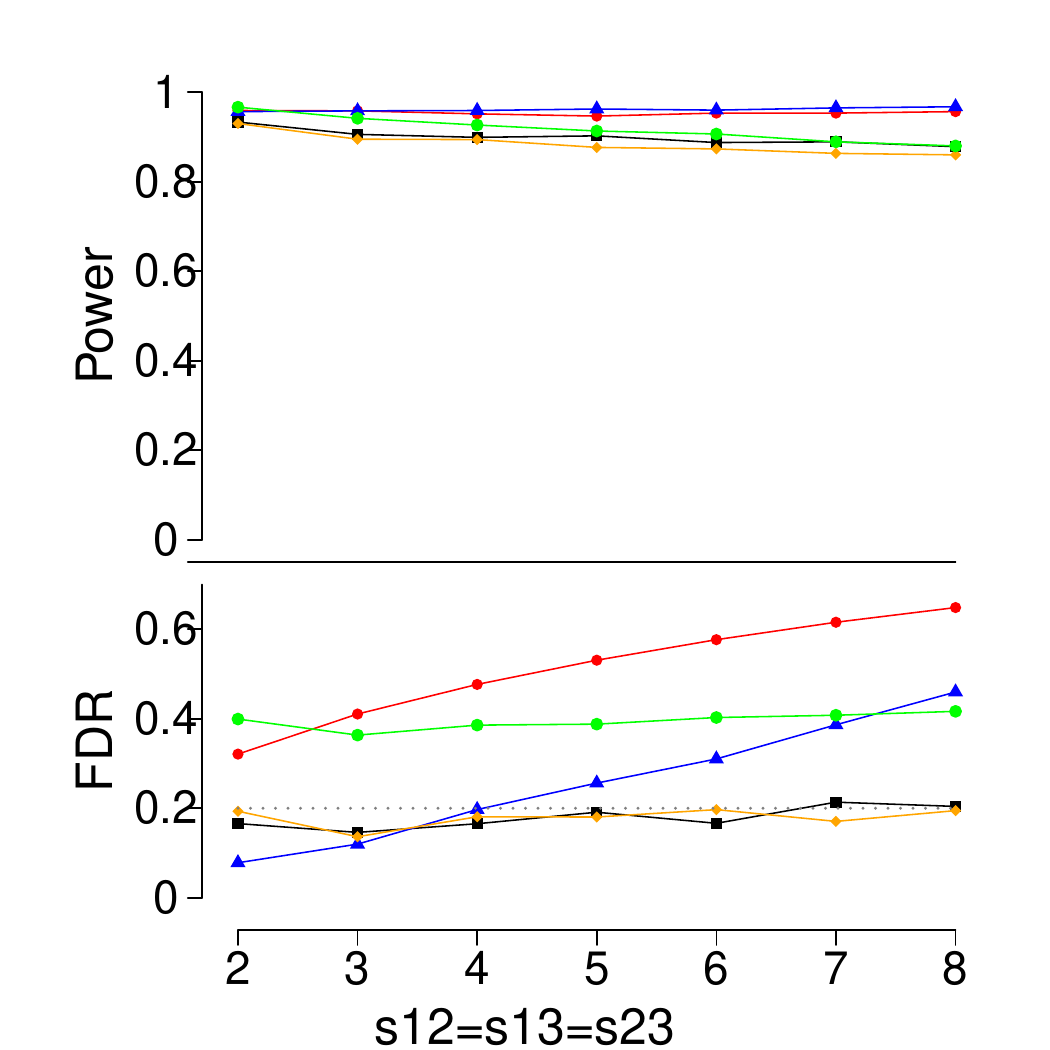}\\
    \includegraphics[scale=0.23]{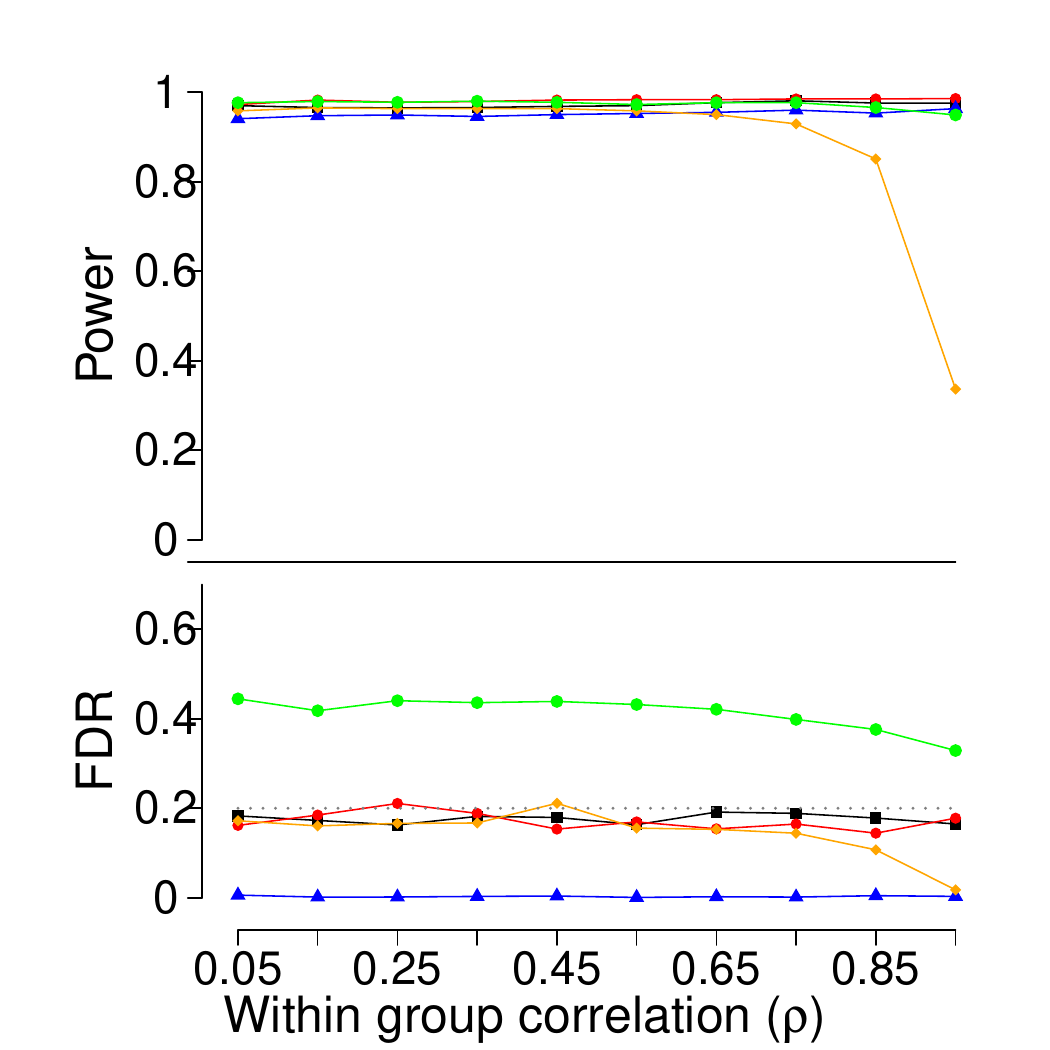}
    \includegraphics[scale=0.23]{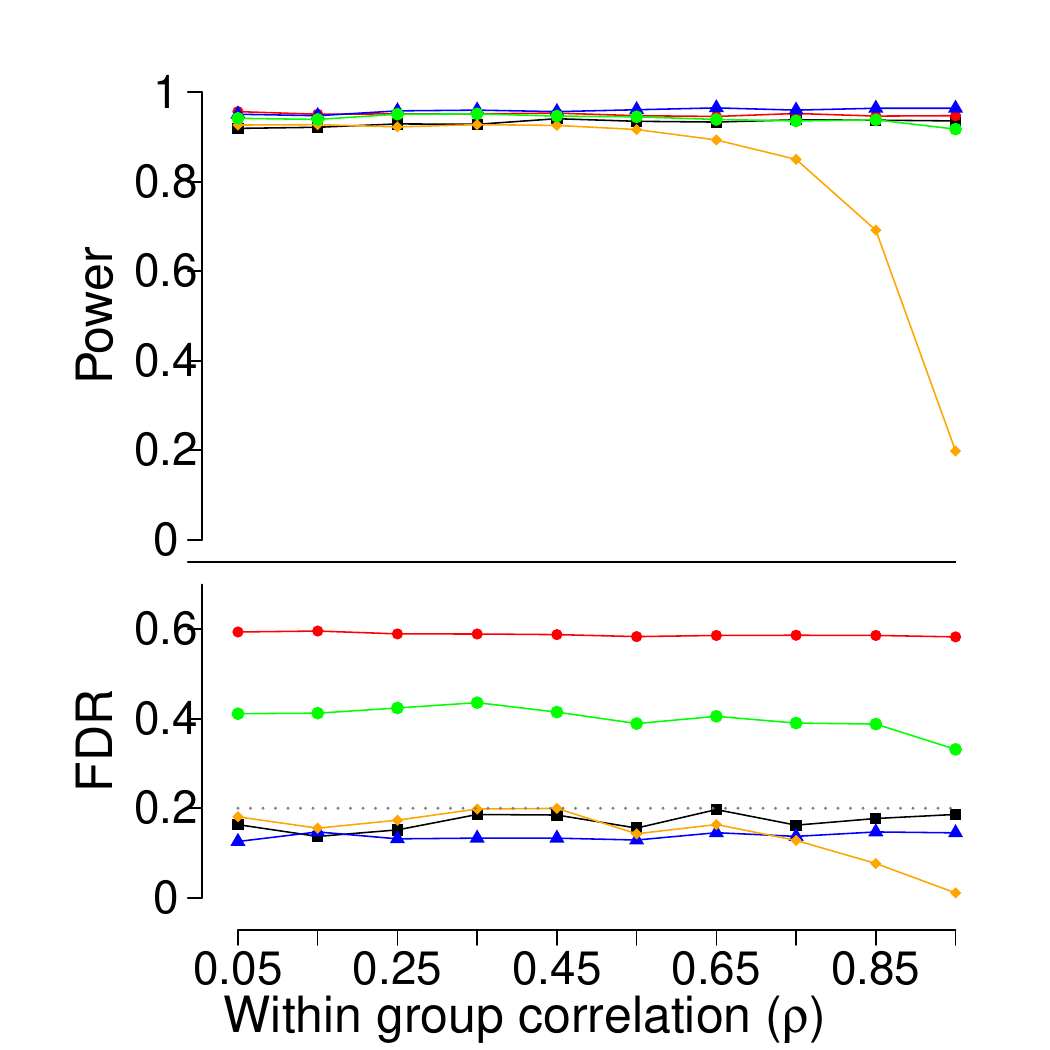}
    \includegraphics[scale=0.23]{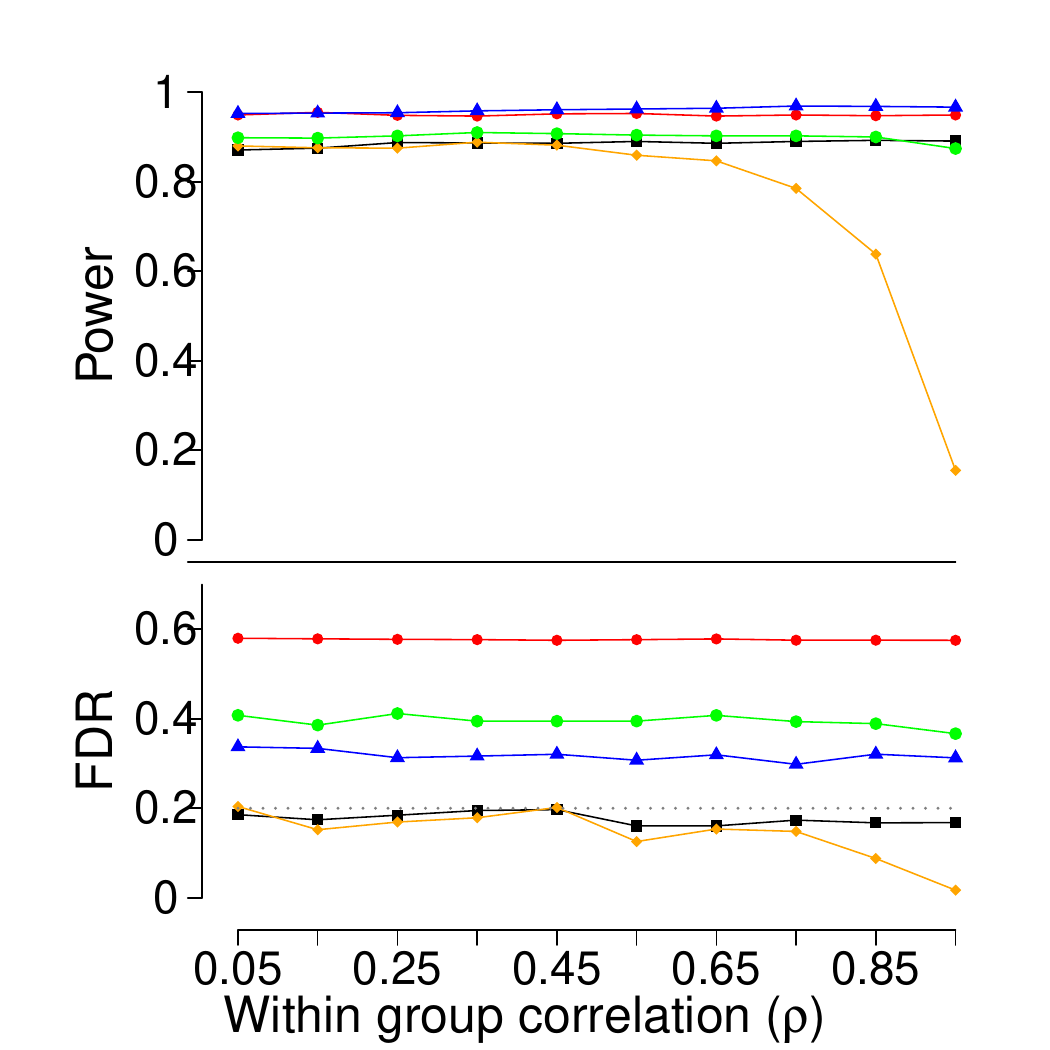}\\
    \includegraphics[scale=0.23]{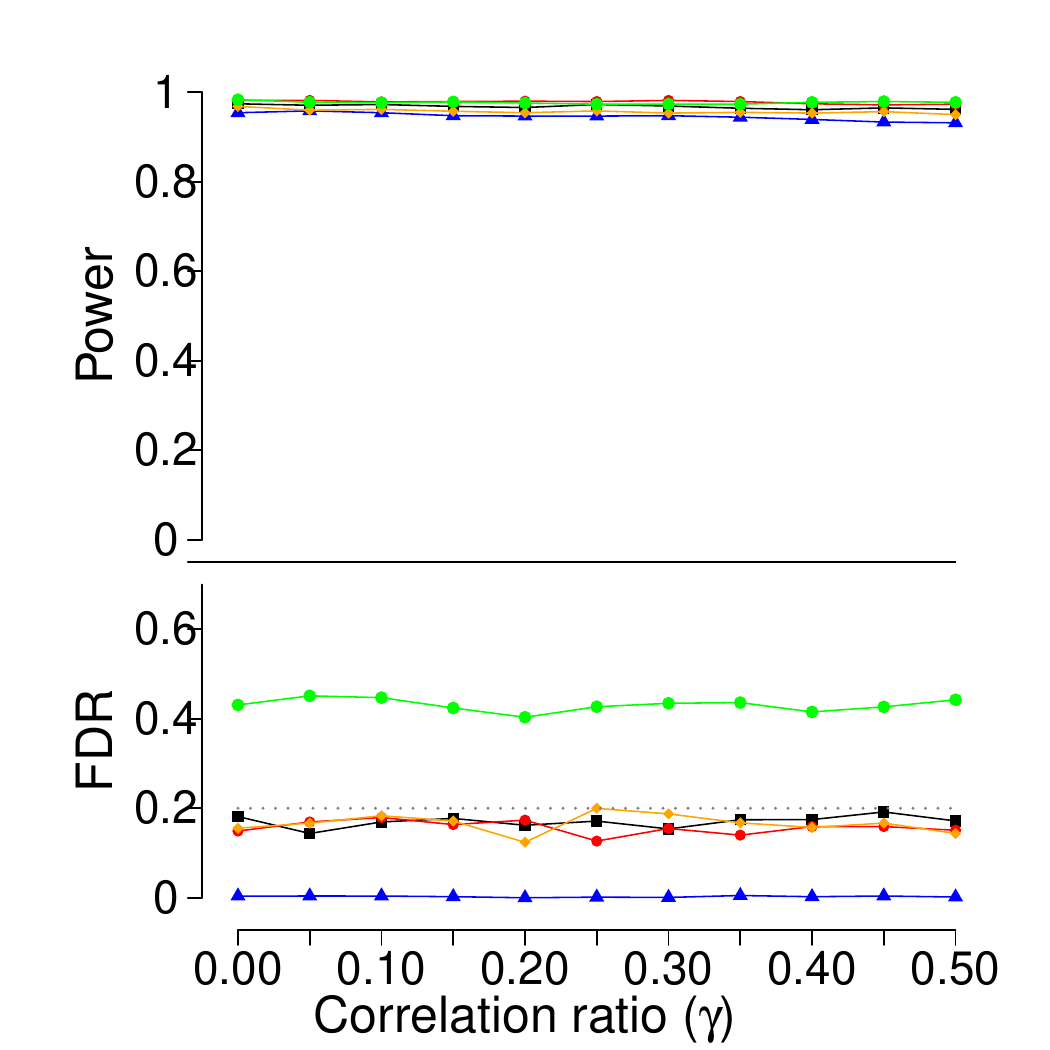}
    \includegraphics[scale=0.23]{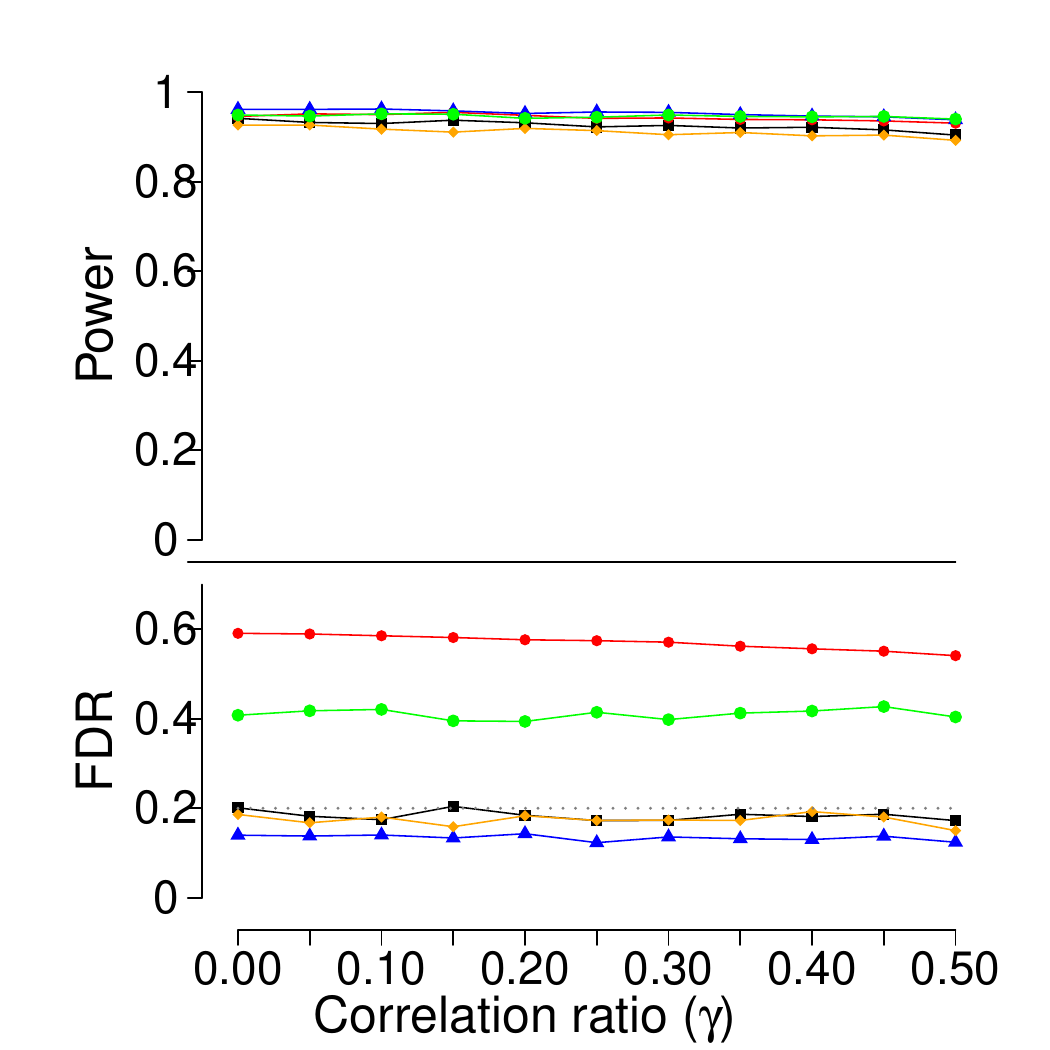}
    \includegraphics[scale=0.23]{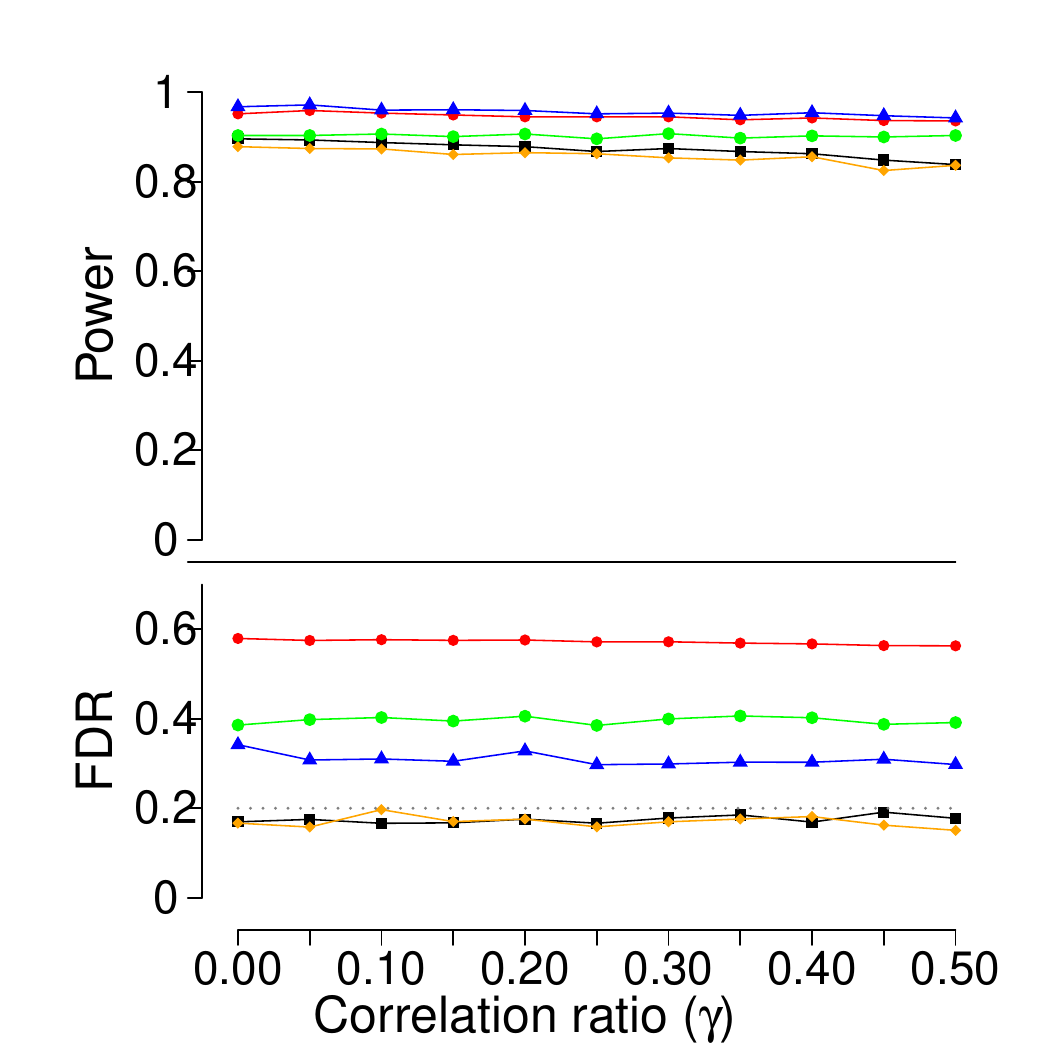}
    
    \caption{{\color{black} The power and the FDR for identifying group level simultaneous signals with data generated from \textbf{Setting 1} for the \textbf{Continuous} models (K=3) on \textbf{Scenario 1} when {\color{black}$n_1=n_2=n3=1000$.}. Left column includes settings with $s_0 \neq 0, s_1=s_2=s_3=s_{12}=s_{13}=s_{23}=0$; middle column includes settings with $s_0=12, s_1=s_2=s_3 \neq 0, s_{12}=s_{23}=2,s_{13}=0$; right column includes settings with $s_0=6, s_1=s_2=s_3=0, s_{12}=s_{13}=s_{23} \neq 0.$}}
    \label{fig:figure3-1}
\end{figure}

\begin{figure}
    \centering
    \includegraphics[scale=0.23]{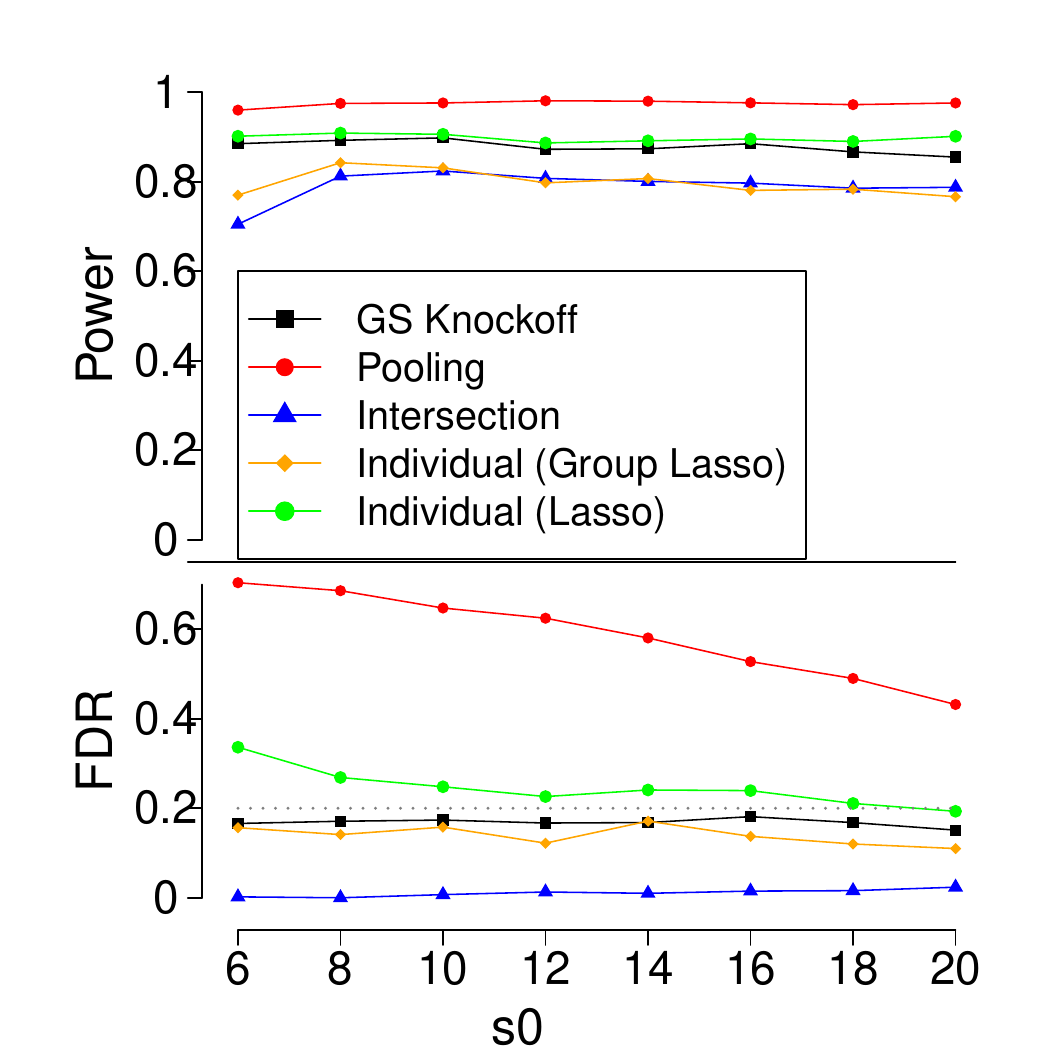}
    \includegraphics[scale=0.23]{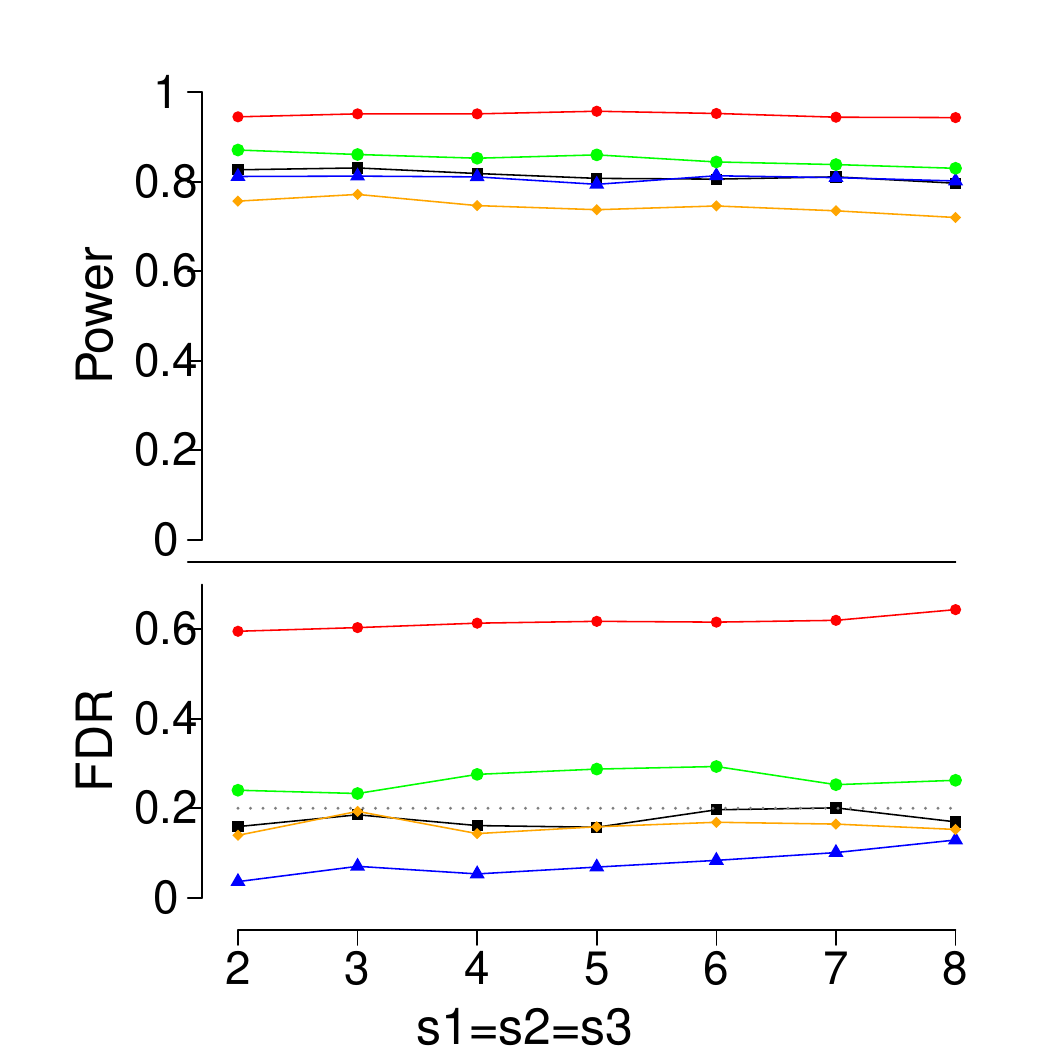}
    \includegraphics[scale=0.23]{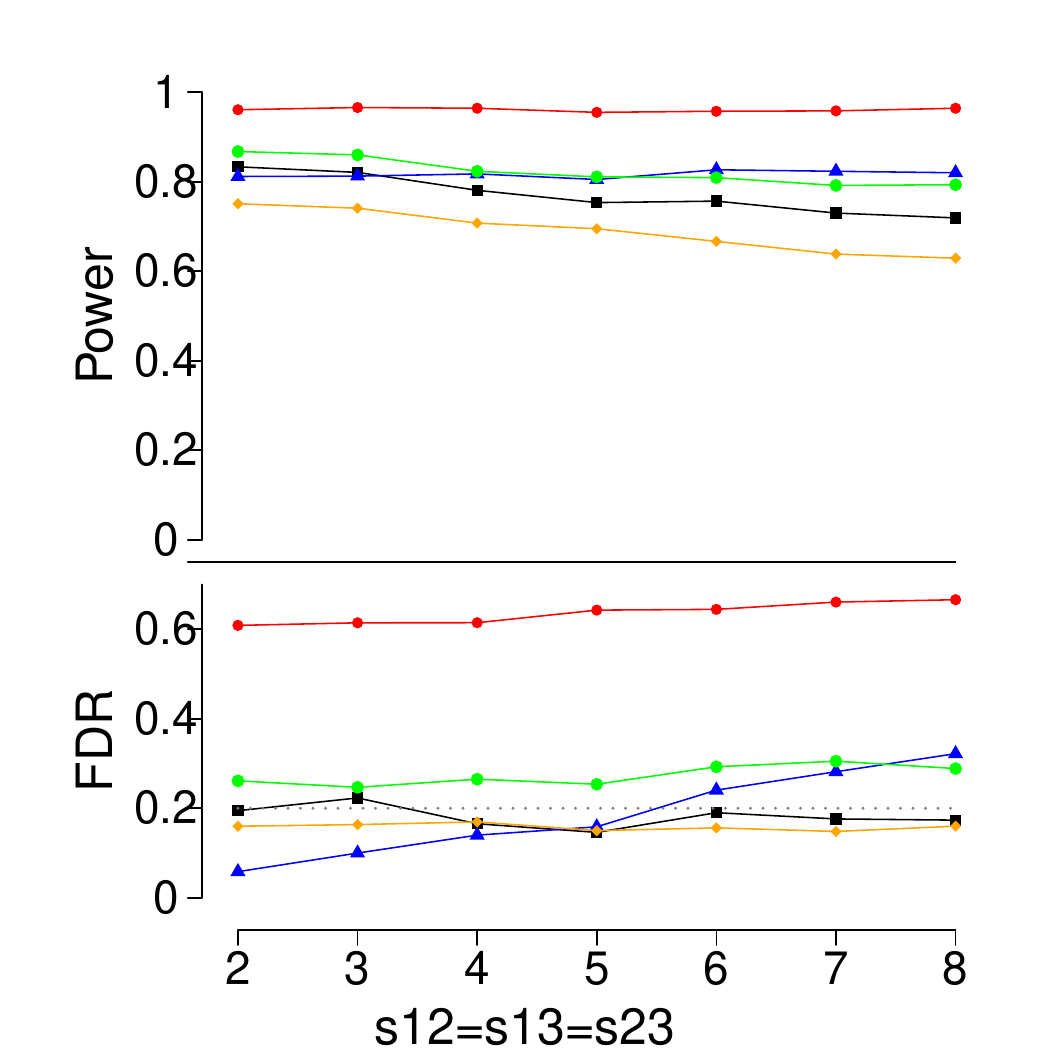}\\
    \includegraphics[scale=0.23]{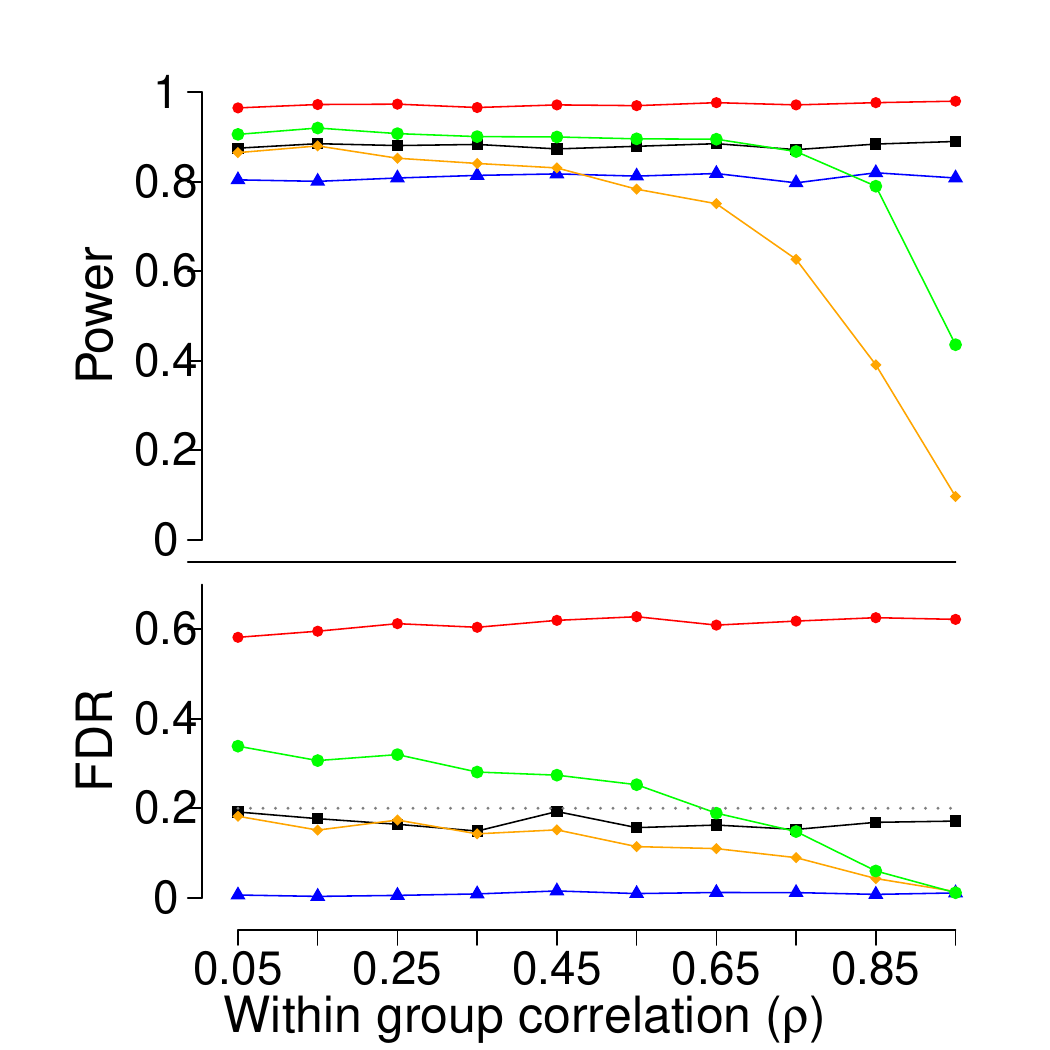}
    \includegraphics[scale=0.23]{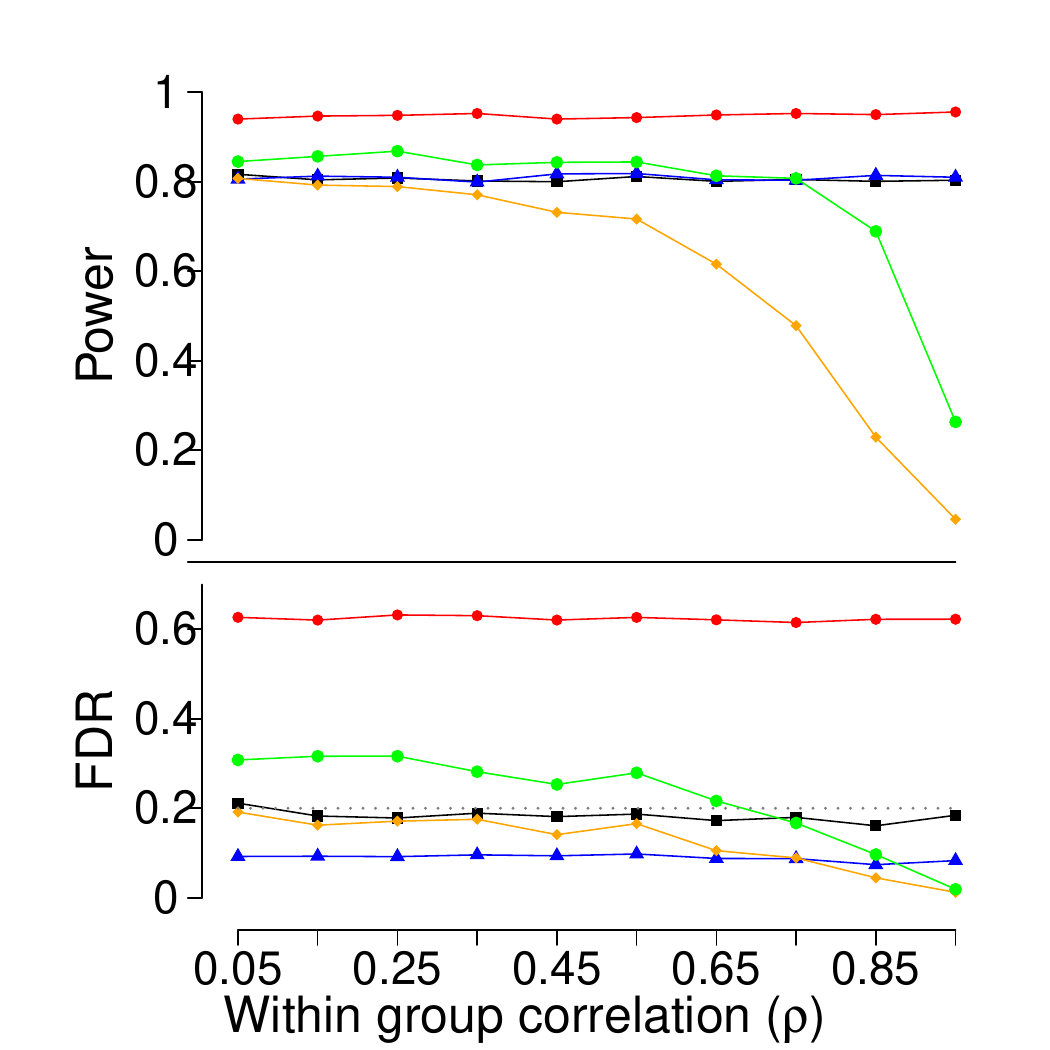}
    \includegraphics[scale=0.23]{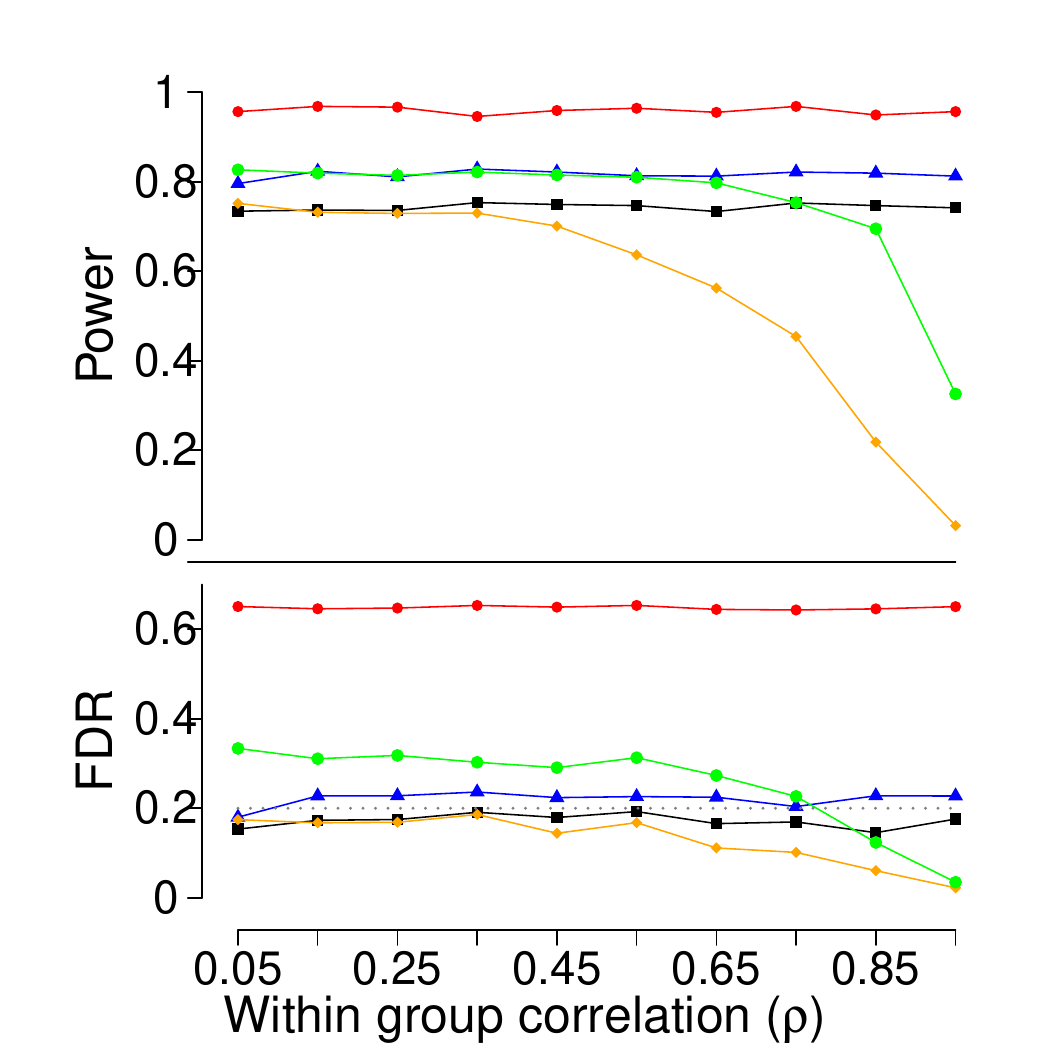}\\
    \includegraphics[scale=0.23]{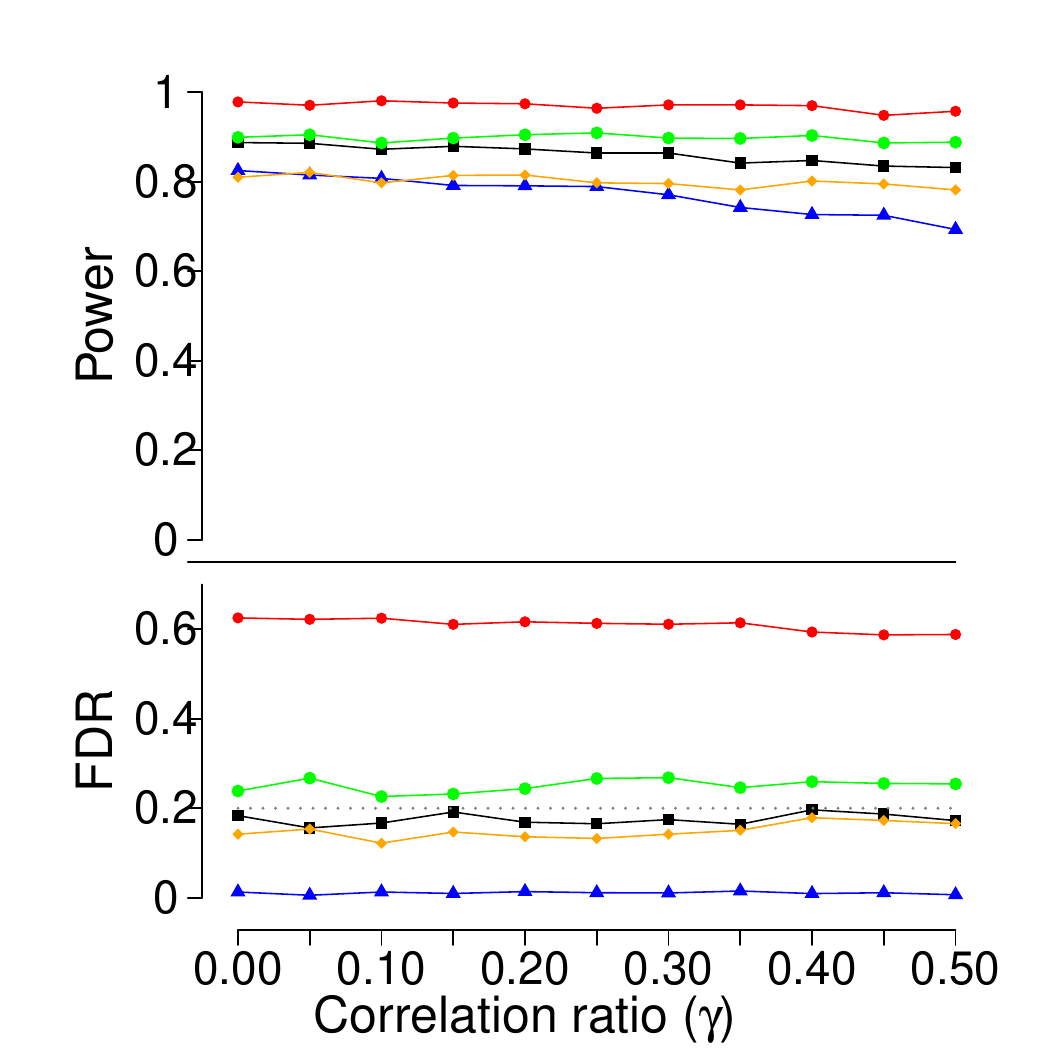}
    \includegraphics[scale=0.23]{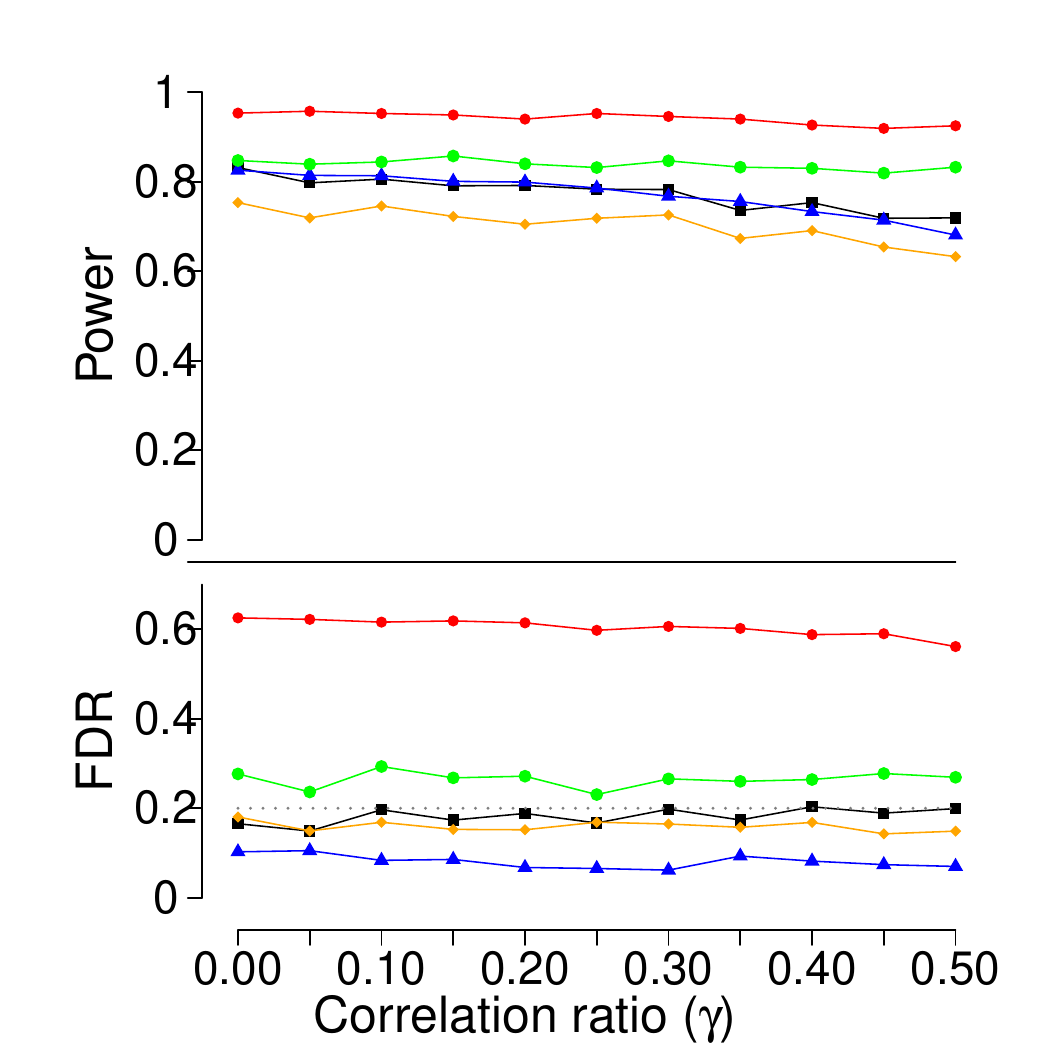}
    \includegraphics[scale=0.23]{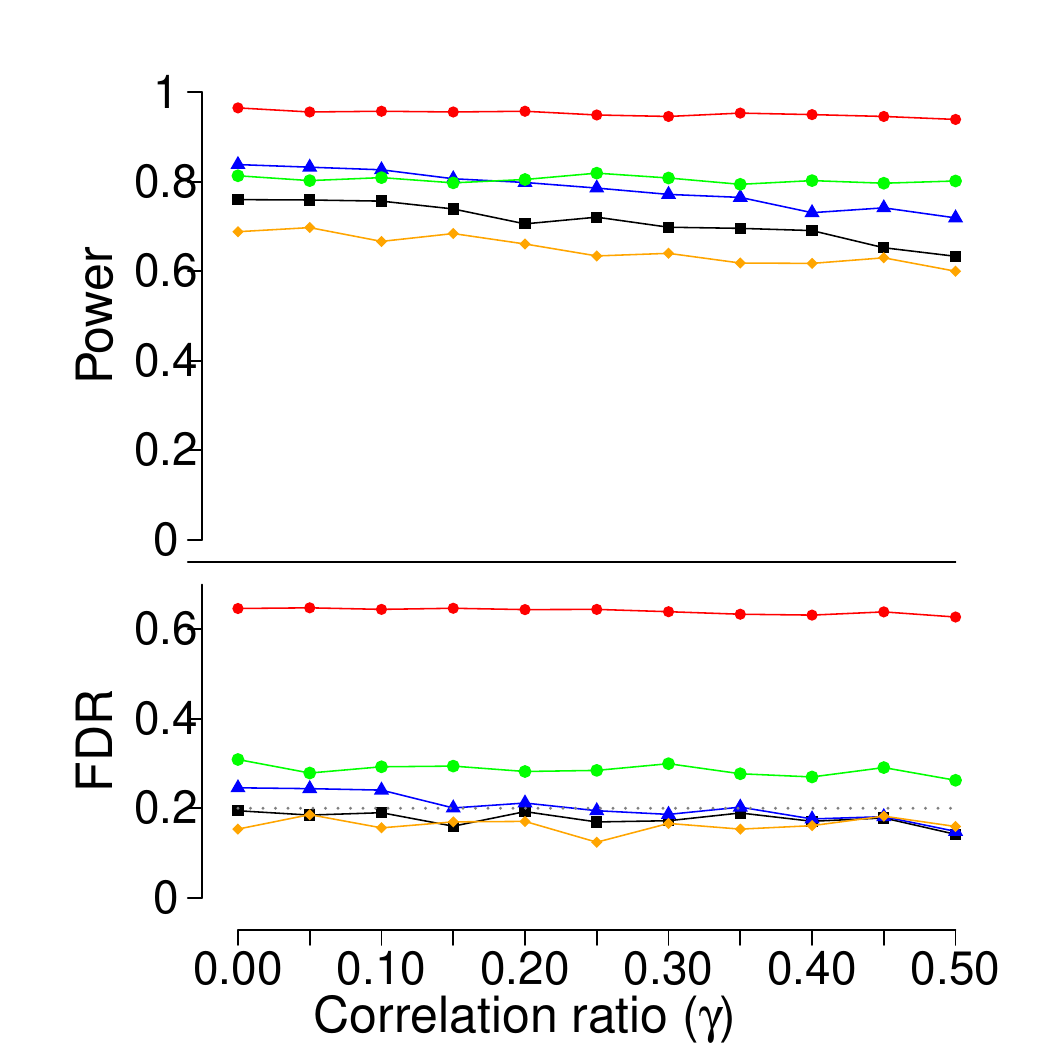}
    
    \caption{{\color{black}The power and the FDR for identifying group level simultaneous signals with data generated from \textbf{Setting 1} for the \textbf{Binary} models (K=3) on \textbf{Scenario 1} when {\color{black}$n_1=n_2=n3=1000$.} Left column includes settings with $s_0 \neq 0, s_1=s_2=s_3=s_{12}=s_{13}=s_{23}=0$; middle column includes settings with $s_0=6, s_1=s_2=s_3 \neq 0, s_{12}=s_{23}=2,s_{13}=0$; right column includes settings with $s_0=12, s_1=s_2=s_3=0, s_{12}=s_{13}=s_{23} \neq 0.$}}
    \label{fig:figure3-2}
\end{figure}

\begin{figure}
    \centering
    \includegraphics[scale=0.23]{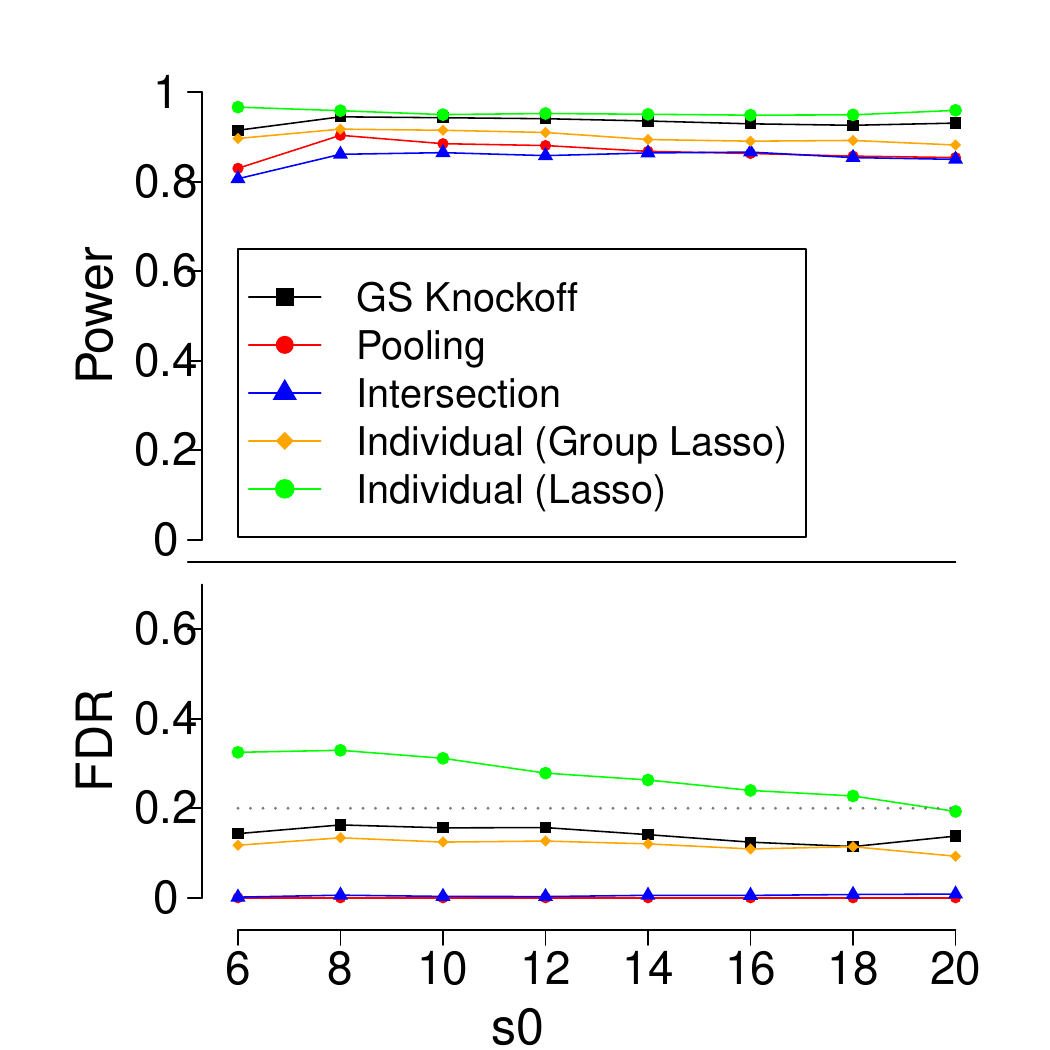}
    \includegraphics[scale=0.23]{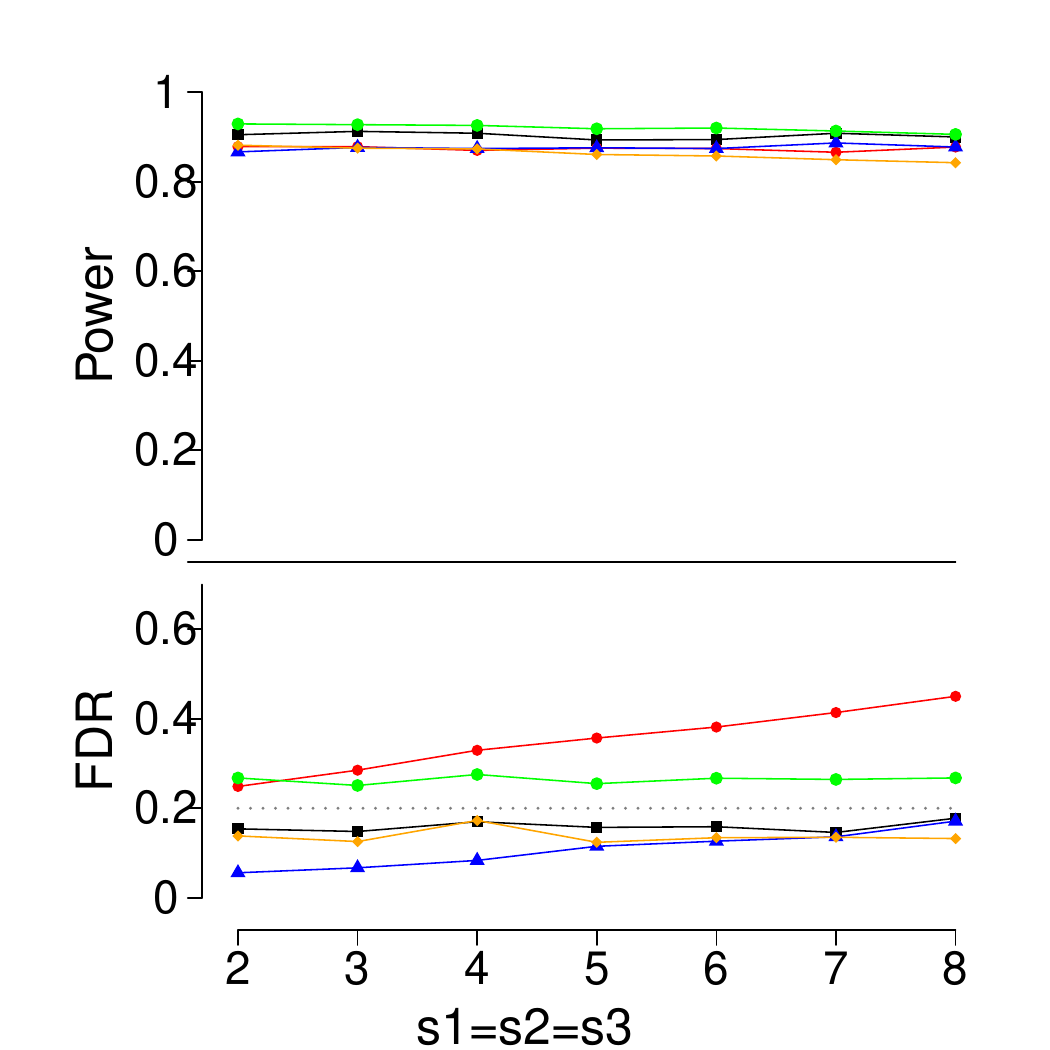}
    \includegraphics[scale=0.23]{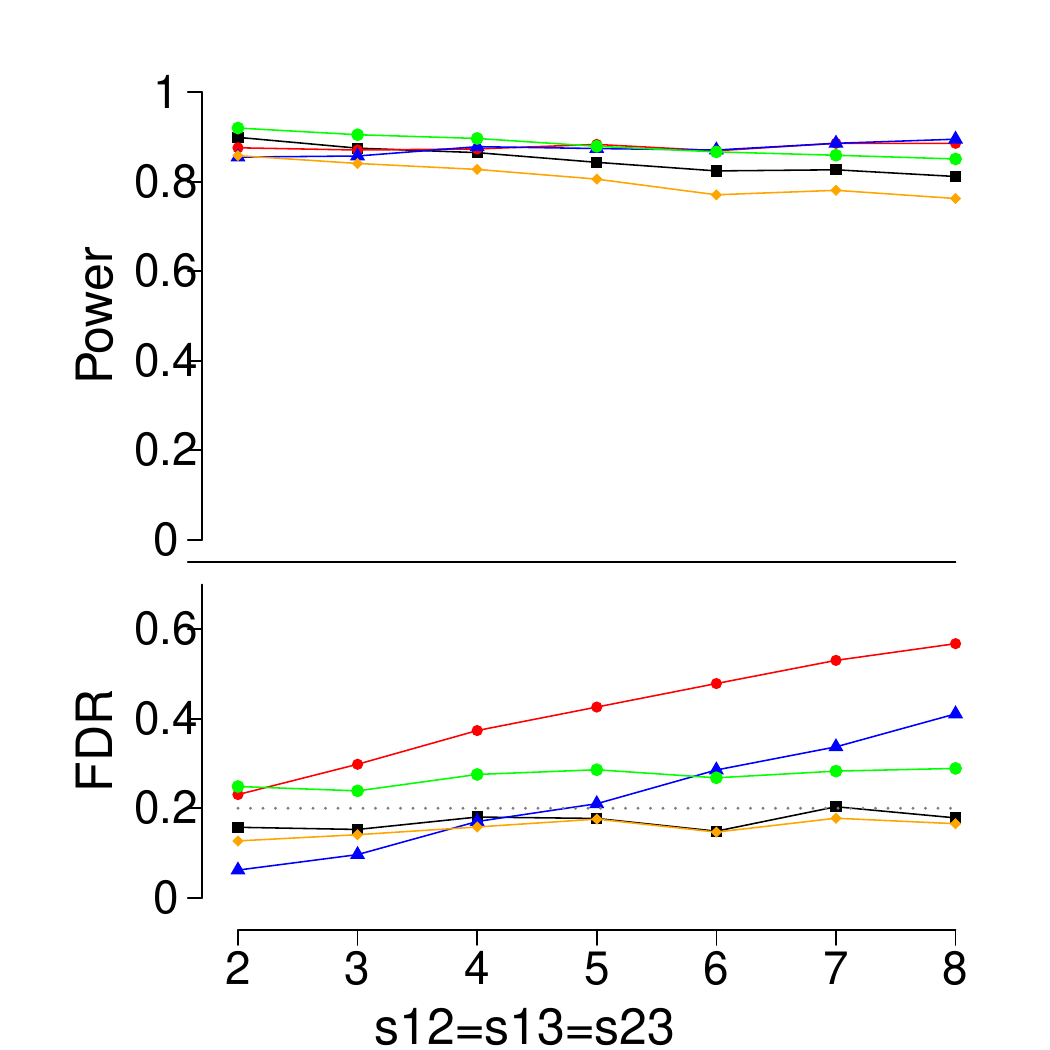}\\
    \includegraphics[scale=0.23]{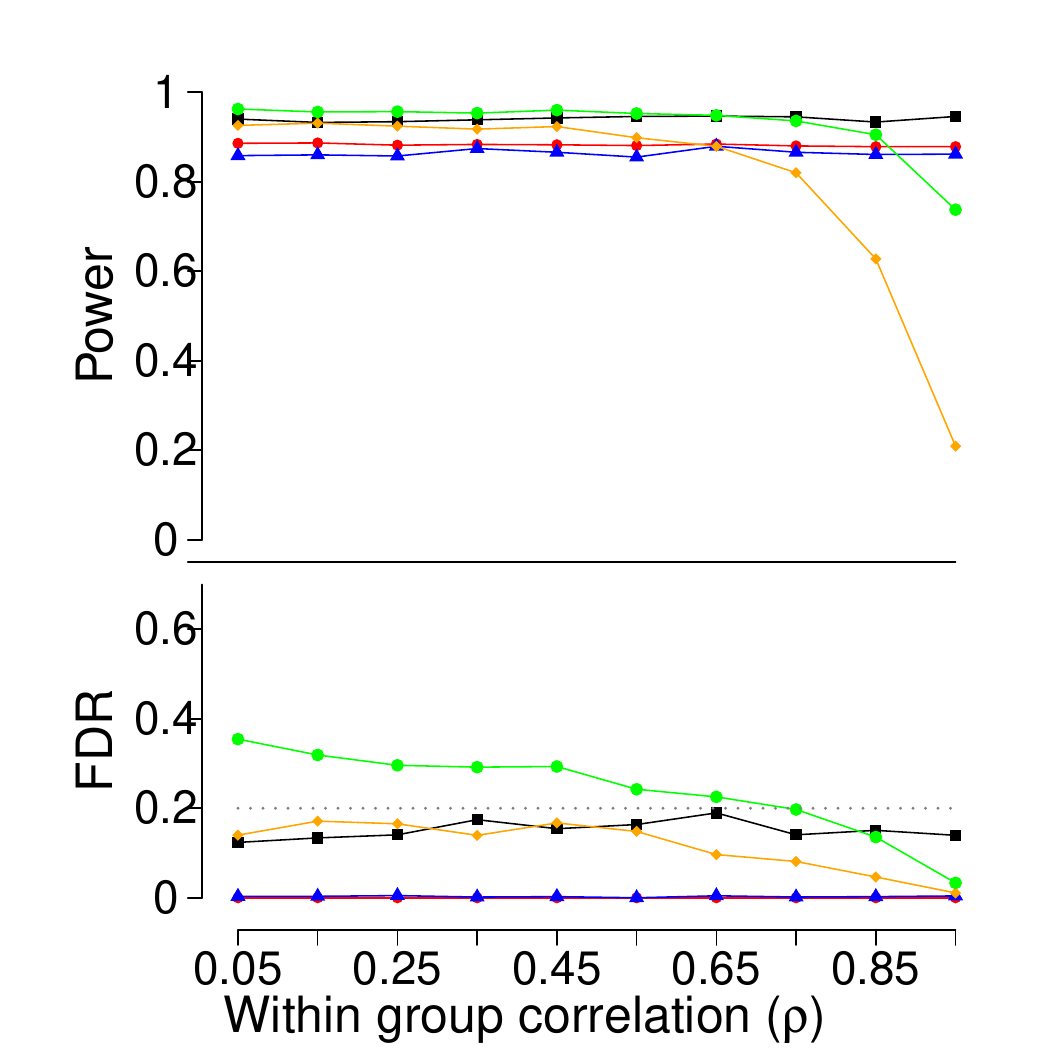}
    \includegraphics[scale=0.23]{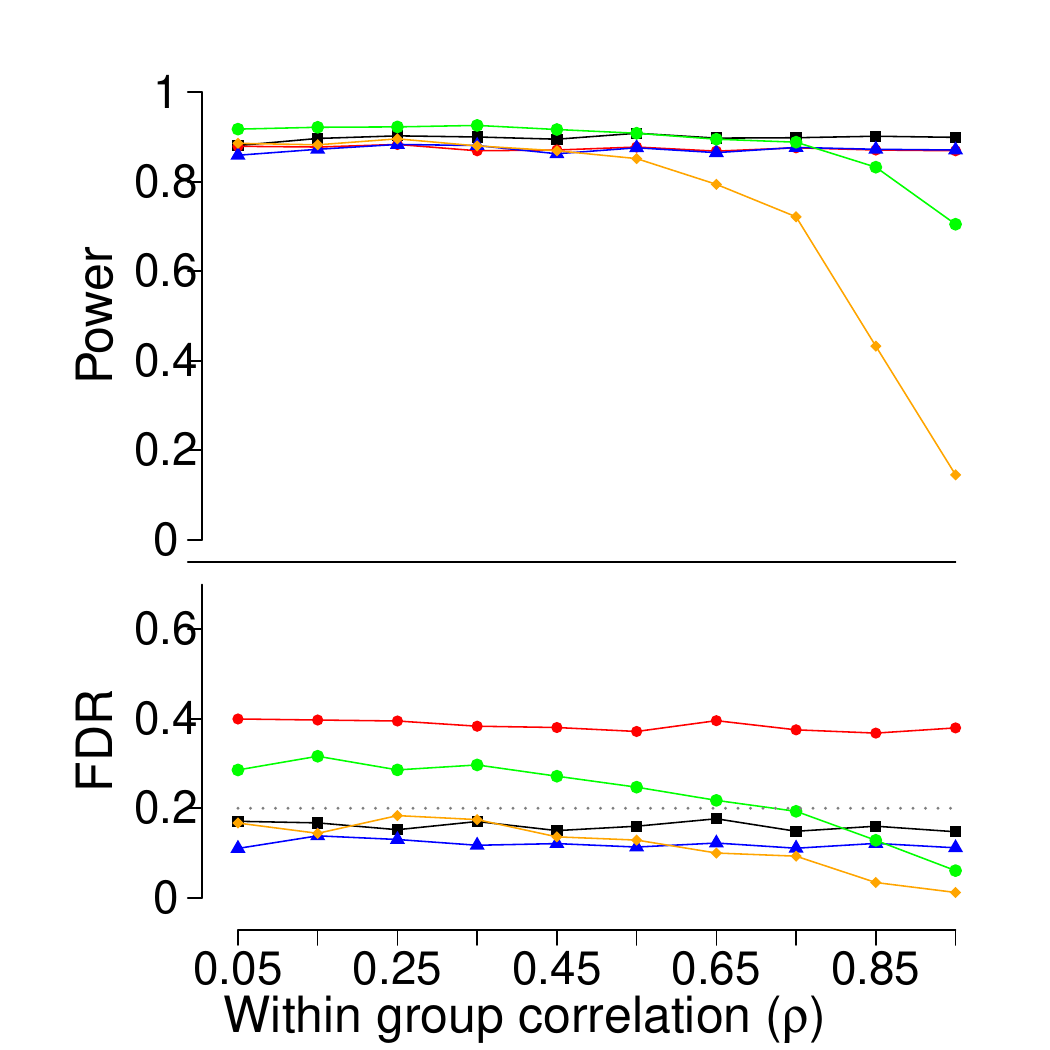}
    \includegraphics[scale=0.23]{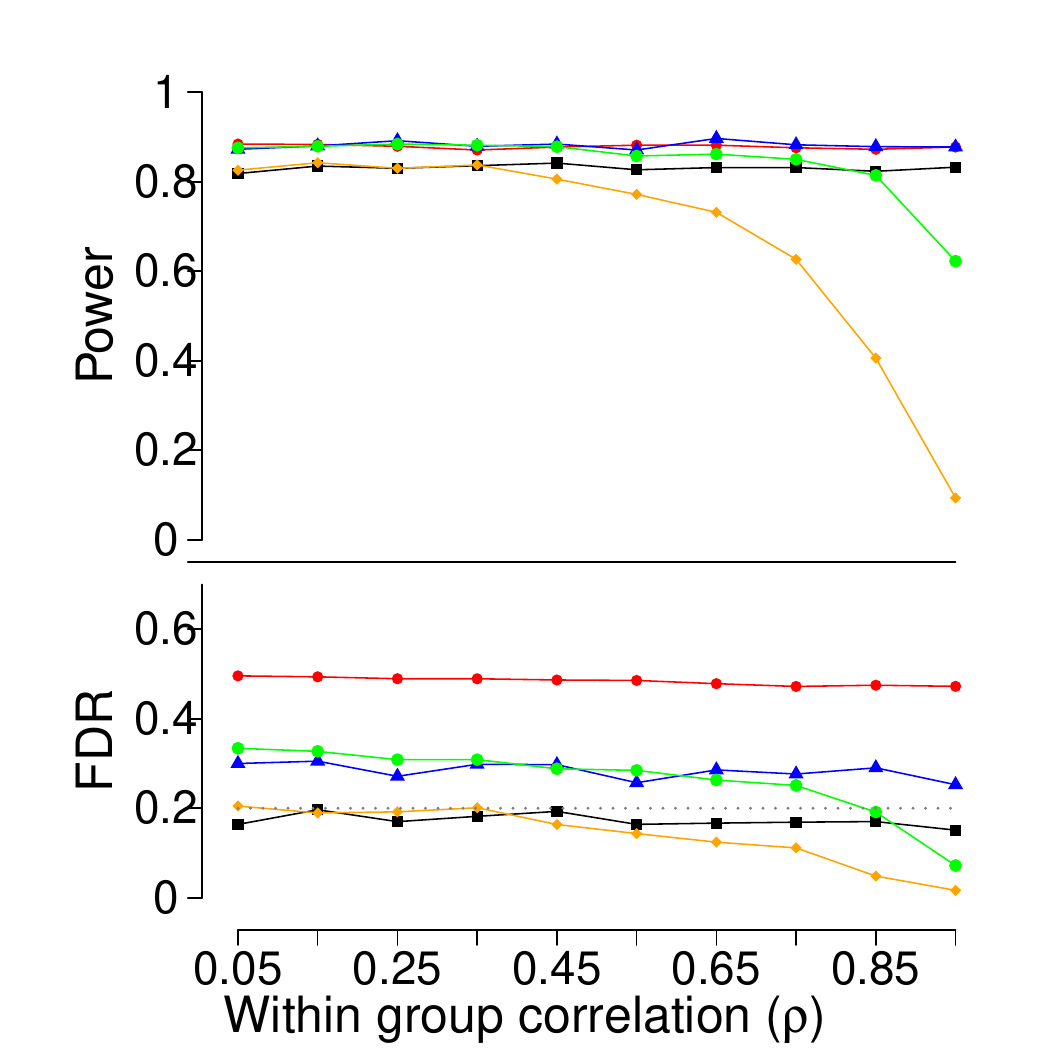}\\
    \includegraphics[scale=0.23]{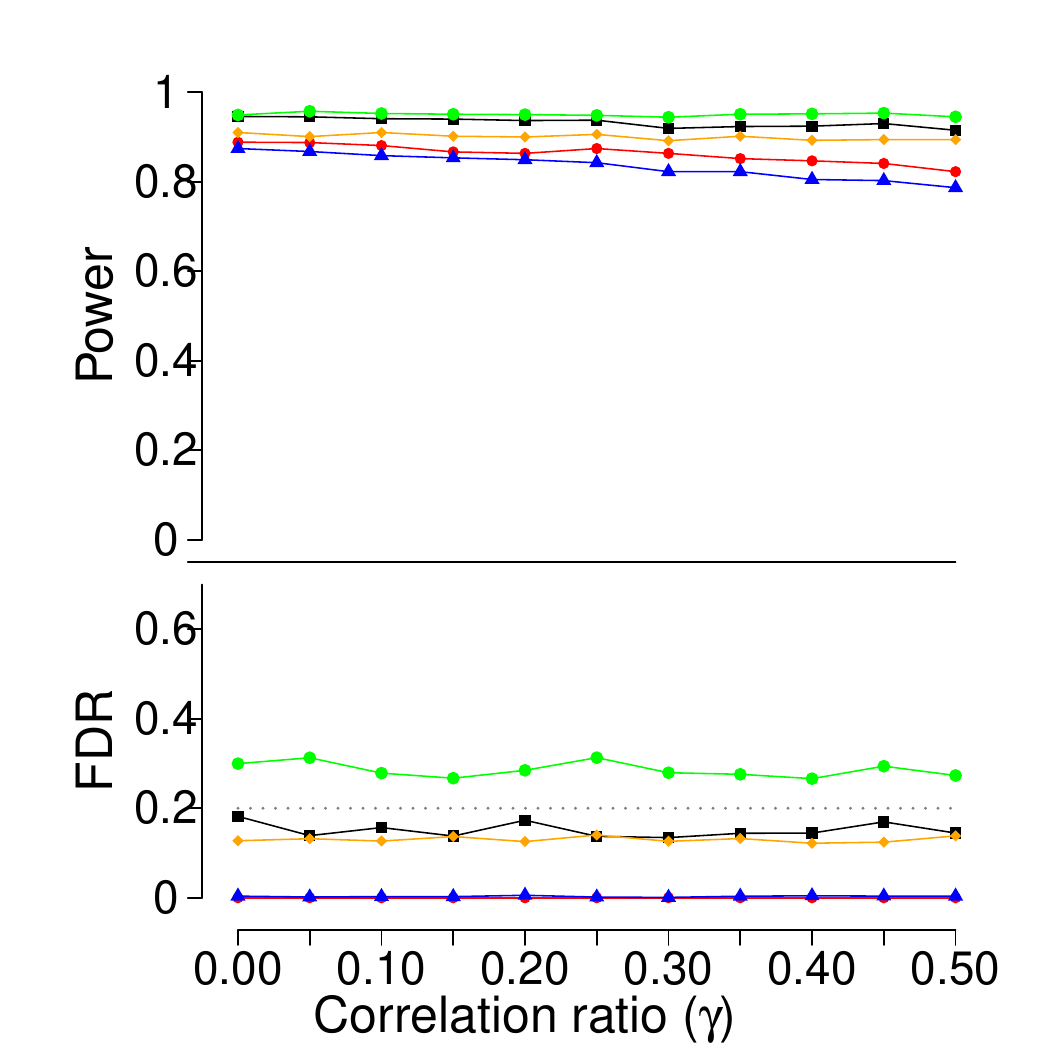}
    \includegraphics[scale=0.23]{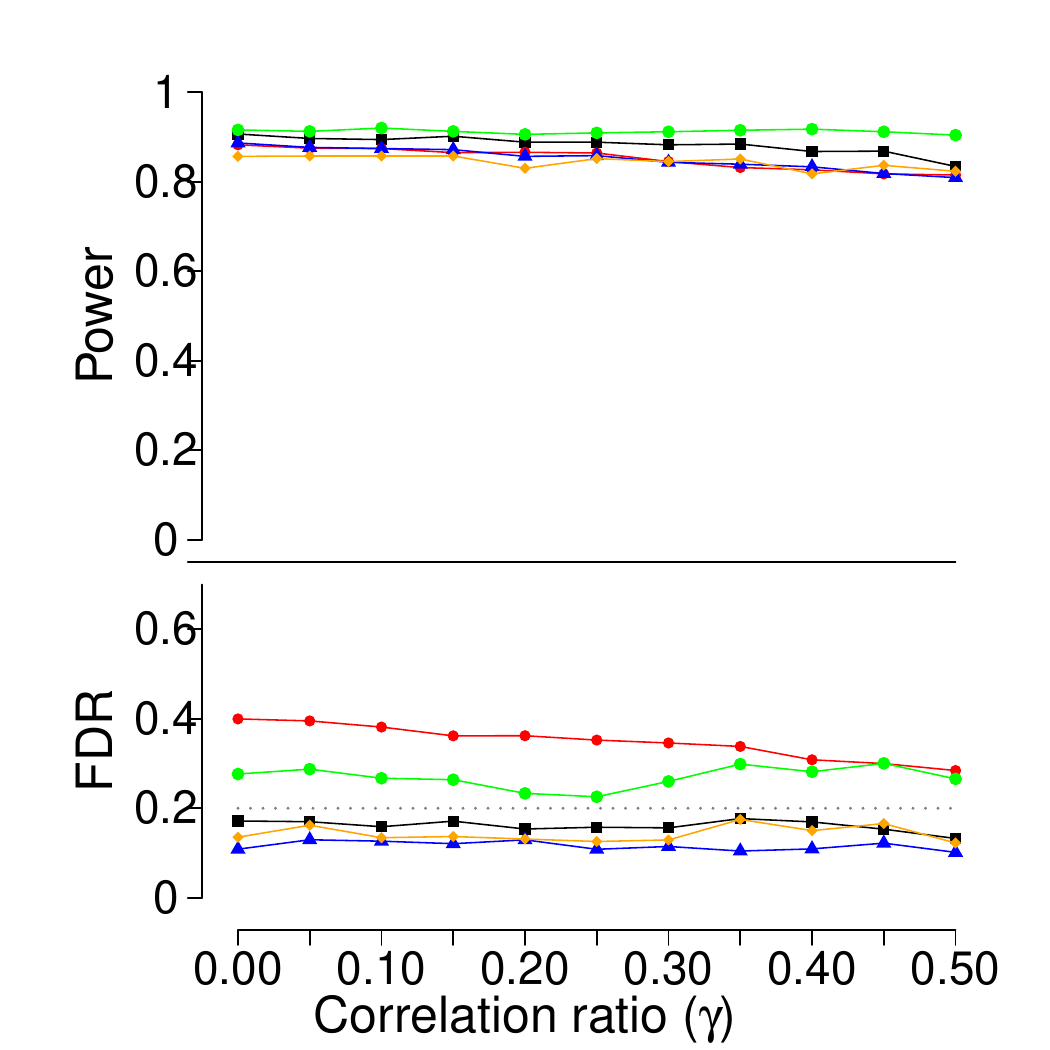}
    \includegraphics[scale=0.23]{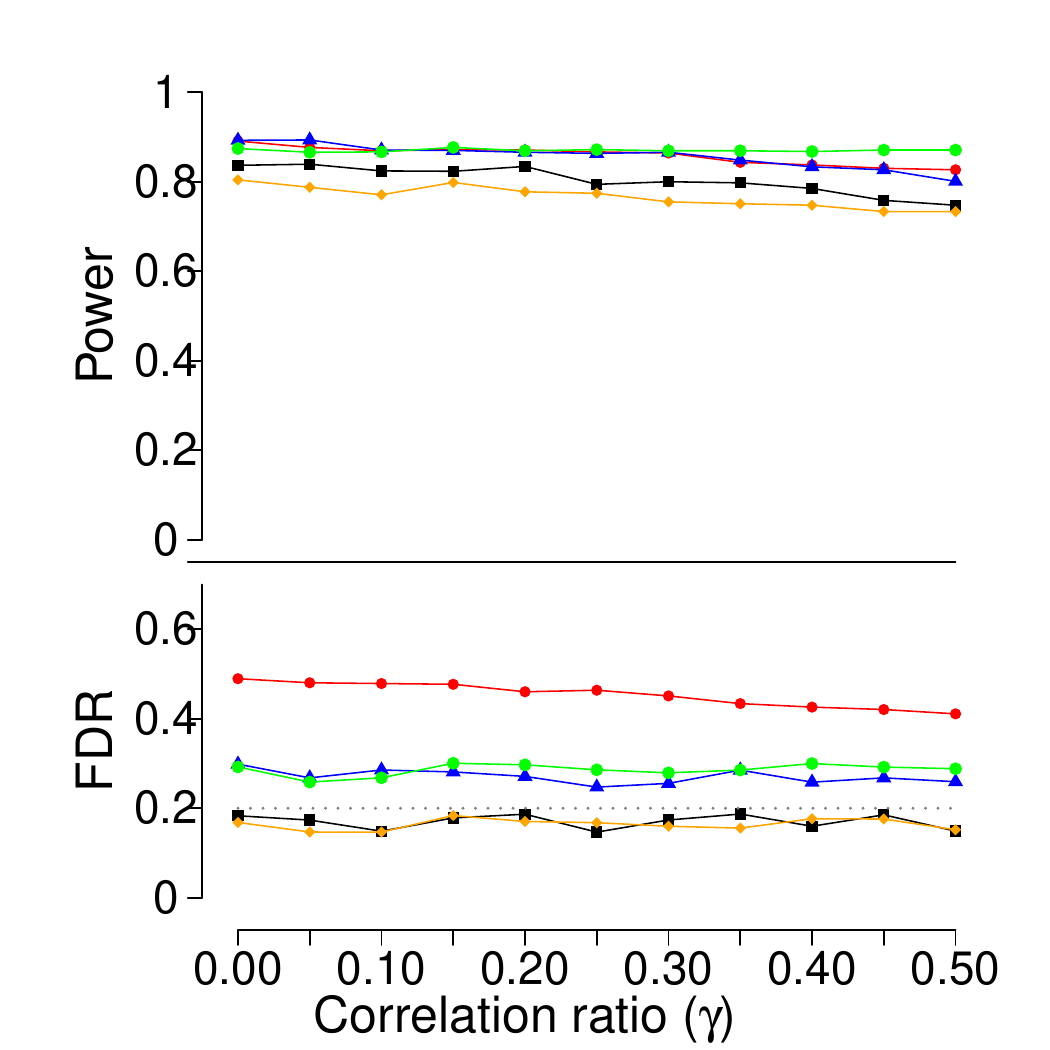}
    
    \caption{{\color{black}The power and the FDR for identifying group level simultaneous signals with data generated from \textbf{Setting 1} for the \textbf{Mixed} models (K=3) on \textbf{Scenario 1} when {\color{black}$n_1=n_2=n3=1000$}. Left column includes settings with $s_0 \neq 0, s_1=s_2=s_3=s_{12}=s_{13}=s_{23}=0$; middle column includes settings with $s_0=12, s_1=s_2=s_3 \neq 0, s_{12}=s_{23}=2,s_{13}=0$; right column includes settings with $s_0=12, s_1=s_2=s_3=0, s_{12}=s_{13}=s_{23} \neq 0.$}}
    \label{fig:figure3-3}
\end{figure}

\begin{figure}
    \centering
    \includegraphics[scale=0.23]{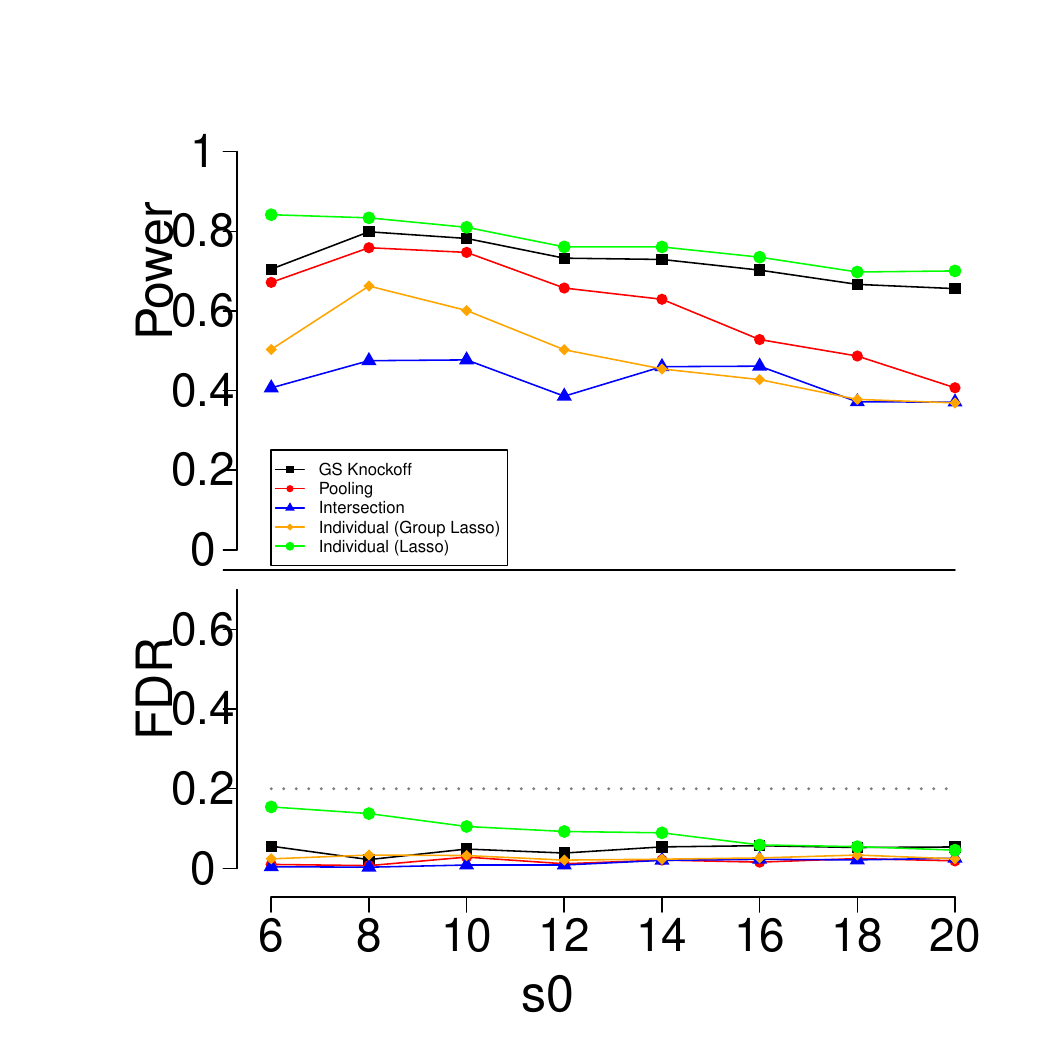}
    \includegraphics[scale=0.23]{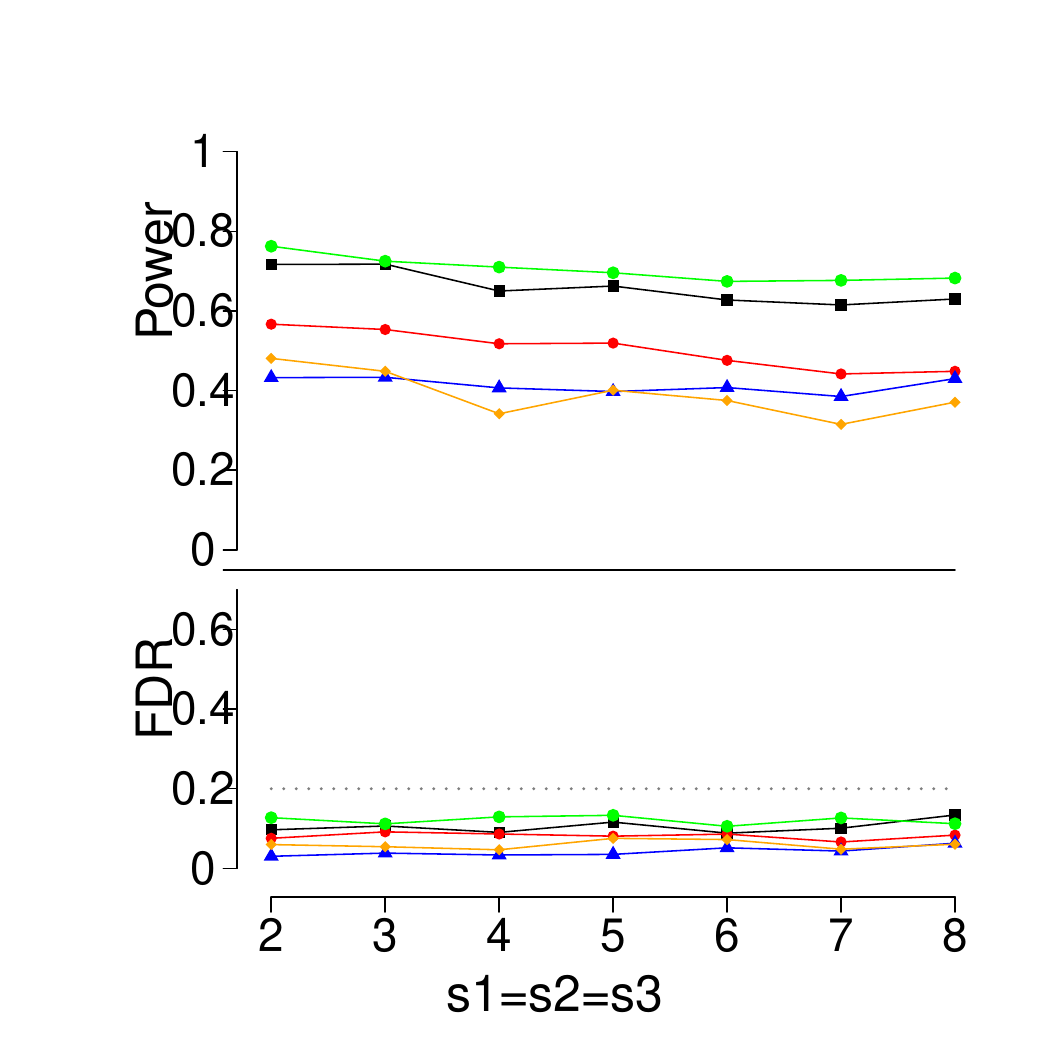}
    \includegraphics[scale=0.23]{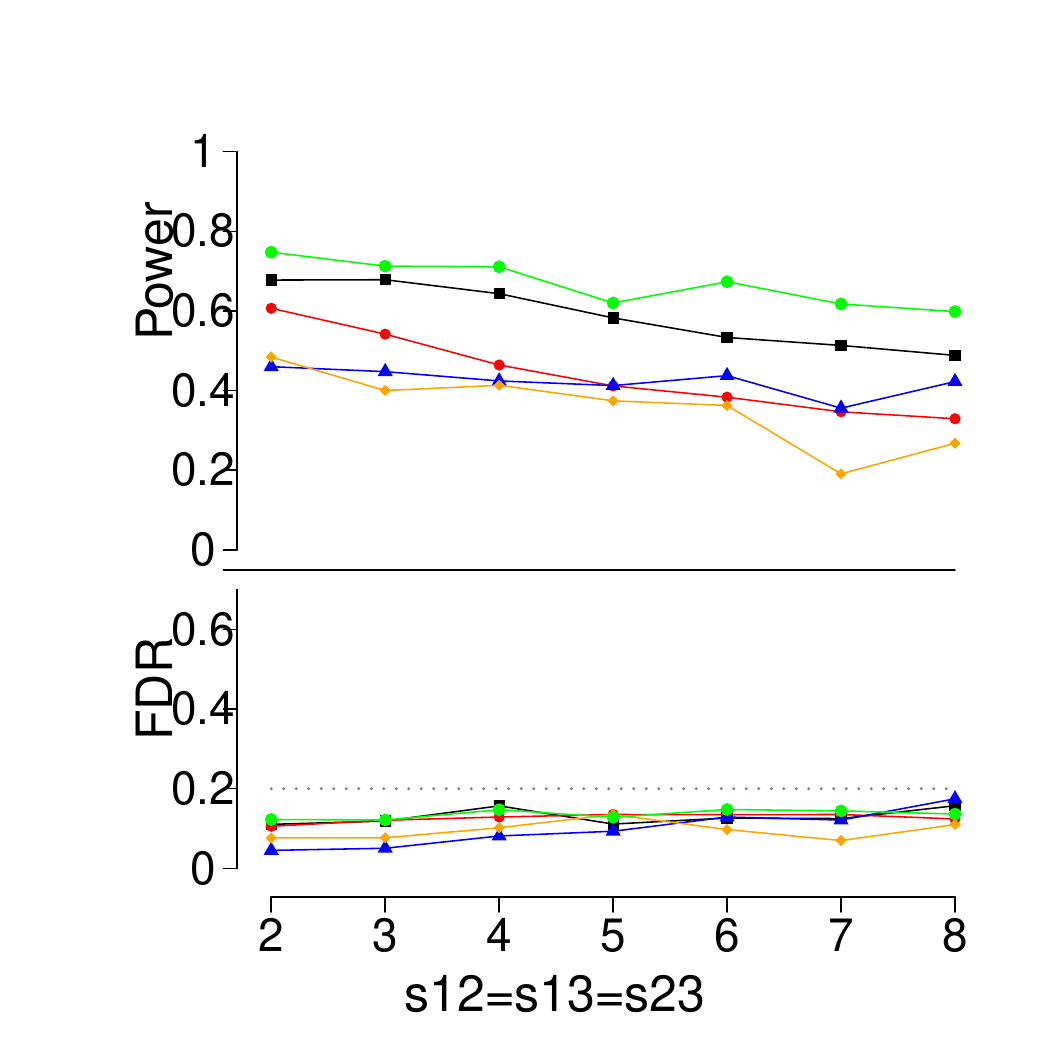}\\
    \includegraphics[scale=0.23]{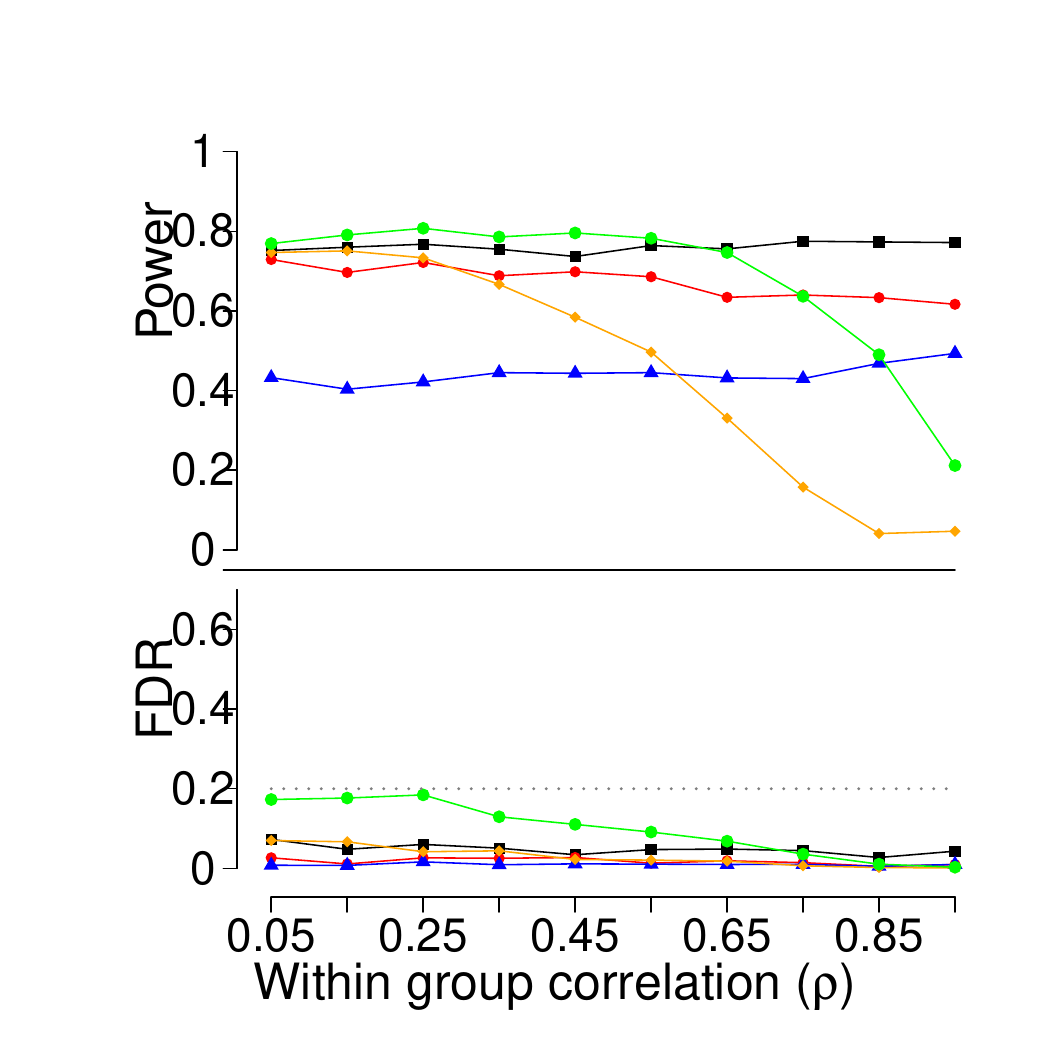}
    \includegraphics[scale=0.23]{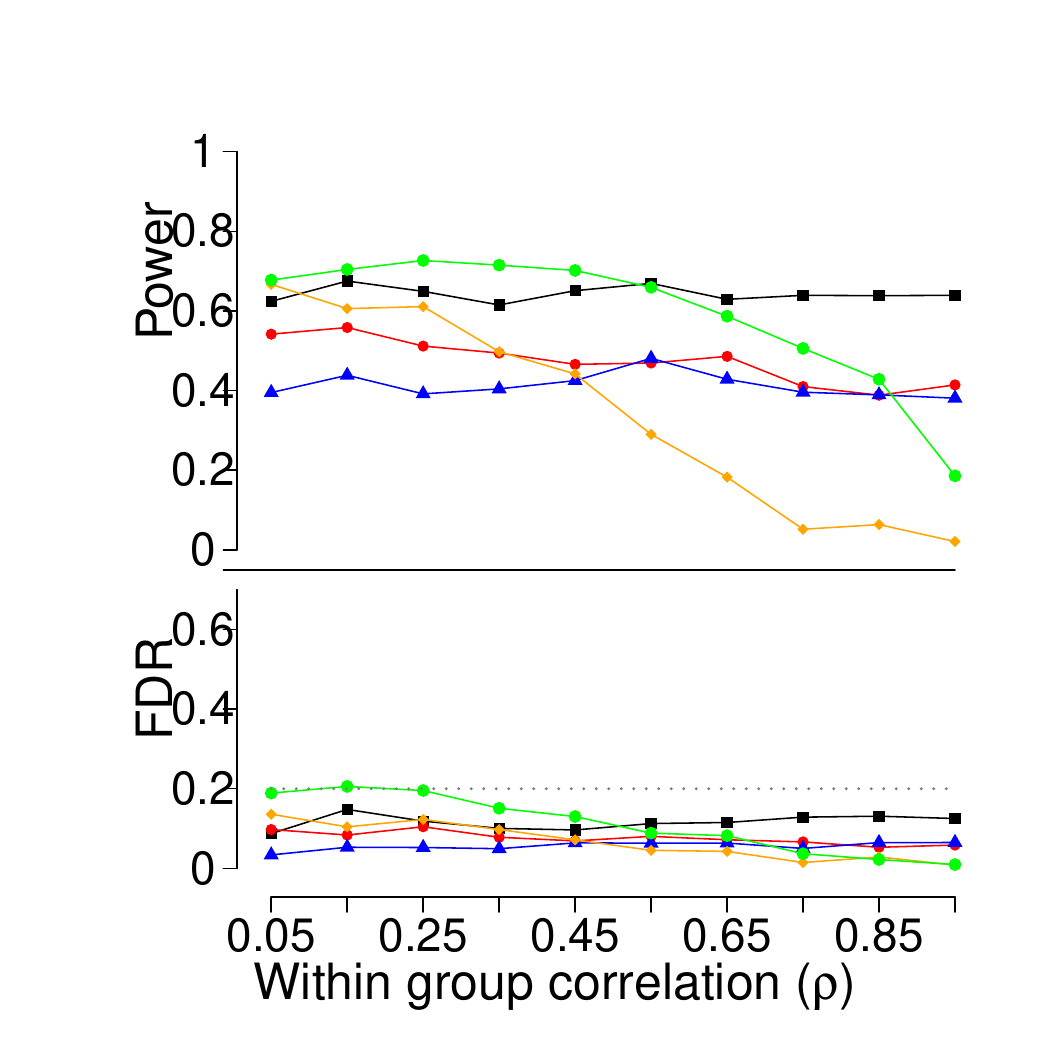}
    \includegraphics[scale=0.23]{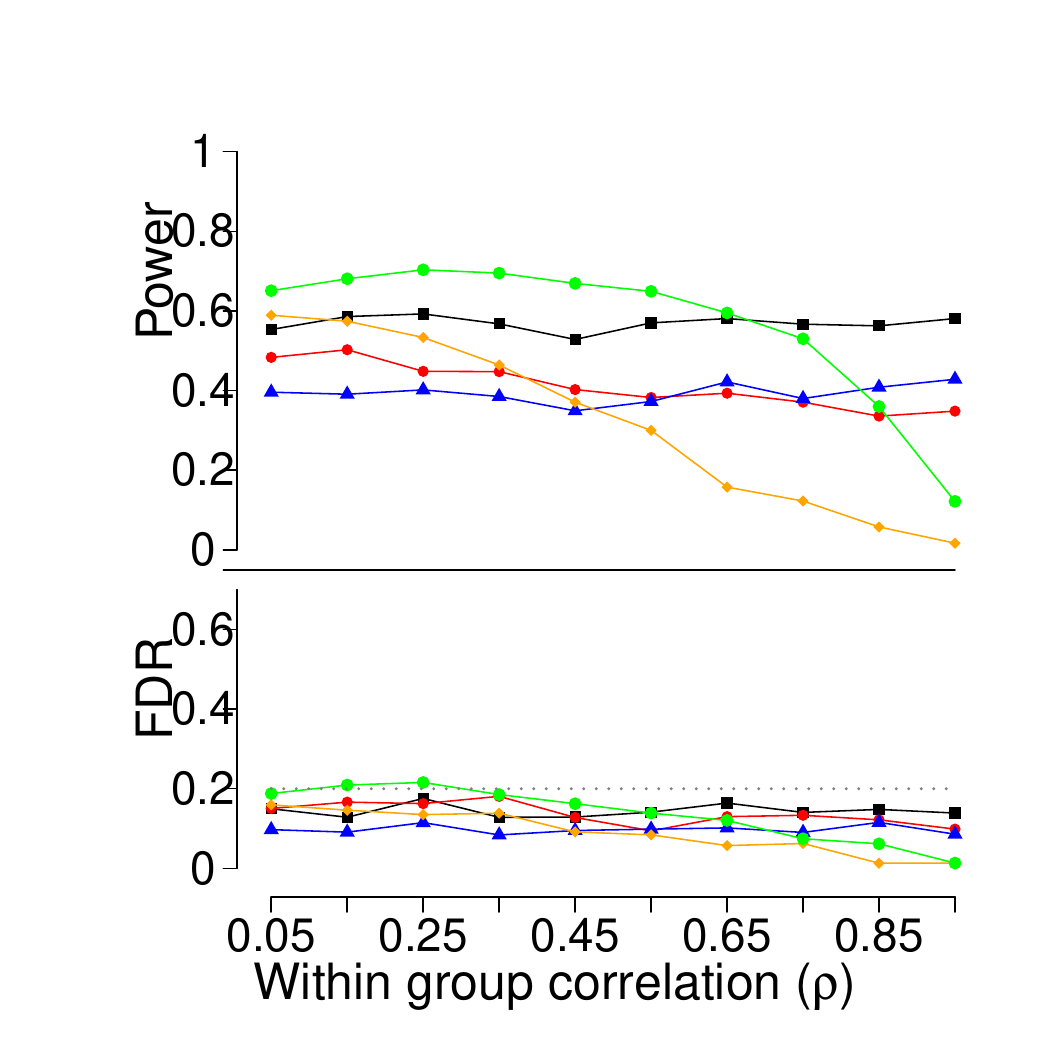}\\
    \includegraphics[scale=0.23]{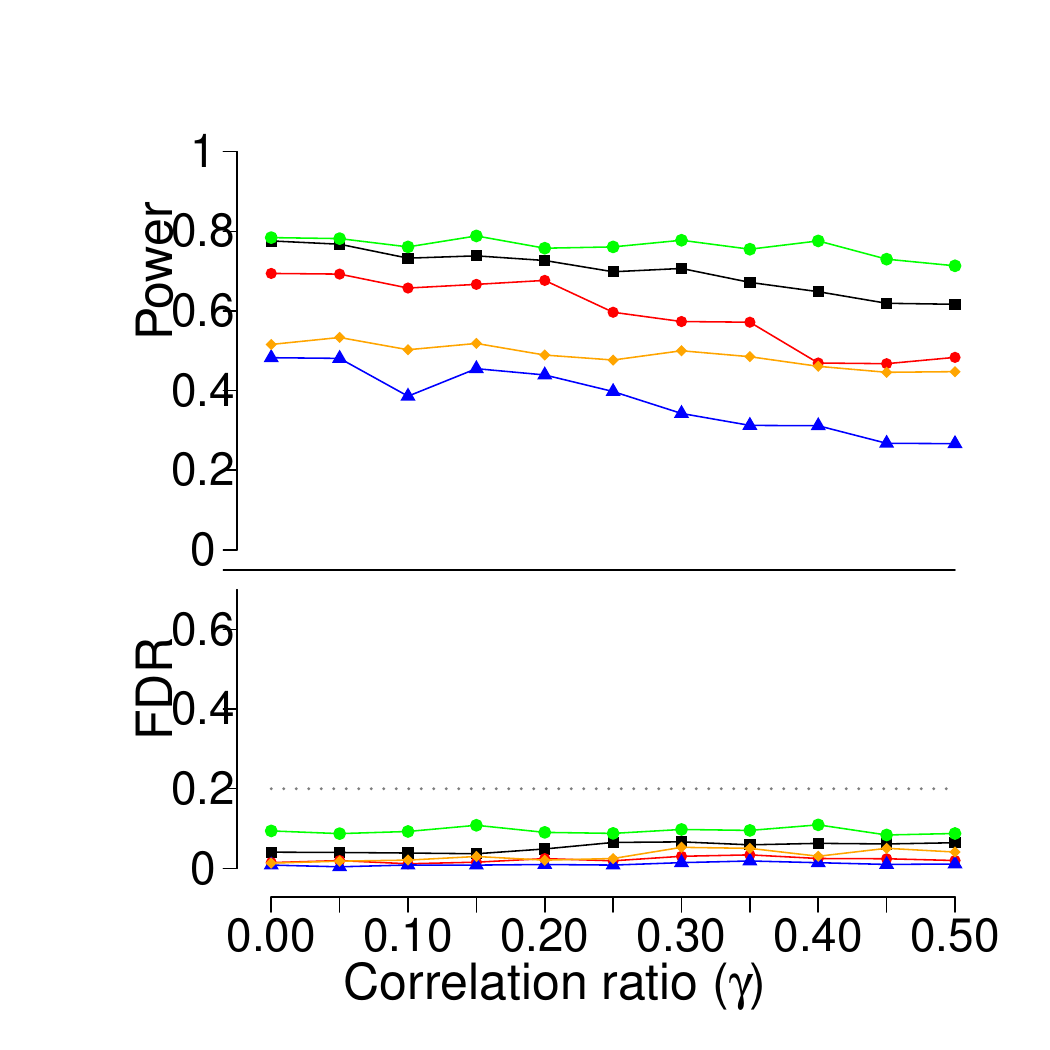}
    \includegraphics[scale=0.23]{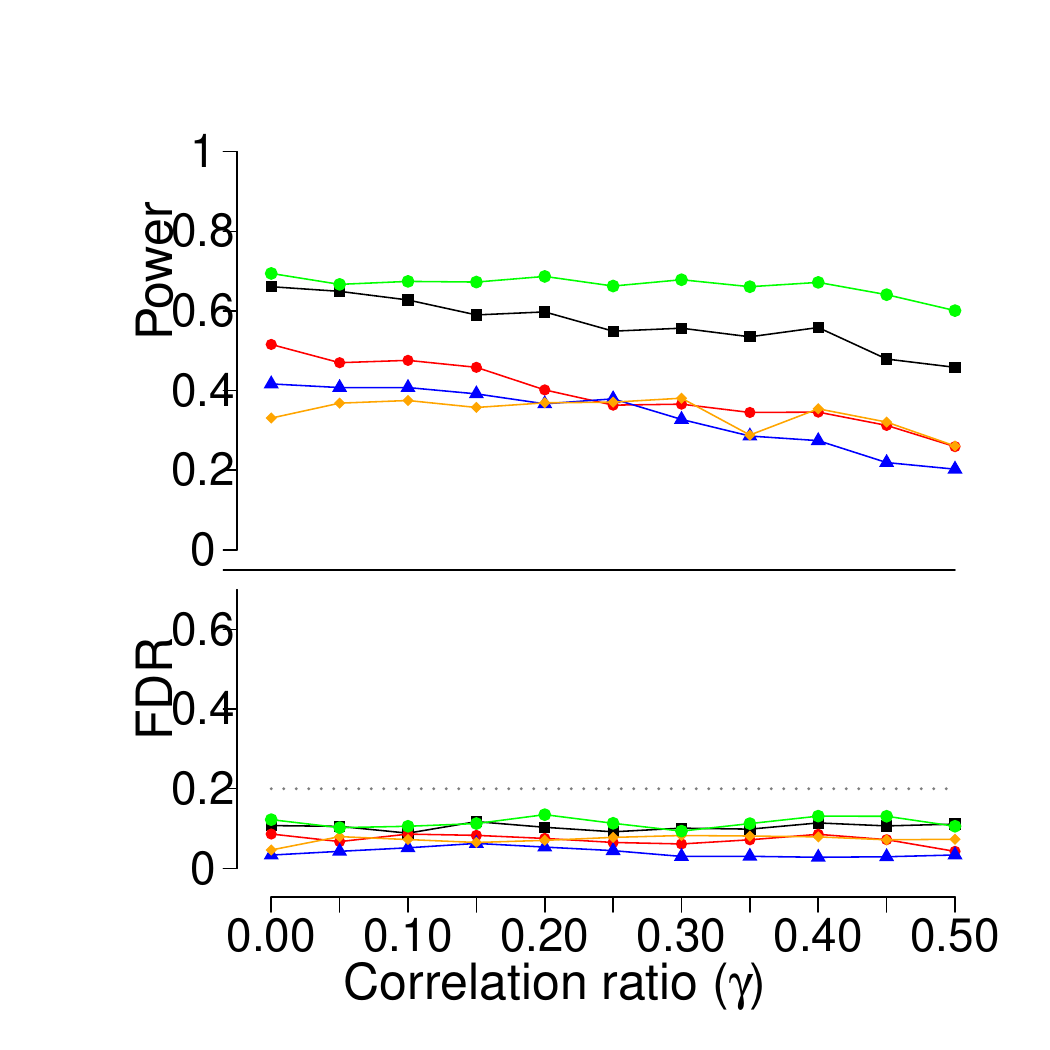}
    \includegraphics[scale=0.23]{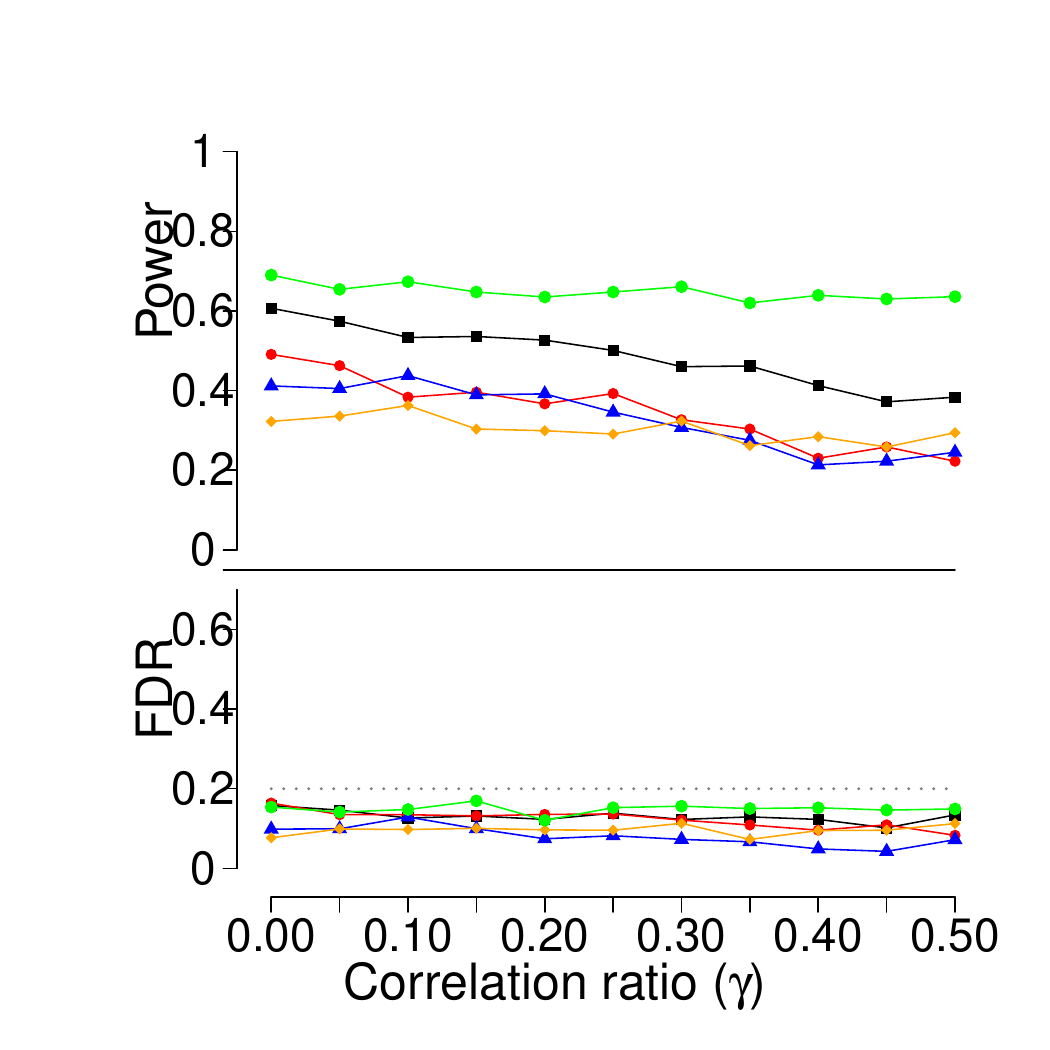}
    
    \caption{{\color{black} The power and the FDR for identifying group level simultaneous signals with data generated from \textbf{Setting 1} for the \textbf{Continuous} models (K=3) on \textbf{Scenario 1} when {\color{black}$n_1=n_2=n3=200$.} Left column includes settings with $s_0 \neq 0, s_1=s_2=s_3=s_{12}=s_{13}=s_{23}=0$; middle column includes settings with $s_0=12, s_1=s_2=s_3 \neq 0, s_{12}=s_{23}=2,s_{13}=0$; right column includes settings with $s_0=6, s_1=s_2=s_3=0, s_{12}=s_{13}=s_{23} \neq 0.$}}
    \label{fig:figure3-}
\end{figure}

\begin{figure}
    \centering
    \includegraphics[scale=0.23]{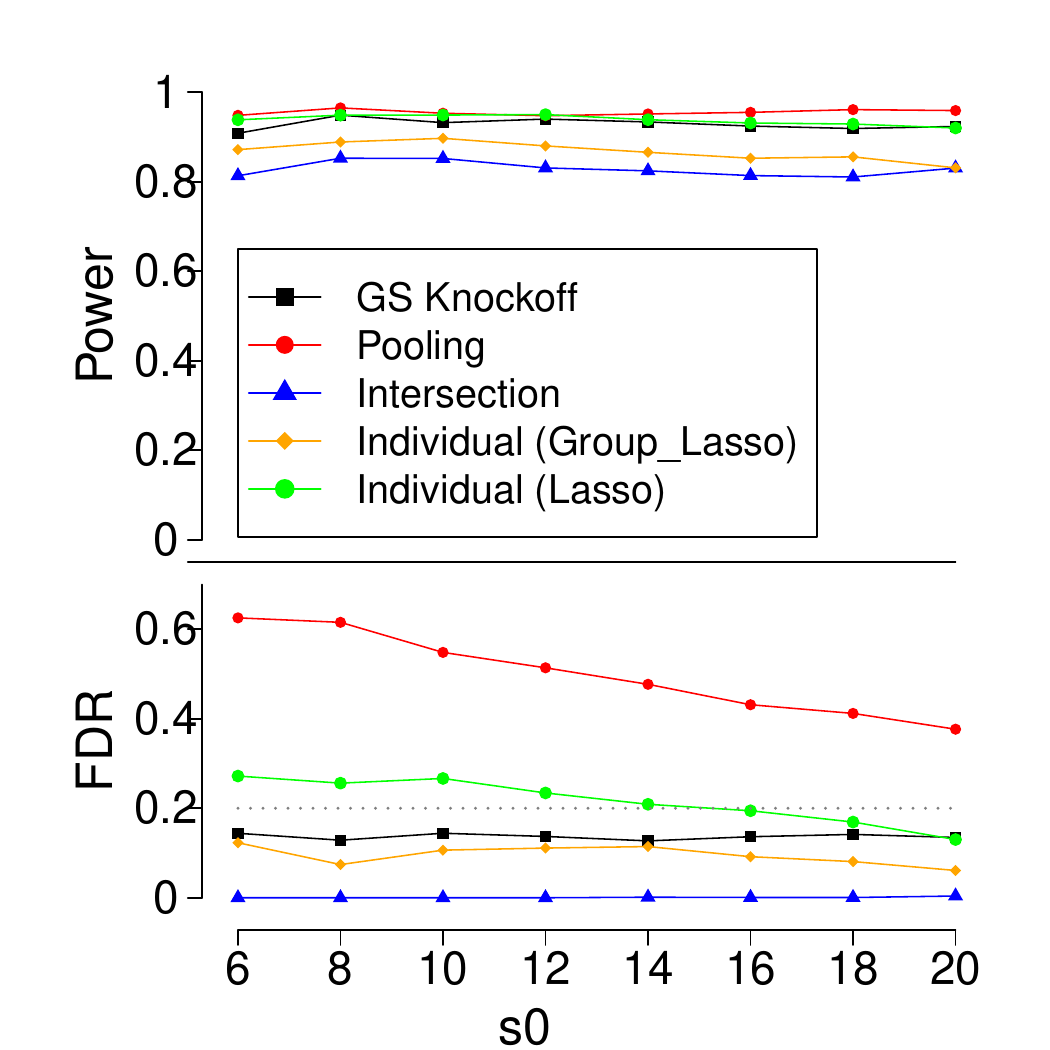}
    \includegraphics[scale=0.23]{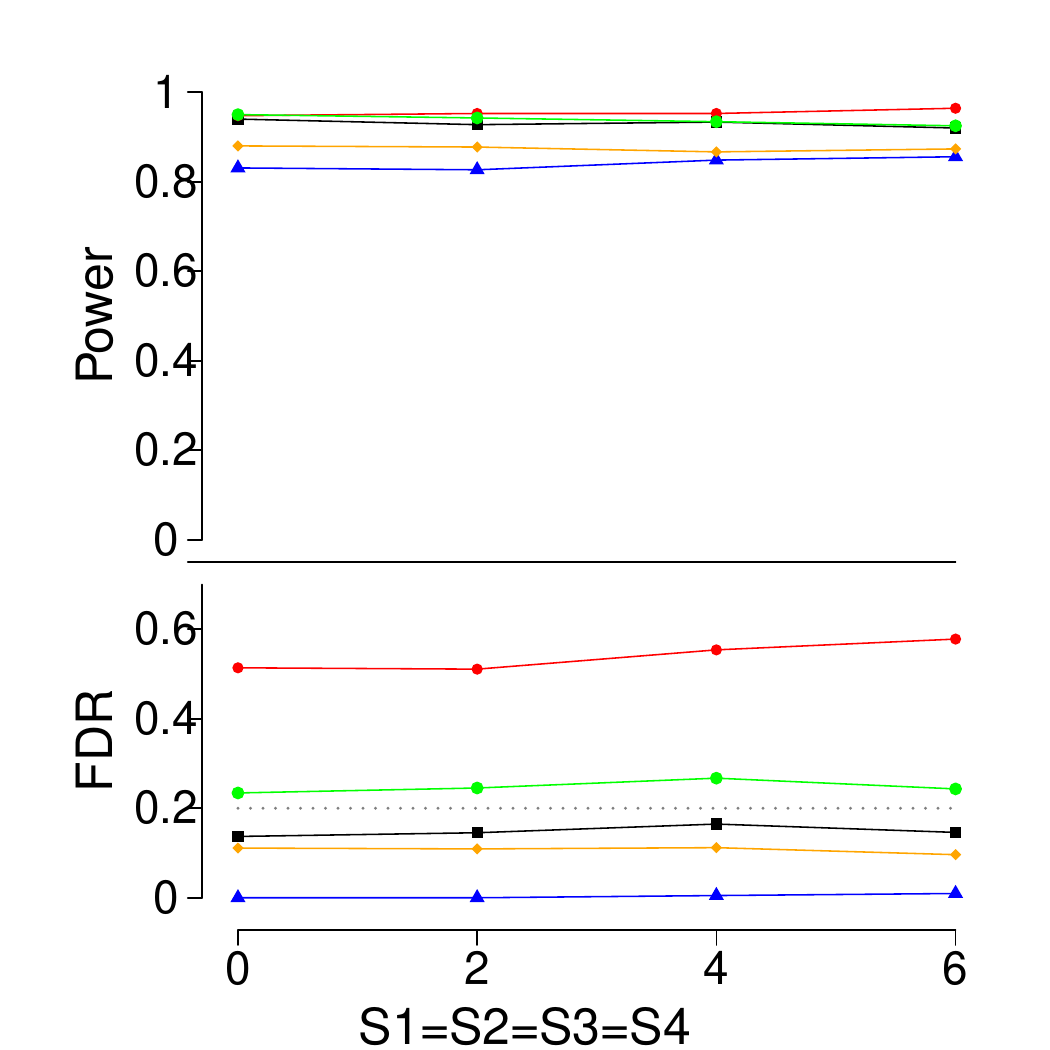}
    \includegraphics[scale=0.23]{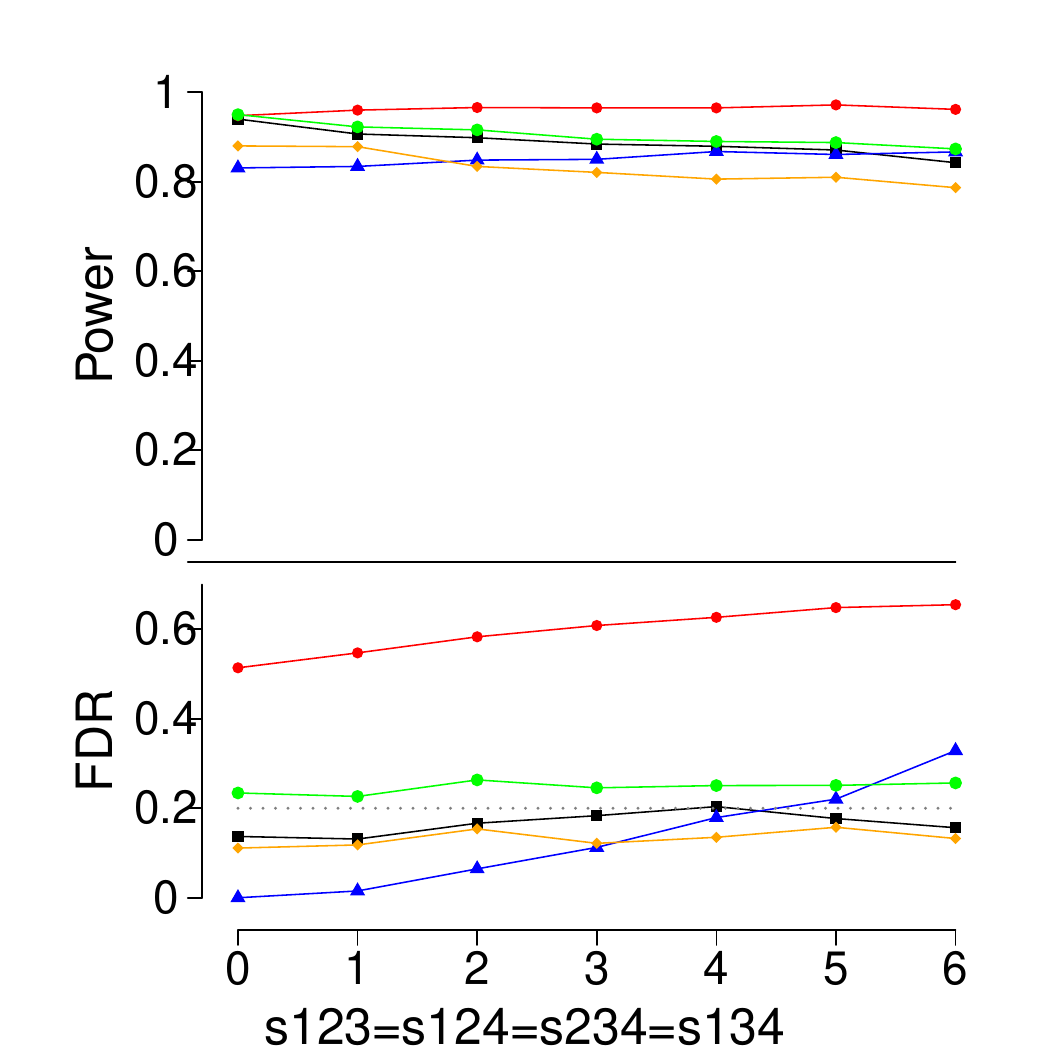}\\
    \includegraphics[scale=0.23]{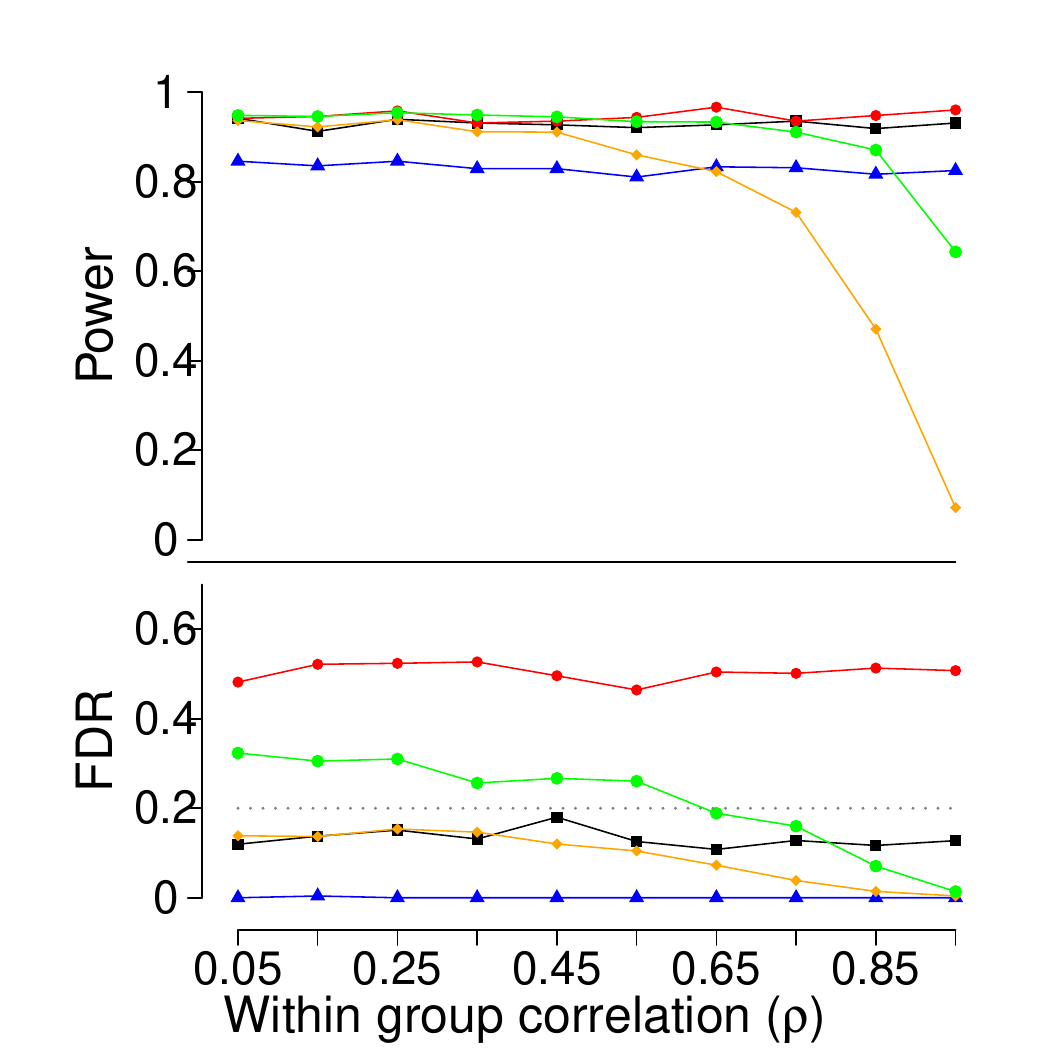}
    \includegraphics[scale=0.23]{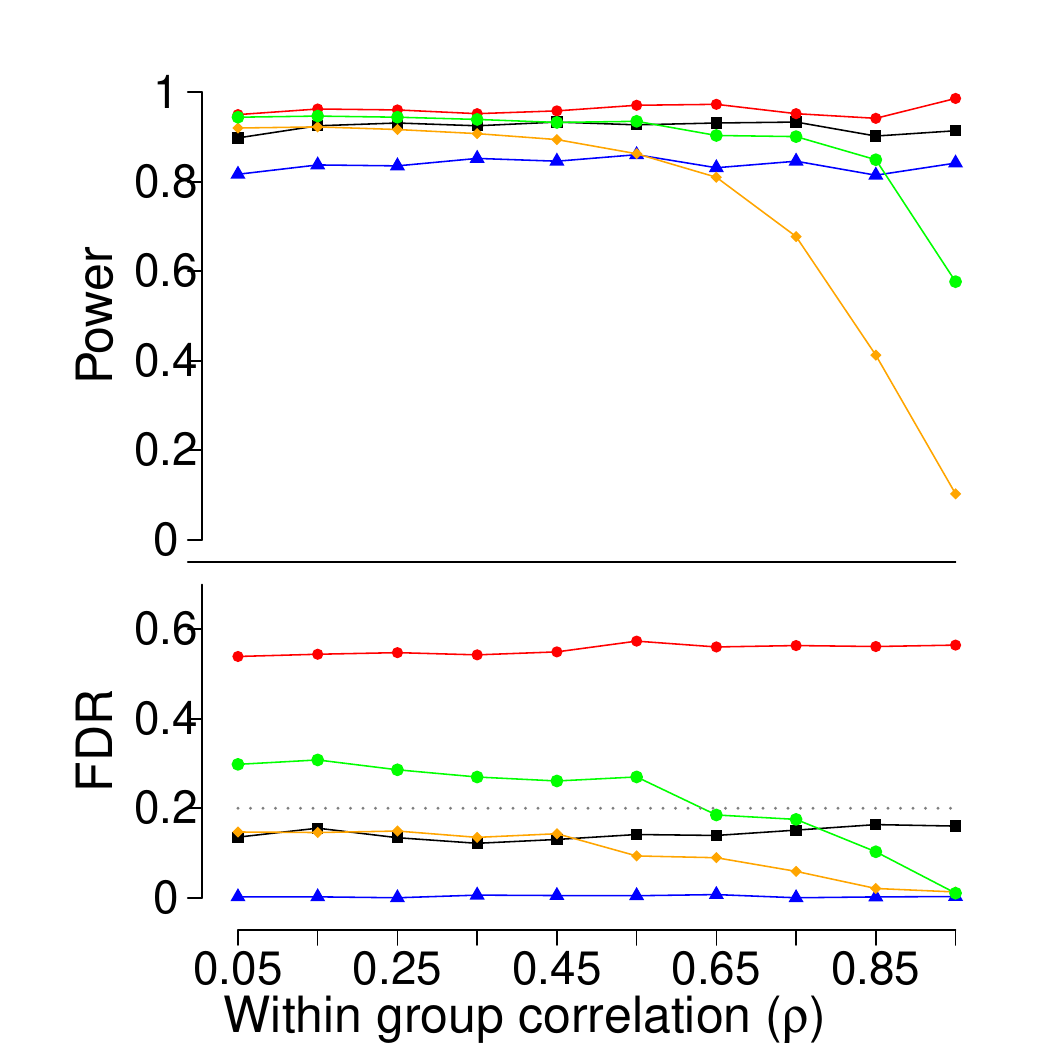}
    \includegraphics[scale=0.23]{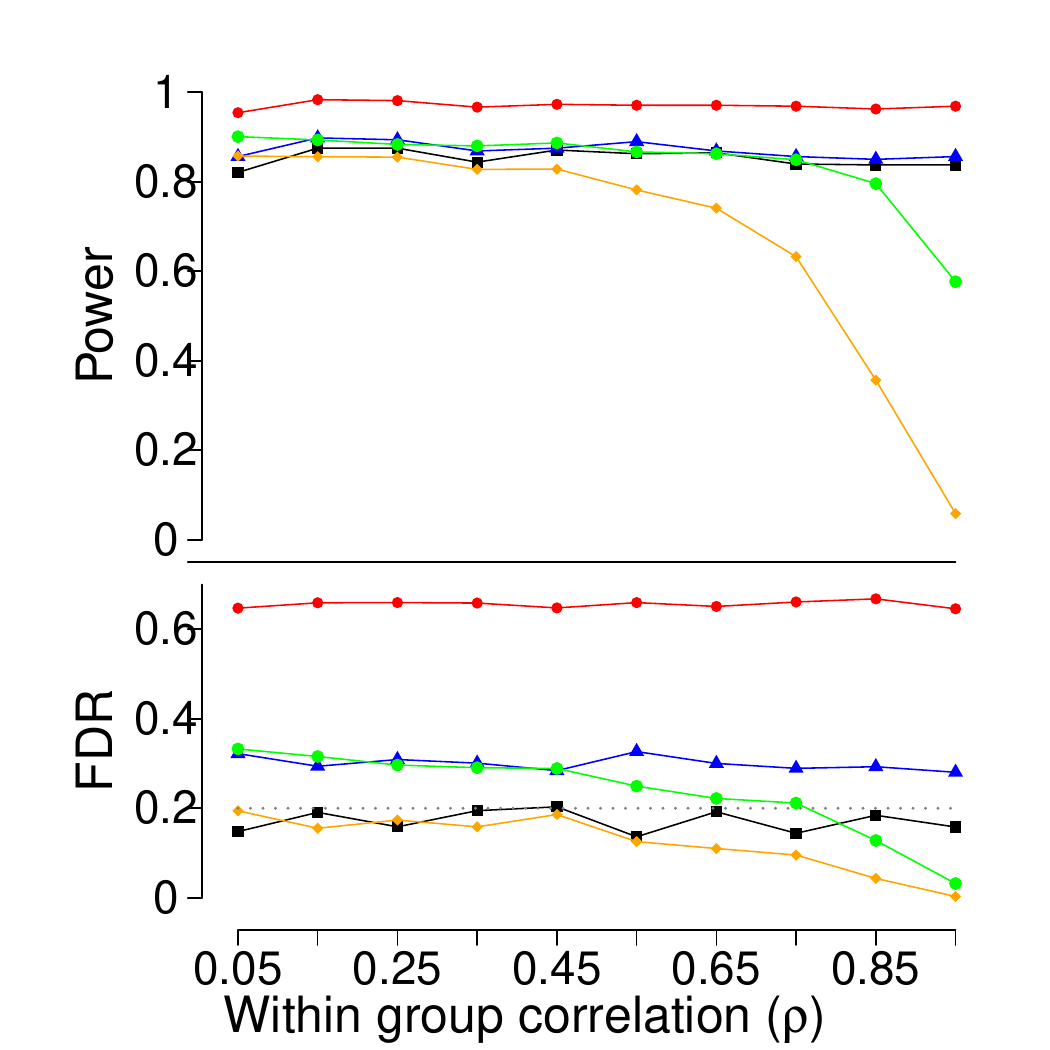}\\
    \includegraphics[scale=0.23]{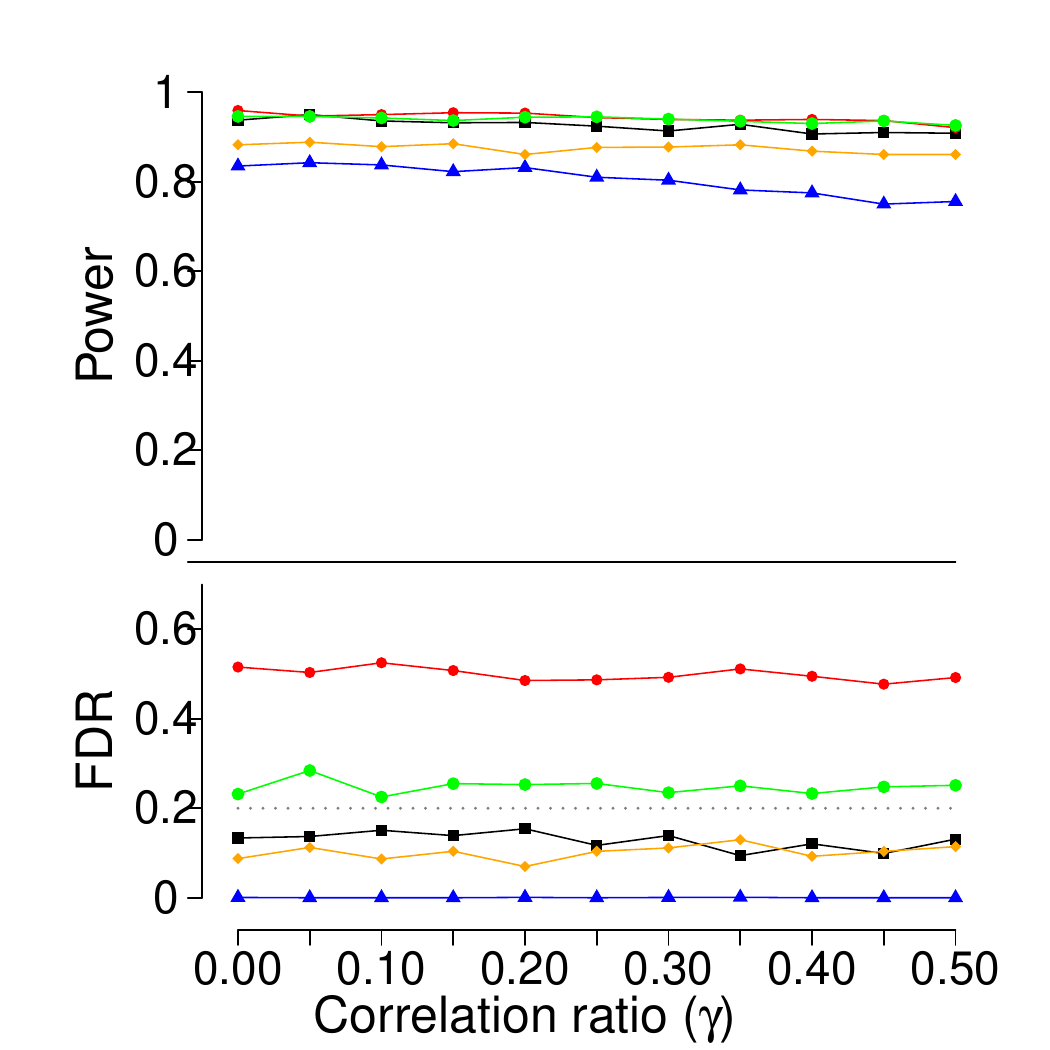}
    \includegraphics[scale=0.23]{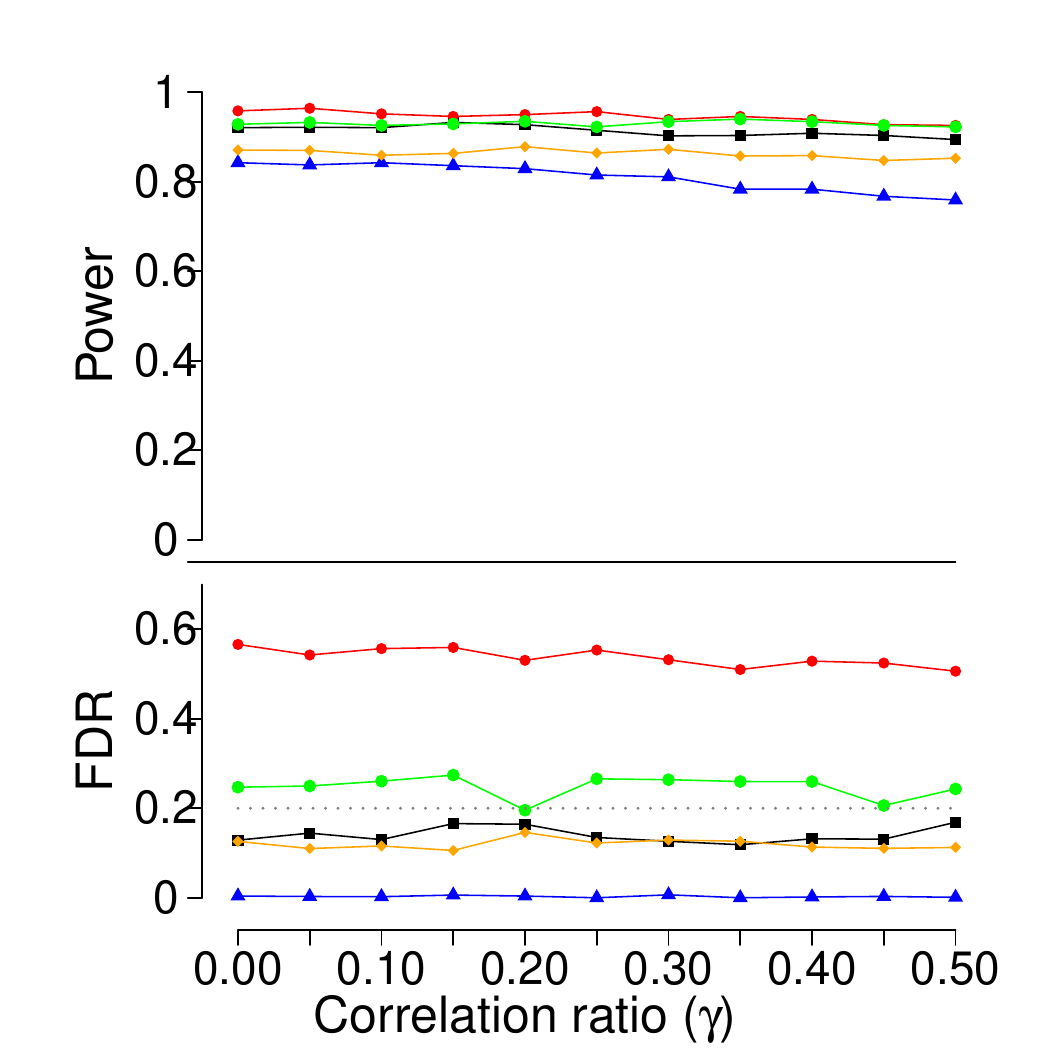}
    \includegraphics[scale=0.23]{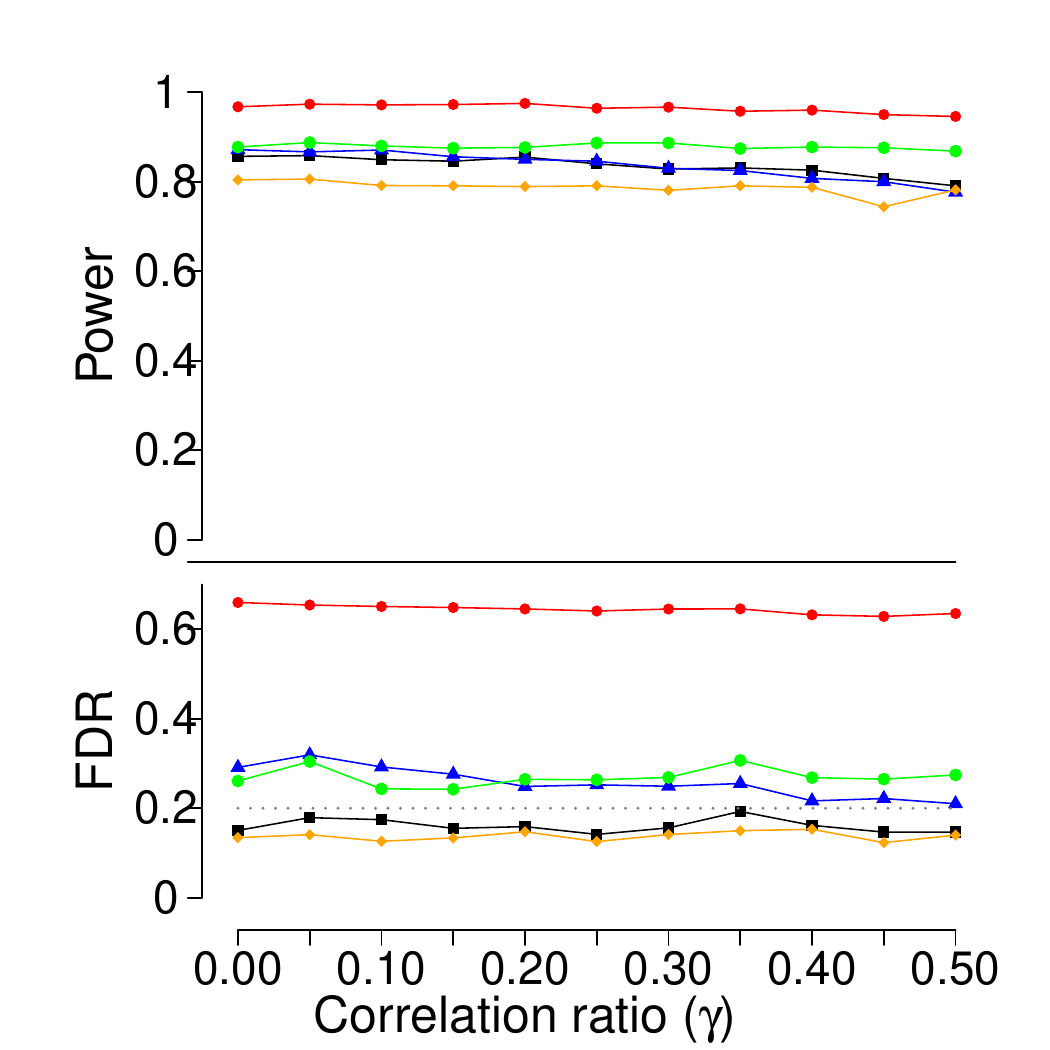}
    
    \caption{{\color{black} The power and the FDR for identifying group level simultaneous signals with data generated from \textbf{Setting 1} for the \textbf{Mixed} models (K=4) on \textbf{Scenario 1}. Left column includes settings with $s_0 \neq 0, s_1=s_2=s_3=s_4=s_{12}=s_{13}=s_{14}=s_{23}=s_{24}=s_{34}=s_{123}=s_{124}=s_{134}=s_{234}=0$; middle column includes settings with $s_0=12, s_1=s_2=s_3=s_4 \neq 0, s_{12}=s_{13}=s_{14}=s_{23}=s_{24}=s_{34}=s_{123}=s_{124}=s_{134}=s_{234}=0$; right column includes settings with $s_0=12, s_1=s_2=s_3=s_4=s_{12}=s_{13}=s_{14}=s_{23}=s_{24}=s_{34}=0, s_{123}=s_{124}=s_{134}=s_{234}\neq 0.$}}
    \label{fig:figure3-4}
\end{figure}

\begin{figure}
    \centering
    \includegraphics[scale=0.23]{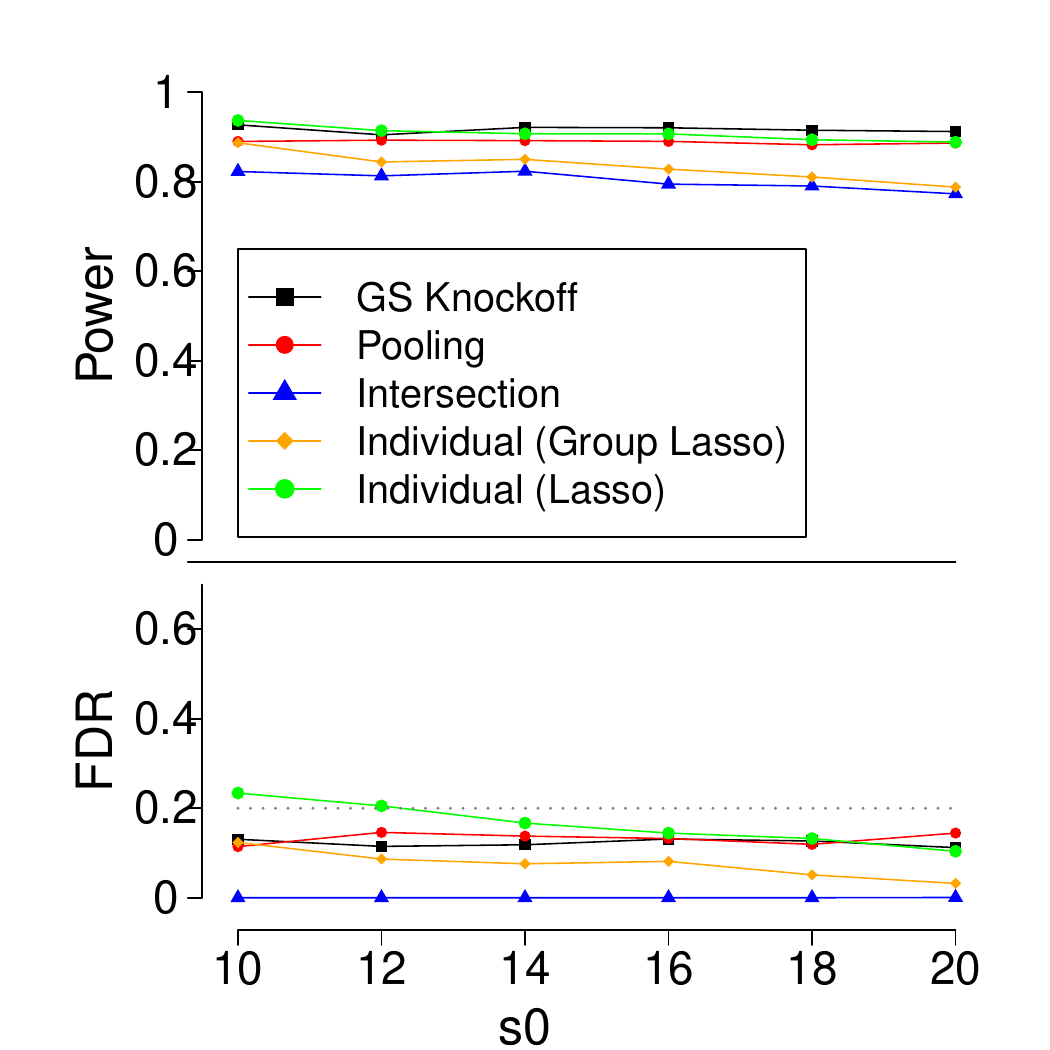}
    \includegraphics[scale=0.23]{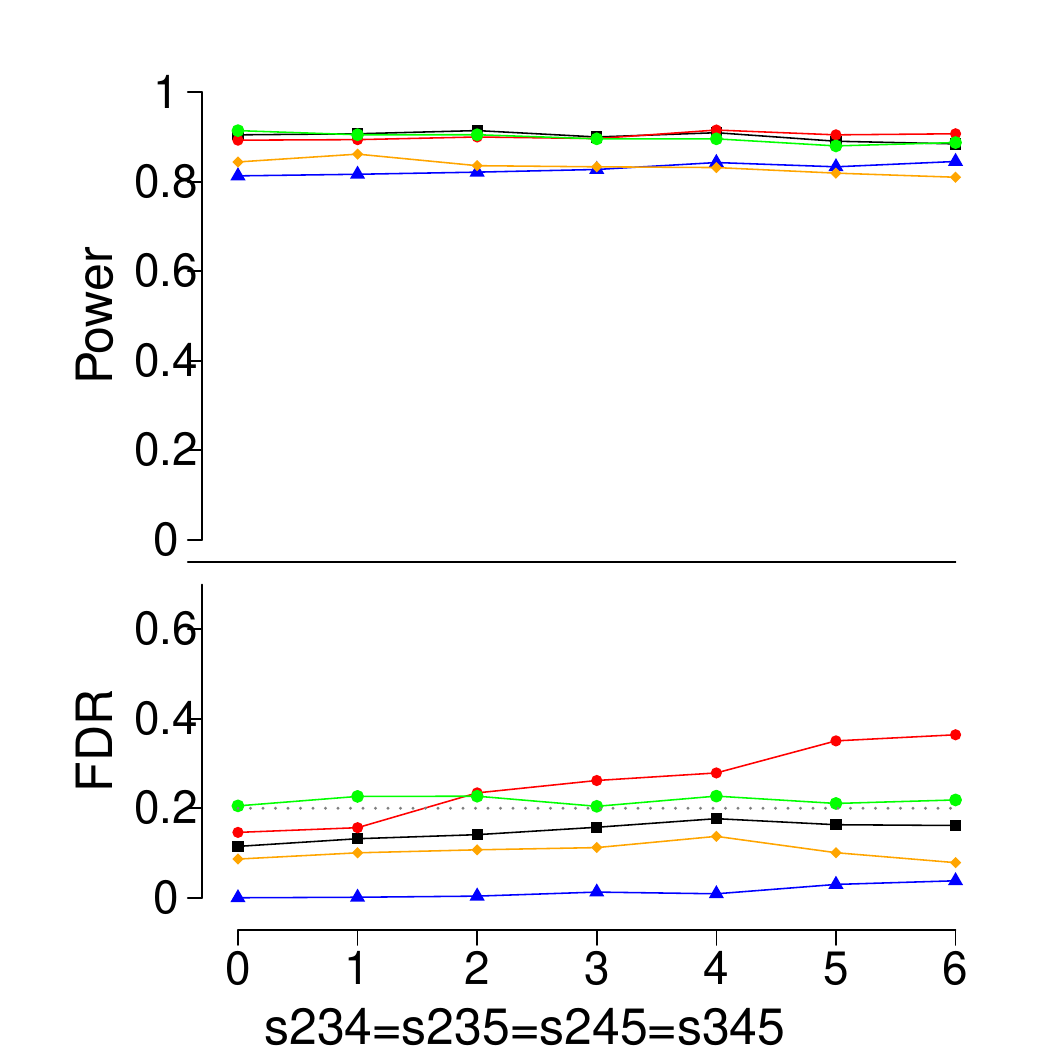}
    \includegraphics[scale=0.23]{images/K=5_Mixed_s4_scale=4_samesig=1.pdf}\\
    \includegraphics[scale=0.23]{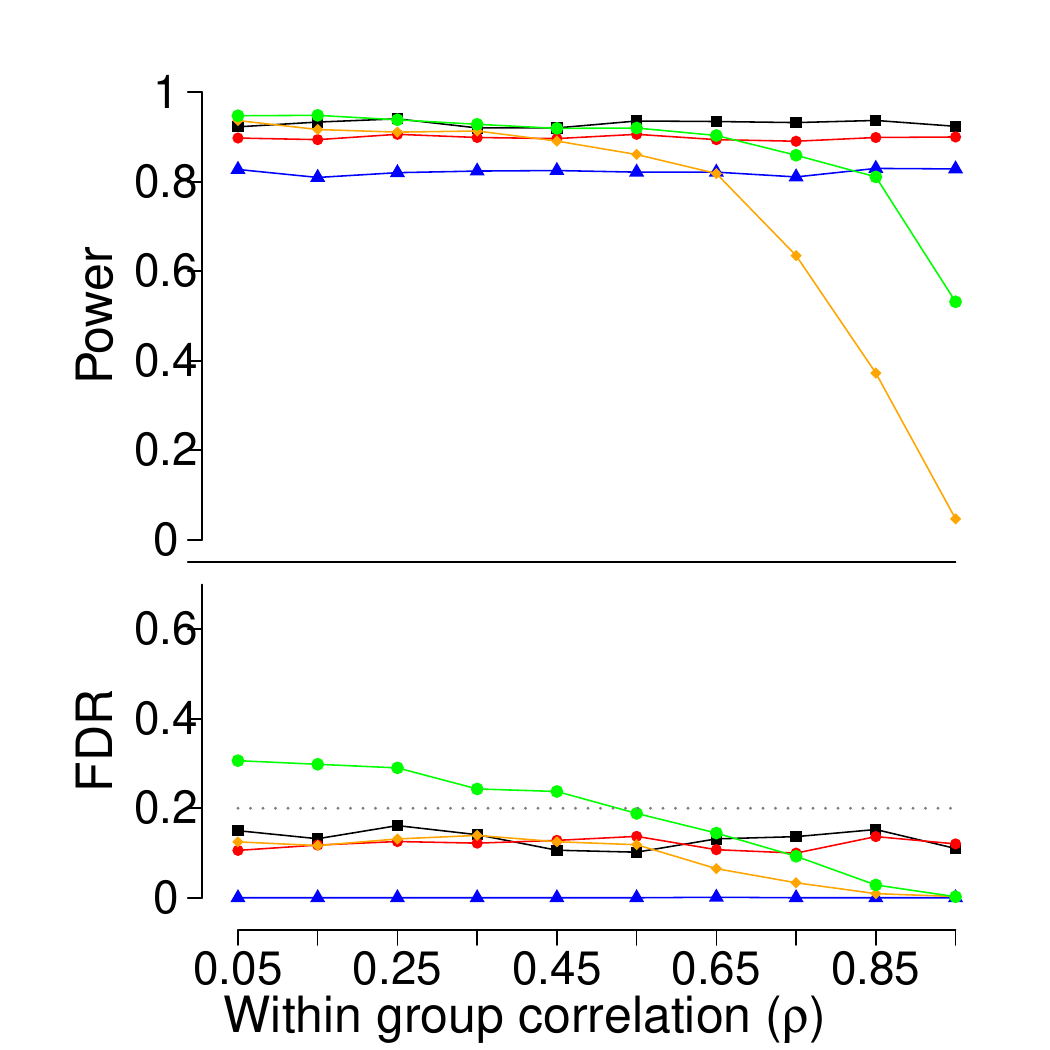}
    \includegraphics[scale=0.23]{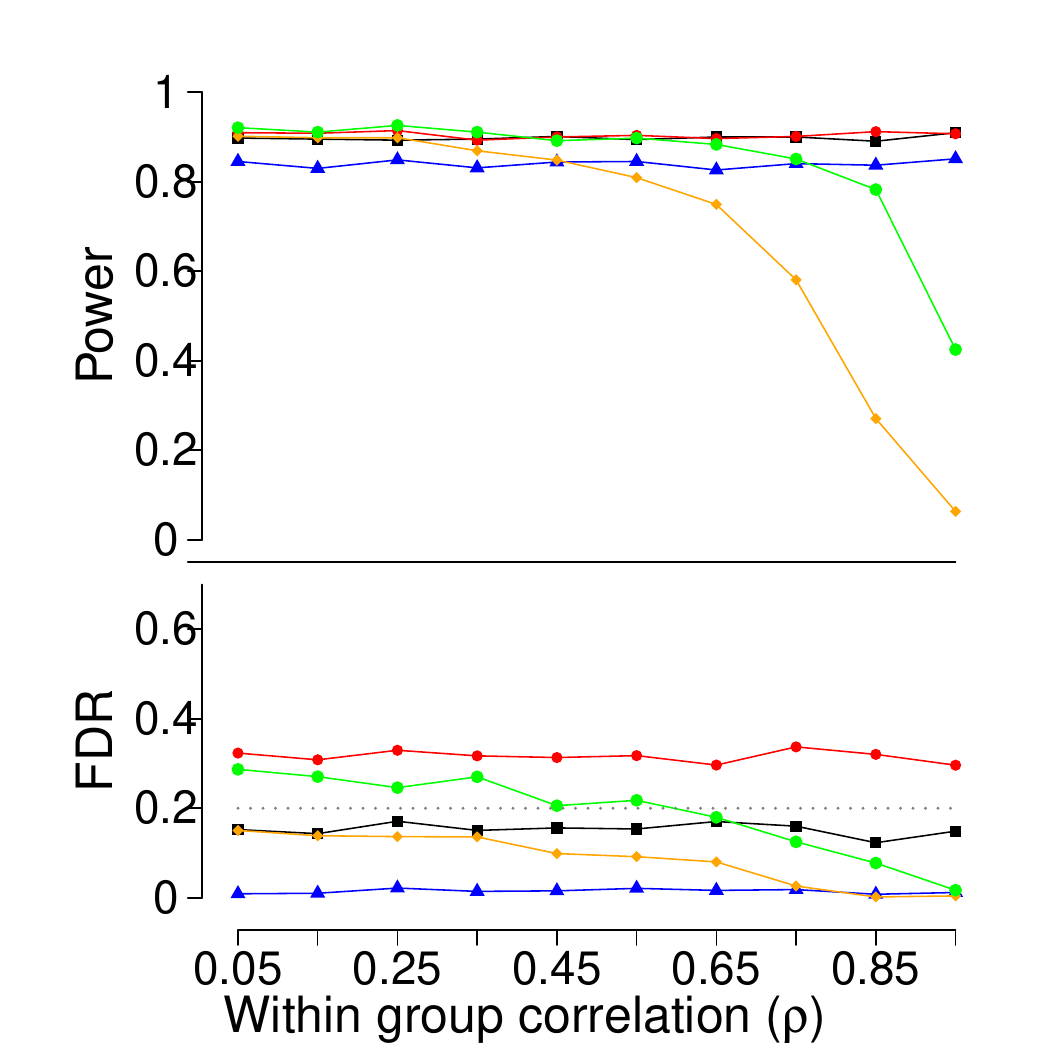}
    \includegraphics[scale=0.23]{images/K=5_Mixed_rho_s4_scale=4_samesig=1.pdf}\\
    \includegraphics[scale=0.23]{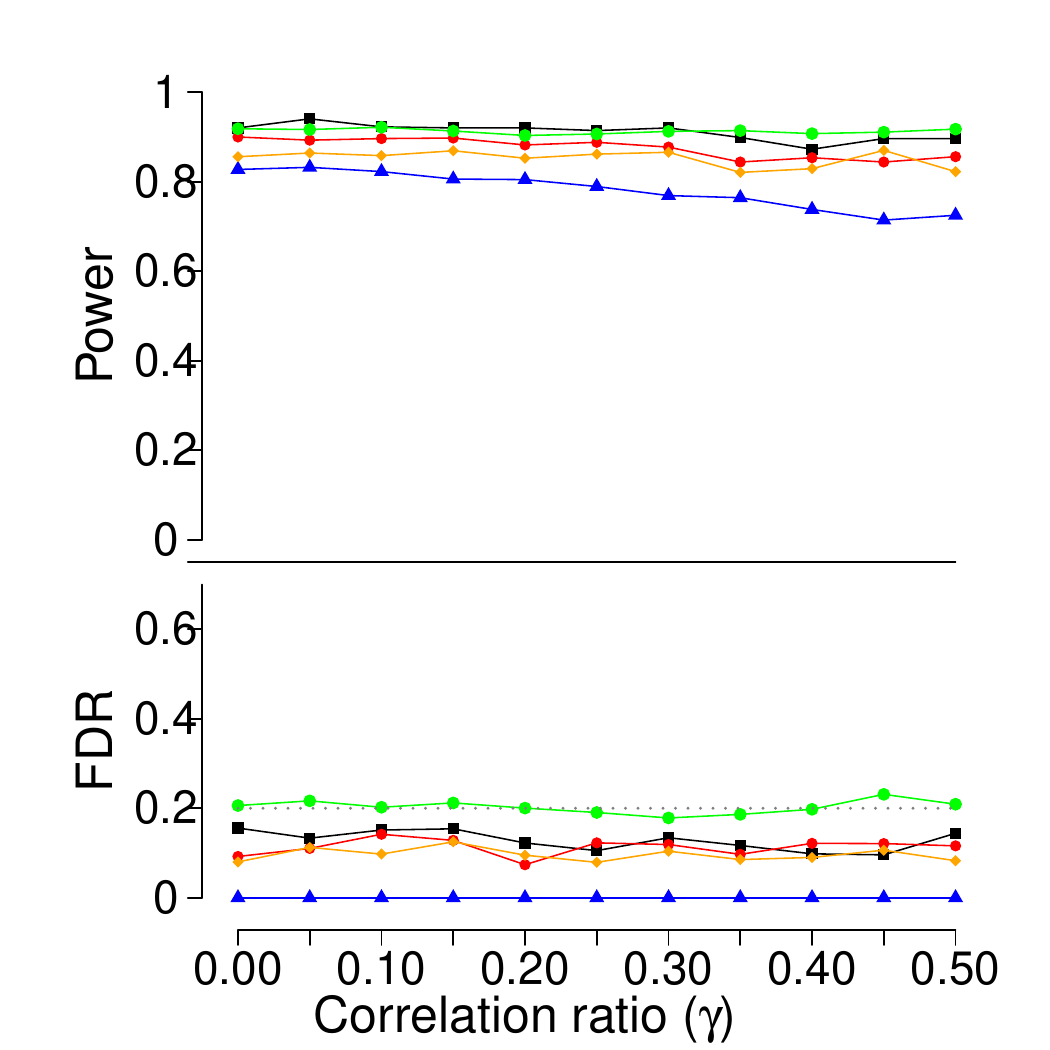}
    \includegraphics[scale=0.23]{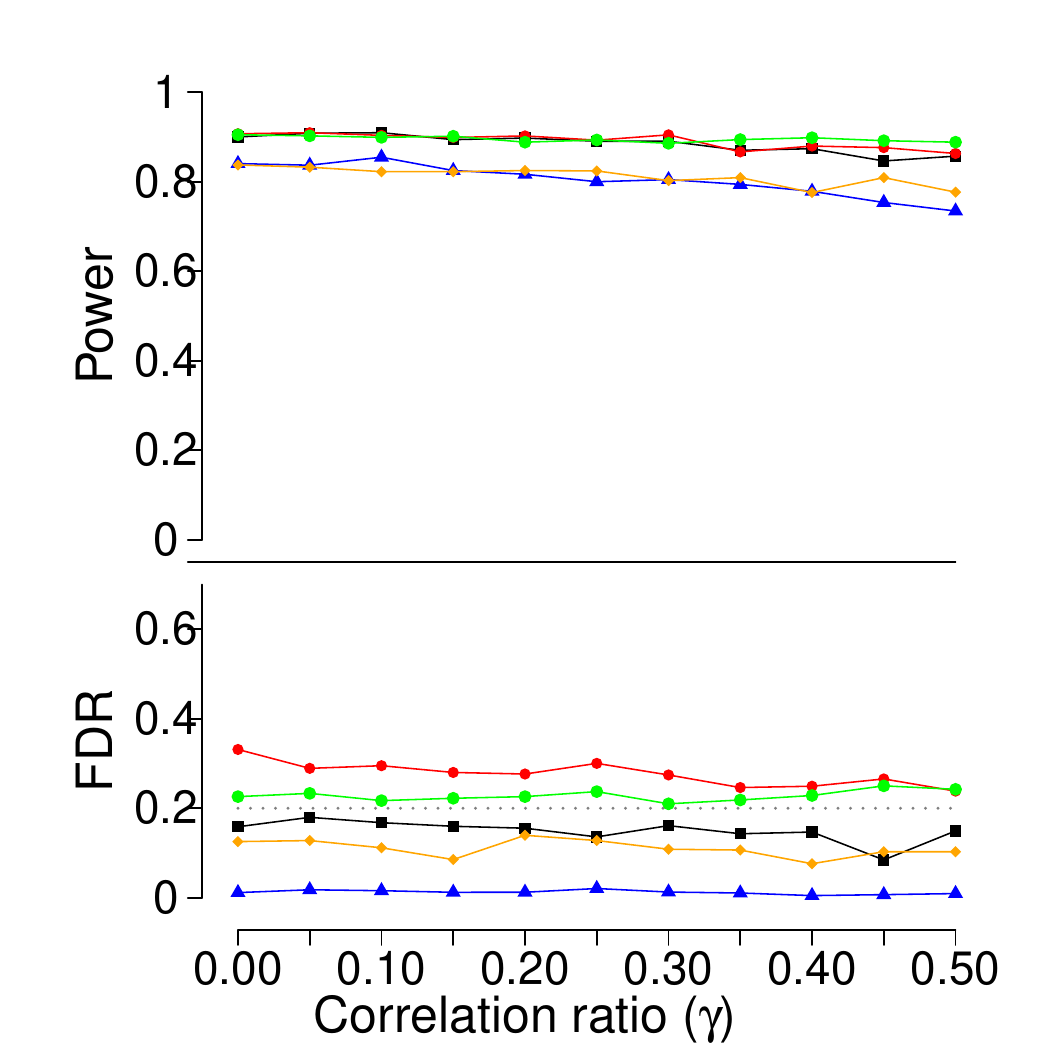}
    \includegraphics[scale=0.23]{images/K=5_Mixed_gamma_s4_scale=4_samesig=1.pdf}
     \caption{{\color{black} The power and the FDR for identifying group level simultaneous signals with data generated from \textbf{Setting 1} for the \textbf{Mixed} models (K=5) on \textbf{Scenario 1}. Left column includes settings with $s_0 \neq 0, s_1=s_2=s_3=s_4=s_5=s_{12}=s_{13}=s_{14}=s_{15}=s_{23}=s_{24}=s_{25}=s_{34}=s_{35}=s_{45}=s_{123}=s_{124}=s_{125}=s_{134}=s_{135}=s_{145}=s_{234}=s_{235}=s_{245}=s_{345}=s_{1234}=s_{1235}=s_{1245}=s_{1345}=s_{2345}=0$; middle column includes settings with $s_0=12, s_{234}=s_{235}=s_{245}=s_{345}\neq 0, s_1=s_2=s_3=s_4=s_5=s_{12}=s_{13}=s_{14}=s_{15}=s_{23}=s_{24}=s_{25}=s_{34}=s_{35}=s_{45}=s_{123}=s_{124}=s_{125}=s_{134}=s_{135}=s_{145}=s_{1234}=s_{1235}=s_{1245}=s_{1345}=s_{2345}=0$; right column includes settings with $s_0 = 12, s_{1234}=s_{1235}=s_{1245}=s_{1345}=s_{2345}\neq 0, s_1=s_2=s_3=s_4=s_5=s_{12}=s_{13}=s_{14}=s_{15}=s_{23}=s_{24}=s_{25}=s_{34}=s_{35}=s_{45}=s_{123}=s_{124}=s_{125}=s_{134}=s_{135}=s_{145}=s_{234}=s_{235}=s_{245}=s_{345}=0.$}}
    \label{fig:figure3-5}
\end{figure}

\begin{figure}
    \centering
    \includegraphics[scale=0.23]{images/K=41_Mixed_s0_scale=4_samesig=1.pdf}
    \includegraphics[scale=0.23]{images/K=41_Mixed_s1_scale=4_samesig=1.pdf}
    \includegraphics[scale=0.23]{images/K=41_Mixed_s3_scale=4_samesig=1.pdf}\\
    \includegraphics[scale=0.23]{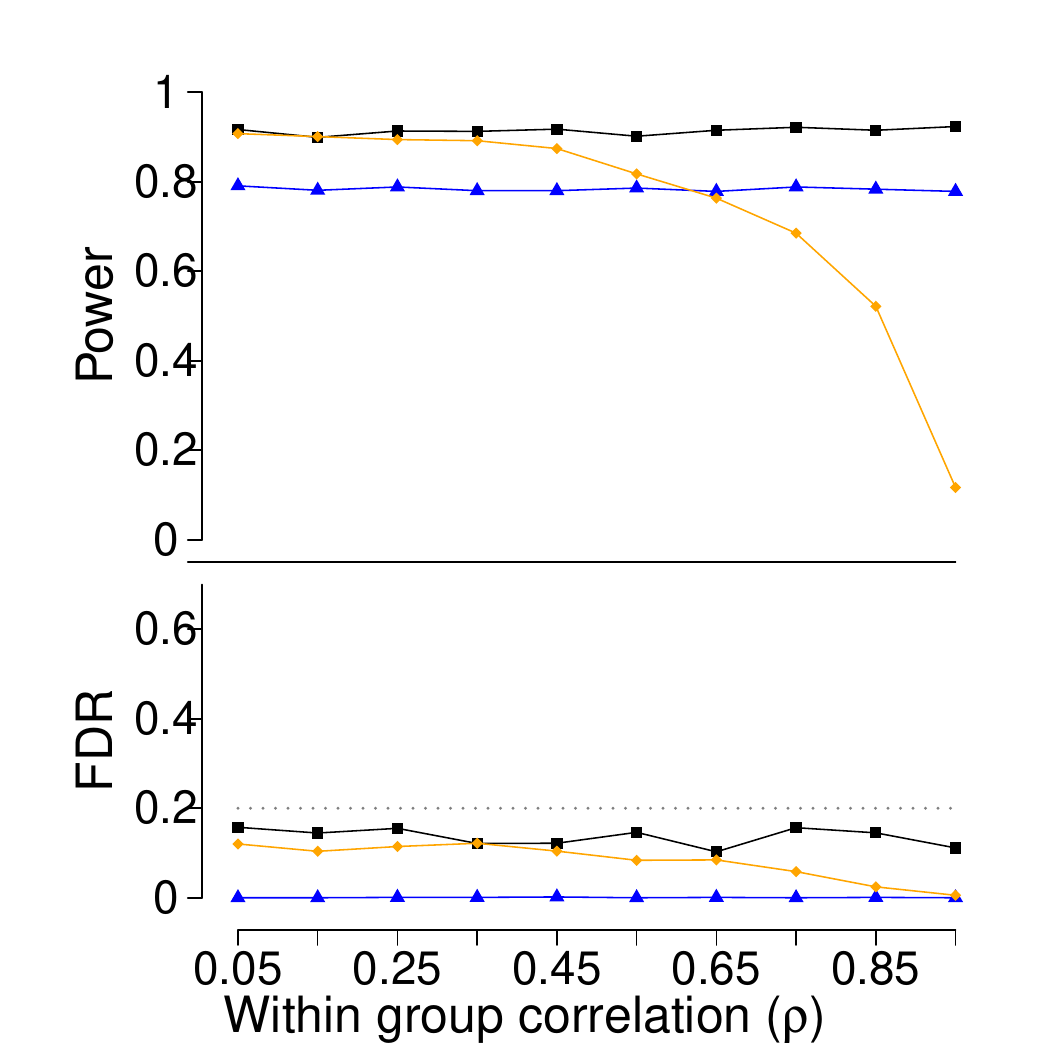}
    \includegraphics[scale=0.23]{images/K=41_Mixed_rho_s1_scale=4_samesig=1.pdf}
    \includegraphics[scale=0.23]{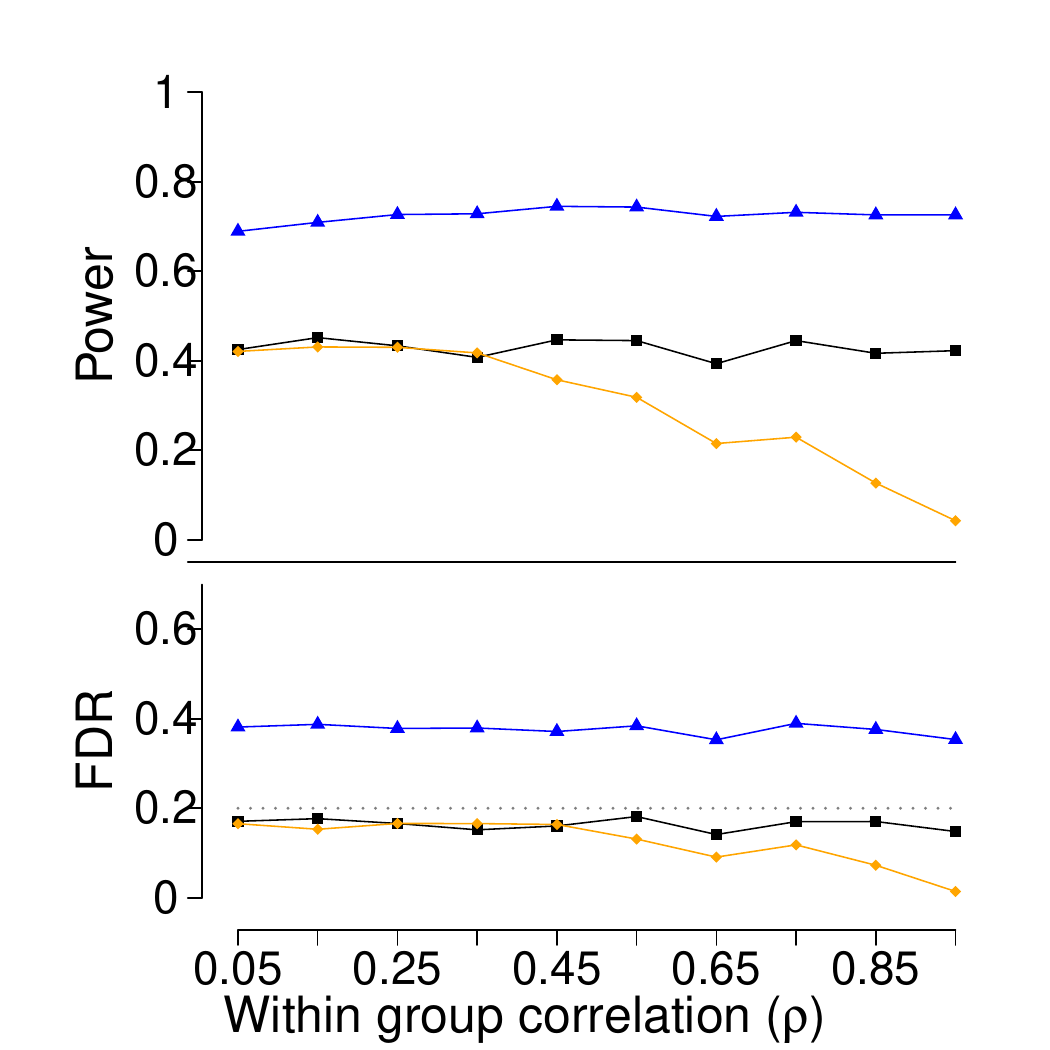}\\
    \includegraphics[scale=0.23]{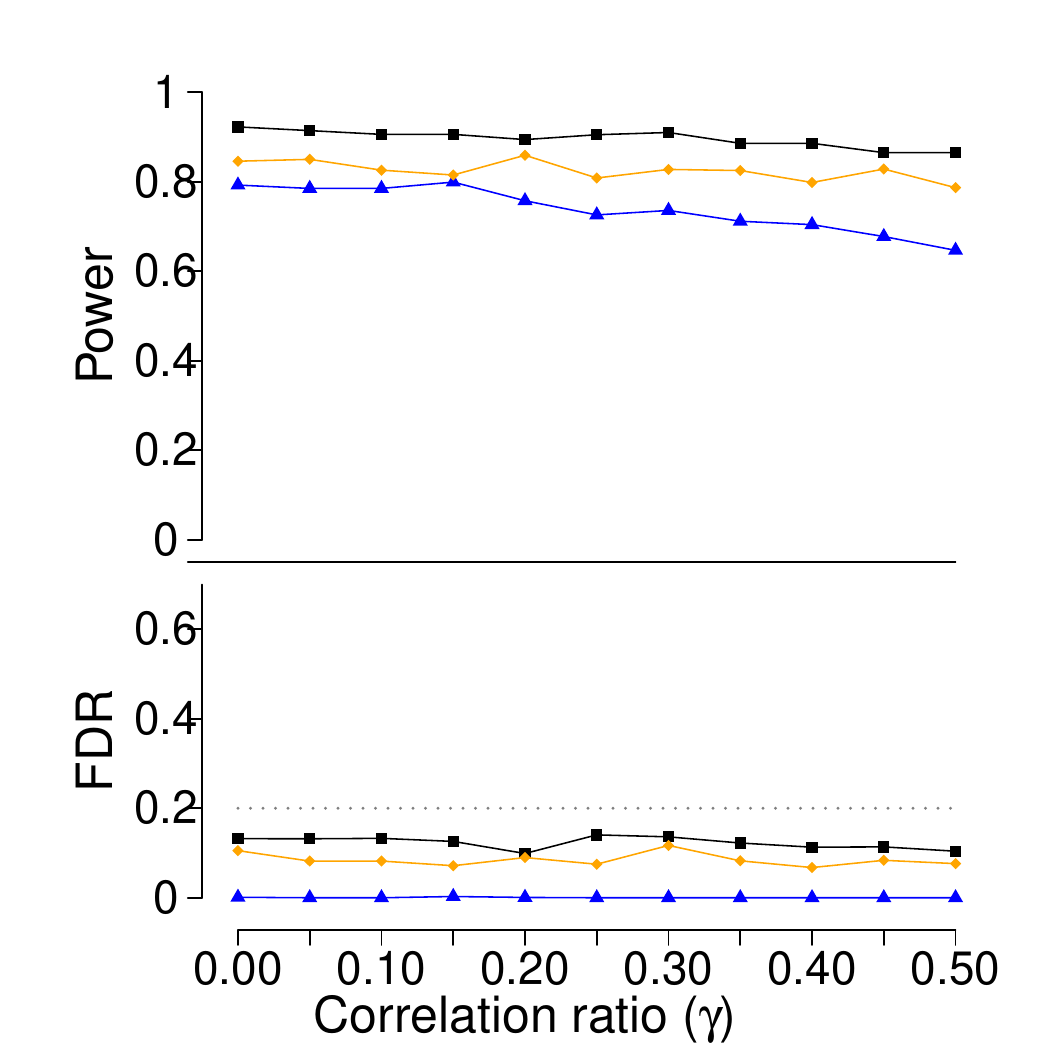}
    \includegraphics[scale=0.23]{images/K=41_Mixed_gamma_s1_scale=4_samesig=1.pdf}
    \includegraphics[scale=0.23]{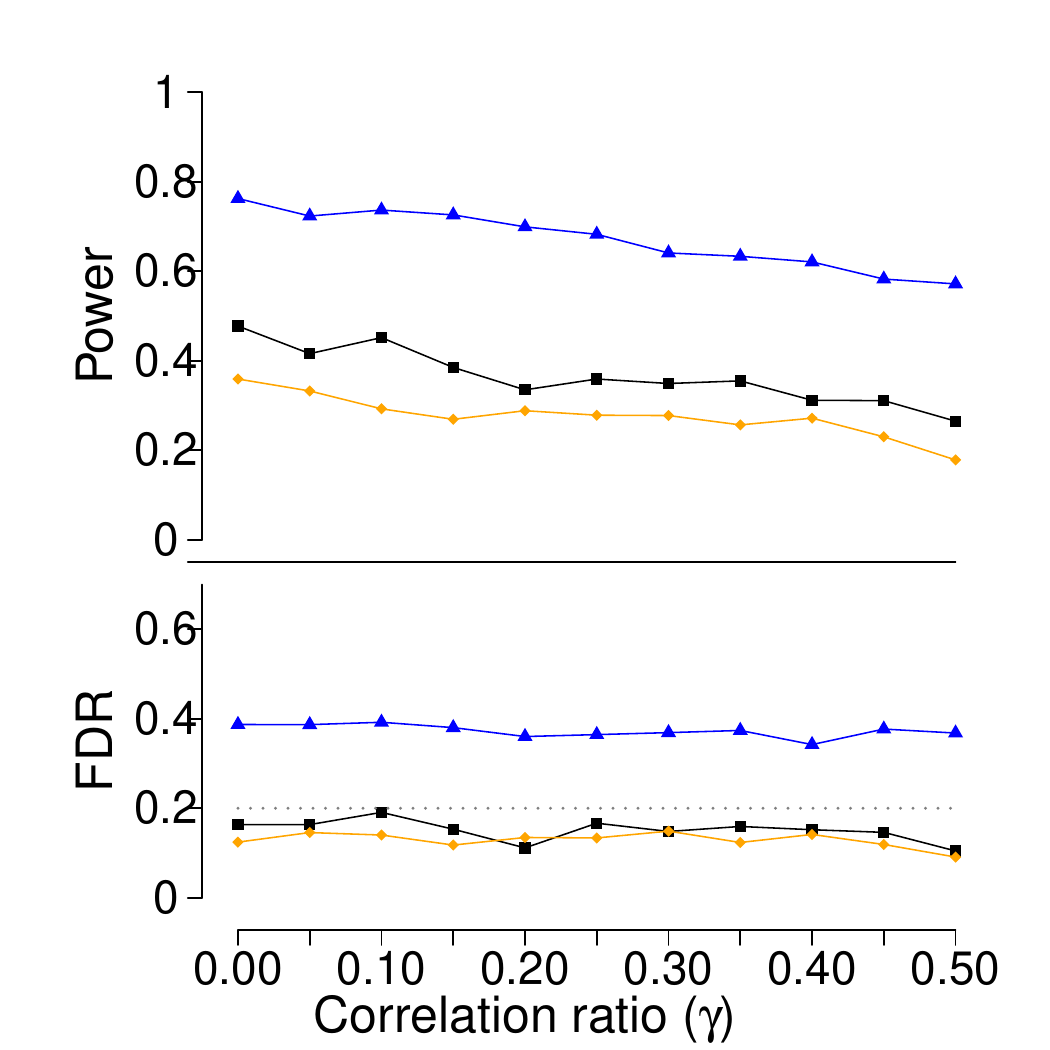}
    \caption{{\color{black} The power and the FDR for identifying group level simultaneous signals with data generated from \textbf{Setting 2} for the \textbf{Mixed} models (K=4) on \textbf{Scenario 1}.} Left column includes settings with $s_0 \neq 0, s_1=s_2=s_3=s_4=s_{12}=s_{13}=s_{14}=s_{23}=s_{24}=s_{34}=s_{123}=s_{124}=s_{134}=s_{234}=0$; middle column includes settings with $s_0=12, s_1=s_2=s_3=s_4 \neq 0, s_{12}=s_{13}=s_{14}=s_{23}=s_{24}=s_{34}=s_{123}=s_{124}=s_{134}=s_{234}=0$; right column includes settings with $s_0=12, s_1=s_2=s_3=s_4=s_{12}=s_{13}=s_{14}=s_{23}=s_{24}=s_{34}=0, s_{123}=s_{124}=s_{134}=s_{234}\neq 0.$}
    \label{fig:figure3-6}
\end{figure}

\begin{figure}
    \centering
    \caption{Cohort construction for N3C Knowledge Store Shared Project.}
    \includegraphics[scale=0.7]{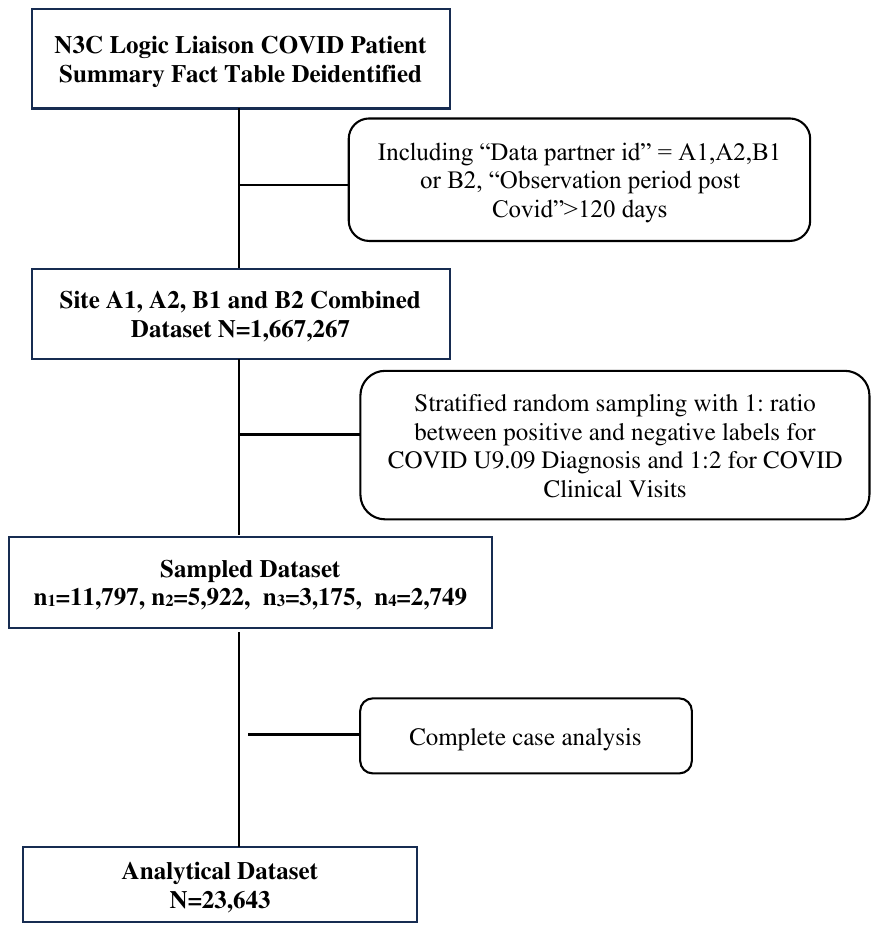}  
    \label{fig:figure6}
\end{figure}